\documentclass[a4paper,twocolumn,accepted=2025-08-20]{quantumarticle}
\pdfoutput=1

\usepackage[utf8]{inputenc}
\usepackage[english]{babel}
\usepackage[numbers,sort&compress]{natbib}

\usepackage{graphicx}
\usepackage{dcolumn}
\usepackage{enumitem}
\usepackage{bbm}
\usepackage{bm}
\usepackage{physics}
\usepackage{subfigure}
\usepackage{mathtools}
\usepackage{xparse}
\usepackage{amssymb,amsmath,amsfonts}
\usepackage{dsfont}
\usepackage[dvipsnames]{xcolor}
\usepackage{hyperref}
\hypersetup{
  colorlinks=true,
  citecolor=blue
}
\usepackage[all]{hypcap}  
\usepackage{soul}
\usepackage{amsthm}
\usepackage{tikz}
\usetikzlibrary{quantikz2}
\usetikzlibrary{calc,arrows.meta,topaths}
\usepackage[normalem]{ulem}

\makeatletter   
\renewcommand\onecolumngrid{
\do@columngrid{one}{\@ne}%
\def\set@footnotewidth{\onecolumngrid}
\def\footnoterule{\kern-6pt\hrule width 1.5in\kern6pt}%
}

\renewcommand\twocolumngrid{
        \def\footnoterule{
        \dimen@\skip\footins\divide\dimen@\thr@@
        \kern-\dimen@\hrule width.5in\kern\dimen@}
        \do@columngrid{mlt}{\tw@}
}%
\makeatother

\newtheorem{theorem}{Theorem}

\newtheorem{lemma}{Lemma}
\newtheorem{proposition}{Proposition}
\newtheorem{assumption}{Assumption}

\DeclarePairedDelimiterX{\set}[1]{\{}{\}}{\setargs{#1}}
\NewDocumentCommand{\setargs}{>{\SplitArgument{1}{;}}m}
{\setargsaux#1}
\NewDocumentCommand{\setargsaux}{mm}
{\IfNoValueTF{#2}{#1} {#1\nonscript\:\delimsize\vert\allowbreak\nonscript\:\mathopen{}#2}}%

\newcommand{\diff}{\mathrm{d}}
\newcommand{\im}{\mathrm{i}}
\newcommand{\e}{\mathrm{e}}
\newcommand{\bs}{\boldsymbol}
\newcommand{\calA}{\mathcal{A}}
\newcommand{\calD}{\mathcal{D}}
\newcommand{\calH}{\mathcal{H}}

\newcommand{\calO}{\mathcal{O}}
\newcommand{\calL}{\mathcal{L}}

\newcommand{\calU}{\mathcal{U}}
\newcommand{\calW}{\mathcal{W}}
\newcommand{\ma}{\mathrm{anc}}
\newcommand{\mE}{\mathrm{E}}
\newcommand{\mL}{\mathrm{L}}
\newcommand{\mR}{\mathrm{R}}
\newcommand{\mix}{\mathrm{mix}}
\renewcommand{\Re}{\mathrm{Re}}

\renewcommand{\Tr}{\mathrm{Tr}}
\renewcommand{\log}{\mathrm{log}}
\DeclareMathOperator*{\E}{\mathbb{E}}
\newcommand{\bE}{\mathbb{E}}
\newcommand{\bZ}{\mathbb{Z}}
\newcommand{\bohrH}{\ensuremath{B_H}}
\DeclareMathOperator{\Spec}{spec}

\DeclareMathOperator{\Var}{Var}
\DeclareMathOperator{\argmin}{argmin}
\newcommand{\lindblad}{\mathcal{L}}

\newcommand*{\expv}[1]{\expval*{#1}}
\newcommand{\SteadyState}{\rho_\infty}

\begin{document}

\title{Lindblad engineering for quantum Gibbs state preparation under the eigenstate thermalization hypothesis}

\author{Eric~Brunner}
\orcid{0000-0001-7631-6528}
\email{eric.brunner@quantinuum.com}
\affiliation{Quantinuum, Partnership House, Carlisle Place, London SW1P 1BX, United Kingdom}

\author{Luuk~Coopmans}
\orcid{0000-0001-6501-5420}
\affiliation{Quantinuum, Partnership House, Carlisle Place, London SW1P 1BX, United Kingdom}

\author{Gabriel~Matos}
\orcid{0000-0002-3373-0128}
\affiliation{Quantinuum, Partnership House, Carlisle Place, London SW1P 1BX, United Kingdom}
\affiliation{Quantinuum, 17 Beaumont St., Oxford OX1 2NA, United Kingdom}

\author{Matthias~Rosenkranz}
\orcid{0000-0002-1605-9141}
\email{matthias.rosenkranz@quantinuum.com}
\affiliation{Quantinuum, Partnership House, Carlisle Place, London SW1P 1BX, United Kingdom}

\author{Frederic~Sauvage}
\orcid{0000-0003-3363-5929}
\affiliation{Quantinuum, Partnership House, Carlisle Place, London SW1P 1BX, United Kingdom}

\author{Yuta~Kikuchi}
\orcid{0000-0002-1802-5260}
\affiliation{Quantinuum K.K., Otemachi Financial City Grand Cube 3F, 1-9-2 Otemachi, Chiyoda-ku, Tokyo, Japan}
\affiliation{Interdisciplinary Theoretical and Mathematical Sciences Program (iTHEMS), RIKEN, Wako, Saitama 351-0198, Japan}

\begin{abstract}
Building upon recent progress in Lindblad engineering for quantum Gibbs state preparation algorithms, we propose a simplified protocol that is shown to be efficient under the eigenstate thermalization hypothesis (ETH). The ETH reduces circuit overheads of the Lindblad simulation algorithm and ensures a fast convergence toward the target Gibbs state. Moreover, we show that the realized Lindblad dynamics exhibits an inherent resilience against stochastic noise, opening up the path to a first demonstration on quantum computers. We complement our claims with numerical studies of the algorithm's convergence in various regimes of the mixed-field Ising model. In line with our predictions, we observe a mixing time scaling polynomially with system size when the ETH is satisfied. In addition, we assess the impact of algorithmic and hardware-induced errors on the algorithm's performance by carrying out quantum circuit simulations of our Lindblad simulation protocol with a local depolarizing noise model. This work bridges the gap between recent theoretical advances in dissipative Gibbs state preparation algorithms and their eventual quantum hardware implementation.
\end{abstract}

\maketitle

\section{Introduction}

The simulation of quantum systems is poised to be one of the most promising applications of quantum computing. This task typically requires the accurate preparation of relevant initial quantum states.
In particular, Gibbs states at an inverse temperature $\beta$ and for a quantum (non-commuting) Hamiltonian $H$,
\begin{align}
\label{eq:gibbs_state}
    \sigma_\beta = \frac{\e^{-\beta H}}{\Tr[\e^{-\beta H}]},
\end{align}
play an essential role in understanding the thermal properties of quantum many-body systems. Furthermore, Gibbs states serve as key resources in a variety of quantum algorithms for optimization problems, machine learning tasks, and more~(see \cite{Dalzell2023, abbasChallengesOpportunitiesQuantum2024} and references therein).
Such demands have led to significant research efforts to develop quantum algorithms for Gibbs state preparation~\cite{Terhal2000, Poulin2009, Temme2009, Bilgin2010, Yung2012, Kastoryano2016, Brandao2016, Chowdhury2016, Motta2019, Gilyen2019, lu_algorithms_2021, Coopmans2022, Holmes2022, Zhang2023,schuckert_probing_2023,ghanem_robust_2023,hemery_measuring_2024}.

Inspired by the successful application of Markov chain Monte Carlo algorithms to a wide range of practical computational problems, including classical Gibbs sampling~\cite{Metropolis1953, Hastings1970}, the authors in~\cite{Temme2009, Yung2012, Moussa2019, Jiang2024} investigated their quantum analogue for sampling from Gibbs states of quantum Hamiltonians.
Inspired by thermalization phenomena occurring in nature~\cite{Lloyd1996}, there has also been a broad interest in utilizing dissipative quantum dynamics, described under the Markov assumption by the Lindblad equation~\cite{Lindblad1975}, to prepare quantum thermal states on a quantum computer~\cite{Terhal2000, Kliesch2011, Childs2016, Cleve2019, Wocjan2021, Shtanko2021, Rall2023, Chen2023thermal, Chen2023efficient, Ding2023, Ding2024, Ding2024open, Gilyen2024, Chen2024random}.
The underlying mechanism is the existence of an (often unique) steady state, which can be made close to a target state through a specific design of the underlying Lindbladian~\cite{Verstraete2009, Kraus2008, Harrington2022}.
Recently, the authors of~\cite{Chen2023thermal, Chen2023efficient} have resolved several outstanding obstacles in the engineering of such specific Lindbladians, which are both efficiently simulatable and have the desired Gibbs state as their unique steady state.

From a complexity theoretical perspective, low-temperature Gibbs states are hard to prepare even with a quantum computer in the worst case~\cite{Bravyi2021}. Still, it is anticipated that quantum algorithms simulating Lindblad dynamics can efficiently prepare certain Gibbs states of interest~\cite{Rouze2024,Bergamaschi2024,Rajakumar2024} that would otherwise be hard to sample from.
The efficiency of these algorithms is assessed in terms of the resources required to implement the Lindblad evolution as a quantum circuit and the mixing time of the dynamics. The latter characterizes the time for any initial state to converge close to the steady state 
of the Lindbladian.

Despite remarkable advances in the design and simulation of adequate Lindbladians, a gap remains between theoretical results in idealized scenarios and their practical application. Such practical considerations include the noise resilience of the proposed algorithms or numerical studies of their convergence.
We address this gap by proposing a version of the Gibbs state preparation algorithms~\cite{Chen2023thermal, Chen2023efficient, Ding2023, Ding2024, Chen2024random} with reduced quantum resource requirements, potentially facilitating a near-term hardware demonstration.
We provide a comprehensive numerical investigation of convergence characteristics of the Lindblad dynamics in various settings, establishing the crucial roles played by both the jump operators and the dynamical properties of the system Hamiltonian.
In addition, we analyze algorithmic errors and the noise susceptibility of our proposal---both analytically and based on simulations of the corresponding quantum circuits.

In Sec.~\ref{sec:preliminaries} we recall necessary background about Lindbladians, their mixing time and steady states, and give a detailed summary of our contributions and their relation to prior works.
In Sec.~\ref{sec:ETH} and \ref{sec:single_ancilla_protocol} we introduce our variant of the dissipative quantum Gibbs state preparation algorithm.
Crucial to our algorithm is the compliance of the system Hamiltonian $H$ and the selected jump operators with the eigenstate thermalization hypothesis (ETH)~\cite{Srednicki1994, Srednicki1995, Srednicki1999, DAlessio2015}.
We show in Sec.~\ref{sec:ETH} that the ETH guarantees fast and accurate convergence of the Lindblad evolution towards the desired Gibbs state.
The circuit implementation of our algorithm, together with a detailed analysis of the algorithmic errors incurred, is given in Sec.~\ref{sec:single_ancilla_protocol}.
To complement our theoretical investigations, in Sec.~\ref{sec:numerics} we perform extensive numerical simulations of the noiseless Lindblad dynamics for a 1D mixed-field Ising model. In
Sec.~\ref{sec:noise}, we analyze the impact of stochastic noise on the convergence accuracy of our protocol, neglecting algorithmic errors such as discretization and Trotter errors. Based on this, in Sec.~\ref{sec:noisy_simulation}, we numerically study the influence of local depolarizing noise by simulating quantum circuits of the full protocol.
In this section, we also study the effect of our algorithm's main parameters on the convergence accuracy and analyze the trade-off between algorithmic and noise-induced errors.

\section{Preliminaries}
\label{sec:preliminaries}

We briefly summarize the concepts of Lindblad dynamics, steady states, detailed balance, and mixing time. These form the basis for the discussions in the subsequent sections.
The Lindblad equation, 
\begin{align}
\label{eq:lindblad_eq}
    \frac{\diff\rho}{\diff t}(t)
    =
    \calL[\rho(t)],
\end{align}
describes the dissipative dynamics of an open quantum system,
which is formally solved as $\rho(t) = \e^{t\calL}[\rho(0)]$ for an initial state $\rho(0)$.
The Lindbladian $\calL$ is the generator of the dynamics and can be written as follows:
\begin{align}
\label{eq:lindbladian_general}
    \calL[\rho] 
    &= 
    - \im[G,\rho]
    + \sum_{a\in\bm{A}}\gamma_a \Big(L^a\rho L^{a\dag} -\frac{1}{2}\{L^{a\dag}L^a,\rho\}\Big).
\end{align}
The first term in Eq.~\eqref{eq:lindbladian_general}, $-\im[G,\rho]$, captures the coherent (unitary) part of the dynamics. $G$ is Hermitian and in the present work we set $G=H$, where $H$ is the system Hamiltonian of the Gibbs state $\sigma_\beta$ in Eq.~\eqref{eq:gibbs_state}.
The remaining terms are responsible for the dissipation along the evolution: the transition term $\sum_{a \in \bm{A}} \gamma_a L^a\rho L^{a\dag}$ and the decay term $-\frac{1}{2}\sum_{a \in \bm{A}}\gamma_a \{L^{a\dag} L^a,\rho\}$. The set of indices $\bm{A}$ selects suitable Lindblad operators $L^a$ that drive the dissipation, and $\gamma_a$ are transition weights.

Quantum detailed balance---a generalization of the classical notion of detailed balance of Markov chains---controls the steady state of the Lindblad dynamics.
A steady state $\sigma$ of $\calL$
is defined by $\e^{t\calL}[\sigma] = \sigma$ for all $t\geq 0$, or, equivalently, $\calL[\sigma] = 0$.
Throughout this work, we assume that the Lindblad dynamics has a full-rank and unique steady state, denoted as $\SteadyState$.
The latter can be ensured by choosing the coherent term and Lindblad operators such that only multiples of the identity commute with them~\cite[Theorem 3]{spohnIrreversibleThermodynamicsQuantum1978}.

We adopt the Kubo-Martin-Schwinger (KMS) inner product to define quantum detailed balance.\footnote{Several other notions exist in the literature~\cite{Alicki1976,Kossakowski1977, Carlen2017}.}
For a full-rank density matrix $\sigma$, the KMS inner product is defined as the weighted scalar product~\cite{Kossakowski1977, Fagnola2007, Fagnola2010, Temme2010},
\begin{align}
\label{eq:DB_inner_product_general}
    \expval{X, Y}_{\sigma} \coloneq \Tr[ X^\dagger \sigma^{1/2} Y \sigma^{1/2} ],
\end{align}
for bounded operators $X, Y$.
We denote by $\calL^\dagger$ the adjoint of $\calL$ with respect to the Hilbert Schmidt inner product, i.e. $\expval*{X, \calL^\dagger[Y]}_\mathrm{HS} = \expval{\calL[X], Y}_\mathrm{HS}$ with $\expval{X, Y}_\mathrm{HS} \coloneq \Tr[ X^\dagger Y ]$.
Then, we say that the Lindbladian $\calL$ obeys the $\sigma$-detailed balance ($\sigma$-DB) condition if $\calL^\dag$ is self-adjoint with respect to the KMS inner product~\eqref{eq:DB_inner_product_general},
\begin{align}
    \label{eq:KMS}
        \expv{ X, \calL^\dag[Y]}_{\sigma}
        =
        \expv{ \calL^\dag[X], Y}_{\sigma}
\end{align}
for all bounded operators $X, Y$.
For a Lindbladian $\calL$ that obeys the $\sigma$-DB condition, we have for all $X$
\begin{align}
\label{eq:self_adjoint_and_steady_state}
\begin{split}
    \langle \calL[\sigma], X\rangle_{\rm HS}
    &= \expv{I, \calL^\dag[X]}_{\sigma}
    = \expv{\calL^\dag[I], X}_{\sigma}
    = 0,
\end{split}
\end{align}
where $I$ is the identity operator.
We used the $\sigma$-DB condition~\eqref{eq:KMS} in the second equality and $\calL^\dag[I]=0$ in the last equality.
When a Lindbladian obeys the $\sigma$-DB condition, the corresponding Lindblad dynamics converges to the steady state $\sigma$ in the infinite time limit~\cite[Proposition 7.5]{wolfQuantumChannelsOperations2012}. 

To prepare the quantum Gibbs state $\sigma_\beta$, Eq.~\eqref{eq:gibbs_state}, we wish to engineer a Lindbladian $\mathcal{L}$ that satisfies the $\sigma_{\beta}$-DB condition. When this is only approximately satisfied, the steady state $\SteadyState$ of $\mathcal{L}$ inevitably deviates from the target $\sigma_\beta$. In such case, it is important to control the \textit{convergence accuracy} $\left\|\SteadyState - \sigma_{\beta} \right\|_1$ of the steady state compared to the Gibbs state, as quantified here via the trace distance.
Furthermore, to use this construction algorithmically, the speed of convergence and the efficient implementability of the Lindblad dynamics are crucial (see Sec.~\ref{sec:ETH}).
We quantify the convergence speed via the mixing time $t_\mathrm{mix}$, i.e. the time it takes to approach the steady state under the Lindblad dynamics from any initial state $\rho$  in trace distance:
\begin{align}
\label{eq:mixing_time}
    t_\mathrm{mix}(\epsilon) 
    \coloneq 
    \inf \left\{t\geq 0 \mid \forall \rho:  \left\|\e^{\calL t}[\rho] - \SteadyState \right\|_1 \le \epsilon \right\}.    
\end{align}

\subsection*{Contribution and relation to prior works}
\label{sec:literature_overview}

Our protocol builds upon recent advances in Lindblad engineering for quantum Gibbs state preparation,
with the aim to reduce circuit complexity and facilitate a potential near-term hardware demonstration.
We follow~\cite{Chen2023thermal}, who propose an efficient quantum algorithm for Gibbs state preparation via simulation of the Lindblad equation~\eqref{eq:lindblad_eq} with a carefully designed (approximately $\sigma_\beta$-DB) $\calL$.
Key to their algorithm is the use of Lindblad operators in the form of filtered operator Fourier transforms~\eqref{eq:L_operators}, which allow to control the trade-off between convergence accuracy and required resources.
In subsequent work, \cite{Chen2023efficient} propose to use $G=G_{\rm CKG}$, Eq.~\eqref{eq:coherent_G_CKG}, as the coherent part of the Lindbladian to satisfy the $\sigma_\beta$-DB condition exactly.
Despite being efficiently implementable, its overheads likely remain prohibitive on near-term quantum hardware.

In our protocol, we simplify the construction of~\cite{Chen2023thermal} by using a discrete set of Lindblad operators (similar to~\cite{Ding2024}), instead of the continuous set used in \cite{Chen2023efficient} (see App.~\ref{app:relation_chen}).
In addition, we avoid the implementation of $G_{\rm CKG}$ by resorting to the ETH (assuming the underlying system Hamiltonian behaves quantum chaotic), which ensures $\sigma_\beta$-DB on average.
The evolution under the Lindbladian~\eqref{eq:lindbladian_general} is simplified by employing the single-ancilla protocol of~\cite{Ding2024open,Pocrnic2023} (similarly used in~\cite{Ding2023} for tasks of ground state preparations), combined with the random selection of a single Lindblad operator~\eqref{eq:L_operators} at each time step.
Such a randomized method is also used in~\cite{Chen2024random} with jump operators constructed from $n$-qubit unitary 2-designs.
In comparison, assuming the ETH allows us to use simple, local Pauli products as jump operators, which can be implemented with shallower circuits than 2-designs in practice.
Finally, \cite{ramkumarMixingTimeQuantum2024} proves a bound on the mixing time similar to ours for random sparse Hamiltonians and unitary 1-design jump operators. While their choice of jump operators is similar to ours, they use the quantum Gibbs sampling algorithm of~\cite{Chen2023efficient} without the simplifications we propose here.
We also point out that the ETH was employed to efficiently prepare the quantum Gibbs state in~\cite{Shtanko2021}, where, instead of evolving the system under dissipative dynamics, they simulate the unitary dynamics of the system and ancilla and trace the latter out only at the end.

We verify the correctness of our protocol and prove that under the ETH the Lindbladian $\calL$, Eq.~\eqref{eq:lindbladian_general}, is $\sigma_\beta$-DB (on average), and derive an upper bound on the convergence accuracy $\left\|\SteadyState - \sigma_{\beta} \right\|_1$ for a concrete realization of $\calL$.
Moreover, we show that, in this case, the mixing time of the Lindblad dynamics is polynomially bounded in the number of qubits $n$.
For this, we tailor the spectral gap analysis of \cite{Chen2021ETH} to our case.
Note that very recently and employing a different technique, \cite{rouzeOptimalQuantumAlgorithm2024} derived a $\log(n)$ mixing time bound for non-commuting local lattice Hamiltonians above a constant critical temperature, improving over bounds based on the spectral gap of the Lindbladian.

Beyond those theoretical results, our main contributions are a comprehensive numerical study of mixing time and convergence accuracy of the Lindblad dynamics, and a thorough analysis of the influence of algorithmic errors and noise on our protocol.

Prior numerical works mainly focus on the spectral gap of the Lindbladian \cite{znidaric_relaxation_2015,Chen2021ETH,smid_polynomial_2025}.
Convergence towards the target Gibbs state via long-time simulations is studied in~\cite{hagan_thermodynamic_2025} for harmonic oscillators and hydrogen chains.
In \cite{Li2024}, considering a similar protocol as ours for the purpose of ground state preparation, the Hartree-Fock method is used to analyze the convergence of energy expectation values.
In general, simulation of dissipative system dynamics is a well-developed field~\cite{plenio_quantum-jump_1998,znidaric_dephasing-induced_2010,finazzi_corner-space_2015,gravina_adaptive_2024,weimer_simulation_2021,sander_large-scale_2025,zhan2025rapid}.
However, the numerical methods developed there often exploit structure of the Lindblad operators, such as sparsity, low-rank approximations or tensor network decompositions, which we do not expect to apply, in general, for Lindblad operators in form of the operator Fourier transforms we consider here
(see App.~\ref{sec:app:improved_classical_simulation_techniques} for a brief discussion of recently proposed tensor network approaches).

In our numerics we focus on the mixed-field Ising model.
We identify quantum chaotic (where ETH and $\sigma_\beta$-DB holds) and non-chaotic regimes of the model based on an eigenstate delocalization analysis~\cite{kolovsky_quantum_2004,atas_multifractality_2012,pausch_chaos_2021,brunner_many-body_2023}.
To efficiently explore mixing time and convergence accuracy across the different dynamical regimes, we develop an adaptive numerical solver that determines for each setting an optimal temporal discretization.
In the chaotic regime and using local jump operators (ETH holds), we numerically observe polynomial (in some cases linear) mixing time scaling.
For non-local jump operators (ETH does not hold), convergence to the Gibbs state is significantly slower.
Interestingly, in the non-chaotic regimes of our model we observe vastly different convergence properties.

To our knowledge, the interplay between algorithmic and noise-induced errors on Gibbs state preparation with the recent Lindblad simulation algorithms has not been studied in detail before. In order to assess noise-induced errors we bound the distance to the target Gibbs state under the influence of a stochastic mixture of unitary noise channels.
In certain regimes the effect of noise can be significantly smaller than expected: noise-induced errors at earlier times are damped through the contractive nature of the Lindblad dynamics showcasing inherent robustness of the protocol to typical sources of hardware noise on near-term quantum computers. Similar effects have been observed, e.g., in~\cite{Raghunandan2020, Polla2021,Mi2023,Cubitt2023,Granet2024,Kashyap2024}. To assess errors from key algorithmic parameters (e.g. Trotter step size) we perform full circuit simulations. We also analyze the trade-offs between algorithmic and noise-induced errors, whereby smaller algorithmic errors tend to increase circuit depth which, in turn, increases noise-induced errors.
Quantifying such trade-offs is crucial for identifying optimal algorithmic parameter configurations for given hardware specifications.

\section{Lindblad engineering with ETH}
\label{sec:ETH}

In this section we introduce our protocol in more detail and discuss basic aspects of the ETH. The ETH will allow us to derive analytical bounds on the accuracy and the speed of convergence towards the target Gibbs state for our Lindbladian~\eqref{eq:lindbladian_general}, which is only approximately detailed balanced.

\subsection{Lindblad operators}

We consider Lindblad operators $\set*{L^a}_{a\in\bs{A}}$ given in the form of a filtered operator Fourier transform (OFT)~\cite{Ding2024},
\begin{align}
\label{eq:L_operators}
    L^a 
    = 
    \int_{-\infty}^{\infty}\diff t\,g(t)A^a(t)
    =
    \sum_{\nu \in \bohrH} \eta_\nu A^a_\nu,
\end{align}
with $A^a(t)\coloneq \e^{\im Ht}A^a\e^{-\im Ht}$, $\eta_\nu\coloneq\int_{-\infty}^{\infty}\diff t\,\e^{\im\nu t}g(t)$ and $A^a_\nu\coloneq\sum_{E_i-E_j=\nu}\Pi_iA^a\Pi_j$, where $\Pi_i$ is the projector onto the eigenstate of $H$ corresponding to the energy $E_i$.  We denote with $\bohrH \coloneq \set*{E_i-E_j\mid E_i,E_j\in \Spec[H]}$ the set of Bohr frequencies.
The set of jump operators $\{A^a\}_{a\in\bs{A}}$ can be chosen arbitrarily as long as it satisfies $\{\sqrt{\gamma_a}A^a\}_{a\in\bs{A}} = \{\sqrt{\gamma_a}A^{a\dag}\}_{a\in\bs{A}}$.
In our case, we will always consider Hermitian jump operators, such that this condition is fulfilled by construction.
In addition, we assume the jump operators to be local, in the sense that they act non-trivially only on a small part of the system. This is a requirements of the ETH, discussed in more detail in Sec.~\ref{sec:convergence_under_ETH}.
The concrete jump operator model used in our numerics is described in Sec.~\ref{sec:numerics:time_evolution}.
Furthermore, to guarantee uniqueness of the steady state we require that only multiples of the identity commute with the jump operators and Hamiltonian~\cite[Theorem 3]{spohnIrreversibleThermodynamicsQuantum1978}.
We choose a filter function
\begin{equation}
\label{eq:g_gaussian}
    g(t) 
    = 
    \left(\frac{\Delta_\mE^2}{2\pi^3}\right)^{1/4}\e^{-\Delta_\mE^2 t^2+\im\beta\Delta_\mE^2 t/2},
\end{equation}
which satisfies the normalization condition $\int_{-\infty}^{\infty}\diff t|g(t)|^2=1$. This choice selects energy transitions in Eq.~\eqref{eq:L_operators} that decrease the energy by roughly $\beta\Delta_\mE^2/2$.
For our analytical and numerical studies, we choose $\Delta_\mE=\sqrt{2}/\beta$, as discussed in App.~\ref{app:ETH_scales}. This choice ensures that the energy transitions induced by the jump operators are well within the energy window defined by the Fourier transform of $g$.
Noting $\eta_\nu=\e^{-\beta\nu/2}\eta_{-\nu}$, Eq.~\eqref{eq:eta_lindblad_filter}, and using $\{\sqrt{\gamma_a}A^a\}_{a\in{\bf A}}=\{\sqrt{\gamma_a}A^{a\dag}\}_{a\in{\bf A}}$, one can readily show that $\sum_{a}\gamma_a\sigma_\beta^{-1/2}L^{a}\sigma_\beta^{1/2} [\cdot]\sigma_\beta^{1/2}L^{a\dag}\sigma_\beta^{-1/2} = \sum_{a}\gamma_aL^{a\dag}[\cdot]L^{a}$.
It then follows that the transition term in Eq.~\eqref{eq:lindbladian_general} satisfies the $\sigma_\beta$-DB condition~\eqref{eq:KMS} as
\begin{align}
\label{eq:DB_transition}
\begin{split}
    &\sum_{a\in\bm{A}}\gamma_a \langle X, L^{a\dag} Y L^{a}\rangle_{\sigma_\beta}
    \\
    &=
    \sum_{a\in\bm{A}}\gamma_a \langle \sigma_\beta^{-1/2}L^{a}\sigma_\beta^{1/2} X\sigma_\beta^{1/2}L^{a\dag}\sigma_\beta^{-1/2}, Y\rangle_{\sigma_\beta}
    \\
    &=
    \sum_{a\in\bm{A}}\gamma_a \langle L^{a\dag} XL^{a}, Y\rangle_{\sigma_\beta}.
\end{split}
\end{align}
It remains to assess under which conditions the decay term in Eq.~\eqref{eq:lindbladian_general} obeys the $\sigma_\beta$-DB condition.

\subsection{Convergence under the ETH}
\label{sec:convergence_under_ETH}

The ETH~\cite{Srednicki1999,DAlessio2015} states that, for a given Hamiltonian $H$ with eigenbasis $\{\ket{E_i}\}$, the matrix elements of a local observable $A$ are expressed as
\begin{align}
\label{eq:ETH}
    \bra{E_i}A\ket{E_j} 
    =
    \calA(E_{i})\delta_{ij} + \frac{f(E_{ij},\nu_{ij})}{\sqrt{D(E_{ij})}}R_{ij},
\end{align}
with $E_{ij} \coloneq (E_i+E_j)/2$ and $\nu_{ij} \coloneq
E_i-E_j$. 
$\calA(E)$ and $f(E,\nu)$ are smooth functions of $E$ and $\nu$.
The density of states $D(E)$ is defined by $D(E)=\sum_{E_i\in{\rm spec}[H]}\tilde{\delta}(E-E_i)$, where $\tilde{\delta}(E-E_i)$ is a smeared delta function so that $D(E)$ becomes a smooth function of $E$.
$R$ is a Hermitian matrix with entries $R_{ij}$ whose real and imaginary
parts are independent random variables and which satisfy $\bE_R[R_{ij}]=0$ and $\bE_R[|R_{ij}|^2]=1$, where $\bE_R$ denotes the average over $R$. We assume that the diagonal vanishes, $R_{ii}=0$, for all $i$. See App.~\ref{app:ETH} for details.

In the following, we assume that for given jump operators the Lindbladian $\calL$, Eq.~\eqref{eq:lindbladian_general}, is well-approximated by a random realization of $\calL$ with jump operators modeled according to Eq.~\eqref{eq:ETH}.
The strategy is to show that the ETH-averaged Lindbladian $\bE_R\calL$ respects the $\sigma_\beta$-DB condition and that the distance between the average $\bE_R\calL$ and a single realization of $\calL$ is bounded by a sufficiently small quantity with high probability. Via this bound we obtain upper bounds on the mixing time and the convergence accuracy for our Lindbladian $\calL$.

\paragraph{Detailed balance ---}
The decay term in Eq.~\eqref{eq:lindbladian_general} obeys the $\sigma_\beta$-DB condition under the ETH average.
To see this, we convert the decay term to,
\begin{align}
\label{eq:DB_decay1}
    &\sum_{a\in\bm{A}}\gamma_a \langle X, \{L^{a\dag}L^a, Y\}\rangle_{\sigma_\beta}
    \nonumber\\
    &=
    \sum_{a\in\bm{A}}\gamma_a \langle \sigma_\beta^{-1/2}L^{a\dag}L^a\sigma_\beta^{1/2}X
    +X\sigma_\beta^{1/2}L^{a\dag}L^a\sigma_\beta^{-1/2}, Y\rangle_{\sigma_\beta}
\end{align}
for all bounded operators $X$, $Y$.
According to the ETH, the matrix elements of the jump operators $A^a$ are given by Eq.~\eqref{eq:ETH}, which leads to,
\begin{align}
\label{eq:DB_decay2}
\begin{split}
    &\bE_R[\sigma_\beta^{-1/2}L^{a\dag}L^a\sigma_\beta^{1/2}]
    \\
    &=
    \sum_{\nu,\nu'}
    \eta_\nu\eta_{\nu'}\e^{\beta\frac{\nu'-\nu}{2}}\bE_R[(A^a_{\nu})^{\dag} A^a_{\nu'}]
    \overset{\text{ETH}}{=}
    \bE_R[L^{a\dag}L^a],
\end{split}
\end{align}
where the last equality follows from $\bE_R[(A^a_{\nu})^{\dag} A^a_{\nu'}]\propto \delta_{\nu,\nu'}$.
Combining Eqs.~\eqref{eq:DB_decay1} and~\eqref{eq:DB_decay2}, we find 
\begin{align}
\label{eq:DB_decay}
\begin{split}
    &\bE_R\Big[
        \sum_{a\in\bm{A}}\gamma_a \langle X, \{L^{a\dag}L^a, Y\}\rangle_{\sigma_\beta}
    \Big]
    \\
    &=
    \bE_R\Big[
        \sum_{a\in\bm{A}}\gamma_a \langle\{L^{a\dag}L^a, X\}, Y\rangle_{\sigma_\beta}
    \Big].
\end{split}
\end{align}
From Eqs.~\eqref{eq:DB_transition} and~\eqref{eq:DB_decay} we conclude that the Gibbs state $\sigma_\beta$ is the steady state of the averaged dissipative Lindbladian, i.e. neglecting the coherent term $-\im [G, \rho]$ in Eq.~\eqref{eq:lindbladian_general}---which in our case is generated by the system Hamiltonian, $G = H$.
Note that $H$ commutes with $\sigma_\beta$, which implies that the coherent term $-\im [H, \rho]$ does not affect the steady state.
Thus, we find that the Gibbs state $\sigma_\beta$ is the steady state of the averaged Lindbladian $\bE_R\calL$,
\begin{align}
\label{eq:fix_average_D}
    \bE_R\calL[\sigma_\beta] = 0.
\end{align}

\paragraph{Mixing time ---}

For a Lindbladian $\calL$ with steady state $\SteadyState$ and gap $\Delta_\calL$, the mixing time is bounded via (see e.g. \cite[Eq.~(104)]{Kastoryano2013})
\begin{align}
\label{eq:mixing_time_and_spectral_gap}
t_\mathrm{mix}(\epsilon)
\le 
\frac{1}{\Delta_\calL}\log\bigg(\frac{2\|\SteadyState^{-1/2}\|_\infty}{\epsilon}\bigg).
\end{align}
The spectral norm $\|\SteadyState^{-1/2}\|_\infty$ typically scales exponentially in the number of qubits $n$.
Thus, we obtain a polynomial dependence of the mixing time on the number of qubits $n$ if the spectral gap of $\calL$ is lower-bounded by $1/\mathsf{poly}(n)$.

In App.~\ref{app:ssec:gap_EL}, we show that $\bE_R\calL$ effectively reduces to a classical Markov chain on the spectrum of $H$~\cite{Chen2021ETH}.
This allows us to employ the well-developed framework of Markov chain conductance to derive an inverse-polynomial lower bound
$\Delta_{\bE_R\calL} \ge \Omega(\beta^3/n)$.
According to Eq.~\eqref{eq:mixing_time_and_spectral_gap}, this implies a bound on the mixing time of ${\bE_R\cal L}$ polynomial in $n$ and $\beta$. To obtain a similar bound for $\calL$, we show in App.~\ref{app:ssec:concentration_ETH_average} that the distance between the Lindbladians $\calL$ and $\bE_R\calL$ is bounded by $\mathcal{O}(\beta/\sqrt{|\bm{A}|})$.

As proven in App.~\ref{app:ssec:spectral_gap_and_mixing_time},
it follows that 
\begin{align}
\label{eq:mixing_time_bound}
    t_\mathrm{mix}(\epsilon) 
    \le
    \calO \left(
    n\beta^2(\beta\|H\|_\infty + \log(1/\epsilon))
    \right),
\end{align}
with high probability,  provided a sufficiently large number of jump operators, $|\bm{A}|=\Omega(n^2\beta^8)$.
See Thm.~\ref{thm:convergence} for a precise statement of the result and the assumptions.
Typically considered Hamiltonians have polynomially bounded norm $\|H\|_\infty$, for example $\|H\|_\infty = \calO(n^k)$ for $k$-local Hamiltonians, or $\|H\|_\infty = \calO(n)$ for geometrically local Hamiltonians.

\paragraph{Convergence accuracy of the steady state ---}
We again invoke the result on the distance between $\calL$ and $\bE_R\calL$.
As proven in App.~\ref{app:ssec:spectral_gap_and_mixing_time},
see Thm.~\ref{thm:convergence},
the trace distance between the Gibbs state $\sigma_\beta$ and the steady state of $\calL$, $\SteadyState$, is bounded by
\begin{align}
\label{eq:upper_bound_trace_distance}
    \|\SteadyState - \sigma_\beta\|_{1}
    \le
    \calO \left( \epsilon + \frac{n\beta^2}{\sqrt{|\bm{A}|}} \left( 
    \beta \|H\|_\infty + \log(1/\epsilon)
    \right) \right),
\end{align}
with high probability.
Thus, a sufficiently large number of jump operators $|\bm{A}| = \Omega\left( n^2\beta^6\|H\|_\infty^2 / \epsilon \right)$ leads to a steady state that is $\epsilon$-close to the Gibbs state, i.e., $\|\SteadyState - \sigma_\beta\|_{1} \le \epsilon$.
In the next section we will see that $|\bm{A}|$ does not enter the complexity scaling of our quantum protocol for the Lindblad simulation.
Therefore, in principle, we can choose $|\bm{A}|$ as large as allowed by the considered jump operator model.

\section{Randomized single-ancilla protocol}
\label{sec:single_ancilla_protocol}

So far we have studied the convergence properties of the Lindblad dynamics. In this section we describe how this evolution can be implemented as a quantum circuit.
In Sec.~\ref{sec:noisy_simulation} we perform full circuit simulations and investigate the algorithmic errors of our protocol.
We adopt a single-ancilla protocol as in~\cite{Ding2023, Ding2024open} and combine it with a randomized scheme~\cite{Campbell2019,Chen2024random}
to simulate the dynamics under the Lindbladian \eqref{eq:lindbladian_general} with $G=H$ and $\gamma_a=\gamma p_a$.
The parameter $\gamma\ge0$ controls the strength of the dissipation and ensures that the probabilities $p_a$ satisfy $0\le p_a\le1$ together with $\sum_ap_a=1$.
Instead of applying the full Lindbladian~\eqref{eq:lindbladian_general}, at each evolution time step we apply a single Lindblad operator $L^a$ sampled with probability $p_a$. We further factorize the coherent and the dissipative parts of the time evolution.
Hence, starting from an initial state $\rho(0)$, the state prepared after $M$ evolution steps of length $\delta t$ each, is given by
\begin{align}
\label{eq:lindblad_random}
    \left(\prod_{i=1}^{M}
    \e^{\delta t\gamma\calD^{a_i}}
    \circ \calU_{\delta t}\right)[\rho(0)],
\end{align}
where we have defined
\begin{align}
\label{eq:random_dissipator}
    &\calD^{a}[\rho]
    \coloneq
    L^a\rho L^{a\dag} - \frac{1}{2}\{L^{a\dag} L^a, \rho\},
    \\
\label{eq:random_dissipator_coh}
    &\calU_{\delta t}[\rho] 
    \coloneq
    \e^{-\im \delta t H}\rho\e^{\im \delta  t H},
\end{align}
and $\{a_1,\dots,a_M\}$ is the set of labels of the randomly sampled Lindblad operators.
The total evolution time is $t = M \delta t$.
By taking the average over the random sampling of Lindblad operators we find,
\begin{align}
\label{eq:random_average}
\begin{split}
    \sum_a p_a\e^{\delta t\gamma\calD^{a}} \circ \calU_{\delta t}[\rho]
    &= \rho + \delta t\calL[\rho] + \calO(\delta t^2)
    \\
    &= \e^{\delta t\calL}[\rho]  + \calO(\delta t^2).
\end{split}
\end{align}
It has further been shown that
the individual trajectory defined by Eq.~\eqref{eq:lindblad_random} approximates the target dynamics with an error measured in the trace distance inversely proportional to $M$ (Theorem~7 in \cite{Chen2024random})\footnote{Results of Ref.~\cite{Chen2024random} are given in terms of a weighted $l^2$-metric that upper bounds the trace distance, see Lemma~5 of Ref.~\cite{Temme2010}.}.
Hence,
for large $M$, a single trajectory describes the average evolution in Eq.~\eqref{eq:random_average} sufficiently well.

To implement each dissipative Lindblad evolution $\e^{\delta t\gamma\calD^a}[\rho]$, we note that, for the dilation~\cite{Cleve2019}
\begin{align}
\label{eq:dilation}
    K^a
    \coloneq
    \ket{1}\bra{0}_\ma\otimes L^a + \ket{0}\bra{1}_\ma\otimes L^{a\dag},
\end{align} 
the following identity holds,
\begin{align}
\label{eq:evolve_dilation}
\begin{split}
    &\Tr_\mathrm{anc}[\e^{-\im \sqrt{\delta t \gamma} K^a}(\ket{0}\bra{0}_\ma\otimes \rho)\e^{\im \sqrt{\delta t \gamma} K^a}]
    \\
    &=
    \e^{\delta t\gamma\calD^a}[\rho] + \calO((\delta t\gamma)^2).
\end{split}
\end{align}
Implementation of the left-hand side of the identity only requires introducing a single ancilla qubit and Hamiltonian simulation.

To retain the algorithmic error of $\calO(\delta t^2)$ in Eq.~\eqref{eq:random_average}, we apply the second-order product formula to implement the evolution under the unitary $\e^{-\im \sqrt{\delta t \gamma} K^a}$~\cite{Ding2023}.
To do so, we first discretize the OFT~\eqref{eq:L_operators} over a restricted domain $[-T,T]$. Taking discretized time steps $\Delta t\coloneq T/S$, we get
\begin{align}
\label{eq:L_operator_discretized}
    \bar{L}^a
    \coloneq
    \sum_{s=-S}^{S} \Delta t_s g(s\Delta t) A^a(s\Delta t),
\end{align} 
with $\Delta t_s\coloneq\Delta t$ for $-S+1\le s\le S-1$ and $\Delta t_s\coloneq\Delta t/2$ for $s=\pm S$.
We call $\Delta t$ the \textit{OFT discretization step}.
Accordingly, the dilation $K^a$, Eq.~\eqref{eq:dilation}, is discretized as
$\bar{K}^a\coloneq \ket{1}\bra{0}_\ma\otimes \bar{L}^a + \ket{0}\bra{1}_\ma\otimes \bar{L}^{a\dag}$.

\begin{figure}
\begin{quantikz}[row sep =0.2cm, column sep=0.4cm]
    & \wireoverride{n}
    & \lstick{$\ket{0}_\ma$} \wireoverride{n} \gategroup[5, steps=3, style={inner xsep=4.3mm,dashed}]{Repeat for $i=1,\dots, M$}
    & \gate[5]{V^{a_i}(\delta t)}
    & \meter{}
    \\
    \lstick[4]{$\rho(0)$}
    &
    & \gate[4]{\e^{-\im H \delta t}}
    &
    &
    &
    & \rstick[4]{$\approx \rho(M \delta t)$}
    \\
    &
    &
    &
    &
    &
    &
    \\
    & \vdots \wireoverride{n}
    & \wireoverride{n}
    & \wireoverride{n}
    & \vdots \wireoverride{n}
    & \wireoverride{n}
    & \wireoverride{n}
    \\
    &
    &
    &
    &
    &
    &
\end{quantikz}
\caption{\label{fig:single_ancilla_circuit}
    The quantum circuit for simulating Lindblad evolution~\eqref{eq:lindblad_random}. The ancilla qubit on the top is traced out by discarding the measurement results.
}
\end{figure}
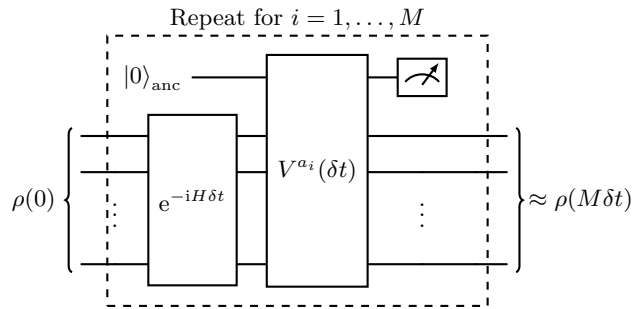

Now, one can implement $\e^{-\im\sqrt{\delta t \gamma}\bar{K}^a}$ by applying the second-order product formula, which we denote by $V^{a}(\delta t)$, such that
\begin{align}
\label{eq:circuit_V_box}
    V^{a}(\delta t)
    =
    \e^{-\im\sqrt{\delta t \gamma}\bar{K}^{a}} 
    +\calO(\delta t^2\gamma^2).
\end{align}
The same ancilla qubit can be reused after resetting it to $\ket{0}$.
The resulting quantum circuit implementing the Lindblad dynamics is sketched in Fig.~\ref{fig:single_ancilla_circuit}.

As apparent from the previous discussion, several approximations are required for the circuit implementation.
In App.~\ref{app:single_gibbs}, we provide a comprehensive study of the resulting errors. 
When accounting for all these, the preparation of the steady state of $\mathcal{L}$ incurs the following algorithmic error:
\begin{align}
\label{eq:error_sources}
\begin{split}
    &t_{\rm mix}\times\calO\left(
        \delta t
        + \frac{\Delta t}{\sqrt{\beta}}\e^{-2(T/\beta)^2}
    \right.
    \\
    &\quad\left.   
        + \sqrt{\beta}|B_H| \e^{-\frac{1}{8}\big(2\pi\frac{\beta}{\Delta t} - 2\beta\|H\|_\infty-1\big)^2}
        + \frac{T\Delta t^2}{\delta t}
    \right).
\end{split}
\end{align}
The precise statement is given in Prop.~\ref{prop:algorithmic_error} in App.~\ref{app:single_gibbs}.
The contribution of the individual error sources are discussed further in App.~\ref{app:alg_errors}, and analyzed in detail in Sec.~\ref{sec:noisy_simulation}, by means of simulations of the corresponding quantum circuits.
For the algorithmic  error~\eqref{eq:error_sources} to be less than~$\epsilon$, the algorithm uses Hamiltonian simulation (i.e. evolution under the system Hamiltonian $H$) for a time scaling as
\begin{align}
\label{eq:runtime_main}
    \Theta\left(\frac{\beta t_\mix^2}{\epsilon}\sqrt{\log\frac{\beta t_{\rm mix}}{\epsilon}}\right).
\end{align}
See Thm.~\ref{thm:qa_complexity} for the precise statement.
Combined with Eq.~\eqref{eq:mixing_time_bound} and also accounting for the convergence accuracy of Eq.~\eqref{eq:upper_bound_trace_distance}, this shows that under the ETH, preparation of the Gibbs state, with $\epsilon$-error, can be implemented efficiently as a circuit with a single additional ancilla.
We note that the currently known best Lindblad simulation algorithm achieves the runtime~\cite{Cleve2019,Chen2023efficient}
\begin{align}
\label{eq:optimal_scaling}
    \tilde{\calO}\left(\beta t_\mix\,{\rm polylog}(1/\epsilon)\right),
\end{align}
where $\tilde{\calO}$ hides the polylogarithmic dependence on $\beta$ and $t_{\rm mix}$. However, this requires an additional overhead circuit for a complex block encoding of the Lindbladian.
The protocol presented in this work is based on the Trotterization and dilation (weak measurement), which
avoid such overhead circuit at the cost of
sacrificing the near-optimal scaling in Eq.~\eqref{eq:optimal_scaling}.

\section{Noiseless numerical study}
\label{sec:numerics}

This section studies the dynamical properties of our randomized Lindbladian protocol. 
To this end, we numerically estimate mixing time, spectral gap and convergence accuracy for the Lindbladian~\eqref{eq:lindbladian_general}, taking into account the random application of a single jump operator at each time step, as in the quantum algorithm in Sec.~\ref{sec:single_ancilla_protocol}.
Here we ignore the Trotterization~\eqref{eq:lindblad_random}, OFT discretization~\eqref{eq:L_operator_discretized}, as well as the influence of noise as they will be the focus of Secs.~\ref{sec:noise}--\ref{sec:noisy_simulation}.

We consider the 1D mixed-field Ising model with open boundary conditions
\begin{equation}
\label{eq:ising_model}
H=-J\sum_{i=0}^{n-2}Z_iZ_{i+1} - h \sum_{i=0}^{n-1}X_i - m\sum_{i=0}^{n-1}Z_i.
\end{equation}
The parameters $h$ and $m$ control the strength of the transverse and longitudinal fields, respectively.
We set the inverse temperature to $\beta = (2J)^{-1}$
throughout our numerical studies.

We analyze various settings defined via: (i) the number of jump operators, (ii) the locality of the jump operators, and (iii) the degree to which the system satisfies the ETH, parametrized by $h$ and $m$
(keeping $J$ fixed), as detailed in Secs.~\ref{sec:numerics:time_evolution}--\ref{sec:numerics:convergence_accuracy}.
We observe a polynomial scaling of the mixing time with $n$ when the ETH holds, confirming our analytical bound in Eq.~\eqref{eq:mixing_time_bound}. 
Moreover, we observe a polynomially decreasing distance between the steady state of the Lindblad dynamics and the target Gibbs state with increasing number of jump operators $|\bm{A}|$ as indicated by the bound in Eq.~\eqref{eq:upper_bound_trace_distance}.
In settings where the ETH does not hold, we observe vastly different convergence properties.

\paragraph*{ETH and quantum chaos ---}
\begin{figure}[t]
    \centering
    \includegraphics[width=0.48\textwidth]{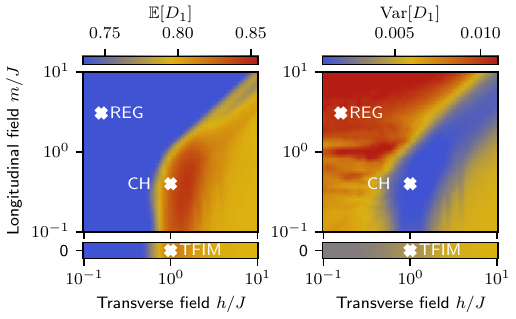}
    \caption{
    Identification of chaotic (point \texttt{CH}) and non-chaotic regimes (\texttt{TFIM}, \texttt{REG}) of the mixed-field Ising model~\eqref{eq:ising_model} as a function of transverse and longitudinal fields via eigenstate delocalization. ETH is expected to hold in the chaotic regime.
    Mean $\mathbb{E}[D_1]$ (left) and variance $\text{Var}[D_1]$ (right) of the fractal dimension $D_1$ are evaluated in the $n$-qubit Z-basis for $n=8$.
    Large mean $\mathbb{E}[D_1]$ combined with small $\text{Var}[D_1]$ signal quantum chaos.
    Other parameter points and eigenstate delocalization in other bases are discussed in App.~\ref{app:numerics_spectral_statistics} (cf. Tab.~\ref{tab:hamiltonian_parameters}).
    }
    \label{fig:spectral_statistics_n_qubits8}
\end{figure}
The ETH is expected to hold with high accuracy in quantum chaotic systems \cite{Srednicki1999}.
For our model~\eqref{eq:ising_model}, we identify sets of parameter points $(h,m)$ with distinct quantum chaotic properties by analyzing the delocalization properties of the energy eigenstates, as used e.g. in \cite{pausch_chaos_2021,brunner_many-body_2023}.
This approach directly assesses the underlying properties of ETH, while simultaneously being independent of a specific choice of observables (jump operators).
Quantum chaos is signaled by a large average $\bE[D_1]$ over the fractal dimension $D_1$, in combination with a drop in variance $\Var[D_1]$.
This behavior signals a strong and uniform delocalization of energy eigenstate.
Definitions and more details are given in App.~\ref{app:numerics_spectral_statistics}.
Figure~\ref{fig:spectral_statistics_n_qubits8} shows $\mathbb{E}[D_1]$ (left) and $\text{Var}[D_1]$ (right) in the parameter range $m/J \in \{0\} \cup [10^{-1}, 10^1]$, $h/J \in [10^{-1}, 10^1]$.
The parameter region $0.7 \lesssim h/J \lesssim 2$, $0.2 \lesssim m/J \lesssim 0.9$ shows the clearest signature of quantum chaos (see also Fig.~\ref{fig:app_spectral_statistics_n_qubits8} for an extended analysis).
Within this {\it chaotic lake}, the ETH ansatz is expected to be valid with with high accuracy.
In turn, the limits $m\rightarrow \infty$ or $h \rightarrow \infty$ are dominated by either the transverse or longitudinal field, whereas in the limit $h,m \rightarrow 0$ the interaction term dominates the dynamics.
In all three cases, eigenstates become more localized and we expect that the ETH is not applicable.
We call these limits the \textit{regular limits} of our model.
Another important case is the limit $m \rightarrow 0$ of the transverse field Ising model, which is an integrable (i.e., non-chaotic) model.
Based on the phase diagram in Fig.~\ref{fig:spectral_statistics_n_qubits8} we select three relevant parameter points, shown as white crosses, for our numerical studies in this section---\texttt{TFIM} (the transverse field Ising model),
\texttt{REG} (dominated by the longitudinal field), and the point \texttt{CH} deeply within the chaotic lake.
Additional parameter points, exploring the remaining regular limits and the chaos transition, are considered in App.~\ref{app:numerics}.
In Tab.~\ref{tab:hamiltonian_parameters} we summarize all considered parameter points.
Furthermore, in App.~\ref{sec:app_mean_level_spacing_ratio} we compare eigenstate delocalization properties to the distribution of level spacing ratios, another frequently used signature of quantum chaos~\cite{oganesyan_localization_2007,atas_distribution_2013}.

\subsection{Simulation of the Lindblad equation}
\label{sec:numerics:time_evolution}

\begin{figure*}
    \centering
    \includegraphics[width=\textwidth]{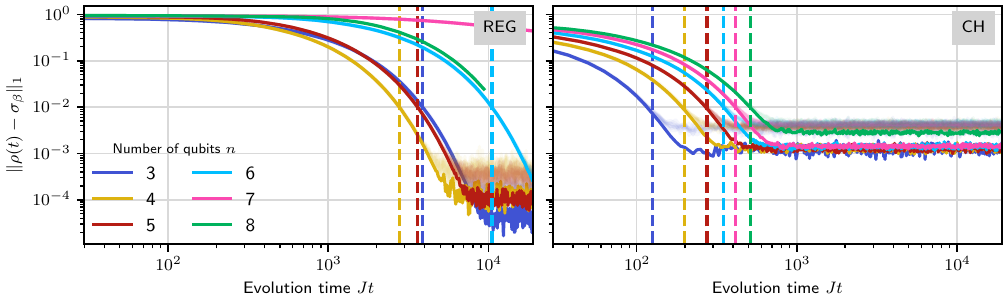}
    \caption{
    Trace distance between the target Gibbs state $\sigma_\beta$ and the state $\rho(t)$ evolved under Lindblad dynamics of the mixed-field Ising model~\eqref{eq:ising_model} with single random jump operators at each time step for $n=3,\dots, 8$ qubits
    and $\beta=(2J)^{-1}$.
    We consider a set of size $|\bm{A}| = 20$ of random $2$-local Pauli jump operators and the Hamiltonian parameter points \texttt{REG} (left, dominated by the longitudinal field) and \texttt{CH} (right, quantum chaotic regime), cf. Fig.~\ref{fig:spectral_statistics_n_qubits8}.
    Thin lines are individual trajectories of the randomized Lindblad simulation protocol (details in App.~\ref{app:numerics_details}), thick lines are trace distances between the Gibbs state and the time-evolved state averaged over trajectories.
    Vertical lines indicate the mixing time estimate in Eq.~\eqref{eq:mixing_time_estimate}.
    The results confirm much faster mixing in the quantum chaotic regime (\texttt{CH}) than in the non-chaotic regime (\texttt{REG}). On the other hand, the system in parameter regime \texttt{REG} converges to a state closer to the Gibbs state as explained in App.~\ref{app:numerics_mixing_time}.
    }
    \label{fig:time_evolution_random_k_body_Pauli_k_order2}
\end{figure*}

In this subsection we study numerically the mixing time and the convergence accuracy $\left\|\SteadyState - \sigma_{\beta} \right\|_1$ of the Lindblad dynamics~\eqref{eq:lindblad_eq}.
The Lindblad operators~\eqref{eq:L_operators} are constructed from a set of random $k$-local Pauli jump operators $A^a$ with $a \in \bm{A}$.
The random jump operator model is detailed in App.~\ref{sec:app:adaptive_rk_scheme}.
For our studies, we vary both $k$ and the number $|\bm{A}|$ of considered jump operators.
For computational efficiency, we ues a randomized simulation strategy. We compute individual trajectories of the dynamics by sampling at each time step a single jump operator $A^a$, $a\in \bm{A}$, with probability $p_a = 1/|\bm{A}|$, and evolve the current state under this single jump operator. This is similar to our quantum protocol in Sec.~\ref{sec:single_ancilla_protocol}.
We use the fourth-order Runge-Kutta method, instead of the Trotterization~\eqref{eq:lindblad_random} used in the quantum algorithm, and evaluate the integral in Eq.~\eqref{eq:L_operators} exactly, instead of using the discretization~\eqref{eq:L_operator_discretized}.
For each parameter point, we simulate the Lindblad evolution up to a maximal number of time steps $N_\mathrm{steps}^\mathrm{max} = 3 \cdot 10^5$. The step size is determined by an adaptive scheme for each setting and system size individually.
Details on the numerical scheme are provided in App.~\ref{app:numerics_details}, with step sizes reported in Tab.~\ref{tab:hamiltonian_parameters}.
The results shown in the main text are computed for a fixed number of random trajectories. A detailed analysis of the influence of the number of trajectories on the dynamics, and a comparison to the Monte Carlo wave function method, is given in App.~\ref{sec:app:adaptive_rk_scheme:comparison_WVMC}.
We highlight that the numerical simulations performed are exact, thus limiting the system sizes that are probed.
In App.~\ref{sec:app:improved_classical_simulation_techniques} we briefly discuss the potential and challenges of recently proposed tensor network simulation techniques for similar systems.

Figure~\ref{fig:time_evolution_random_k_body_Pauli_k_order2} exemplifies the evolution of the trace distance $\Vert \rho(t) - \sigma_\beta\Vert_1$ for the two parameter points \texttt{REG} and \texttt{CH} (cf. Tab.~\ref{tab:hamiltonian_parameters}) and increasing number of qubits $n=3, \dots, 8$.
As initial state we fix the maximally mixed state, $\rho(0) = I / 2^n$ with $I$ the $2^n\times 2^n$ identity, and we fix for each $n$ a set of random $(k=2)$-local jump operators of size $|\bm{A}| = 20$.
A single trajectory of the dynamics for a given $n$ is generated by sampling a random jump operator from this set in each time step.
The trace distances for several of these trajectories are plotted as thin lines (varying $n$ is indicated by color).
The time-evolved state is computed as an average of the trajectory states at each time point (see App.~\ref{app:numerics_details} for details).
The trace distances of this state to the Gibbs state are shown as bold lines.
In all cases, the averaged time-evolved state is significantly closer to the Gibbs state than the individual trajectories.

Because of the randomization, the Lindblad dynamics does not exactly converge to $\SteadyState$ but plateaus at a value between $10^{-4}$ and $10^{-2}$.
Moreover, the exact steady state of Eq.~\eqref{eq:lindblad_eq} only approximates the Gibbs state $\sigma_\beta$.
However, as we will see in Fig.~\ref{fig:distance_steady_gibbs_random_k_body_Pauli_k_order2}, the convergence accuracy $\Vert \SteadyState - \sigma_\beta \Vert_1$ is well below the observed plateau in Fig.~\ref{fig:time_evolution_random_k_body_Pauli_k_order2} in all cases.
Hence, with our limited number of trajectories in Fig.~\ref{fig:time_evolution_random_k_body_Pauli_k_order2}, we cannot dynamically resolve the difference between the two states.
Further note that for the point \texttt{REG} and $n=8$, the time evolution already stops at $Jt \approx 10^{4}$. This is because this setting requires a much smaller step size for the simulation, cf. Tab.~\ref{tab:hamiltonian_parameters}, such that we only reach a time of $Jt \approx 10^{4}$ within the maximal number of steps $N_\mathrm{steps}^\mathrm{max}$.

\subsection{Mixing time estimate and spectral gap}
\label{sec:numerics:mixing_time}

We set $\epsilon=10^{-2}$ and estimate the mixing time~\eqref{eq:mixing_time} by
\begin{equation}
\label{eq:mixing_time_estimate}
\hat{t}_\mathrm{mix} \coloneq \inf \lbrace t > 0: \Vert \rho(t) - \sigma_\beta \Vert_1 < 10^{-2} \rbrace \,,
\end{equation}
starting from the maximally mixed state, $\rho(0) = I/2^n$, instead of maximizing over all initial states.
In a slight abuse of notation, $\rho(t)$ denotes the evolved state under the numerical scheme, and not the state under the exact dynamics~\eqref{eq:lindblad_eq}.
The mixing time estimate $\hat{t}_\mathrm{mix}$ is highlighted in Fig.~\ref{fig:time_evolution_random_k_body_Pauli_k_order2} by vertical dashed lines.
We generally observe that the early time dynamics ($t \leq \hat{t}_\mathrm{mix}$) coincides for all trajectories (thin lines).
Hence, a very small number of trajectories (even only a single one) suffices to accurately compute the estimate $\hat{t}_\mathrm{mix}$.
As expected, $\hat{t}_\mathrm{mix}$ increases with $n$.
The quantum chaotic parameter point \texttt{CH} exhibits a much smaller mixing time than the regular limit \texttt{REG}.
Interestingly, the trace distance $\Vert \rho(t) - \sigma_\beta\Vert_1$ converges to a smaller value for \texttt{REG} than for \texttt{CH}, which we will discuss in more detail in Fig.~\ref{fig:distance_steady_gibbs_random_k_body_Pauli_k_order2}.

\begin{figure*}
    \centering
    \includegraphics[width=\textwidth]{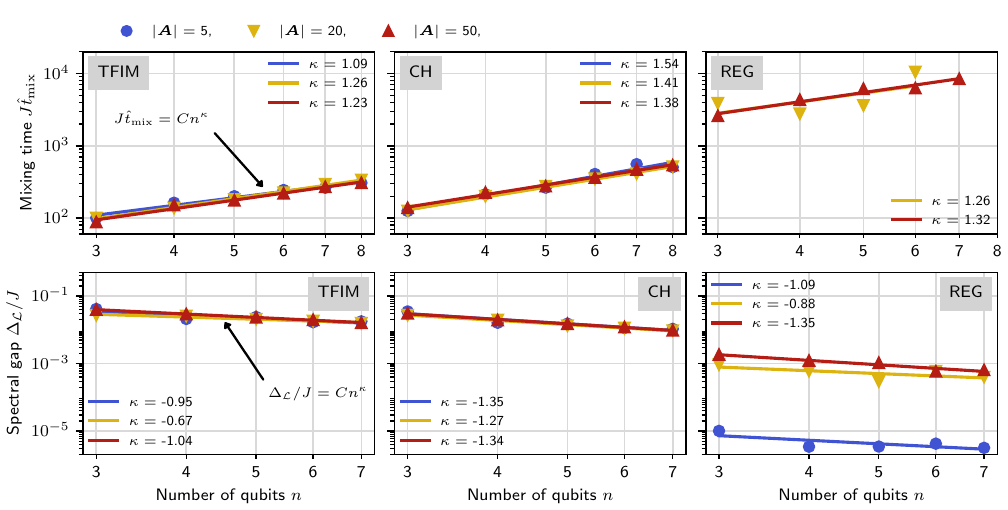}
    \caption{
    Scaling of mixing time (top) and spectral gap of Lindbladian $\calL$ (bottom) with $n$.
    We consider sets of random $2$-local Pauli jump operators of size $|\bm{A}| = 5, 20, 50$ and the mixed-field Ising model~\eqref{eq:ising_model} at parameter points \texttt{TFIM} (left, transverse field Ising model), \texttt{CH} (middle, quantum chaotic regime) and \texttt{REG} (right, dominated by the longitudinal field), cf. Fig.~\ref{fig:spectral_statistics_n_qubits8}.
    We use $\beta=(2J)^{-1}$.
    The mixing time and spectral gap data fit well a polynomial $\propto n^\kappa$ (lines). The transverse-field (\texttt{TFIM}) and chaotic (\texttt{CH}) regimes show much faster mixing and correspondingly larger spectral gap than the longitudinal-field-dominated regime (\texttt{REG}).
    The number of jump operators $|\bm{A}|$ has a negligible effect except for \texttt{REG}, $|\bm{A}|=5$.
    In this case the dynamics does not converge within the maximal simulation time and the spectral gap is significantly smaller than in other cases.
    }
    \label{fig:mixing_time_spectral_gap_random_k_body_Pauli_k_order2}
\end{figure*}
From the bound in Eq.~\eqref{eq:mixing_time_bound} we expect the mixing time to scale as $\calO(\mathsf{poly}(n))$.
Here we study numerically the scaling of $\hat{t}_\mathrm{mix}$ and the spectral gap $\Delta_\calL$ with the number of qubits $n$ for different numbers of $(k=2)$-local jump operators, $|\bm{A}| = 5,20,50$, at the parameter points \texttt{TFIM}, \texttt{CH} and \texttt{REG} (cf. Fig.~\ref{fig:spectral_statistics_n_qubits8}).
As before, the initial state is maximally mixed.
Other parameter points, other values of $k$, as well as simulations with initial state $\ket{0}^{\otimes n}$ are discussed in App.~\ref{app:numerics_mixing_time}.

Figure~\ref{fig:mixing_time_spectral_gap_random_k_body_Pauli_k_order2} shows the mixing time $\hat{t}_\mathrm{mix}$ (top row) and the spectral gap $\Delta_{\calL}$ (bottom row) as a function of $n$ for different $|\bm{A}|$.
We compute $\Delta_{\calL}$ from the full Lindbladian~\eqref{eq:lindbladian_general}, containing all Lindblad operators $L^a, a \in \bm{A}$ (details in App.~\ref{app:numerics_details_Lindbladian_gap}).
For each $|\bm{A}|$ and each parameter point the legend shows the leading exponent $\kappa$ of the polynomial dependence of $\hat{t}_\mathrm{mix}, \Delta_\calL $ on $n$, obtained by fitting the data to $\sim n^\kappa$.
In almost all cases, this assumption is justified based on our numerical findings.
We note that the number of jump operators $|\bm{A}|$ does not have a significant effect on $\hat{t}_\mathrm{mix}$ and $\Delta_\calL$.

We observe small mixing times for the chaotic point \texttt{CH} with a polynomial scaling exponent of approximately $\kappa \approx 1.4$.
Other chaotic parameter points show a similarly fast convergence (see App.~\ref{app:numerics_mixing_time}).
In contrast, the regular limit \texttt{REG} shows over an order of magnitude larger mixing times with slightly smaller scaling exponent $\kappa \approx 1.3$.
Note, however, that in this case and for $n=8$ or $|\bm{A}| = 5$,
the time-evolved state does not converge to a distance $\Vert \rho(t) - \sigma_\beta \Vert_1$ below $\epsilon = 10^{-2}$ within the maximum time horizon considered.
Even larger mixing time estimates are obtained in the limit of large transverse field, $h \gg m, J$, as discussed in App.~\ref{app:numerics_details}.
Interestingly, the integrable point \texttt{TFIM}, corresponding to the transverse-field Ising model, shows the fastest convergence with exponent $\kappa \approx 1.2$.
Moreover, for \texttt{REG} we observe a non-monotonic dependency of mixing time on $n$ (although showing an overall increasing trend with $n$). The mechanism behind this is a complex interplay of Bohr frequencies of the model, the filter functions used in Eq.~\eqref{eq:L_operators}, and the energy distribution of the initial state.
In App.~\ref{app:bohr_frequencies} we discuss this effect in more detail and show that, as a consequence, the scaling for the regular point \texttt{REG} is strongly susceptible to the initial state, while the chaotic point \texttt{CH} is relatively robust against the choice of initial state.

The situation is mirrored for the spectral gap $\Delta_\calL$. As expected, we obtain large gaps with a decay exponent between $-2\leq \kappa \leq -1$ for the chaotic parameter point \texttt{CH}, and comparably smaller gaps for \texttt{REG}.
Interestingly, in this case and for $|\bm{A}| =5$, the spectral gap drops roughly two orders of magnitude in comparison $|\bm{A}| = 20, 50$.
This is in line with the above observed extremely long convergence time in this case.
The integrable limit \texttt{TFIM} exhibits the largest spectral gap, accompanied by the smallest (in modulus) decay exponents between $-1 \leq \kappa \leq -0.5$.

In App.~\ref{app:non-local_jump_operators} we further explore the influence of initial state and locality $k$ of the jump operators
(see App.~\ref{sec:app:adaptive_rk_scheme} for a detailed description of the jump operator model we use).
We generally observe an increase of mixing time for larger $k$.
Moreover, the zero state $\ket{0}^{\otimes n}$ typically converges faster to the steady state than the maximally mixed state.
In App.~\ref{app:coherent_term} we study the influence of the coherent term $-\im[H, \rho]$ in the Lindbladian~\eqref{eq:lindbladian_general}. In the derivations of our analytical results in Eqs.~\eqref{eq:mixing_time_bound} and~\eqref{eq:upper_bound_trace_distance}, it turns out the coherent term does not play a role. Nevertheless, we find numerically that the coherent term significantly improves the convergence towards the steady state of $\calL$.

\subsection{Distance between steady state and Gibbs state}
\label{sec:numerics:convergence_accuracy}

\begin{figure*}
    \centering
    \includegraphics[width=\textwidth]{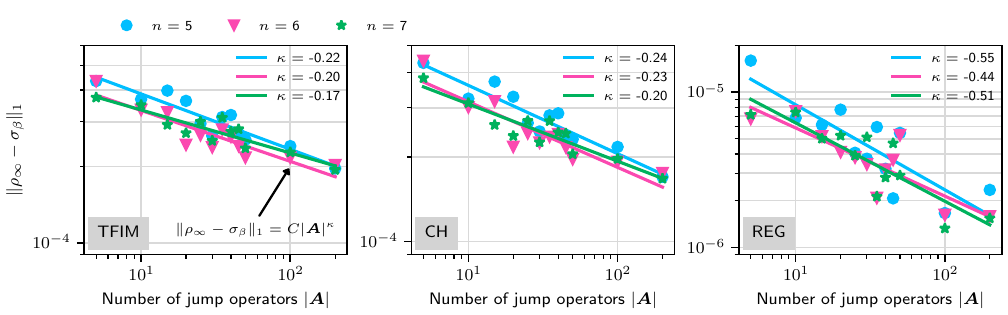}
    \caption{
    Scaling of trace distance between the exact steady state of the Lindblad dynamics $\SteadyState$ and target Gibbs state $\sigma_\beta$ with the number of jump operators $|\bm{A}|$.
    We consider random $2$-local Pauli jump operators, $n = 5,6,7$ qubits,
    $\beta=(2J)^{-1}$,
    and the mixed-field Ising model~\eqref{eq:ising_model} at parameter points \texttt{TFIM} (left, transverse field Ising model), \texttt{CH} (middle, quantum chaotic regime) and \texttt{REG} (right, dominated by the longitudinal field), cf. Fig.~\ref{fig:spectral_statistics_n_qubits8}.
    Polynomial fits $\propto |\bm{A}|^\kappa$ (lines) confirm a decrease with exponent $\kappa \approx -0.2$ for \texttt{TFIM} and \texttt{CH} slightly smaller in modulus than expected from the bound~\eqref{eq:upper_bound_trace_distance} ($\kappa=-1/2$).
    The integrable point \texttt{REG} shows a faster decrease ($\kappa\approx -1/2$) and a significantly smaller distance between the steady state and Gibbs state.
    Appendix~\ref{app:numerics_mixing_time} explains this observation and presents further results.
    }
    \label{fig:distance_steady_gibbs_random_k_body_Pauli_k_order2}
\end{figure*}

According to the bound~\eqref{eq:upper_bound_trace_distance}, the trace distance between the steady state and Gibbs state, $\Vert \SteadyState - \sigma_\beta \Vert_1 \leq \calO(\mathsf{poly}(n)/\sqrt{|\bm{A}|})$, depends inversely on the number of jump operators $|\bm{A}|$ when $|\bm{A}|$ is not sufficiently large.

Here we study numerically the scaling of $\Vert \SteadyState - \sigma_\beta \Vert_1$ with $|\bm{A}| \in [5, 200]$ for $n=5,6,7$ at the parameter points \texttt{TFIM}, \texttt{CH} and \texttt{REG}.
The results are shown in Fig.~\ref{fig:distance_steady_gibbs_random_k_body_Pauli_k_order2}.
As expected, the distance between both states decreases in all cases for increasing $|\bm{A}|$.
For the points \texttt{TFIM} and \texttt{CH} we observe a clear linear decrease on the double-logarithmic scale. The leading order exponent $\kappa$ is obtained from a fit of the data to $\sim |\bm{A}|^\kappa$ (solid lines) and given in the legend.
The data shows a smaller (in modulus) exponent, around $\kappa\approx -0.2$, compared to the scaling $\kappa = -1/2$ suggested by the analytical upper bound~\eqref{eq:upper_bound_trace_distance}.
Note that the dynamics is approximated with a randomized protocol with fixed number of trajectories, as described in App.~\ref{sec:app:adaptive_rk_scheme}.
A detailed analysis of the decay of convergence accuracy with increasing number of trajectories is provided in App.~\ref{sec:app:adaptive_rk_scheme:comparison_WVMC}.

The data for the parameter point \texttt{REG} shows strong variation. However, we still discern a clear decrease with $|\bm{A}|$. In this case, the leading order exponent is approximately $\kappa\approx -1/2$.
Overall, we observe a much smaller trace distance between steady state and Gibbs state in the regular limit \texttt{REG} ($m/J\gg 1$, $h/J\ll 1$), compared to the other two points.
This is in line with Fig.~\ref{fig:time_evolution_random_k_body_Pauli_k_order2}, where we also observed a convergence to a lower value of $\Vert \rho(t) - \sigma_\beta \Vert_1$ for \texttt{REG}.
The same effect is present for the complementary regular limit $h/J\gg 1$, $m/J\ll 1$ (cf. Fig.~\ref{fig:app_distance_steady_gibbs_random_k_body_Pauli_k_order2}, App.~\ref{app:numerics_mixing_time}). The likely cause is that, effectively, only one energy transition contributes to the Lindblad operator which leads to an effective Lindbladian with the Gibbs state as its steady state with a high accuracy. For details, see
App.~\ref{app:bohr_frequencies}.

Interestingly, we do not observe an increase of $\Vert \SteadyState - \sigma_\beta \Vert_1$ for an increasing number of qubits $n$, suggesting that this empirical distance has milder dependence on $n$ than our analytical upper bound~\eqref{eq:upper_bound_trace_distance}.
This is a promising observation for the practical implementability of our approach since it potentially allows us to choose a number of jump operators $|\bm{A}|$ independently of $n$ and still retain $\Vert \SteadyState - \sigma_\beta \Vert_1 = \calO(\epsilon)$, as discussed below Eq.~\eqref{eq:upper_bound_trace_distance}.

\section{Resilience against hardware noise}
\label{sec:noise}

In addition to errors incurred by the approximations necessary for circuit implementation, any hardware implementation will inevitably encounter imperfections due to noise. 
As much as for the algorithmic errors in Sec.~\ref{sec:single_ancilla_protocol}, it is crucial to quantify and understand these errors, which is the purpose of this section.

To facilitate the exposition we focus on the effect of noise, considering a stochastic noise model (i.e. a stochastic mixture of unitary channels including the identity), and ignore other algorithmic errors.
The noise model contains typical noise channels such as local and global depolarization and arbitrary Pauli noise channels
representative of noise obtained after application of randomized compiling~\cite{wallman_noise_2016} techniques, that are commonly used to mitigate coherent errors.
Under this assumption, we derive two independent bounds on the noise-induced Gibbs state preparation error.
The first one is restricted to stochastic noise, Eq.~\eqref{eq:main_global_depo}, and explicitly takes into account the damping of early-time errors through the contractive nature of the Lindblad channel.
The second bound is adapted from \cite{Chen2023thermal} and is applicable to more generic noise channels, as discussed in more depth in App.~\ref{app:noise_comp}.
In certain regimes, the former can result in a significantly tighter bound.
To make this study concrete, we will evaluate deviations incurred by noise resorting to protocol characteristics extracted from the numerical studies presented in Sec.~\ref{sec:numerics} and extrapolated to larger system sizes.

\subsection{Setup}\label{sec:noise_setup}
A noisy realization of the protocol consists of noiseless short-time Lindblad evolution interleaved with the noise channel.
In the noiseless case, the dynamics corresponding to $M$ steps of evolution, each of duration $\delta t$, is given by
$\Gamma_M \coloneq (\e^{\delta t\lindblad})^M$. In contrast, we model the noisy dynamics as $\widetilde{\Gamma}_M \coloneq (\Lambda_{\lambda} \circ \e^{\delta t\lindblad})^M$, with noise channel
\begin{align}
\label{eq:main_global_depo}
    \Lambda_{\lambda}[X] \coloneq (1-\lambda) X + \sum_{l} \lambda_l U_l X U_l^{\dagger}.
\end{align}
Each $U_l$ is unitary, $\lambda_l \in [0, 1]$, and $\lambda = \sum_l \lambda_l \in [0,1]$ is the overall error probability.

Let us define $\rho_M\coloneq\Gamma_M[\rho(0)]$ and $\tilde{\rho}_M\coloneq \widetilde{\Gamma}_M[\rho(0)]$, the states obtained after $M$ steps of noiseless and noisy evolution, respectively, both starting from the same initial state $\rho(0)$.
Under the noiseless dynamics, $\rho_M$ converges to the target Gibbs state $\sigma_\beta$ up to a convergence accuracy~\eqref{eq:upper_bound_trace_distance}, which is neglected for the purpose of this section.
In trace distance this convergence is captured by
\begin{align}
\label{eq:main_bound_td}
    \|\rho_M - \sigma_\beta \|_1 \leq B \e^{-\alpha M} ,
\end{align}
We note that for this bound to hold for an arbitrary input state $\rho(0)$, $B$ can be exponentially large in the system size.
Below we will fit $B$ and $\alpha$ to numerical data obtained from our simulations and, when limiting ourselves to a single input state, show that $B$ can be significantly smaller in practice.

\subsection{Bounds on the convergence accuracy}\label{sec:noise_bounds}
As detailed in App.~\ref{app:noise_conv}--\ref{app:gen_noise_comp}, through Eq.~\eqref{eq:main_bound_td} and the use of triangle inequalities, we can upper bound the distance $\|\tilde{\rho}_M  - \sigma_\beta \|_1$
between the state prepared by $M$ steps of noisy dynamics and the Gibbs state by 
\begin{align}
\label{eq:main_dk}
\begin{split}
    &\widetilde{B}_M 
    \coloneq B\left[u_0^M \left(1- \frac{\lambda}{1-u_0}\right) + \frac{\lambda}{1-u_0} \right]
    \\
    &\text{with   } u_0\coloneq (1-\lambda)\e^{- \alpha} \in [0,1).
\end{split}
\end{align}
Given that $0\le u_0<1$, $\widetilde{B}_M$ decreases monotonically with $M$ towards the asymptotic bound $\widetilde{B}_{\infty}$, limiting the distance between the noisy steady state $\tilde\rho_\infty$ and the Gibbs state,
\begin{align}
\label{eq:main_asympt_distance}
\norm{\tilde\rho_\infty  - \sigma_\beta}_1 \leq \widetilde{B}_{\infty} \coloneq B \frac{\lambda}{1-u_0}.
\end{align}
The precise statement is given in Thm.~\ref{thm:noise_bound}.
This bound depends both on the error probability $\lambda$ and the convergence rate $\alpha$. 
As expected, the smaller the error per step of evolution, the closer $\tilde{\rho}_\infty$ is to the Gibbs state.
Similarly, the bound decreases as the convergence rate increases, and simply becomes $B\lambda$ for $\alpha \rightarrow \infty$.
Note that for given $B, \alpha$ and sufficiently small $\lambda$ (i.e. $B \lambda / (1- u_0) \leq \norm{ \rho_\mathrm{noise} - \sigma_\beta }_1/2$) the noisy steady state $\tilde{\rho}_\infty$ is closer to the Gibbs state than to the steady state
$\rho_{\rm noise}$
of the noise channel.
This is in contrast to unitary evolution, which always converges to the fixed point of the noise channel for any finite level of noise $\lambda$.

In App.~\ref{app:noise_comp}, we derive a second bound on the noise-induced Gibbs state preparation error, Eq.~\eqref{eq:app_bounds_generic}, which holds for more generic noise models.
This bound is based on the operator distance between the Lindblad evolution quantum channels of the ideal and noisy dynamics.
For the derivation we adapt \cite{Chen2023thermal} (Lemma II.1), which is stated for continuous time, to our discrete setting and the considered noise channel~\eqref{eq:main_global_depo}.

\begin{figure*}
\centering
\includegraphics[width=1\textwidth]{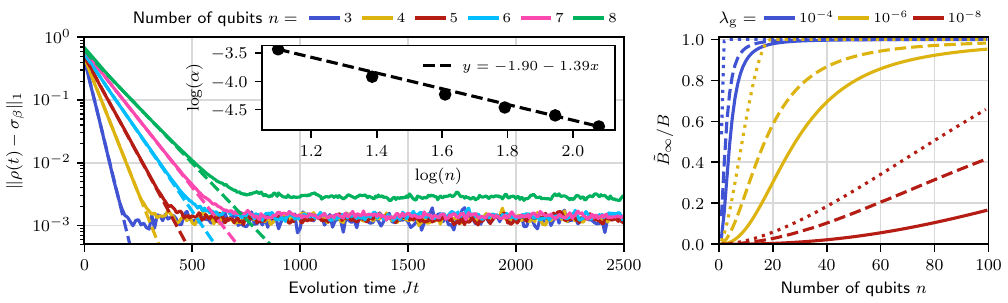}
\caption{
Noise study for the mixed-field Ising model~\eqref{eq:ising_model}.
(Left panel) Fitting of the convergence dynamics in the chaotic regime (\texttt{CH} data from Fig.~\ref{fig:time_evolution_random_k_body_Pauli_k_order2}). We report the trace distance between the state prepared by noiseless evolution and the target Gibbs state, for different system sizes ($n=3$ to $8$ qubits, colors in legend)
and set $\beta=(2J)^{-1}$ as before.
A fit of the form of the right hand side of Eq.~\eqref{eq:main_bound_td} (dashed lines) matches closely the convergence dynamics. (Inset) the convergence rates $\alpha$ are found to scale as $\mathcal{O}(n^{-1.39})$. (Right panel) Extrapolating the convergence rates $\alpha$ and prefactors $B$ to larger system sizes (up to $n=100$), and assuming a number of $2$-qubit gates $N_{\rm g}(n)=50n$ per evolution step ($J\delta t=1$), we can evaluate the bounds~\eqref{eq:main_asympt_distance} on the distance between the steady state of the noisy dynamics and the Gibbs state. 
These are normalized
by $B$ and plotted (solid lines) for different values of the error probabilities $\lambda_{\rm g}$ per noisy gate.
For comparison, we show bounds based on the operator distance between noisy and noiseless Lindblad channels, Eq.~\eqref{eq:app_bounds_generic}, with the same strength (dotted lines), and for corresponding unitary circuits (dashed lines). These are further detailed in the main text and appendices (App.~\ref{app:noise_comp} and~\ref{app:noise_comp_unit}).
}
\label{fig:noise_study}
\end{figure*}

\subsection{Results}

In order to quantify noise-induced errors and, ultimately, to assess the viability of our dissipative Gibbs state preparation protocol, one wishes to evaluate the bound~\eqref{eq:main_asympt_distance} at varied system sizes, noise levels, and algorithmic parameters (for instance, the discretization of the Lindblad operators~\eqref{eq:L_operator_discretized} that affects the gate count of the circuit and thus the strength of the noise per step of evolution).
To illustrate applications of these bounds, we resort to values of $B$ and $\alpha$, extrapolated to large $n$ from the numerical data discussed in Sec.~\ref{sec:numerics}.
We use the data from the chaotic setting \texttt{CH} in Fig.~\ref{fig:time_evolution_random_k_body_Pauli_k_order2}, for $n=3, \dots, 8$, and starting from a fixed initial state $\rho(0)=I/2^n$.
In Fig.~\ref{fig:noise_study} (left panel) we fit the right-hand side of Eq.~\eqref{eq:main_bound_td} (dashed lines) to the numerical data (solid lines).
From these, a dependency on $n$ is extracted. Results are displayed for $\alpha$ in the inset of Fig.~\ref{fig:noise_study}~(left panel), showing that a geometric fit closely matches the data, and for $B$ in Fig.~\ref{fig:noise_app_fit}~(right panel) of App.~\ref{app:noise_study}.
Using these fitted $\alpha, B$ in our bound~\eqref{eq:main_asympt_distance} is justified only if we additionally assume that Eq.~\eqref{eq:main_bound_td} holds with these parameters for all $\rho(0)$ that can be generated from $I/2^n$ by application of $\Gamma$. For a more detailed discussion we refer to App.~\ref{app:gen_noise_comp}.

For the noise channel~\eqref{eq:main_global_depo} the error probability $\lambda$ for one step of evolution can be related to the number $N_{\rm g}$ of noisy operations (e.g., $2$-qubit gates) required for the circuit implementation of $\e^{\delta t\lindblad}$ and the error $\lambda_{\rm g}$ per operation: $1-\lambda = (1-\lambda_{\rm g})^{N_{\rm g}}$. We fix the error per noisy gate $\lambda_{\rm g}$ and assume a number of gates  $N_{\rm g}=50n$ per unit of time ($J\delta t = 1$).
This linear dependency is taken to reflect that the number of terms in the Hamiltonian~\eqref{eq:ising_model} grows linearly with the system size. In turn, the number of gates required to perform unitary evolution (see Eq.~\eqref{eq:W_evolve} in Appendix~\ref{app:single_gibbs}) would also increase linearly.
The coefficient $50$ is taken as a representative example---in concrete implementations it will depend on the circuit details and their compilations.
Other choices, grounded in algorithmic and system details or scalings of $\alpha$ and $B$, can be made. Further full circuit simulations under a local depolarizing $2$-qubit gate noise model are presented in Sec.~\ref{sec:local_depolaizing_noise_simulations},
where we report actual gate counts resulting from compilation.

The bound~\eqref{eq:main_asympt_distance} is reported in units of $B$ in Fig.~\ref{fig:noise_study}~(right panel, solid lines), and evaluated for different values of the $2$-qubit gate error, $\lambda_{\rm g}=10^{-4},10^{-6}$ and $ 10^{-8}$. Such values span error rates representative of the transition from quantum platforms with high-fidelity $2$-qubit operations on physical qubits expected in the near term, to platforms with a limited number of quantum-error-corrected qubits in the medium term.
As can be seen, for a noise strength $\lambda_{\rm g}=10^{-4}$ (blue line), except for the smaller system sizes,
the prepared steady state
quickly becomes indistinguishable from the maximally mixed state, which corresponds to a value $\widetilde{B}_{\infty}/B=1$ for the chosen initial state $\rho(0) = I/2^n$.
As the $2$--qubit gate errors decrease to $\lambda_{\rm g}=10^{-6}$, the range of sizes for which the prepared state remains sufficiently close to the target Gibbs states increases to a few tens of qubits. 
Finally, for $\lambda_{\rm g}=10^{-8}$ it would become possible to prepare Gibbs states of up to $n=100$ qubits with relatively low error found to be $\|\tilde{\rho}_\infty - \sigma_\beta\|_1 \approx 0.2$.

We stress that the previous evaluation of the bounds relies on extrapolation of the convergence rates from fixed input state and relatively small system sizes, up to $n=8$, to much larger values of $n$.
Still, through bounds on the spectral gap derived in App.~\ref{app:ssec:gap_EL} and~\ref{app:ssec:spectral_gap_and_mixing_time}, and corresponding assumptions, we expect the convergence rate to scale geometrically as $\calO(n^{-c})$, as used here, albeit with potentially different values of the exponent $c$.

To put these results into perspective, we include our second bound, Eq.~\eqref{eq:app_bounds_generic}, which in principle can be adapted to more generic noise models than the stochastic model given in Eq.~\eqref{eq:main_global_depo}.
The second bound is distinct from Eq.~\eqref{eq:main_asympt_distance} and is tighter for large enough values of $B$. However, for relatively small values of $B$, in regimes where noise effects are comparatively small, it can significantly overestimate the impact of noise.
This is observed in our numerics.
As can be seen in Fig.~\ref{fig:noise_study}~(right panel), the bounds for generic noise (dotted lines) are substantially looser than the ones from Eq.~\eqref{eq:main_asympt_distance}:  $3$ to $70$ times larger for the regime of small errors $\tilde{B}_{\infty} \leq0.2$.
Notably, the latter accounts for the fact that errors (especially the ones occurring at early times in the dynamics) are tempered by the subsequent steps of evolutions that always tend towards the steady state of $\calL$.
This feature is not captured by the more generic bounds.

Finally, we also incorporate deviations that would be entailed for a unitary Gibbs state preparation protocol with comparable circuit complexity (dashed lines and detailed further in App.~\ref{app:noise_comp_unit}).
As can be seen, these are also significantly larger than the bounds of Eq.~\eqref{eq:main_asympt_distance} obtained for Lindblad evolution and stochastic noise: $2$ to $50$ times larger for the regime of small errors $\tilde{B}_{\infty} \leq0.2$.
While such a study has its limitations---some amount of non-unitarity would be required when preparing Gibbs states---it exemplifies
the enhanced resilience of the dissipative protocol compared to the unitary case, and adds to the body of work identifying this effect~\cite{Raghunandan2020,Polla2021,Mi2023,Cubitt2023,Granet2024}.
Conceptually, the existence of a non-trivial steady state, distinct from the steady state of the noise channel, ensures that the long-time evolution does not accumulate errors and still retains information about the Gibbs state~$\sigma_\beta$.

Overall, the bounds in Eqs.~\eqref{eq:main_asympt_distance} and~\eqref{eq:app_bounds_generic} allow us to assess the viability of the here considered dissipative Gibbs state preparation protocols on near-term quantum hardware.
We saw that errors induced by noise can be significantly less detrimental than expected, especially when the dominant contribution from the noise is stochastic. 
Going forward, combining such noise estimations with the algorithmic errors analysis of Sec.~\ref{sec:single_ancilla_protocol}
will be key in determining the optimal algorithmic parameters for our protocol.
Finally, we note, that the noise considered here was adversarial, in that our Lindbladian was not designed to account for it.  Engineering Lindbladians for Gibbs state preparation taking into account pre-characterized noise (even if approximately) may open up the path to even more resilient protocols.

\section{Quantum circuit simulation}
\label{sec:noisy_simulation}

In Sec.~\ref{sec:numerics} we investigated the dynamical properties of our Lindbladian~\eqref{eq:lindbladian_general}, neglecting realistic error sourced incurred by the circuit implementation.
In this section we focus on circuit simulations of our protocol (Sec.~\ref{sec:single_ancilla_protocol}). These allow us to study the influence of the main algorithmic hyperparameters and analyze the algorithm's performance under realistic noise models.
In Sec.~\ref{sec:alg_errors} we compare our analytical bounds derived in Sec.~\ref{sec:single_ancilla_protocol} to the error obtained from noiseless circuit simulation.
Furthermore, in Sec.~\ref{sec:local_depolaizing_noise_simulations},
we support our theoretical treatment of hardware noise (Sec.~\ref{sec:noise}) by noisy circuit simulations, assuming a local depolarizing two-qubit-gate noise model.
The circuit simulations are performed with \texttt{qujax}~\cite{qujax}, a Python circuit simulation package based on JAX~\cite{jax2018github}, which allows efficient parallelization and batching of circuit runs for different hyperparameter configurations.
As in the previous section, we focus on the parameter point \texttt{CH} of the Ising model~\eqref{eq:ising_model}.
For the jump operators $A^a$ in the discretized Lindblad operators~\eqref{eq:L_operator_discretized} we use the same model of $k$-local Pauli operators as in Sec.~\ref{sec:numerics} (see also App.~\ref{sec:app:adaptive_rk_scheme} for a precise definition).
The compiled circuits used in this section are accessible via \url{https://github.com/CQCL/lindblad_engineering_circuits}.
There, we also provide a circuit visualization of a single step of trotterized Lindblad evolution for $n=2$.

\subsection{Algorithmic errors}
\label{sec:alg_errors}

\begin{figure*}[t]
    \centering
    \includegraphics[width=\textwidth]{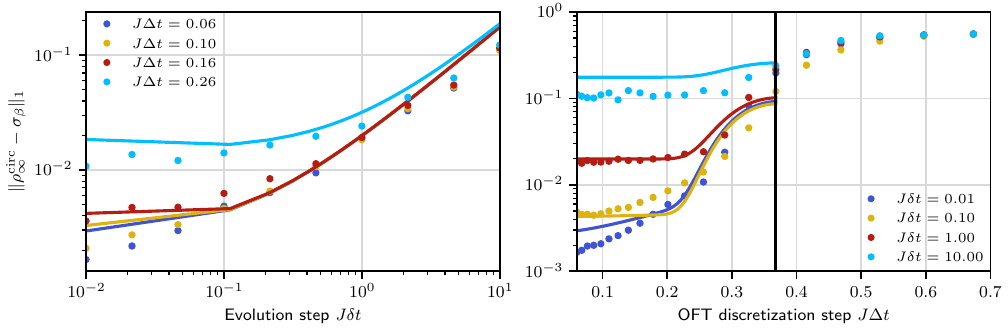}
    \caption{
    Algorithmic errors of the randomized single-ancilla Lindblad simulation protocol for the mixed-field Ising model~\eqref{eq:ising_model} with $n=5$ qubits
    and $\beta=(2J)^{-1}$.
    Trace distance between the steady state $\rho_\infty^{\rm circ}$ simulated by the noiseless circuit (simulation time $Jt=500$) and the target Gibbs state $\sigma_\beta$ as a function of Trotter evolution step $\delta t$ (left panel) and OFT discretization step $\Delta t$ (right panel). The fit (solid lines) obtained with the ansatz in Eq.~\eqref{eq:algorithmic_errors_fit} captures the algorithmic error data (dots) reasonably well.
    The black vertical line in the right panel indicates the maximum $\Delta t$ for the fit ansatz to be valid. Fitted parameters and further details are in App.~\ref{app:circ_sim_alg_errors}.
    }
    \label{fig:algorithmic_error_scaling_OFT_discretization_and_trotter_step}
\end{figure*}
We explore the interplay of different sources of algorithmic errors and quantitatively investigate how each source impacts the accuracy of our Gibbs state preparation protocol
presented in Sec.~\ref{sec:single_ancilla_protocol}.

\paragraph*{Error sources ---}
We recall the upper bound on the Lindblad simulation error~\eqref{eq:error_sources},
\begin{align}
\label{eq:error_sources2}
\begin{split}
    &t_{\rm mix}\times\calO\left(
        \delta t
        + \frac{T\Delta t^2}{\delta t}
        + \frac{\Delta t}{\sqrt{\beta}}\e^{-2(T/\beta)^2}
    \right.
    \\
    &\quad\left.   
        + \sqrt{\beta}|B_H| \e^{-\frac{1}{8}\big(2\pi\frac{\beta}{\Delta t} - 2\beta\|H\|_\infty-1\big)^2}
    \right),
\end{split}
\end{align}
with evolution step $\delta t$ and OFT discretization step $\Delta t$ appearing in Eqs.~\eqref{eq:lindblad_random} and~\eqref{eq:L_operator_discretized},
and the number of Bohr frequencies $|B_H|$ for a Hamiltonian $H$.
See also Prop.~\ref{prop:algorithmic_error} for a detailed discussion.
The first term is due to Trotterization, random sampling of Lindblad operators~\eqref{eq:lindblad_random}, and dilation~\eqref{eq:evolve_dilation}. The approximation of $\e^{-\im\sqrt{\delta t\gamma}K^a}$ in Eq.~\eqref{eq:evolve_dilation} induces the second, third and fourth terms, each of which arises as follows
(see Prop.~\ref{prop:algorithmic_error}
for more details):
Approximating the Lindblad operators~\eqref{eq:L_operators} via Eq.~\eqref{eq:L_operator_discretized} by restricting the integration domain to $[-T, T]$, and discretizing it into $2S$ time steps of size $\Delta t = T/S$ induces a truncation  (third term) and discretization error (fourth term).
The implementation of the operator $A^a(s\Delta t)$ appearing in Eq.~\eqref{eq:L_operator_discretized} requires coherent evolution under $\e^{-\im Hs\Delta t}$, that is approximated by a second-order product formula, and thus, induces the second term in Eq.~\eqref{eq:error_sources2}.

\paragraph*{Numerical setup ---}
The Gibbs state preparation error of our circuit is quantified by the distance $\|\rho_\infty^{\rm circ} - \sigma_\beta\|_1$ between the approximate steady state $\rho_\infty^{\rm circ}$ of the noiseless quantum circuit, Eq.~\eqref{eq:lindblad_random}, after Trotterization and OFT discretization and the Gibbs state $\sigma_\beta$.
For the following noiseless circuit simulations we use the mixed-field Ising Hamiltonian~\eqref{eq:ising_model} with $n = 5$ qubits and
choose $\beta=(2J)^{-1}$ as before.
We set
the cutoff time of the discretized OFT~\eqref{eq:L_operator_discretized} to $JT = 1.6$. 
This value ensures that the bulk of the Gaussian filter $g(t)$, Eq.~\eqref{eq:g_gaussian}, is captured (i.e. $\big(\int_{-\infty}^\infty \diff t\, g(t) - \int_{-T}^T \diff t\, g(t)\big)/\int_{-\infty}^\infty \diff t\, g(t) \approx 10^{-7}$).
Furthermore, we fix the maximum simulation time to $Jt = 500$, which is sufficiently large for the dynamics to converge to its steady state $\rho_\infty^{\rm circ}$ for the chosen Hamiltonian parameter point (cf. Fig.~\ref{fig:time_evolution_random_k_body_Pauli_k_order2}).
Thus, the number $M$ of evolution steps, Eq.~\eqref{eq:lindblad_random}, is given by $M = 500 / J\delta t$ for fixed step size $\delta t$.
Moreover, we choose $|\bm{A}| = 10$ random jump operators and average all simulations over $10$ repetitions.

To investigate the applicability of our theoretical bound~\eqref{eq:error_sources2}, we
compute the distances $\|\rho_\infty^{\rm circ} - \sigma_\beta\|_1$ on a two-dimensional discretized grid of $(\delta t, \Delta t)$-values in the range $10^{-2} \leq J\delta t \leq 10^1$ and $0.06 \leq J\Delta t \leq 1$.
We fit this data with the function
\begin{align}
\label{eq:algorithmic_errors_fit}
\begin{split}
    &f_{\alpha_1, \alpha_2, \alpha_3, \alpha_4} (\delta t, \Delta t) = 
    \\
    & 
    \alpha_1 + \alpha_2 \delta t + \alpha_3 T\frac{\Delta t^2}{\delta t} + \alpha_4 \sqrt{\beta} |B_H|\e^{-\frac{1}{8}\big(2\pi\frac{\beta}{\Delta t}-2\beta\|H\| - 1 \big)^2},
\end{split}
\end{align}
which captures the error scaling of Eq.~\eqref{eq:error_sources2} for fixed evolution time.
Details are given in App.~\ref{app:circ_sim_alg_errors_fit}.
Since we use a large integration window with $JT = 1.6$, the truncation error term $\propto\e^{-2(T/\beta)^2}$ in Eq.~\eqref{eq:error_sources2} is negligible relative to the other error sources, and we drop its contribution in Eq.~\eqref{eq:algorithmic_errors_fit}.
The parameter $\alpha_1$ has been introduced to account for the convergence accuracy, Eq.~\eqref{eq:upper_bound_trace_distance}, of the ideal Lindblad evolution.

\paragraph*{Results ---}
In Fig.~\ref{fig:algorithmic_error_scaling_OFT_discretization_and_trotter_step} (left panel) we report the distances $\|\rho_\infty^{\rm circ} - \sigma_\beta\|_1$ (indicated as dots) as a function of the evolution step size $\delta t$ for several values of the OFT discretization step size $\Delta t$ (colors in the legend).
We observe that the fitted function $f$ [Eq.~\eqref{eq:algorithmic_errors_fit}, solid lines] captures the approximately polynomial error increase for large $\delta t$ ($J\delta t \gtrsim 1$), although with larger slope than the data.
For small $\delta t$ ($J\delta t \lesssim 0.1$) there are two distinct behaviors. If $\Delta t$ is sufficiently large, the third term of Eq.~\eqref{eq:algorithmic_errors_fit} (proportional to $\alpha_3$) is large and its $1/\delta t$ dependence counteracts the $\delta t$ dependence of the second term.
In this regime, the error can increase with decreasing evolution step size $\delta t$, which is visible
in the fit corresponding to $J\Delta t = 0.26$ (left end of the light blue line).
The circuit simulations data (dots) partially reflects this behavior: for $J\Delta t \geq 0.16$ the error decreases only mildly with decreasing $\delta t$.
The contribution from the third term of Eq.~\eqref{eq:algorithmic_errors_fit} can be suppressed for sufficiently small $J\Delta t \lesssim 0.1$, which results in an overall error decay with $\delta t$ according to the second term.
Note that the observed error (dots) decreases faster than predicted by the fit in this regime, which is expected since the function we fit is an upper bound.

In Fig.~\ref{fig:algorithmic_error_scaling_OFT_discretization_and_trotter_step} (right panel) we show the error dependence on the OFT discretization step $\Delta t$ for four values of the evolution step $\delta t$ (given in the legend).
Several characteristic regimes can be identified.
For small $J\delta t = 0.01$ (dark blue dots) and $J\Delta t \lesssim 0.25$, the third term of Eq.~\eqref{eq:algorithmic_errors_fit} controls the weak decrease of $\|\rho_\infty^{\rm circ} - \sigma_\beta \|_1$ as $\Delta t$ decreases.
For larger OFT discretization step, the fourth term of Eq.~\eqref{eq:algorithmic_errors_fit} dominates, leading to a rapid increase of the error for $0.25 \lesssim J\Delta t \lesssim 0.4$.
A qualitatively different behavior is observed for larger $J\delta t \geq 1$. In this case, the second term of Eq.~\eqref{eq:algorithmic_errors_fit} exceeds the third term (dependent on $\Delta t$), and the error plateaus and becomes independent of $\Delta t$ for $J\Delta t \lesssim 0.25$.
A more detailed discussion of the individual error contributions is given in App.~\ref{app:circ_sim_alg_errors_decomposition}.

Overall, our circuit simulations show that our analytical bounds capture the different algorithmic error sources reasonably well and allow us to understand the individual contributions.
In the noiseless case, we can control the algorithmic errors by reducing evolution step $\delta t$ and OFT discretization step $\Delta t$, at the expense of increased circuit depth scaling polynomially in the inverse error, cf. Eq.~\eqref{eq:runtime_main}.
However, when executing on noisy quantum computers the increased circuit depth generally will lead to larger errors induced by noise.
These two trends, namely the reduction of algorithmic errors and the increase of noise-induced errors, counteract each other and need to be balanced for an optimal overall accuracy of the prepared state. We study this interplay in the next subsection.

\subsection{Local depolarizing noise}
\label{sec:local_depolaizing_noise_simulations}

To complete our error analysis of the randomized single-ancilla Gibbs state preparation protocol, we perform circuit simulations with local depolarizing noise.
To reflect the realistic implementation on hardware, the circuits are first compiled to the native gate set of the Quantinuum's H1 architecture using the TKET compiler~\cite{Sivarajah_2021} with optimization level $2$.
Furthermore, after each 2-qubit gate in the compiled circuit, we apply a noise channel
\begin{align}
\label{eq:main_local_depo}
    \Lambda_{\rm 2Q}[X] \coloneq (1-\lambda_{\rm g}) X + \lambda_{\rm g} \Tr_{\rm 2Q}[X] \otimes \frac{I_{\rm 2Q}}{4},     
\end{align}
controlled by the error probability $\lambda_{\rm g} \in [0,1]$ and where $\Tr_{\rm 2Q}$ indicates tracing out the 2-qubit Hilbert space the gate acts on, while $I_{\rm 2Q}$ is the identity operator on such space.
While here we consider a local depolarization noise channel~\eqref{eq:main_local_depo}, we note that the bounds of Eq.~\eqref{eq:main_dk} would evaluate to the same value for other stochastic noise channels with the same probability of no-error $1-\lambda_g$.
We choose an OFT discretization step size $J\Delta t = 0.2$. As in Section~\ref{sec:alg_errors}, all simulations are performed for the mixed-field Ising model~\eqref{eq:ising_model} with $n=5$ qubits,
$\beta=(2J)^{-1}$,
cutoff time $JT=1.6$, $|\bm{A}| = 10$, maximum simulation time $Jt = 500$, and we average all simulations over $10$ repetitions.

Results of the circuit simulations with noise are displayed in Fig.~\ref{fig:trace_vs_noise}. 
Denoting as $\tilde\rho_\infty^\text{circ}$ the output of the noisy quantum circuit simulation, we evaluate the trace distance $\|\tilde\rho_\infty^\text{circ} - \sigma_{\beta}\|_1$ to the Gibbs state as a function of the noise parameter $\lambda_{\rm g} \in [10^{-6}, 10^{-4}]$. This distance quantifies the overall error, including both algorithmic and noise contributions, in the preparation of the target Gibbs state. Results (indicated as circles) for different values of the evolution steps $J\delta t = 1, 3$ and $5$ are reported (colors in legend). As can be seen in the figure, for noise strengths $\lambda_{\rm g} < 10^{-5}$, lower values of $\delta t$ systematically yield smaller trace distances. As the noise increases, however, larger evolution steps result in smaller distances. This highlights trade-offs between algorithmic and noise errors: while the algorithmic error scales with the size of the evolution steps, the effect of the noise decreases with it, such that in certain noise regimes adopting larger evolution steps becomes beneficial. 

To validate the noise analysis of Sec.~\ref{sec:noise}, we evaluate the bounds $\tilde{B}_\infty(N_{\rm g})$ derived in Eq.~\eqref{eq:main_asympt_distance} for stochastic noise. 
These bounds are computed using an error rate $\lambda=1-(1-\lambda_{\rm g})^{N_{\rm g}}$ for $N_{\rm g}$ the number of 2-qubit gates per evolution step. For $J\delta t=\{1, 3, 5\}$, the gate counts per step of the compiled circuits are $N_{\rm g}=\{308, 484, 644\}$, respectively.
To be comparable to the trace distance, the bounds $\tilde{B}_\infty(N_{\rm g})$ are shifted by the algorithmic error $d_0$ that is obtained in the noiseless scenario (i.e. by setting the noise strength $\lambda_{\rm g}=0$) for each of the values of $\delta t$ probed. 
As seen in Fig.~\ref{fig:trace_vs_noise}, $\tilde{B}_\infty(N_{\rm g}) + d_0$ (dotted lines) always upper bounds the overall errors.
Furthermore, the trade-off between algorithmic and noise errors, whereby larger evolution steps can incur smaller overall errors, is also captured by these bounds, albeit at slightly shifted values of $\lambda_{\rm g}$.

Finally, we perform a fit of the errors obtained with the bounds $\tilde{B}_\infty(N_{\rm eff})$ from Eq.~\eqref{eq:main_asympt_distance} for an \emph{effective} 2-qubit gate count $N_{\rm eff}$, rather than the true number of $2$-qubit gates $N_{\rm g}$. 
For $J\delta t=\{1, 3, 5\}$, we obtain fitted $N_{\rm eff}=\{135, 193, 247\}$ respectively. The resulting bounds, again shifted by the algorithmic errors $d_0$, are reported as solid lines. This shows that, for the noise model simulated, the functional dependence of the bounds $\tilde{B}_\infty$ captures the actual errors remarkably well.  

Overall, these simulations allow us to validate the noise analysis of Sec.~\ref{sec:noise}. 
We show that, while overestimating the actual noise, the bound $\tilde{B}_\infty$ already captures important trade-offs between algorithmic and noise-induced errors. Considering these trade-offs will be crucial when aiming to run specific circuits on quantum hardware. 
Furthermore, we see that the bound captures the actual errors remarkably well using an effective number of gates smaller than the real number of gates. 
Understanding this reduction in relation to the prior works on noise-resilience of quantum simulation~\cite{Childs2001, Roland2005, Trivedi2022, Kechedzhi2023, Kashyap2024, Schiffer2024, Granet2024dilution, Chertkov2024} is left to future work.

\begin{figure}[t]
    \centering
    \includegraphics[width=0.48\textwidth]{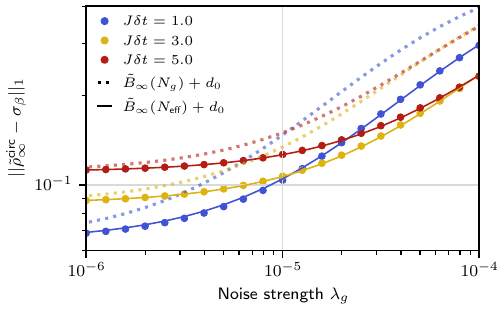}
    \caption{
    Circuit simulations with depolarizing noise applied to the 2-qubit gates for $n=5$ qubits
    and $\beta=(2J)^{-1}$.
    The distance $\|\tilde{\rho}^{\rm circ}_\infty- \sigma_\beta\|_1$ between the target Gibbs state and the output of the noisy circuits (circles) is reported as a function of the noise strength $\lambda_{\rm g}$ in Eq.~\eqref{eq:main_local_depo} for different values of the evolution step $\delta t$ (colors in legend). 
    Further details regarding the parameters used for the circuit construction are provided in the main text.
    These distances account for both algorithmic and noise errors in the preparation of the Gibbs state.
    In addition, we report the bounds $\tilde{B}_\infty$ on the noise errors from Eq.~\eqref{eq:main_asympt_distance}, shifted by the algorithmic error $d_0$, which is evaluated for $\lambda_g = 0$.
    We report these shifted bounds both for the actual number of gates $N_{\rm g}$ used in the compiled circuits (dashed lines) and for an \emph{effective} number of gates $N_{\rm eff}$ (solid lines), which is obtained by fitting the bound $\tilde{B}_\infty + d_0$ to the data. Both $N_g$ and the fitted values $N_{\rm eff}$ are provided in the main text for all the $\delta t$ studied.
    }
    \label{fig:trace_vs_noise}
\end{figure}

\section{Conclusion}

Our contribution bridges the gap between the predominantly theory-driven literature on dissipative quantum Gibbs state preparation algorithms, and the engineering of concrete Lindbladians and protocols that exhibit good convergence in practice, with circuits implementable on foreseeable quantum hardware.
To this end, we combine several recent developments in the literature \cite{Chen2023thermal, Chen2023efficient, Ding2023, Ding2024, Chen2024random, Chen2021ETH} and propose a variant of this class of algorithms with reduced implementation cost (Secs.~\ref{sec:ETH} and~\ref{sec:single_ancilla_protocol}).
Through a numerical analysis (Sec.~\ref{sec:numerics}), we establish the crucial influence of the dynamical properties of the underlying system Hamiltonian and the Lindblad operators on mixing time and convergence accuracy.
In line with our theoretical analysis, we observe a weak, almost linear, polynomial system-size scaling of the mixing time for systems obeying the ETH.
In contrast, the Lindblad dynamics exhibits vastly different convergence characteristics in the non-chaotic limits of our model.

In realistic scenarios, the overall Gibbs state preparation error is controlled by an interplay of algorithmic and hardware-induced errors---for example smaller evolution steps in Eq.~\eqref{eq:lindblad_random} reduce the Trotterization error but, at the same time, increase the gate count of the full circuit. 
To obtain a good understanding of those error sources,
we investigate their impact both on a theoretical level and based on circuit simulations (Secs.~\ref{sec:single_ancilla_protocol}, \ref{sec:noise} and~\ref{sec:noisy_simulation}).
We numerically demonstrate that our theoretical error bounds provide a good description of the actual errors in noiseless circuit simulations.
Further, we show that the Lindblad dynamics can exhibit an inherent resilience against incoherent stochastic noise, due to the presence of a nontrivial steady state of the noisy dissipative dynamics.
Unlike in the unitary case, noise-induced errors at earlier times can be damped significantly through the dissipative character of the dynamics, which is a promising observation for the successful demonstration of this class of algorithms on near-term hardware.
We support our theoretical noise analysis with circuit simulations considering a local depolarizing noise model.

The successful demonstration of dissipative quantum Gibbs state preparation algorithms on foreseeable hardware requires a deep understanding of many contributing factors---dynamical and algorithmic design properties (Hamiltonian, Lindblad operators, initial state, algorithmic parameters), and of the expected hardware errors.
This work lays an important foundation to achieve this goal. Future work needs to address the design of suitable initial states implementable with shallow circuits, and a detailed account of the trade-offs between algorithmic and noise-induced errors for realistic hardware-specific noise models.
To further mitigate the influence of noise, one may design the Lindbladian to take into account and counteract pre-characterized stochastic noise channels.
Furthermore, the influence of coherent errors, and potential routes to convert those into stochastic noise similar to randomized compiling~\cite{wallman_noise_2016,hashim_randomized_2021}, need to be explored.
Finally, to assess the performance of the algorithm, efficient ways to certify the prepared Gibbs state, e.g. based on recent Hamiltonian learning results~\cite{bakshi_learning_2023,Rajakumar2024}, are required.

\begin{acknowledgments}

We thank Daniel Stilck Fran\c{c}a, Tomoya Hayata, and Maria Tudorovskaya for insightful discussions. We thank Marcello Benedetti and Henrik Dreyer for their feedback on this manuscript.

\end{acknowledgments}

\onecolumngrid
\clearpage
\appendix

\section{Protocol details and relation to the CKG algorithm}
\label{app:chen_protocol}

After briefly reviewing the quantum Gibbs sampling algorithm proposed in~\cite{Chen2023thermal,Chen2023efficient}, we clarify its relation to our protocol.

\subsection{CKG Quantum Gibbs sampling algorithm}

The CKG algorithm~\cite{Chen2023efficient} uses the Lindbladian
\begin{align}
\label{eq:lindbladian_CKG}
\begin{split}
    \calL_\mathrm{CKG}[\rho] 
    &= 
    - \im[G,\rho]
    + \sum_{a\in\bm{A}}\int_{-\infty}^{\infty}\diff\omega\, \gamma(\omega) \Big(L^a(\omega)\rho L^{a}(\omega)^\dag -\frac{1}{2}\{L^{a}(\omega)^\dag L^a(\omega),\rho\}\Big)
    \\
    &= 
    - \im[G,\rho]
    + \sum_{a\in\bm{A}}\sum_{\nu_1,\nu_2} \alpha_{\nu_1,\nu_2}\Big(A^a_{\nu_1}\rho (A^{a}_{\nu_2})^\dag -\frac{1}{2}\{(A^a_{\nu_1})^\dag A^{a}_{\nu_2},\rho\}\Big),
    \\
    L^a(\omega) 
    &=
    \sum_{\nu\in \bohrH}\hat{g}_{\rm CKG}(\nu-\omega) A_\nu,
\end{split}
\end{align}
where $\alpha_{\nu_1,\nu_2}\coloneq\int\diff\omega \gamma(\omega)\hat{g}_{\rm CKG}(\nu_1-\omega)\hat{g}_{\rm CKG}(\nu_2-\omega)$ and $\bohrH$ is the set of Bohr frequencies of the Hamiltonian $H$ of the target Gibbs state $\sigma_\beta$.
With the choice of transition weight $\gamma(\omega)= \e^{-\frac{(\omega+\omega_\gamma)^2}{2\Delta_\gamma^2}}$ and frequency-domain filter function $\hat{g}_{\rm CKG}(\omega) = \frac{1}{(2\pi\Delta_\mE^2)^{1/4}}\e^{-\frac{\omega^2}{4\Delta_\mE^2}}$ for parametrization satisfying $\beta=\frac{2\omega_\gamma}{\Delta_\mE^2 + \Delta_\gamma^2}$, one can readily confirm the identity,
\begin{align}
\label{eq:DB_alpha}
    \alpha_{\nu_1,\nu_2}
    =
    \e^{-\beta(\nu_1+\nu_2)/2}\alpha_{-\nu_1,-\nu_2}.
\end{align}
This identity ensures that the transition term $\sum_{a\in\bm{A}}\sum_{\nu_1,\nu_2} \alpha_{\nu_1,\nu_2}A^a_{\nu_1}[\cdot] (A^{a}_{\nu_2})^\dag$ obeys the $\sigma_\beta$-DB condition~\eqref{eq:KMS}, while the decay term $-\frac{1}{2}\sum_{a\in\bm{A}}\sum_{\nu_1,\nu_2}\{L^{a}(\omega)^\dag L^a(\omega),\cdot\}$ does not.
To enforce the $\sigma_\beta$-DB condition for the full Lindbladian $\calL_\mathrm{CKG}[\cdot]$, the authors in~\cite{Chen2023efficient} designed the coherent term as
\begin{align}
    G 
    =
    G_{\rm CKG}
    \coloneq
    \frac{\im}{2}\sum_{a\in\bm{A}}\sum_{\nu_1,\nu_2}\alpha_{\nu_1,\nu_2} \tanh\Big(\frac{\beta(\nu_1-\nu_2)}{4}\Big) (A^a_{\nu_1})^\dag A^{a}_{\nu_2},
\end{align}
which, combined with the decay term, obeys the $\sigma_\beta$-DB condition, and thus, so does the Lindbladian $\calL_\mathrm{CKG}$,
\begin{align}
    \langle X, \calL_\mathrm{CKG}^\dag[Y]\rangle_{\sigma_\beta}
    =
    \langle \calL_\mathrm{CKG}^\dag[X], Y\rangle_{\sigma_\beta}
\end{align}
for any bounded operators $X, Y$.

\subsection{Relation to the present protocol}
\label{app:relation_chen}

In the present work, we choose a $\delta$-function for the transition weight
\begin{align}
\label{eq:delta_lindblad_filter}
    \gamma(\omega)
    =
    \delta(\omega + \omega_\gamma)
    \quad \text{with}\quad 
    \omega_\gamma 
    = 
    \beta\Delta_\mE^2/2.
\end{align}
With such a choice, the coefficients $\alpha_{\nu_1,\nu_2}$ factorize as
\begin{align}
\label{eq:eta_lindblad_filter}
    \alpha_{\nu_1,\nu_2} 
    = 
    \eta_{\nu_1}\eta_{\nu_2}
    \quad\text{with}\quad
    \eta_{\nu}
    \coloneq
    \hat{g}_{\rm CKG}(\nu+\omega_\gamma)
    =
    \frac{1}{(2\pi\Delta_\mE^2)^{1/4}} \e^{-\frac{(\nu+\beta\Delta_\mE^2/2)^2}{4\Delta_\mE^2}},
\end{align}
and satisfy the identity~\eqref{eq:DB_alpha}.
Furthermore, the Lindbladian~\eqref{eq:lindbladian_CKG} is reduced to
\begin{align}
\label{eq:lindbladian_CKG_factorize}
    \calL_\mathrm{CKG}[\rho] 
    &= 
    - \im[G,\rho]
    + \gamma\sum_{a\in\bm{A}}\Big(L^a\rho L^{a\dag} -\frac{1}{2}\{L^{a\dag} L^a,\rho\}\Big),
    \\
    L^a
    &= 
    \int_{-\infty}^{\infty}\diff t\,g(t)A^a(t)
    =
    \sum_{\nu\in \bohrH}\eta_\nu A^a_\nu,
    \\
    \label{eq:coherent_G_CKG}
    G
    &=
    G_{\rm CKG} 
    =
    \frac{\im}{2}\sum_{a\in\bm{A}}\sum_{\nu_1,\nu_2 \in \bohrH}\eta_{\nu_1}\eta_{\nu_2} \tanh\left(\frac{\beta(\nu_1-\nu_2)}{4}\right) (A^a_{\nu_1})^\dag A^{a}_{\nu_2},
\end{align}
with $A^a(t)\coloneq \e^{\im Ht}A^a\e^{-\im Ht}$, $\eta_\nu\coloneq\int_{-\infty}^{\infty}\diff t\,\e^{\im\nu t}g(t)$, $A^a_\nu\coloneq\sum_{E_i-E_j=\nu}\Pi_iA^a\Pi_j$ and $g(t)$ as defined in Eq.~\eqref{eq:g_gaussian}.
Equation~\eqref{eq:lindbladian_CKG_factorize} is identical to our Lindbladian~\eqref{eq:lindbladian_general} except for the coherent term involving $G_{\rm CKG}$ and the fact that transition rates $\gamma_a$ in Eq.~\eqref{eq:lindbladian_general} formally depend on $a$ because subsequently we write $\gamma_a=\gamma p_a$ and sample individual Lindblad operators according to the discrete distribution $p_a$.

Next, we show that the coherent part~\eqref{eq:coherent_G_CKG} vanishes under the ETH average, i.e. $\bE_R[G_{\rm CKG}]=0$.
To this end, we assume that the jump operators $A^a$ satisfy the ETH~\eqref{eq:ETH},
\begin{align}
\label{eq:app_ETH}
    \bra{E_i}A^a\ket{E_j} 
    =
    \calA^a(E_{i})\delta_{ij} + \frac{f^a(E_{ij},\nu_{ij})}{\sqrt{D(E_{ij})}}R^a_{ij},
\end{align} 
where the random variables $R^a_{ij}$ are independent and satisfy $\bE_R[R^a_{ij}]=0$ together with $\bE_R[|R^a_{ij}|^2]=1$, and $D(E)$ is the density of states. Hence, the ETH average of a product of jump operators is given by
\begin{equation}\label{eq:deriv_vanishing_G_ETH}
    \begin{split} 
        \mathbb{E}_R[(A^a_{\nu_2})^\dagger A_{\nu_1}^a]
        &= 
        \sum_{\substack{E_k - E_i = \nu_2 \\ E_k - E_j = \nu_1}} \delta_{ki}\delta_{kj} \calA^a(E_{i}) \calA^a(E_{k}) \ket{E_i}\bra{E_j}
        + 
        \frac{f^a_{ki} f^a_{kj}}{\sqrt{D(E_{ki})D(E_{kj})}} \mathbb{E}_R [R^{a\ast}_{ki} R^a_{kj}] \ket{E_i}\bra{E_j}
        \\
        &= 
        \sum_{\substack{E_k - E_i = \nu_2 \\ E_k - E_j = \nu_1}} \ket{E_i}\bra{E_i}\Big(
        \delta_{\nu_1,0}\delta_{\nu_2,0} (\calA^a(E_i))^2
        + 
        \delta_{\nu_1, \nu_2} \frac{|f^a_{ki}|^2}{D(E_{ki})}
        \Big).
    \end{split}
\end{equation}
In the first line, we use the fact that the cross terms involving $\calA^a$ and $R^a$ are first-order in the random matrix elements $R^a_{ij}$ and, hence, vanish in the ETH average.
For the second line note that the $\delta_{ki}\delta_{kj}$ in the first term enforces $E_k = E_i = E_j$, such that $\nu_1 = \nu_2 = 0$.
Moreover, $\mathbb{E}_\mathrm{R} [R^{a\ast}_{ki} R^a_{kj}]=\delta_{ij}$, which enforces $E_i = E_j$ and, thus, $\nu_1 = \nu_2$.
Importantly, both of these contributions are zero whenever $\nu_1 \neq \nu_2$.
From this, and the fact that $\tanh(0)=0$, it directly follows that the coherent term~\eqref{eq:coherent_G_CKG} vanishes under the ETH:
\begin{equation}
    \begin{split}
        \mathbb{E}_R[G]
        \propto \frac{\im}{2}\sum_{\nu\in \bohrH}\tanh\left(\frac{\beta\nu}{4}\right)
        \sum_{\nu_1 - \nu_2 = \nu} \eta_{\nu_1} \eta_{\nu_2} \delta_{\nu_1, \nu_2} 
        = 
        0.
    \end{split}
\end{equation}

The above argument confirms that the average dissipative Lindbladian $\bE_R \calD$, with
\begin{align}
\label{eq:lindbladian_dissipative}
    \calD[\rho] 
    &= 
    \sum_{a\in\bm{A}}\gamma_a \Big(L^a\rho L^{a\dag} -\frac{1}{2}\{L^{a\dag}L^a,\rho\}\Big).
\end{align}
[cf. Eq.~\eqref{eq:lindbladian_general}], is $\sigma_\beta$-DB without introducing the coherent term $-\im[G_{\rm CKG}, \cdot]$, as we saw this in Sec.~\ref{sec:convergence_under_ETH}.
Since the Gibbs state $\sigma_\beta$ commutes with $H$, a coherent term $-\im[G, \cdot]$ generated by the system Hamiltonian $G=H$ does not change the steady state, and $\bE_R\calL =  - \im [H, \cdot] + \bE_R\calD$ remains $\sigma_\beta$-DB under the ETH average under a slightly generalized definition of detailed balance.
To see this, note that, in general, $(\calL^\dagger)^{\dagger_\mathrm{KMS}} = \sigma_\beta^{-1/2} \calL [\sigma_\beta^{1/2} \, \cdot \, \sigma_\beta^{1/2}] \sigma_\beta^{-1/2}$ is the adjoint of $\calL^\dagger$ (the Hilbert-Schmidt adjoint of $\calL$) with respect to the KMS inner product $\expval{\cdot, \cdot}_{\sigma_\beta}$.
In our case both generators are not equal, as required by our KMS detailed balance condition~\eqref{eq:KMS}.
Their difference (under the ETH average) is given by $\bE_R\calL^\dagger - (\bE_R\calL^\dagger)^{\dagger_\mathrm{KMS}} = 2\im [H, \cdot]$, as can be readily confirmed. This is a special case of the more general version of quantum detailed balance introduced in \cite{Fagnola2007} (Def.~27).
In fact, it is easy to see [similar as in Eq.~\eqref{eq:self_adjoint_and_steady_state}], that under this condition
\begin{align}
\begin{split}
    \langle \bE_R\calL[\sigma_\beta], X\rangle_{\rm HS}
    &= \expv{I, \bE_R\calL^\dag[X]}_{\sigma_\beta}
    = \expv{(\bE_R\calL^\dag)^{\dag_\mathrm{KMS}} [I], X}_{\sigma_\beta}
    = \expv{\bE_R\calL^\dag[I] - 2\im [H, I], X}_{\sigma_\beta}
    = 0 
\end{split}
\end{align}
for all $X$, which confirms that $\sigma_\beta$ is the steady state of $\bE_R\calL$.

Thus, by replacing $G_\mathrm{CKG}$ in Eq.~\eqref{eq:lindbladian_CKG_factorize} with the system Hamiltonian $H$, we end up with our Lindbladian~\eqref{eq:lindbladian_general}. 
For our protocol, we set $\Delta_\mE=\sqrt{2}/\beta$ (see App.~\ref{app:ETH_scales} where we elaborate on this choice of parameter), which results in the filter function [cf. Eq.~\eqref{eq:g_gaussian}]
\begin{align}
\label{eq:Gibbs_filter_appA}
    g(t)
    =
    \frac{1}{\pi^{3/4}\beta^{1/2}}
    \e^{\im t/\beta}\e^{-2t^2/\beta^2}.
\end{align}
With this, our Lindblad operators take the form [cf. Eq.~\eqref{eq:L_operators}]
\begin{align}
\label{eq:lindblad_OFT_gaussian_appA}
    L^a 
    = 
    \int_{-\infty}^{\infty}\diff t\,g(t)A^a(t)
    = 
    \left(\frac{\beta^2}{4\pi}\right)^{1/4}
    \sum_{\nu\in \bohrH}\e^{-\frac{(\beta\nu+1)^2}{8}}
    A^a_\nu = \sum \eta_\nu A^a_\nu,
\end{align}
with the frequency-domain filter function 
\begin{align}
\label{eq:eta_nu}
    \eta_\nu
    =
    \left(\frac{\beta^2}{4\pi}\right)^{1/4}
    \e^{-\frac{(\beta\nu+1)^2}{8}}.
\end{align}

\section{Details of single-ancilla Gibbs state preparation protocol}
\label{app:single_gibbs}

We provide details on the implementation of the steps of Lindblad evolution that are core to our Gibbs-state preparation protocol by closely following~\cite{Ding2023}, where the authors employed a single-ancilla protocol to prepare ground states.
A thorough error analysis of each approximation necessary for a quantum circuit implementation is provided (Sec.~\ref{app:discretization} to Sec.~\ref{app:integral_error}).
A summary of these contributions is reported in Sec.~\ref{app:alg_errors} together with the overall circuit complexity,
which we state here as a main result:
\begin{theorem}\label{thm:qa_complexity}
    Let $\calL$ be the Lindbladian given by Eq.~\eqref{eq:lindbladian_app} with the Lindblad operators~\eqref{eq:lindblad_OFT_gaussian}. For simulation time $t$ and inverse temperature $\beta$, there exists a quantum algorithm that simulates a Lindblad dynamics $\e^{\calL t}[\cdot]$ with an error $\epsilon$ in trace distance, whose cost in terms of the Hamiltonian simulation time is given by $\Theta\big(\frac{\beta t^2}{\epsilon}\sqrt{\log\frac{\beta t}{\epsilon}}\big)$.
\end{theorem}

The theorem is proven in the remainder of this App.~\ref{app:single_gibbs}. See Eq.~\eqref{eq:runtime} for the final result.
A decomposition of the algorithmic error into individual contributions is given in Prop.~\ref{prop:algorithmic_error} at the end of this appendix.
In the following, we aim at implementing the evolution $\e^{t\calL}$ for a Lindbladian $\calL$ over a simulation time $t=M \delta t$.
The Lindbladian~\eqref{eq:lindbladian_general} with $G=H$ can be recast as
\begin{equation}\label{eq:lindbladian_app}
    \calL[\rho] = \sum_{a\in\bm{A}} p_a \left( - \im[H,\rho]
    + \gamma \Big(L^a\rho L^{a\dag} -\frac{1}{2}\{L^{a\dag}L^a,\rho\}\Big)\right),
\end{equation}
in terms of a discrete probability distribution over the jump operator indices $\set{p_a}_{a \in \bm{A}}$ and where the strength $\gamma$ of the dissipative terms has been renormalized accordingly.
For the purpose of implementation, we approximate the evolution under such Lindbladian through
\begin{align}
\label{eq:lindblad_random_app}
    \e^{M \delta t\calL}
    &\approx
    \left(\prod_{i=1}^{M}
    \e^{\delta t \gamma\calD^{a_i}}
    \circ \calU_{\delta t}\right)
 \text{where} \quad
    \calD^{a}[\rho]
    \coloneq
    L^a\rho L^{a\dag} - \frac{1}{2}\{L^{a\dag} L^a, \rho\},\quad 
    \text{and} \quad 
    \calU_{\delta t}[\rho] 
    \coloneq 
    \e^{-\im \delta t H}\rho\e^{\im \delta t H},
\end{align}
At each evolution step we randomly sample a Lindblad operator yielding a set $\set{L^{a_i}}_{i=1, \hdots, M}$ where each of the indices $a_i$ has been drawn with probability $p_{a_i}$.
The approximation in Eq.~\eqref{eq:lindblad_random_app} holds upon taking this average.
Using the Taylor expansion, one can see that, on average, such sampling incurs errors that scale with $\delta t^2$ as we have
\begin{align}
\label{eq:app_random_average}
\begin{split}
    \sum_a p_a\e^{\delta t\gamma\calD^{a}} \circ \calU_{\delta t}[\rho]
    &= \rho + \delta t\calL[\rho] + \calO(\delta t^2)
    = \e^{\delta t\calL}[\rho]  + \calO(\delta t^2).
\end{split}
\end{align}
While the second approximation holds upon taking the average over trajectories with randomly sampled jump operators, it is expected to hold even with a single trajectory if the number of steps $M$, at fixed evolution time $t$, is sufficiently large. The authors in~\cite{Chen2024random} have shown that the individual trajectory approximates the target dynamics with an error measured in weighted $l^2$-distance inversely proportional to $M$ (Theorem~7 in \cite{Chen2024random}).

\subsection{Regularization of the integral in the operator Fourier transforms}
\label{app:discretization}

Recall from Eq.~\eqref{eq:L_operators}, that the individual Lindblad operator $L^a$ can be expressed as an operator Fourier transform via the integral
\begin{align}
\label{eq:L_operators_app}
    L^a 
    = 
    \int_{-\infty}^{\infty}\diff t\,g(t)A^a(t),
    \quad \text{where}\quad 
    A^a(t)\coloneq \e^{\im Ht}A^a\e^{-\im Ht}.
\end{align}
To realize such an integral as a quantum circuit, we need to regularize it. This is achieved by restricting its domain to $[-T,T]$ and discretizing it using the trapezoidal rule. 
For a step size $\Delta t\coloneq T/S$, we obtain
\begin{align}
\label{eq:L_int_discrete}
\begin{split}
    \bar{L}^a
    &\coloneq
    \frac{\Delta t}{2}g_{-S}\e^{-\im HS\Delta t} A^a\e^{\im HS\Delta t}
    +\sum_{s=-S+1}^{S-1} \Delta t g_s\e^{\im Hs\Delta t} A^a\e^{-\im Hs\Delta t}
    +\frac{\Delta t}{2} g_S\e^{\im HS\Delta t} A^a\e^{-\im HS\Delta t}
    \\
    &=
    \sum_{s=-S}^{S} \Delta t_s g_s\e^{\im Hs\Delta t} A^a\e^{-\im Hs\Delta t},
\end{split}
\end{align} 
where we have defined $g_s\coloneq g(s\Delta t)$ and $\Delta t_s\coloneq \Delta t$ for $-S+1\le s\le S-1$ or $\Delta t_s\coloneq\Delta t/2$ for $s=\pm S$.

Accordingly, the Hermitian dilation operator $K^a$~\eqref{eq:dilation} corresponding to $L^a$ becomes
\begin{align}
\label{eq:discrete_K}
\begin{split}
    \bar{K}^a
    &
    \coloneq 
    \ket{1}\bra{0}_\ma\otimes \bar{L}^a + \ket{0}\bra{1}_\ma\otimes \bar{L}^{a\dag}
    =
    X_\ma\otimes \frac{\bar{L}^a+\bar{L}^{a\dag}}{2} - \im Y_\ma\otimes \frac{\bar{L}^a-\bar{L}^{a\dag}}{2}
    \\
    &=
    \sum_{s=-S}^{S}\Delta t_s
    (\mathrm{Re}[g_s]X_\ma +\mathrm{Im}[g_s]Y_\ma)  \otimes \e^{\im Hs\Delta t} A^a\e^{-\im Hs\Delta t}
    \\
    &\eqcolon
    \sum_{s=-S}^{S} \bar{K}^a_s,
\end{split}
\end{align}
where $A^a$ is assumed to be Hermitian.
Errors entailed by the discretization of the jump operators will be quantified in Sec.~\ref{app:integral_error}, but for now we proceed with the implementation of the Lindblad evolution.

\subsection{Second-order product formula for Lindblad evolution}

Having defined the discretized Lindblad $\bar{L}^a$ and dilation operators $\bar{K}^a$, we aim at implementing a step of dissipative evolution $\e^{\delta t \gamma\calD^{a}}$ appearing in Eq.~\eqref{eq:lindblad_random_app}. Recall from Eq.~\eqref{eq:evolve_dilation}
that this can be achieved, up to an error scaling as $(\delta t \gamma)^2$, through the unitary evolution $\e^{-\im\sqrt{\delta t \gamma}K^a}$ acting on the system together with a single additional qubit:
\begin{align}
\label{eq:app_evolve_dilation}
    &\Tr_\mathrm{anc}[\e^{-\im \sqrt{\delta t \gamma} K^a}(\ket{0}\bra{0}_\ma\otimes \rho)\e^{\im \sqrt{\delta t \gamma} K^a}]
    =
    \e^{\delta t\gamma\calD^a}[\rho] + \calO((\delta t \gamma)^2).
\end{align}
This evolution is implemented for the discretized $\bar{K}_a$, as a second-order product formula yielding
\begin{align}
\label{eq:second_order_app}
    \e^{-\im\sqrt{\delta t \gamma}\bar{K}^a} 
    = 
    \e^{-\im\sqrt{\delta t \gamma}\sum_{s=-S}^{S} \bar{K}^a_s} 
    \approx 
    \prod_{s}^{\rightarrow}\e^{-\im\frac{\sqrt{\delta t \gamma}}{2}\bar{K}^a_s}
    \prod_{s}^{\leftarrow}\e^{-\im\frac{\sqrt{\delta t \gamma}}{2}\bar{K}^a_s},
\end{align}
where we  defined the ordered products $\prod_s^{\leftarrow}O(s)=O(S)\cdots O(-S+1)O(-S)$ and $\prod_s^{\rightarrow}O(s)=O(-S)\cdots O(S-1)O(S)$.
The resulting Trotter error is given by
\begin{align}
    \prod_{s}^{\rightarrow}\e^{-\im\frac{\sqrt{\delta t \gamma}}{2}\bar{K}^a_s}
    \prod_{s}^{\leftarrow}\e^{-\im\frac{\sqrt{\delta t \gamma}}{2}\bar{K}^a_s}
    -
    \e^{-\im\sqrt{\delta t \gamma}\bar{K}^a}
    =
    (\delta t \gamma)^{3/2} \sum_{s_1,s_2,s_3}c_{s_1s_2s_3}\bar{K}^a_{s_1}\bar{K}^a_{s_2}\bar{K}^a_{s_3}
    +
    \calO((\delta t \gamma)^2),
\end{align}
with some coefficients $c_{s_1s_2s_3}$.
Furthermore, noting that $(\bra{0}_\ma\otimes I)\bar{K}^a_{s_1}\bar{K}^a_{s_2}\bar{K}^a_{s_3}(\ket{0}_\ma\otimes I)=0$, we find
\begin{align}
\label{eq:app_dilated_evolution}
\begin{split}
    &\Tr_\ma\left[
        \prod_{s}^{\rightarrow}\e^{-\im\frac{\sqrt{\delta t \gamma}}{2}\bar{K}^a_s}
        \prod_{s}^{\leftarrow}\e^{-\im\frac{\sqrt{\delta t \gamma}}{2}\bar{K}^a_s} 
        (\ket{0}\bra{0}_\ma\otimes \rho)
        \prod_{s}^{\rightarrow}\e^{\im\frac{\sqrt{\delta t \gamma}}{2}\bar{K}^a_s}
        \prod_{s}^{\leftarrow}\e^{\im\frac{\sqrt{\delta t \gamma}}{2}\bar{K}^a_s}  
    \right]
    \\
    &=
    \Tr_\ma[\e^{-\im \sqrt{\delta t \gamma} \bar{K}^a}(\ket{0}\bra{0}_\ma\otimes \rho)\e^{\im \sqrt{\delta t \gamma} \bar{K}^a}]
    + \calO((\delta t \gamma)^2).
\end{split}
\end{align}

We simplify the evolution operator and estimate the evolution time required for each step of Lindblad evolution. Recalling that the evolution under a single $\bar{K}^a_s$ (out of a total of $S$) takes the form,
\begin{align}\label{eq:app_one_step_dillat}
\begin{split}
    \e^{-\im\frac{\sqrt{\delta t \gamma}}{2}\bar{K}^a_s}
    &=
    \exp\left[
        -\im\frac{\sqrt{\delta t \gamma}}{2}
        \Delta t_s
    (\mathrm{Re}[g_s]X_\ma +\mathrm{Im}[g_s]Y_\ma) \otimes \e^{\im Hs\Delta t} A^a\e^{-\im Hs\Delta t}
    \right]
    \\
    &=
    (I_\ma\otimes \e^{\im Hs\Delta t})
    \underbrace{\e^{-\im\frac{\sqrt{\delta t \gamma}}{2} \Delta t_s(\mathrm{Re}[g_s]X_\ma +\mathrm{Im}[g_s]Y_\ma) \otimes A^a}}_{\eqcolon B^a_s}
    (I_\ma\otimes \e^{-\im Hs\Delta t}),
\end{split}
\end{align}
we can rewrite Eq.~\eqref{eq:second_order_app} as
\begin{align}\label{eq:app_many_step_dillat}
\begin{split}
    \prod_{s}^{\rightarrow}\e^{-\im\frac{\sqrt{\delta t \gamma}}{2}\bar{K}^a_s}
    \prod_{s}^{\leftarrow}\e^{-\im\frac{\sqrt{\delta t \gamma}}{2}\bar{K}^a_s}
    &=
    (I_\ma\otimes \e^{-\im HS\Delta t}) 
    \underbrace{\left(\prod_{s}^{\rightarrow} B^a_s (I_\ma\otimes \e^{\im H\Delta t})\right)
    \left(\prod_{s}^{\leftarrow} (I_\ma\otimes \e^{-\im H\Delta t}) B^a_s\right)}_{\eqcolon\,V^a(\delta t)}
    (I_\ma\otimes \e^{\im HS\Delta t}).
\end{split}
\end{align}
Therefore, a single step of Lindblad evolution~\eqref{eq:lindblad_random_app}, which consists in a unitary part $\calU_{\delta t}[\cdot]=\e^{-\im H\delta t}[\cdot]\e^{\im H\delta t}$ and a dissipative part~\eqref{eq:app_dilated_evolution}, is given by the superoperator 
\begin{align}
\label{eq:W_evolve}
    &\calW^a(\delta t)[\rho]
    \nonumber\\
    &\coloneq
    \Tr_\ma[(I_\ma\otimes \e^{-\im HS\Delta t}) 
    V^a(\delta t)
    (I_\ma\otimes \e^{\im HS\Delta t})
    (\ket{0}\bra{0}\otimes\calU_{\delta t}[\rho])
    (I_\ma\otimes \e^{-\im HS\Delta t}) 
    V^a(\delta t)^{\dag}
    (I_\ma\otimes \e^{\im HS\Delta t})]
    \nonumber\\
    &\,=
    \e^{-\im HS\Delta t}\Tr_\ma[ 
        V^a(\delta t)
        (\ket{0}\bra{0}\otimes \e^{\im HS\Delta t}\e^{-\im H \delta t}\rho \e^{\im H 
        \delta t} \e^{-\im HS\Delta t})
        V^a(\delta t)^\dag
    ]\e^{\im HS\Delta t}.
\end{align}
Since we sequentially apply $M$ steps of evolution $\calW^a(\delta t)$, the operators $\e^{-\im HS\Delta t}$ and $\e^{\im HS\Delta t}$ cancel out, except at the first and last steps.
Still, we can show that these two steps can also be discarded. For that, suppose we evolve an initial state $\rho(0)$ for evolution time $M\delta t$ that is larger than the mixing time $t_{\rm mix}$.
By definition of the mixing time, any initial state $\rho$ evolved for a time larger than $t_{\rm mix}$ approximates the Gibbs state $\sigma_\beta=\e^{-\beta H}/\Tr[\e^{-\beta H}]$. That is,
\begin{align}
    \left(\prod_{i=1}^{M}\e^{\delta t\calD^{a_i}}\circ \calU_{\delta t}\right)[\rho] 
    \overset{\text{Trotterization}}{\approx}
    \left(\prod_{i=1}^{M}\calW^{a_i}(\delta t)\right)[\rho]
    \approx
    \sigma_\beta.
\end{align}
Taking $\rho= \calU_{S \Delta t} [\rho(0)]$, and using the fact that $\sigma_\beta= \calU_{-S \Delta t} [\sigma_\beta]$ as $H$ and $\sigma_\beta$ commute, this implies
\begin{align}
\label{eq:w_to_wtilde}
    \calU_{-S \Delta t}\circ \left(\prod_{i=1}^{M}\calW^{a_i}(\delta t)\right)\circ \calU_{S \Delta t}[\rho(0)] = \left(\prod_{i=1}^{M}\widetilde{\calW}^{a_i}(\delta t)\right)[\rho(0)]
    \approx
    \sigma_\beta,
\end{align}
where we have defined
\begin{align}
\label{eq:step_W}
    \widetilde{\calW}^a(\delta t)[\rho]
    \coloneq \calU_{-S \Delta t}\circ \calW^{a}(\delta t) \circ \calU_{S \Delta t}[\rho] =
    \Tr_\ma[ 
        V^a(\delta t)
        (\ket{0}\bra{0}\otimes\e^{-\im H \delta t}\rho \e^{\im H \delta t})
        V^a(\delta t)^\dag
    ].
\end{align}
Therefore, the evolution by $\widetilde{\calW}^a[\rho]$, Eq.~\eqref{eq:step_W}, induces the same Lindblad evolution as the one by $\calW^a[\rho]$, Eq.~\eqref{eq:W_evolve}, but does not require any of the unitary $\e^{\pm \im HS\Delta t}$ evolution.
The resulting quantum circuit is sketched as follows:
\begin{align}\label{eq:circuit_app}
\begin{quantikz}[row sep =0.2cm]
    & \wireoverride{n}
    & \lstick{$\ket{0}_\ma$} \wireoverride{n} \gategroup[5, steps=3, style={dashed, inner xsep=5.5mm}]{Repeat for $i=1,\dots,M$}
    & \gate[5]{V^{a_i}(\delta t)}
    & \meter{}
    \\
    \lstick[4]{$\rho(0)$}
    &
    & \gate[4]{\e^{-\im H \delta t}}
    &
    &
    &
    & \rstick[4]{$\approx\rho(M\delta t)$}
    \\
    &
    &
    &
    &
    &
    &
    \\
    & \vdots \wireoverride{n}
    & \wireoverride{n}
    & \wireoverride{n}
    & \vdots \wireoverride{n}
    & \wireoverride{n}
    & \wireoverride{n}
    \\
    &
    &
    &
    &
    &
    &
\end{quantikz}
\end{align}
where $V^{a_i}(\delta t)$ is given by
\begin{align}
\label{eq:circ_U_D}
V^{a_i}(\delta t) =
\begin{quantikz}[row sep=0.2cm, column sep=0.2cm, transparent]
    \lstick{\text{anc}}
    & \gate[3, label style={yshift=1em}]{B_{-S}^{a_i}}
    &
    & \gate[3, label style={yshift=1em}]{B_{-S+1}^{a_i}}
    & \quad\dots\quad
    & \gate[3, label style={yshift=1em}]{B_S^{a_i}}
    & \gate[3, label style={yshift=1em}]{B_S^{a_i}}
    &
    & \quad\dots\quad
    & \gate[3, label style={yshift=1em}]{B_{-S}^{a_i}}
    &
    \\
    & \linethrough
    & \gate[4]{\e^{-\im H\Delta t}}
    & \linethrough
    & \quad\dots\quad
    & \linethrough
    & \linethrough
    & \gate[4]{\e^{\im H\Delta t}}
    & \quad\dots\quad
    & \linethrough
    &
    \\
    &
    &
    &
    & \quad\dots\quad
    &
    &
    &
    & \quad\dots\quad
    &
    &
    \\
    & \vdots \wireoverride{n}
    & \wireoverride{n}
    & \wireoverride{n}
    & \wireoverride{n}
    & \wireoverride{n}
    & \wireoverride{n}
    & \wireoverride{n}
    & \wireoverride{n}
    & \vdots \wireoverride{n}
    & \wireoverride{n}
    \\
    &
    &
    &
    & \quad\dots\quad
    &
    &
    &
    & \quad\dots\quad
    &
    &
\end{quantikz}
\end{align}

In the circuit depiction of Eq.~\eqref{eq:circ_U_D}, the jump operator $A^{a_i}$, used in the definition of any of the $B_s^{a_i}$ as per Eq.~\eqref{eq:app_one_step_dillat}, has been taken to be a one-local operator acting on the second system qubit from the top.
More generally, these jump operators could act on more than one qubit, but in the present work, we restrict them to be local, i.e. acting on $\mathcal{O}(1)$ qubits to be aligned with the ETH.
Hence, the cost of implementing any of the $V^{a_i}(\delta t)$ is dominated by the $2S$ steps of Hamiltonian simulation $\calU_{\Delta t} = \e^{\im H\Delta t}$.The Hamiltonian simulation of $\calU_{\Delta t}$ and $\calU_{\delta t}$ in Eq.~\eqref{eq:w_to_wtilde}, e.g. again via Trotterization, induces an error which can be controlled by choosing an appropriately fine evolution step. We discuss this in more detail below Eq.~\eqref{app:conditions_overall}, where we summarize all algorithmic errors.
In the following, we attempt to identify the choice of $S$.

\subsection{Truncation and discretization of integral}
\label{app:integral_error}

So far, the circuit implementation presented was general to the evolution under the Lindbladian of Eq.~\eqref{eq:lindbladian_app} together with the Lindblad operators obtained through the generic operator Fourier transform of Eq.~\eqref{eq:L_operators_app}. We now specialize to our protocol and specify the Lindblad operator~\eqref{eq:L_operators_app}.
Let us recall our definition of the filter function [Eq.~\eqref{eq:Gibbs_filter_appA}],
\begin{align}
\label{eq:Gibbs_filter}
    g(t)
    =
    \frac{1}{\pi^{3/4}\beta^{1/2}}
    \e^{\im t/\beta}\e^{-2t^2/\beta^2},
\end{align}
resulting in Lindblad operators of the form
\begin{align}
\label{eq:lindblad_OFT_gaussian}
    L^a 
    = 
    \int_{-\infty}^{\infty}\diff t\,g(t)A^a(t)
    = 
    \left(\frac{\beta^2}{4\pi}\right)^{1/4}
    \sum_{\nu\in \bohrH}\e^{-\frac{(\beta\nu+1)^2}{8}}
    A^a_\nu.
\end{align}
We prove the following lemma, which plays a key role in determining the complexity of the protocol.
\begin{lemma}
\label{lemma:discretization}
    Let $g(t)$ be the filter function defined in Eq.~\eqref{eq:Gibbs_filter}.
    Then, provided that
    \begin{align}
    \label{eq:S_and_delta_t_app}
        S
        = 
        \Omega\left(\beta\|H\|_\infty\log\frac{\beta}{\epsilon'}\right),
        \qquad
        \Delta t 
        = 
        \Theta\left(\frac{\beta}{\beta\|H\|_\infty+\sqrt{n+\log\frac{\beta}{\epsilon'}}}\right)
    \end{align}
    in \eqref{eq:L_int_discrete}, we have the regularization error
    \begin{align}
        \|L - \bar{L}\|_\infty = \calO(\epsilon').
    \end{align}
\end{lemma}

\begin{proof}
    Introducing $\bar{L}_\infty\coloneq \sum_{s=-\infty}^{\infty} \Delta t_s g_s\e^{\im Hs\Delta t} A\e^{-\im Hs\Delta t}$, we separate the discretization and truncation errors,
    \begin{align}
    \label{eq:triangle}
        \|L - \bar{L}\|_\infty
        \le
        \|L - \bar{L}_\infty\|_\infty + \|\bar{L}_\infty - \bar{L}\|_\infty.
    \end{align}

    We start with bounding the truncation error $\|\bar{L}_\infty - \bar{L}\|_\infty$ as follows:
    \begin{align}
    \label{eq:triangle_2nd}
        \|\bar{L}_\infty - \bar{L}\|_\infty
        \le
        \frac{\|A\|_\infty}{\pi^{3/4}\beta^{1/2}}\Delta t\sum_{|s|>S}\e^{-2(s\Delta t/\beta)^2}
        <
        \frac{2\Delta t}{\pi^{3/4}\beta^{1/2}}\sum_{s=S}^{\infty}
        \e^{-2S(\Delta t/\beta)^2s}
        =
        \frac{2\Delta t}{\pi^{3/4}\beta^{1/2}}
        \frac{\e^{-2(S\Delta t/\beta)^2}}{1-\e^{-2S(\Delta t/\beta)^2}}
        =
        \epsilon',
    \end{align}
    for
    \begin{align}\label{eq:S_app}
        S
        =
        \Omega\left(
            \frac{\beta}{\Delta t}\sqrt{\log\Big(\frac{\Delta t}{\sqrt{\beta}\epsilon'}\Big)}
        \right),
    \end{align}
    where we used $\|A\|_\infty=\calO(1)$ assuming $A$ is $\calO(1)$-local.

    Next, we bound the discretization error,
    \begin{align}
    \label{eq:triangle_1st}
        \|L - \bar{L}_\infty\|_\infty
        \le
        \Big\| \int_{-\infty}^{\infty}\diff t\, g(t) A(t)
        - \Delta t\sum_{s=-\infty}^{\infty} g(s\Delta t)A(s\Delta t) \Big\|_\infty.
    \end{align}
    Using the Poisson summation formula,\footnote{
        For a function $h(\tau)$, the Poisson summation formula is given by,
        \begin{align}
        \label{eq:poisson_sum_footnote}
            \sum_{s=-\infty}^{\infty}h(s)
            =
            \sum_{k=-\infty}^{\infty}\int_{-\infty}^{\infty}\diff\tau\, \e^{-\im 2\pi k\tau}h(\tau).
        \end{align}
        Setting $h(s)=\Delta t\, g(s\Delta t)A(s\Delta t)$ in  Eq.~\eqref{eq:poisson_sum_footnote}, we arrive at Eq.~\eqref{eq:posson_sum}.
    }
    we have,
    \begin{align}
    \label{eq:posson_sum}
    \begin{split}
        \Delta t\sum_{s=-\infty}^{\infty}g(s\Delta t)A(s\Delta t)
        &=
        \int_{-\infty}^{\infty}\diff t\, g(t) A(t)
        +
        \sum_{k=\bZ\backslash\{0\}}\int_{-\infty}^{\infty}\diff t\, \e^{-\im 2\pi k t/\Delta t}g(t) A(t)
        \\
        &=
        \int_{-\infty}^{\infty}\diff t\, g(t) A(t)
        +
        \sum_{\nu\in \bohrH}\sum_{k=\bZ\backslash\{0\}}\int_{-\infty}^{\infty}\diff t\, \e^{-\im 2\pi k t/\Delta t}g(t) \e^{\im\nu t}A_\nu,
    \end{split}
    \end{align}
    leading to
    \begin{align}
    \label{eq:app_discretization_error}
    \begin{split}
        \|L - \bar{L}_\infty\|_\infty
        &\le
        \frac{1}{\pi^{3/4}\beta^{1/2}}
        \Big\|
            \sum_{\nu\in \bohrH} A_\nu\sum_{|k|>0}
            \int_{-\infty}^{\infty}\diff t\, \e^{-\im 2\pi k t/\Delta t} \e^{\im t/\beta}\e^{-2t^2/\beta^2}\e^{\im\nu t}
        \Big\|_\infty
        \\
        &=
        \frac{\beta^{1/2}}{2^{1/2}\pi^{1/4}}
        \Big\|
            \sum_{\nu\in \bohrH} A_\nu\sum_{|k|>0}
            \e^{-\frac{1}{8}\big(2\pi \frac{\beta}{\Delta t}k-\beta\nu-1\big)^2}
        \Big\|_\infty
        \\
        &\hspace{-1em}\overset{\nu\le 2\|H\|_\infty}{\le}
        \frac{\beta^{1/2}|\bohrH|\|A\|_\infty}{2^{1/2}\pi^{1/4}}
        \sum_{|k|>0}
        \e^{-\frac{1}{8}\big(2\pi \frac{\beta}{\Delta t}k-2\beta\|H\|_\infty-1\big)^2}
        \\
        &\le
        \frac{(2\beta)^{1/2}|\bohrH|\|A\|_\infty}{\pi^{1/4}}\sum_{k=1}^{\infty}
        \e^{-\frac{1}{8}\big(2\pi \frac{\beta}{\Delta t}-2\beta\|H\|_\infty-1\big)\big(2\pi \frac{\beta}{\Delta t}k-2\beta\|H\|_\infty-1\big)}
        \\
        &\le
        \frac{(2\beta)^{1/2}|\bohrH|\|A\|_\infty}{\pi^{1/4}}
        \frac{\e^{-\frac{1}{8}\big(2\pi \frac{\beta}{\Delta t}-2\beta\|H\|_\infty-1\big)^2}}
        {1 - \e^{-\frac{1}{8}\big(2\pi \frac{\beta}{\Delta t}-2\beta\|H\|_\infty-1\big)\big(2\pi \frac{\beta}{\Delta t}\big)}}.
    \end{split}
    \end{align}
    As the number of Bohr frequencies satisfies $|\bohrH|\le 4^n$, it suffices to choose
    \begin{align}\label{eq:delta_t_app}
        \Delta t 
        = 
        \frac{2\pi\beta}{\sqrt{8\log\big(
        \frac{(2\beta)^{1/2}|\bohrH|\|A\|_\infty}{\pi^{1/4}\epsilon'}
        \big)}+2\beta\|H\|_\infty+1}
        = 
        \Theta\left( \frac{\beta}{\beta\|H\|_\infty + \sqrt{n+\log
        \frac{\beta}{\epsilon'}
        }}
        \right),
    \end{align}
    to ensure $\|L - \bar{L}_\infty\|_\infty\le\epsilon'$.
    Combining Eqs.~\eqref{eq:S_app} and ~\eqref{eq:delta_t_app} we obtain Eq.~\eqref{eq:S_and_delta_t_app}.
\end{proof}

Denoting as $\rho_i\coloneq \rho(i\delta t)$ the state after $i$ steps of ideal evolution, with $\rho_0 = \rho(0)$, and by $\bar{\rho}_i\coloneq \ket{0}\bra{0}_\ma\otimes \rho_i$ and
using Lemma~\ref{lemma:discretization}, we bound the effect of the regularization error on the dissipative evolution~\eqref{eq:app_evolve_dilation},
\begin{align}\label{eq:app_trunc_and_discr}
\begin{split}
    &\big\| 
        \Tr_\ma[\e^{-\im \sqrt{\delta t \gamma} K^a}\bar{\rho}_i \e^{\im \sqrt{\delta t \gamma} K^a}]
        - 
        \Tr_\ma[\e^{-\im \sqrt{\delta t \gamma} \bar{K}^a} \bar{\rho}_i \e^{\im \sqrt{\delta t \gamma} \bar{K}^a}]
    \big\|_1
    \\
    &=
    \left\|\Tr_\ma\big[
        \e^{-\im \sqrt{\delta t \gamma} K^a} \big(
            \bar{\rho}_i
            - 
            \e^{\im \sqrt{\delta t \gamma} K^a}\e^{-\im \sqrt{\delta t \gamma}\bar{K}^a} \bar{\rho}_i \e^{\im \sqrt{\delta t \gamma}\bar{K}^a} \e^{-\im \sqrt{\delta t \gamma} K^a}
        \big) \e^{\im \sqrt{\delta t \gamma} K^a}
    \big]\right\|_1
    \\
    &\hspace{-.6em}\overset{\eqref{eq:int_com}}{=}
    \Big\| 
        \int_{0}^{\sqrt{\delta t \gamma}}\diff\tau\, \Tr_\ma\big[
        \e^{\im(\tau-\sqrt{\delta t \gamma}) K^a} [\bar{K}^a-K^a, \e^{-\im\tau\bar{K}^a}\bar{\rho}_i\e^{\im \tau\bar{K}^a}]  \e^{-\im(\tau-\sqrt{\delta t \gamma}) K^a} 
        \big]
    \Big\|_1
    \\
    &\leq
    \Big\| 
        \int_{0}^{\sqrt{\delta t \gamma}}\diff\tau\, \Tr_\ma\big[
        \e^{\im(\tau-\sqrt{\delta t \gamma}) K^a}
        [\bar{K}^a-K^a, \e^{-\im\tau\bar{K}^a}\bar{\rho}_i\e^{\im\tau\bar{K}^a} -\bar{\rho}_i]  
        \e^{-\im(\tau-\sqrt{\delta t \gamma}) K^a}
    \big]\Big\|_1
    \\
    &\quad+
    \Big\| 
        \int_{0}^{\sqrt{\delta t \gamma}}\diff\tau\, \Tr_\ma\big[
        \e^{\im(\tau-\sqrt{\delta t \gamma}) K^a}
        [\bar{K}^a-K^a, \bar{\rho}_i]  
        \e^{-\im(\tau-\sqrt{\delta t \gamma}) K^a}
    \big]\Big\|_1
    \\
    &\hspace{-.6em}\overset{\eqref{eq:b37}}{=}
    \calO(\|\bar{K}^a-K^a\|_\infty\delta t\gamma).
\end{split}
\end{align}
In the third line, we used the identity,
\begin{align}
\label{eq:int_com}
    \e^{\im tB}\e^{-\im tC}A\e^{\im tC}\e^{-\im tB} - A
    =
    \im\int_0^t\diff\tau\, \e^{\im\tau B}[B-C, \e^{-\im\tau C}A\e^{\im\tau C}] \e^{-\im\tau B},
\end{align}
for non-commuting operators $A$, $B$, and $C$.
In the last inequality, we used
\begin{align}\label{eq:b37}
\begin{split}
    \Big\| 
        \Tr_\ma\big[
        \e^{\im(\tau-\sqrt{\delta t \gamma}) K^a}
        [\bar{K}^a-K^a, \e^{-\im\tau\bar{K}^a}\bar{\rho}_i\e^{\im\tau\bar{K}^a} -\bar{\rho}_i]  
        \e^{-\im(\tau-\sqrt{\delta t \gamma}) K^a}
    \big]\Big\|_1
    &=
    \calO(\|\bar{K}^a-K^a\|_\infty\cdot\tau\|\bar{K}^a\|_\infty),
    \\
    \Big\| 
        \Tr_\ma\big[
        \e^{\im(\tau-\sqrt{\delta t \gamma}) K^a}
        [\bar{K}^a-K^a, \bar{\rho}_i]  
        \e^{-\im(\tau-\sqrt{\delta t \gamma}) K^a}
    \big]\Big\|_1
    &=
    \calO(\|\bar{K}^a-K^a\|_\infty\cdot(\tau-\sqrt{\delta t \gamma})\|K^a\|_\infty),
\end{split}
\end{align}
and $\|\bar{K}^a\|=\|K^a\|=\|A^a\|=\calO(1)$.
Finally, combining Eq.~\eqref{eq:discrete_K} and Lemma~\ref{lemma:discretization} together with Eq.~\eqref{eq:app_trunc_and_discr}, we find
\begin{align}
\label{eq:regularization_error}
\begin{split}
    &\big\| 
        \Tr_\ma[\e^{-\im \sqrt{\delta t \gamma} K^a}\bar{\rho}_i \e^{\im \sqrt{\delta t \gamma} K^a}]
        - 
        \Tr_\ma[\e^{-\im \sqrt{\delta t \gamma} \bar{K}^a} \bar{\rho}_i \e^{\im \sqrt{\delta t \gamma} \bar{K}^a}]
    \big\|_1
    \\
    &=
    \calO(\|\bar{K}^a-K^a\|_\infty\delta t\gamma)
    =
    \calO(\|\bar{L}^a-L^a\|_\infty\delta t\gamma)
    =
    \calO(\epsilon'\delta t\gamma).
\end{split}
\end{align}

\subsection{Algorithmic errors}
\label{app:alg_errors}

Combining all the previous results, we assess the resources needed for the Gibbs-state preparation protocol. 
We start by relating the overall error $\epsilon$ in the approximate Lindblad evolution algorithm to the total evolution time $t$, the number of performed steps $M$ (or equivalently the size $\delta t$ of these steps) and the error $\epsilon '$ resulting from the regularization of the integral appearing in App.~\ref{app:integral_error}. We then provide bounds on the runtime of the Gibbs state preparation algorithm in terms of total Hamiltonian simulation time, and also the required number of steps of Hamiltonian simulation $\e^{\im H\Delta t}$, as seen in Eq.~\eqref{eq:circ_U_D}.

\subsubsection{Lindblad simulation}
The approximate Lindblad evolution incurs an error in the evolved state that is now quantified.
The trace distance between the ideal $\rho(M \delta t)\coloneq\e^{M \delta t\calL}[\rho(0)]$ and the approximately evolved state $\bE_{\{a_i\}}\big[\prod_{i=1}^{M}\calW^{a_i}(\delta t)[\rho(0)]\big]$~\eqref{eq:W_evolve}, with the average $\bE_{\{a_i\}}[\cdot]$ taken over trajectories associated with the randomly sampled jump operators $\{A^{a_1},\dots,A^{a_M}\}$, is
\begin{align}\label{eq:dist_true_evolved}
    \Big\|
    \rho(M\delta t) - \bE_{\{a_i\}}\Big[\prod_{i=1}^{M}\calW^{a_i}(\delta t)[\rho_0]\Big]
    \Big\|_1.
\end{align}
Let us now start by recalling all the sources of errors contributing to Eq.~\eqref{eq:dist_true_evolved}
using $\rho_i\coloneq \rho(i\delta t)$, $\rho_0 = \rho(0)$, and $\bar{\rho}_i\coloneq \ket{0}\bra{0}_\ma\otimes \rho_i$ as before. 
We need to account for the following:
\begin{itemize}

\item Randomized applications of the jump operators.
\begin{align}
\label{eq:error_random_jump}
    \|\e^{\delta t\calL}[\rho_i] - \bE_a[\e^{\delta t\gamma\calD^a}\circ\calU_{\delta t}[\rho_i]]\|_1
    =
    \calO(\delta t^2),
\end{align}
according to Eq.~\eqref{eq:app_random_average}.

\item Dilation of the dissipative evolution.
\begin{align}
\label{eq:error_dilation}
    \big\| 
        \e^{\delta t\gamma\calD^a}
        - 
        \Tr_\ma[\e^{-\im \sqrt{\delta t \gamma}K^a}  \bar{\rho}_i \e^{\im \sqrt{\delta t \gamma} K^a}]
    \big\|_1
    = 
    \calO((\delta t \gamma)^2),
\end{align}
according to Eq.~\eqref{eq:app_evolve_dilation}, and noting that the unitary evolution $\calU_{\delta t}$ does not affect the errors.

\item Truncation and discretization of the operator Fourier transform. 
\begin{align}
\label{eq:error_regularization}
    \big\| 
        \Tr_\ma[\e^{-\im \sqrt{\delta t \gamma} K^a}\bar{\rho}_i \e^{\im \sqrt{\delta t \gamma} K^a}]
        - 
        \Tr_\ma[\e^{-\im \sqrt{\delta t \gamma} \bar{K}^a} \bar{\rho}_i \e^{\im \sqrt{\delta t \gamma} \bar{K}^a}]
    \big\|_1
    =
    \calO(\epsilon'\delta t\gamma)
\end{align}
according to Eq.~\eqref{eq:regularization_error}.

\item Trotterization.
\begin{align}
\label{eq:error_trotterization}
\begin{split}
    &\Big\|\Tr_\ma\Big[
        \Big(\prod_{s}^{\rightarrow}\e^{-\im\frac{\sqrt{\delta t \gamma}}{2}\bar{K}^a_s}
        \prod_{s}^{\leftarrow}\e^{-\im\frac{\sqrt{\delta t \gamma}}{2}\bar{K}^a_s}\Big)
        \bar{\rho}_i
        \Big(\prod_{s}^{\rightarrow}\e^{\im\frac{\sqrt{\delta t \gamma}}{2}\bar{K}^a_s}
        \prod_{s}^{\leftarrow}\e^{\im\frac{\sqrt{\delta t \gamma}}{2}\bar{K}^a_s} \Big)
    \Big]
    -
    \Tr_\ma[\e^{-\im \sqrt{\delta t \gamma} \bar{K}^a}\bar{\rho}_i\e^{\im \sqrt{\delta t \gamma} \bar{K}^a}]
    \Big\|_1
    \\
    &=
    \calO((\delta t \gamma)^2),
\end{split}
\end{align}
according to Eq.~\eqref{eq:app_dilated_evolution}.
\end{itemize}

We are now ready to combine all these errors. First, notice that at any step $i=1, \dots, M$ we have
\begin{align}
\begin{split}
    &\Big\|\rho_i - \bE_{\{a_j\}} \Big[\prod_{j=1}^{i}\calW^{a_j}(\delta t)[\rho_0]\Big]\Big\|_1
    \\
    &= \Big\|\rho_i 
    - \bE_{a_i} \Big[\calW^{a_i}(\delta t)[\rho_{i-1}]\Big] 
    + \bE_{a_i} \Big[
      \calW^{a_i}(\delta t)\Big[\rho_{i-1}
    - \bE_{\{a_j\}}\Big[\prod_{j=1}^{i-1}\calW^{a_j}(\delta t)[\rho_0]\Big]\Big]\Big\|_1
    \\
    &\le
     \Big\|\rho_i 
    - \bE_{a_i} \Big[\calW^{a_i}(\delta t)[\rho_{i-1}]\Big] \Big\|_1 
    + \Big\|\rho_{i-1} -\bE_{\{a_j\}}\Big[\prod_{j=1}^{i-1}\calW^{a_j}(\delta t)[\rho_0]\Big]\Big\|_1.
\end{split}
\end{align}
The last line is obtained through the triangle inequality and the fact that $\E_{a_i} [\calW^{a_i}(\delta t)]$ is a quantum channel together with the contractivity of the trace distance under quantum channels. After recursive use of this inequality, we can bound the total error through
\begin{align}
\label{eq:app_circuit_final_error}
    &\Big\|\rho(M\delta t) - \bE_{\{a_i\}}\Big[\prod_{i=1}^{M}\calW^{a_i}(\delta t)[\rho_0]\Big]\Big\|_1
    \nonumber\\
    &\le
    \sum_{i=1}^{M}\big\|\rho_i - \bE_{a_i}\big[\calW^{a_i}(\delta t)[\rho_{i-1}]\big]\big\|_1
    \nonumber\\
    &\le
    \sum_{i=1}^{M}\Big(
    \big\|\rho_i - \bE_{a_i}\big[\e^{\delta t\gamma\calD^a}\circ\calU_{\delta t}[\rho_{i-1}]\big]\big\|_1
    \nonumber\\
    &\quad+
    \bE_{a_i}\big\| 
        \e^{\delta t\gamma\calD^{a_i}}\circ \calU_{\delta t}[\rho_{i-1}] 
        - 
        \Tr_\ma[\e^{-\im \sqrt{\delta t \gamma} K^a}\bar{\rho}_{i-1} \e^{\im \sqrt{\delta t \gamma} K^a}]
    \big\|_1
    \nonumber\\
    &\quad+
    \bE_{a_i}\big\| 
        \Tr_\ma[\e^{-\im \sqrt{\delta t \gamma} K^a}\bar{\rho}_i \e^{\im \sqrt{\delta t \gamma} K^a}]
        - 
        \Tr_\ma[\e^{-\im \sqrt{\delta t \gamma} \bar{K}^a} \bar{\rho}_i \e^{\im \sqrt{\delta t \gamma} \bar{K}^a}]
    \big\|_1
    \nonumber\\
    &\quad+
    \bE_{a_i}\Big\|
        \Tr_\ma[\e^{-\im \sqrt{\delta t \gamma} \bar{K}^a}\bar{\rho}_i\e^{\im \sqrt{\delta t \gamma} \bar{K}^a}]
        -
        \Tr_\ma\Big[
        \Big(\prod_{s}^{\rightarrow}\e^{-\im\frac{\sqrt{\delta t \gamma}}{2}\bar{K}^a_s}
        \prod_{s}^{\leftarrow}\e^{-\im\frac{\sqrt{\delta t \gamma}}{2}\bar{K}^a_s}\Big)
        \bar{\rho}_i
        \Big(\prod_{s}^{\rightarrow}\e^{\im\frac{\sqrt{\delta t \gamma}}{2}\bar{K}^a_s}
        \prod_{s}^{\leftarrow}\e^{\im\frac{\sqrt{\delta t \gamma}}{2}\bar{K}^a_s} \Big)
    \Big]
    \Big\|_1
    \Big)
    \nonumber\\
    &\le
    \calO(M \delta t^2 + \epsilon' M \delta t).
\end{align}
Hence, we can guarantee an error $\calO(\epsilon)$ in the prepared state, after a total Lindblad evolution time $t=M\delta t$, provided that
\begin{align}\label{app:conditions_overall}
    \epsilon' = \calO\left(\frac{\epsilon}{t}\right),
    \qquad
    M = \calO\left(\frac{t^2}{\epsilon}\right).
\end{align}

So far, we have not specified how to implement the coherent evolutions ${\cal U}_{\delta t}$ and ${\cal U}_{\Delta t}$ in Eqs~\eqref{eq:random_dissipator_coh} and~\eqref{eq:w_to_wtilde}. If we employ the second-order product formula with $r_\delta$ steps to approximately implement ${\cal U}_{\delta t}$, the error is $\calO((\|H\|_\infty\delta t)^3/r_\delta^2)$. Similarly, applying the second-order product formula to ${\cal U}_{\Delta t}$ with $r_\Delta$ steps induces the error $\calO((\|H\|_\infty\Delta t)^3/r_\Delta^2)$.
Note that for each of the $M$ steps of size $\delta t$, ${\cal U}_{\Delta t}$ has to be applied $4S$ times in Eq.~\eqref{eq:w_to_wtilde}.
Thus, the error accumulated during $M$ steps of the Lindblad evolution is 
\begin{align}
\label{eq:error_trotter_coherent}
    \calO\left(
        M \frac{(\|H\|_\infty\delta t)^3}{r_\delta^2} + MS\frac{(\|H\|_\infty\Delta t)^3}{r_\Delta^2}
    \right).
\end{align}
Choosing $r_\delta = \calO(\sqrt{M(\|H\|_\infty\delta t)^3/\epsilon})$ and $r_\Delta = \calO(\sqrt{MS(\|H\|_\infty\Delta t)^3/\epsilon})$ suffices to bound the error by $\epsilon$.

\subsubsection{Gibbs state preparation}

Provided a quantum system described by $H$ and jump operators $\{A^a\}$ obey the ETH~\eqref{eq:app_ETH}, one can prepare the Gibbs state $\sigma_\beta$ by simulating the approximate Lindblad dynamics~\eqref{eq:lindblad_random_app} for a mixing time $t_{\rm mix}=M_{\rm mix}\delta t$ of the Lindbladian $\calL$~\eqref{eq:lindbladian_app}. 
Here, we provide the runtime of the quantum operation~\eqref{eq:lindblad_random_app} to simulate the corresponding Lindblad equation in terms of the total Hamiltonian simulation time
\begin{align}
\label{eq:runtime}
    M_{\rm mix} \times (\delta t + \Theta(S\Delta t))
    =
    \Theta\left(\frac{\beta t_\mix^2}{\epsilon}\sqrt{\log\frac{\beta t_{\rm mix}}{\epsilon}}\right),
\end{align}
to prepare an $\epsilon$-precise Gibbs state in trace distance. For the simulation time corresponding to the dissipative part, we used $S$ and $\Delta t$ satisfying Lemma~\ref{lemma:discretization} with $\epsilon'$ given by Eq.~\eqref{app:conditions_overall}.
The resulting circuit, from Eqs.~\eqref{eq:circuit_app} and ~\eqref{eq:circ_U_D}, uses a total of
\begin{align}
    M_{\rm mix} \times \Theta(S)
    =
    \Theta\left(\frac{\beta t_\mix^2\|H\|_\infty}{\epsilon}\log\frac{\beta t_{\rm mix}}{\epsilon}\right)
\end{align}
applications of $B_s (I_\ma\otimes \e^{\im H\Delta t})$ and $(I_\ma\otimes \e^{-\im H\Delta t}) B_s$.

\vspace{1em}

For the purpose of the error analysis conducted in Sec.~\ref{sec:noisy_simulation}, we collect all the sources of algorithmic error along the Lindblad simulation for time $t_{\rm mix}=M_{\rm mix}\delta t$.
While we stated the runtime without specifying a Hamiltonian simulation algorithm in Eq.~\eqref{eq:runtime}, we here adopt the second-order Trotter formula as discussed around Eq.~\eqref{eq:error_trotter_coherent} to be aligned with the numerical simulations presented in the main text.
Adding Eqs.~\eqref{eq:error_random_jump}, \eqref{eq:error_dilation}, 
\eqref{eq:error_regularization},\eqref{eq:error_trotterization}, and \eqref{eq:error_trotter_coherent}, we obtain the total algorithmic error, summarized below.

\begin{proposition}
\label{prop:algorithmic_error}
Let $\calL$ be the Lindbladian given by Eq.~\eqref{eq:lindbladian_app} with the Lindblad operators~\eqref{eq:lindblad_OFT_gaussian}, mixing time $t_\mathrm{mix}$ and inverse temperature $\beta$.
Let $|B_H|$ be the number of Bohr frequencies for the Hamiltonian $H$.
Moreover, let $\delta t$ be the evolution step size in Eq.~\eqref{eq:lindblad_random}, $\Delta t$ and $T$ be the OFT discretization step and integral range limit appearing in Eq~\eqref{eq:L_operator_discretized},
and $r_\delta$ and $r_\Delta$ as given in Eq.~\eqref{eq:error_trotter_coherent}.
The total algorithmic error of the algorithm in Thm.~\ref{thm:qa_complexity} is given by
\begin{align}
\label{eq:app_error_sources}
    t_{\rm mix}\times\calO\left(
        \delta t
        + \frac{\Delta t}{\sqrt{\beta}}\e^{-2(T/\beta)^2}
        + \sqrt{\beta}|B_H| \e^{-\frac{1}{8}\big(2\pi\frac{\beta}{\Delta t} - 2\beta\|H\|_\infty-1\big)^2}
        + \frac{\delta t^2}{r_\delta^2} 
        + \frac{T\Delta t^2}{\delta t\, r_\Delta^2}
    \right).
\end{align}
\end{proposition}
For $\epsilon'$ in the truncation and discretization error~\eqref{eq:error_regularization}, we used Eqs.~\eqref{eq:triangle_2nd} and \eqref{eq:app_discretization_error}, resulting in the second and third terms in Eq~\eqref{eq:app_error_sources}.
In Eq.~\eqref{eq:error_sources} of the main text, we drop the second last term $\propto \delta t^2$ as this is subleading to the first term $\propto\delta t$, and set $r_\Delta=1$ to align with the setup adopted in the main text.

\section{Convergence of Lindblad dynamics under the ETH}
\label{app:ETH}
This appendix details the convergence analysis of the Lindblad dynamics towards the Gibbs state, assuming that the eigenstate thermalization hypothesis (ETH) holds.
For this, we follow the strategy of~\cite{Chen2021ETH} and compute the spectral gap of the ETH-averaged Lindbladian $\bE_R\calL$ and relate it, via a bound on the channel distance between $\bE_R\calL$ and the actual Lindbladian $\calL$, to the mixing time of $\calL$.
Throughout this appendix we set $\gamma_a=1/|{\bm A}|$ in Eqs.~\eqref{eq:lindbladian_general} and~\eqref{eq:lindbladian_dissipative}.

In Sec.~\ref{app:ETH_scales} we restate the ETH, and discuss the energy scales at play together with the assumptions made.
In Sec.~\ref{app:dos} we introduce the density of states and summarize the properties relevant for our subsequent derivations.
Following this, we compute the spectral gap of the averaged Lindbladian $\bE_R\calL$ and the distance between $\bE_R\calL$ and $\calL$ in Secs.~\ref{app:ssec:gap_EL} and~\ref{app:ssec:concentration_ETH_average}.
Based on these derivations, we then prove in Sec.~\ref{app:ssec:spectral_gap_and_mixing_time} our final mixing time bound for the actual Lindbladian $\calL$ [Eq.~\eqref{eq:mixing_time_bound}], and the closeness of the steady state of $\calL$ to the target Gibbs state [Eq.~\eqref{eq:upper_bound_trace_distance}],
which stated in the following theorem.
\begin{theorem}\label{thm:convergence}
    Suppose a Hamiltonian $H$ and a set of Hermitian operators $\{A^a\}_{a=1,\dots,|{\bf A}|}$ obey the ETH,
    $\bra{E_i}A^a\ket{E_j} = \calA(E_{i})\delta_{ij} + D(E_{ij})^{-\frac{1}{2}}f(E_{ij},\nu_{ij})R_{ij}$,
    with $E_{ij} \coloneq (E_i+E_j)/2$, such that $E_i = E_{ii}$, and $\nu_{ij} \coloneq E_i-E_j$.
    Furthermore suppose the function $f(E,\nu)$ satisfies Assumption~\ref{as:f_ETH} and the density of states $D(E_{ij})$ satisfies Assumptions~\ref{as:dos_ratio} and \ref{as:dos_gibbs}. Then, for an $n$-qubit Gibbs state $\sigma_\beta$ and $\epsilon>0$, there is a Lindbladian $\calL$ admitting $\SteadyState$ as a steady state such that,
    \begin{align}
        \label{eq:thm:convergence_accuracy}
        \|\SteadyState - \sigma_\beta\|_{1}
        \le
        \calO \left( \epsilon + \frac{n\beta^2}{\sqrt{|\bm{A}|}} \left( 
        \beta \|H\|_\infty + \log(1/\epsilon)
        \right) \right)
    \end{align}
    with high probability.
    Moreover, the dynamics $\e^{t\calL}$ prepares a state $\epsilon$-close to $\SteadyState$ in trace distance for any $t\ge t_{\rm mix}$ with
    \begin{align}
        \label{eq:thm:mixing_time}
        t_{\rm mix} \le \calO \left(n\beta^2\big(\beta\|H\|_\infty + \log(1/\epsilon)\big)\right).
    \end{align}
\end{theorem}
We prove this theorem in the remainder of App.~\ref{app:ETH}. See Eqs.~\eqref{eq:app:mixing_time_bound_final} and~\eqref{eq:app:trace_distance_simplified} for the final results.

\subsection{Eigenstate thermalization hypothesis}
\label{app:ETH_scales}

The ETH, originally proposed by Srednicki~\cite{Srednicki1999}, hypothesizes that, for a Hamiltonian $H$, the matrix elements of a local operator $A$ in the eigenbasis $\{\ket{E_i}\}$ of $H$ is expressed as [Eq.~\eqref{eq:ETH}]
\begin{align}
\label{eq:ETH_app}
    A_{ij} 
    \coloneq
    \bra{E_i}A\ket{E_j} 
    =
    \calA(E_{i})\delta_{ij} 
    + \frac{f(E_{ij},\nu_{ij})}{\sqrt{D(E_{ij})}}R_{ij},
\end{align}
with $E_{ij} \coloneq (E_i+E_j)/2$, such that $E_i = E_{ii}$, and $\nu_{ij} \coloneq E_i-E_j$. We will also sometimes use $D_{ij}\coloneq D(E_{ij})$ and $f_{ij}\coloneq f(E_{ij}, \nu_{ij})$.
Both $\calA(E)$ and $f(E,\nu)$ are smooth functions of $E$ and $\nu$.
We defer the discussion of the properties of the density of states $D(E)$ to Sec.~\ref{app:dos}.
$R$ is a Hermitian matrix with entries $R_{ij}$ whose real and imaginary parts are independent random variables that satisfy $\bE_R[R_{ij}]=0$ and $\bE_R[|R_{ij}|^2]=1$. Furthermore, its diagonal entries are set to $R_{ii}=0$ for all $i$.

 Given the ETH ansatz~\eqref{eq:ETH_app} and properties of the random variables $R_{ij}$, one can verify several identities that will be useful later on when performing averages over entries of $R$. First, we have 
\begin{equation}\label{eq:average_formula_1}
    \bE_R[A_{ij}A_{kl}^*] = \delta_{ij} \delta_{kl} \calA (E_i) \calA(E_k) + \delta_{ik} \delta_{jl} \frac{f(E_{ij},\nu_{ij})}{\sqrt{D(E_{ij})}}.
\end{equation}
Furthermore, recalling that $A_\nu = \sum_{E_i - E_j = \nu} \Pi_j A \Pi_i$, with $\Pi_i$ the projector onto the Hamiltonian's eigenspace with energy $E_i$, one can verify that
\begin{equation}\label{eq:average_formula_2}
    \bE_R[A_{\nu_1} \rho A^{\dagger}_{\nu_2}] = \delta_{\nu_1, \nu_2} \bE_R [A_{\nu_1} \rho A^{\dagger}_{\nu_1}],\quad \textrm{and } \; \bE_R \left[\{A^{\dagger}_{\nu_1} A^{\dagger}_{\nu_1}, \rho\}\right] = \delta_{\nu_1, \nu_2} \bE_R \left[\{A^{\dagger}_{\nu_1} A^{\dagger}_{\nu_1}, \rho \} \right].
\end{equation}

To proceed with the concrete calculation of the mixing time, we make a couple of simplifying assumptions on the non-universal function $f(E,\nu)$.
\begin{assumption}\label{as:f_ETH}
The function $f(E,\nu)$~in Eq.~\eqref{eq:ETH_app} satisfies the following properties.
\begin{enumerate}[label=(\alph*)]
\item\label{as:f_ETH_a} $f(E,\nu)=f(\nu)$ is independent of $E$, and supported on $\nu\in[-\Delta_\mathrm{RMT},\Delta_\mathrm{RMT}]$, where it is flat.

\item\label{as:f_ETH_b} The width of the support scales with $\beta$ as $\Delta_{\rm RMT}=\Theta(1/\beta)$, and the function takes an $(n,\beta)$-independent constant value $f_0$ on the support. 
\end{enumerate}
\end{assumption}

To motivate the assumptions we note that an operator $A$, such that $\|A^2\|_\infty=\calO(1)$, obeys $\Tr[\sigma_\beta A^2]\le\|A^2\|_\infty=\calO(1)$. With the ETH, one can express the trace in the integral form,
\begin{align}
\label{eq:f_scale_derivation}
    \Tr[\sigma_\beta A^2]
    =
    \int_{-\infty}^{\infty} \diff\nu\, \e^{\beta \nu/2} |f(\nu)|^2
    =
    \calO(1).
\end{align} 
To ensure that the integral converges to a finite value, we need  $|f(\nu)|\xrightarrow{\nu\to\infty} o(\e^{-\beta\nu/4})$~\cite{Srednicki1999, Murthy2019}. 
Assumption~\ref{as:f_ETH}\ref{as:f_ETH_b} provides a simple restriction on $f(\nu)$ so that this constraint is met. While we made specific choices of $\Delta_{\rm RMT}$ and $|f(\nu)|$, any polynomial dependence of them on $n$ also leads to a polynomial upper bound on the mixing time. In the following, we shall provide computational resources with and without Assumption~\ref{as:f_ETH}\ref{as:f_ETH_b}.

\begin{figure}
    \centering
    \begin{tikzpicture}[domain=-4:4,samples=100]
        \newcommand*{\betaE}{1.3}
        \pgfmathsetmacro{\sigmaE}{sqrt(2)/\betaE}
        \pgfmathsetmacro{\meanE}{-\betaE*\sigmaE^2/2}
        \draw[->] (-4, 0) -- (4.5,0) node[right] {$\nu$};
        \draw[->] (0,0) -- (0,3);
        \draw[thick] plot (\x, {2/(1+exp(-5*(\x+3)))/(1+exp(5*(\x-3)))});
        \node[above] at (1, 2) {$|f(\nu)|^2$};
        \draw[thick] plot (\x, {2.5*exp(-(\x-\meanE)^2/(2*\sigmaE^2)))});
        \draw[|-|] ($({\meanE-\sigmaE},1.3)$) -- node[midway,below] {$\Delta_\text{E}=\Theta(\Delta_\text{RMT})$} ($({\meanE+\sigmaE},1.3)$);
        \node[above] at (\meanE,2.5) {$\eta_\nu^2$};
        \draw (\meanE,-.1) node[below]
        {$-\frac{\Delta_\mE}{\sqrt{2}}=\Theta(\Delta_\text{RMT})$} -- ++(0,
        .2);
        \draw[|-|] (-4, -0.9) -- node[midway,below] {$\Delta_\text{RMT} = \Theta(1/\beta)$} ++(4,0);
    \end{tikzpicture}
    \caption{\label{fig:scales} 
    Squared Gaussian filter function $\eta_\nu^2$~\eqref{eq:eta_nu} from the operator Fourier transform and the function $|f(\nu)|^2$ in the off-diagonal ETH~\eqref{eq:ETH}. The
    relationship between their supports $\Delta_\mE = \sqrt{2}/\beta =
    \Theta(\Delta_\text{RMT})$ guarantees significant overlap between the
    two functions.
    }
\end{figure}
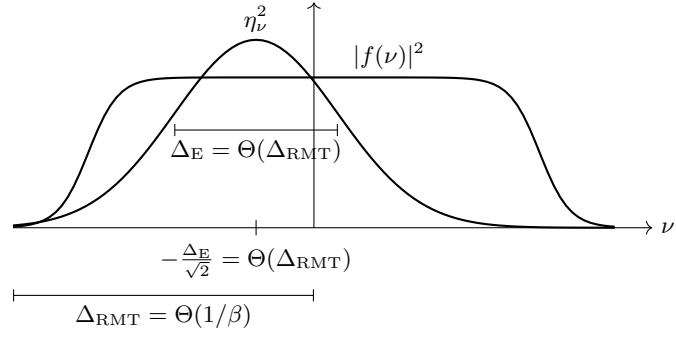

We illustrate the relevant scales under Assumption~\ref{as:f_ETH}\ref{as:f_ETH_b} in Fig.~\ref{fig:scales}.
Given the above condition on $|f(\nu)|$, we choose the width of our filter function $\eta_\nu$~\eqref{eq:eta_nu} to be $\Delta_\mE = \Theta(\Delta_\mathrm{RMT})=\Theta(1/\beta)$.
This, in turn, leads to the filter being peaked at $\nu=-\beta\Delta_\mE^2/2=-\Delta_\mE/\sqrt{2}$.
Figure~\ref{fig:scales} illustrates that there is a significant overlap between the supports of $\eta_\nu^2$ and $|f(\nu)|^2$, which guarantees efficient transition of states upon applications of the jump operators.
Quantitatively, we can derive upper and lower bounds on the overlap integrals,
\begin{align}
\label{eq:app:overlap_integral_limits}
    \Gamma\coloneq\int_0^{\Delta_\mathrm{RMT}} \diff\nu\, \eta_{-\nu}^2 |f(\nu)|^2
    \qquad \text{and} \qquad
    \Gamma_{\infty}\coloneq \int_{-\infty}^\infty \diff\nu\, \eta_\nu^2 |f(\nu)|^2,
\end{align}
which will show up frequently in this appendix. First, we can relate these two integrals through
\begin{equation}\label{eq:rel_gammas}
    \Gamma \leq \Gamma_\infty \leq 2 \Gamma,
\end{equation}
where we used the fact that $f$ is only supported on $\nu\in[-\Delta_\mathrm{RMT},\Delta_\mathrm{RMT}]$, from Assumption~\ref{as:f_ETH}\ref{as:f_ETH_a}, and the profile of $\eta_{\nu}$~\eqref{eq:eta_nu} that has most of its mass on the range $[-\Delta_\mathrm{RMT},0]$.
Hence it is sufficient to derive bounds for $\Gamma$.
For its upper bound, we insert Eq.~\eqref{eq:eta_nu} and obtain
\begin{align}
\label{eq:app:overlap_upper_bound}
\begin{split}
    \Gamma \le 
\frac{\beta}{\sqrt{4\pi}}\int_{-\Delta_\mathrm{RMT}}^{\Delta_\mathrm{RMT}} \diff\nu\,
    \e^{-((\beta\nu)^2+1)/4}\e^{\beta\nu/2} |f(\nu)|^2
    \leq
    \frac{\beta}{\sqrt{4\pi}} \int_{-\Delta_\mathrm{RMT}}^{\Delta_\mathrm{RMT}} \diff\nu\, \e^{\beta \nu/2} |f(\nu)|^2 
    \overset{\eqref{eq:f_scale_derivation}}{=} 
    \calO(\beta),
\end{split}
\end{align}
On the other hand, Assumption~\ref{as:f_ETH}\ref{as:f_ETH_b} allows one to lower bound the overlap,
\begin{align}
\label{eq:app:overlap_lower_bound}
\begin{split}
    \Gamma \ge
    \frac{\beta\e^{-(\beta\Delta_{\rm RMT})^2/4} \e^{-1/4} }{\sqrt{4\pi}}\int_0^{\Delta_\mathrm{RMT}} \diff\nu\,
    |f(\nu)|^2 =
    \frac{\e^{-(\beta\Delta_{\rm RMT})^2/4}\beta\e^{-1/4}}{\sqrt{4\pi}}\Delta_\mathrm{RMT}|f_0|^2
    = 
    \Theta(1) ,
\end{split}
\end{align}
using $\Delta_\mathrm{RMT} =\Theta(1/\beta)$ in the last equality. 

\subsection{Density of states}
\label{app:dos}

We introduce the density of states
\begin{align}
\label{eq:dos_integral}
    D(E)
    =
    \sum_{E_i\in{\rm spec}[H]} \tilde{\delta}(E-E_i).
\end{align}
Here $\tilde{\delta}(E-E_i)$ is a smeared delta function ensuring that
$D(E)$ is a smooth function of $E$.
The following assumption on the density of states guarantees that the ratios
appearing in the conductance calculations in Sec.~\ref{app:ssec:gap_EL} are of order~1.
\begin{assumption}[Bounded ratio of density of states]\label{as:dos_ratio}
    The ratio of densities of states is uniformly bounded such that for all $|E-E'|\leq \Delta_\mathrm{RMT}$
    \begin{equation}
        \frac{D(E)}{D(E')} \leq R_D = \Theta(1).
    \end{equation} 
\end{assumption}
We note that in \cite{Chen2021ETH}, the spectrum of the Hamiltonian is truncated to exclude pathological cases
that would invalidate this bound on the ratio close to the extremal values of the spectrum.

The normalized density of Gibbs states is defined by
\begin{align}\label{eq:norm_desity_gibbs}
    D_\beta(E) 
    \propto
    \e^{-\beta E}D(E)
\end{align}
with the normalization taken to satisfy $\int_{-\infty}^\infty \diff E D_\beta(E) = 1$. We require that the bulk of the density of Gibbs states is lower-bounded by a characteristic scale $1/\Delta_\text{spec}$ and that the tails decay sufficiently fast. This is captured by the following assumption. 
\begin{assumption}[Density of Gibbs states]
\label{as:dos_gibbs}
    The (normalized) density of Gibbs states $D_\beta(E)$, defined in Eq.~\eqref{eq:norm_desity_gibbs}, satisfies the following properties.
    \begin{enumerate}[label=(\alph*)]
        \item\label{as:dos_gibbs_bulk} There exists an interval $[E_\mL, E_\mR]$ that
        contains more than half the weight, such that $\int_{E_\mL}^{E_\mR}
        \diff E D_\beta(E) \geq 1/2$, and for all $E \in [E_\mL,
        E_\mR]$ we have
        \begin{equation}
            D_\beta(E) \ge D_\beta^{\rm min} = \Omega(\Delta_\mathrm{spec}^{-1}).
        \end{equation}
        
        \item\label{as:dos_gibbs_tail}  The right tail $[E_\mR, \infty)$ and left tail
        $(-\infty, E_\mL]$ decay such that for all $E \in [E_\mR,
        \infty)$ we have
        \begin{align}
            \int_E^\infty\diff E' D_\beta(E') 
            &= 
            \mathcal{O} \left( \Delta_\mathrm{spec} D_\beta(E) \right)
        \end{align}
        and for all $E \in (-\infty, E_\mL]$ we have
        \begin{align}
            \int_{-\infty}^E\diff E' D_\beta(E') 
            &= 
            \mathcal{O} \left( \Delta_\mathrm{spec} D_\beta(E) \right).
        \end{align}
    \end{enumerate}
\end{assumption}

To gain intuition about these assumptions, it is instructive to think of the density of states as a Gaussian
\begin{align}
    D(E) 
    \propto
    \e^{-\frac{(E-E_\infty)^2}{2\Delta_\mathrm{spec}^2}},
\end{align}
where $E_\infty\coloneq\Tr[H]/2^n$ is the energy at infinite temperature. Then the density
of Gibbs states is a Gaussian shifted by $-\beta\Delta_\text{spec}^2$
\begin{align}
\label{eq:gaussian_dos}
    D_\beta(E) 
    =
    \frac{1}{\sqrt{2\pi\Delta_\mathrm{spec}^2}}\e^{-\frac{(E-E_\infty+\beta\Delta_\mathrm{spec}^2)^2}{2\Delta_\mathrm{spec}^2}}.
\end{align}
Figure~\ref{fig:density} illustrates those densities.
It can be shown that those Gaussian densities fulfill Assumptions~\ref{as:dos_ratio} and~\ref{as:dos_gibbs}~\cite{Chen2021ETH}.
For instance, the density~\eqref{eq:gaussian_dos} indeed obeys Assumption~\ref{as:dos_gibbs}\ref{as:dos_gibbs_tail},
\begin{align}
    \begin{split}
    \int_E^\infty \diff E' D_\beta(E')
    &= 
    \frac{1}{\sqrt{2\pi\Delta_\text{spec}^2}} \int_{E}^\infty \diff E'\e^{-\frac{(E'-E_\infty + \beta\Delta_\text{spec}^2)^2}{2\Delta_\text{spec}}} 
    = 
    \frac{1}{\sqrt{\pi}} \int_{|E - E_\infty + \beta\Delta_\text{spec}^2|/\sqrt{2\Delta_\text{spec}^2}}^\infty \diff u\,\e^{-u^2}
    \\
    &\leq 
    \frac{1}{2\sqrt\pi w} \int_{w}^\infty \diff u\, 2u \e^{-u^2}
    \qquad 
    \left(w=|E - E_\infty + \beta\Delta_\text{spec}^2|/\sqrt{2\Delta_\text{spec}^2}\,\right)
    \\
    &= 
    \frac{1}{2\sqrt\pi w} \int_{w^2}^{\infty} \diff u\, \e^{-u}
    \\
    &= 
    \frac{\Delta_\text{spec}^2}{|E - E_\infty + \beta\Delta_\text{spec}^2|} D_\beta(E)
    \\
    &\leq
    \Delta_\text{spec} D_\beta(E).
\end{split}
\end{align}
The last line follows from $\Delta_\text{spec} \leq |E - E_\infty + \beta\Delta_\text{spec}^2|$ in the tail $E \geq E_\mR$.

\begin{figure}
    \centering
    \begin{tikzpicture}[domain=-4:4,samples=100]
        \newcommand*{\meanbeta}{-1.5}
        \draw[->] (-4, 0) -- (4,0) node[right] {$E$};
        \draw[->] (0,-.1) node[below] {$E_\infty$} -- (0,4);
        \draw (\meanbeta, -0.1) node[below] {$E_\infty -
        \beta\Delta_\text{spec}^2$} -- ++(0, .2);
        \node[above] at (0.5, 3) {$D(E)$};
        \node[above] at (\meanbeta, 3) {$D_\beta(E)$};
        \draw[thick] plot (\x, {3*exp(-0.5*(\x)^2)});
        \draw[thick,dashed] plot (\x, {3*exp(-0.5*(\x-\meanbeta)^2)});
        \draw[|-|] ($(\meanbeta, 1.6)+(-1,0)$) -- node[midway,above] {$\Delta_\text{spec}$}
            ++(2, 0);
    \end{tikzpicture}
    \caption{\label{fig:density} Example of the density of states $D(E)$ and density of the
     Gibbs state $D_\beta(E)$ that we consider.}
\end{figure}
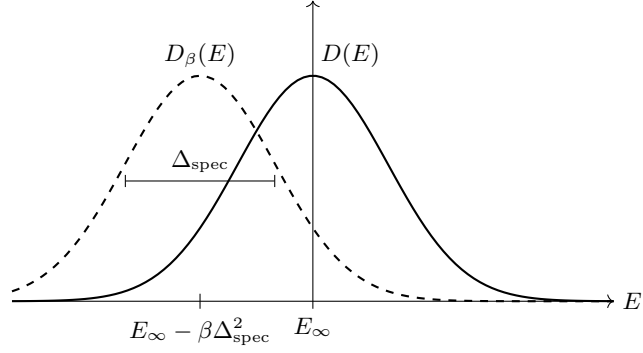

\subsection{Spectral gap of \texorpdfstring{$\bE_R\calL$}{}}
\label{app:ssec:gap_EL}

In this section, we derive a lower bound on the spectral gap of the Lindbladian $\bE_R\calL$ under the ETH average. After working out the action of the average Lindbladian onto eigenstates projectors, in Sec.~\ref{sec:app_action_average}, we proceed in two steps: First, in Sec.~\ref{sec:app:map_classical_mc}, we show that the average Lindbladian $\bE_R\calL$ can be mapped to a classical Markov chain admitting the Gibbs state as its exact steady state. 
In the second step, in Sec.~\ref{eq:app:cond_bounds} we work out bounds on the conductance of the Markov chain, which allows us to bounds its spectral gap through Cheeger's inequality.
Subsequently, in Sec.~\ref{app:ssec:concentration_ETH_average}, we derive an upper bound on the distance between the Lindbladian $\calL$ and the averaged one $\bE_R\calL$, which implies that the steady state $\SteadyState$ of $\calL$ is close to the Gibbs state $\sigma_\beta$.

\subsubsection{Action of the averaged Lindbladian}\label{sec:app_action_average}
We wish to evaluate the action of the averaged Lindbladian $\bE_R\calL$, with $\calL$ given in Eq.~\eqref{eq:lindbladian_app}, onto any of the operators $\ket{E_i}\bra{E_j}$.
Setting $\gamma_a=1/|{\bm A}|$ in Eq.~\eqref{eq:lindbladian_dissipative}, or equivalently $p_a=1/|{\bm A}|$ and $\gamma=1$ in Eq.~\eqref{eq:lindbladian_app}, together with our choice of jump operators in Eq.~\eqref{eq:lindblad_OFT_gaussian_appA}, we can write the dissipative part of our Lindbladian as
\begin{align}
    \label{eq:dissipator_app_convergence_under_ETH}
    \calD[\rho]
    =
    \frac{1}{|{\bm A}|}\sum_a
    \Big(L^a\rho L^{a\dag} - \frac{1}{2}\{L^{a\dag} L^a, \rho\}\Big)
    =
    \frac{1}{|{\bm A}|}\sum_a\sum_{\nu_1,\nu_2}\eta_{\nu_1}\eta_{\nu_2}
    \Big(
        A^a_{\nu_1}\rho A^{a\dag}_{\nu_2} - \frac{1}{2}\{A^{a\dag}_{\nu_1} A^a_{\nu_2}, \rho\}
    \Big).
\end{align}
Taking an average over the random matrices elements appearing in Eq.~\eqref{eq:ETH_app}, and using Eq.~\eqref{eq:average_formula_2}, leads to 
\begin{align}
\label{eq:davies_dissipator_ETH_average}
    \bE_R\calD[\rho]
    =
    \frac{1}{|{\bm A}|}\sum_a\sum_{\nu}\eta_{\nu}^2
    \bE_R\Big[
        A^a_{\nu}\rho A^{a\dag}_{\nu} - \frac{1}{2}\{A^{a\dag}_{\nu} A^a_{\nu}, \rho\}
    \Big],
\end{align}
that has the form of the Davies generator~\cite{Davies1974, Davies1976}.
Let us define $D_{ij} \coloneq D(E_{ij})$ together with $\eta_{ij}\coloneq\eta_{\nu_{ij}}$ and $f^{a}_{ij}\coloneq f^{a}(\nu_{ij})$. Recall that $\nu_{ij} = E_i - E_j$. Using Eq.~\eqref{eq:ETH_app}, we can see that
\begin{align}
\begin{split}
    \sum_\nu \eta_\nu^2 \bE_R[A^a_\nu \ket{E_i}\bra{E_j} A_\nu^{a\dagger}]
    =
    \eta_0^2\calA^{a}(E_i)\calA^{a}(E_j)\ket{E_i}\bra{E_j}
    + 
    \delta_{ij} \sum_{k\neq i}\frac{\eta_{ki}^2|f^{a}_{ki}|^2}{D_{ki}}\ket{E_k}\bra{E_k} .
\end{split}
\end{align}
Similarly, for the second term appearing in Eq.~\eqref{eq:davies_dissipator_ETH_average}, we obtain
\begin{align}
\begin{split}
    \sum_\nu \eta_\nu^2 \bE_R[\lbrace A_\nu^{a\dagger} A_\nu^a, \ket{E_i}\bra{E_j} \rbrace]
    =& 
    \left(\eta_0^2 \left( \calA^a(E_i)^2 + \calA^a(E_j)^2 \right)
     + 
    \sum_{\nu\neq 0} \eta_{\nu}^2 \left( \frac{|f^a_{i (i+\nu)}|^2}{D_{i(i+\nu)}} + \frac{|f^a_{j (j+\nu)}|^2}{D_{j(j+\nu)}} \right)\right) \ket{E_i}\bra{E_j} .
\end{split}
\end{align}
Finally, we note that $[H, \ket{E_i}\bra{E_j}] = (E_i-E_j) \ket{E_i}\bra{E_j}$. 

Putting everything together we can evaluate the action of $\bE_R \calL$ onto any of the $\ket{E_i}\bra{E_j}$ operators. For the case where $i=j$ we get:
\begin{align}
\label{eq:average_lindbladian_on_eigenstate}
    \bE_R\calL[\ket{E_i}\bra{E_i}] = \frac{1}{|{\bm A}|}\sum_a
    \sum_{j\neq i}\frac{\eta_{ji}^2|f^a_{ji}|^2}{D_{ji}}
    (\ket{E_j}\bra{E_j} - \ket{E_i}\bra{E_i}).
\end{align}
On the other hand for $i\neq j$, after some reordering, we obtain:
\begin{align}\label{eq:_average_lindbladian_on_off_diag}
    &\bE_R\calL[\ket{E_i}\bra{E_j}]
    =
    \left(
        -\im (E_i - E_j)
        - \frac{1}{2|{\bm A}|} \sum_a \left( \eta_0^2 \left( \calA^a(E_i) - \calA^a(E_j) \right)^2
        +  
        \sum_{m\neq i} \eta_{mi}^2 \frac{|f^a_{i m}|^2}{D_{im}} + \sum_{l\neq j} \eta_{lj}^2 \frac{|f^a_{jl}|^2}{D_{jl}}
    \right)\right) \ket{E_i}\bra{E_j}.
\end{align}

In the following, we will often need to bound the terms $\sum_{k\neq i}\frac{\eta_{ki}^2|f_{ki}|^2}{D_{ki}}$, for a fixed $i$, appearing in Eqs.~\eqref{eq:average_lindbladian_on_eigenstate} and \eqref{eq:_average_lindbladian_on_off_diag}. To do so, we first approximate the sums as follows:
\begin{align}
\label{eq:approx_sum_eta_f_square}
    \sum_{k\neq i}\frac{\eta_{ki}^2|f_{ki}|^2}{D_{ki}}
    = 
    \int_{-\infty}^{\infty}\diff E_k\,D(E_k) \eta_{ki}^2\frac{|f_{ki}|^2}{D\left(E_{ki}\right)}
    = \int_{E_i-\Delta_\mathrm{RMT}}^{E_i + \Delta_\mathrm{RMT}}\diff E_k\,D(E_k) \eta_{ki}^2\frac{|f_{ki}|^2}{D\left(E_{ki}\right)}.
\end{align}
We first replaced the sum by an integral
that holds for a smooth function $g(E)$.
Then, given that $f(\nu)$ is supported on the interval $[-\Delta_\mathrm{RMT}, \Delta_\mathrm{RMT}]$, we restricted the integral's interval.
Finally, using the fact that the ratio of densities is bounded through Assumption~\eqref{as:dos_ratio} and the definition of $\Gamma$ in Eq.~\eqref{eq:app:overlap_integral_limits}, we get
\begin{align}
\label{eq:approx_sum_eta_f_square_bound}
    \frac{\Gamma}{R_D} \leq \sum_{k\neq i}\frac{\eta_{ki}^2|f_{ki}|^2}{D_{ki}}
    \le
    R_D \Gamma.
\end{align}

\subsubsection{Mapping to a classical Markov chain}\label{sec:app:map_classical_mc}

Here, we show that the average Lindbladian $\bE_R\calL$ defines the generator of a classical Markov chain on the spectrum of $H$.
The transition rates $q_{i\to j}$ between eigenstate projectors $\ket{E_i}\bra{E_i}$ are readily obtained from Eq.~\eqref{eq:average_lindbladian_on_eigenstate} as
\begin{align}
\label{eq:transition_matrix}
\begin{split}
    q_{i\to j}
    &\coloneq
    \Tr\left[\ket{E_j}\bra{E_j}\bE_R\calL[\ket{E_i}\bra{E_i}]\right]
    = 
    \frac{\eta_{ji}^2|f_{ji}|^2}{D_{ji}}
    - \delta_{ij} \sum_{k\neq i} \frac{\eta_{ki}^2|f_{ki}|^2}{D_{ki}} \,,
\end{split}
\end{align}
where for the sake of simplicity we assumed that $f^a=f$ for any $a\in {\bm A}$. Moving to the off-diagonal elements $\ket{E_i}\bra{E_j}$, we saw from Eq.~\eqref{eq:_average_lindbladian_on_off_diag} that these were eigenvectors of $\bE_R \calL$.
From Eq.~\eqref{eq:approx_sum_eta_f_square_bound}, the minimum (in magnitude) of the real part of these eigenvalues is lower bounded through
\begin{align}\label{eq:bound_gap_simplified_offdiag}
\begin{split}
    \Delta^{\rm off} &\coloneq \min_{ij} \frac{1}{2|{\bm A}|} \sum_a \left( \eta_0^2 \left( \calA^a(E_i) - \calA^a(E_j) \right)^2
        +  
        \sum_{m\neq i} \eta_{mi}^2 \frac{|f^a_{i m}|^2}{D_{im}} + \sum_{l\neq j} \eta_{lj}^2 \frac{|f^a_{jl}|^2}{D_{jl}}
    \right)
     \geq \frac{\Gamma}{R_D}. 
\end{split}
\end{align}
Under the time evolution with $\e^{t\bE_R \calL}$, the imaginary part of the eigenvalues results in a phase, while the real part leads to a decay of the amplitude on the off-diagonal operators $\ket{E_i}\bra{E_j}$. 
Later on, we will discuss the timescales of such decay, but for now focus on
the classical Markov chain defined by Eq.~\eqref{eq:transition_matrix}, which is just the restriction of $\bE_R\calL$ to eigenstate projectors. This converges to the Gibbs state $\sigma_\beta$, since $\bE_R\calL$ is $\sigma_\beta$-DB, as discussed in Sec.~\ref{sec:convergence_under_ETH},
and the mixing time of the chain is controlled by the transition rates~\eqref{eq:transition_matrix}.

We shift and rescale the transition rates $q_{i\to j}$ to obtain a stochastic transition matrix of the classical Markov chain:
\begin{align}
\label{eq:spectral_gap_of_average_L_rshift}
    P_{i\to j}
    \coloneq
    \frac{q_{i\to j} + r\delta_{ij}}{r},\quad
\text{with}\;
    r 
    \coloneq
    \max_i \sum_{k\neq i} \frac{\eta_{ki}^2|f_{ki}|^2}{D_{ki}}.
\end{align}
To see that with such choice $P_{i\to j}$ is indeed stochastic\footnote{
    We say $P_{i\to j}$ is stochastic if it satisfies $\sum_j P_{i\to j} = 1$ for all $i$ and any $P_{i\to j}\geq 0$.
}
we first note that the transition rates~\eqref{eq:transition_matrix} satisfy $\sum_j q_{i\to j} = 0$. Hence, $\sum_j P_{i\to j} = 1$ for any choice of $r>0$. It remains to show that the chosen $r$ ensures that any $P_{i \to j}\geq 0$.
This can be seen by noting that $r$ is equal to $- \max_i q_{i \to i}$ and that the $q_{i \to i}$ are the only negative transition rates. Finally, according to Eq.~\eqref{eq:approx_sum_eta_f_square_bound}, we readily see that $r$ is bounded through
\begin{align}
\label{eq:spectral_gap_of_average_L_rshift_upper}
    \frac{\Gamma}{R_D} \leq r 
    \leq
    R_D \Gamma,
\end{align}

Given the spectral gap $\Delta_P$ of the Markov chain with transition matrix $P$, the gap of the diagonal part of $\bE_R \calL$ is $\Delta^{\rm diag} = r\Delta_P$, and 
we now wish to bound $\Delta_P$.
To this end, we relate the gap of the Markov chain with its conductance $\Phi$ defined as
\begin{align}
    \Phi(S) \coloneq \frac{Q(S, \bar{S})}{\pi_S},
    \qquad
    Q(S_1,S_2) \coloneq \sum_{i\in S_1, j\in S_2}\pi_i P_{i\to j}.
\end{align}
Here, $S \subset \Spec(H)$ refers to a subset of the spectrum of $H$ and $\bar{S}$ to its complement. We further defined the probability $\pi_S\coloneq\sum_{i\in S}\pi_i$
for the stationary distribution $\pi_i$ of the Markov chain. This stationary distribution $\pi_i$ satisfies the detailed balance condition given by $\pi_i P_{i\to j} = \pi_j P_{j\to i}$ for any $i, j\in \Spec(H)$.
Defining the minimum conductance (also called bottleneck ratio) as
\begin{align}
\label{eq:conductance}
    \phi \coloneq \min_{\pi_S\le 1/2}\frac{Q(S,\bar{S})}{\pi_S},
\end{align}
the spectral gap $\Delta_{P}$ is bounded via Cheeger's inequality as
\begin{align}
\label{eq:Cheeger}
    \frac{\phi^2}{2} \le \Delta_{P} \le 2\phi.
\end{align}

\subsubsection{Conductance calculation}\label{eq:app:cond_bounds}
In the following we aim at bounding the conductance~\eqref{eq:conductance}. For that, we introduce a set of non-overlapping energy intervals $\{S_i\}_i$, each of which has size $\Delta_\mathrm{RMT}$ as sketched in Fig.~\ref{fig:interval}. The intervals $S_i$ are labeled in increasing order of energies, i.e., $E_i < E_{i+1}$ for any $E_i\in S_i$ and $E_{i+1}\in S_{i+1}$.
Furthermore, we assume for now that the subset $S$ appearing in the definition of the conductance has the form $S_\ell\cup S_{\ell+1}\cup\dots\cup S_{r}$ supported on the energy interval $[E_\ell,E_r]$, such that $\pi_S\le 1/2$.
That is, we assume that $S$ is contiguous over the energy spectrum. When $S$ is not contiguous, we can decompose it into contiguous subsets and combine their conductance bounds, as discussed later.
We divide the analysis of the conductance into four cases~\cite{Chen2021ETH} that cover all the possible locations of $E_l$ and $E_r$ with respect to the interval $[E_L, E_R]$ that was defined in  Assumption~\ref{as:dos_gibbs}.

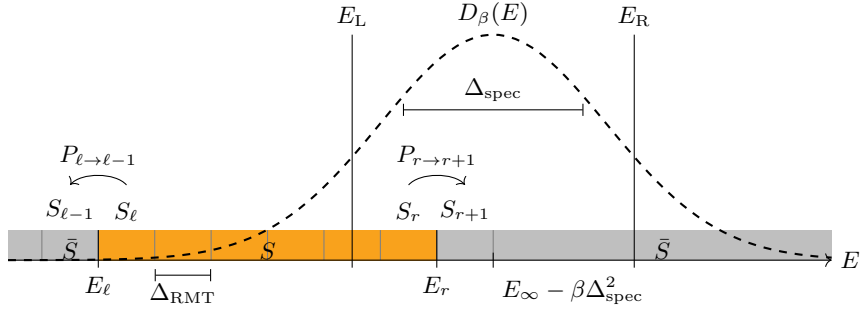
\begin{figure}
    \centering
    \begin{tikzpicture}[domain=-6.3:1,samples=100,xscale=1.5]
        \newcommand*{\meanbeta}{-2}
        \newcommand*{\EL}{-5.5}
        \newcommand*{\ER}{-2.5}
        \newcommand*{\tEL}{-3.25}
        \newcommand*{\Sheight}{0.4}
        \newcommand*{\ymax}{3}
        \newcommand*{\DeltaRMT}{0.5}
        \fill[YellowOrange] (\EL,0) rectangle (\ER,\Sheight);
        \fill[lightgray] (-6.3, 0) rectangle (\EL, \Sheight);
        \fill[lightgray] (\ER, 0) rectangle (1, \Sheight);
        \pgfmathsetmacro{\tmpL}{\EL-\DeltaRMT}
        \pgfmathsetmacro{\tmpR}{\ER+\DeltaRMT}
        \pgfmathsetmacro{\tER}{\tEL+2.5}
        \foreach \x in {\tmpL, \EL, ..., \tmpR} {
            \draw[gray] (\x, 0) -- (\x, \Sheight);
        }
        \node[above] at ($({\EL+(\ER-\EL)/2}, -.1)$) {$S$};
        \node[above] at ($({\EL-\DeltaRMT/2}, -.1)$) {$\bar S$};
        \node[above] at ($({\ER+4*\DeltaRMT}, -.1)$) {$\bar S$};
        \draw (\EL, -.1) node[below] {$E_\ell$} -- (\EL, \Sheight);
        \draw (\ER, -.1) node[below] {$E_r$} -- (\ER, \Sheight);
        \draw (\tEL, -.1) -- (\tEL, \ymax) node[above] {$E_\mL$};
        \draw (\tER, -.1) -- (\tER, \ymax) node[above] {$E_\mR$};
        \node[above] at ($({\EL-\DeltaRMT/2}, \Sheight)$) {$S_{\ell-1}$};
        \node[above] at ($({\EL+\DeltaRMT/2}, \Sheight)$) {$S_\ell$};
        \path[->] ($({\EL+\DeltaRMT/2}, {2.5*\Sheight})$) edge[bend right=60]
        node[midway,above] {$P_{\ell \rightarrow \ell-1}$} ($({\EL-\DeltaRMT/2},
        {2.5*\Sheight})$);
        \path[->] ($({\ER-\DeltaRMT/2}, {2.5*\Sheight})$) edge[bend left=60]
        node[midway,above] {$P_{r \rightarrow r+1}$} ($({\ER+\DeltaRMT/2}, {2.5*\Sheight})$);
        \node[above] at ($({\ER-\DeltaRMT/2}, \Sheight)$) {$S_r$};
        \node[above] at ($({\ER+\DeltaRMT/2}, \Sheight)$) {$S_{r+1}$};
        \draw[->] (-6.3, 0) -- (1,0) node[right] {$E$};
        \draw (\meanbeta, -0.1) node[below right] {$E_\infty - \beta\Delta_\text{spec}^2$} -- ++(0, .2);
        \node[above] at (\meanbeta, \ymax) {$D_\beta(E)$};
        \draw[|-|] ($({\EL+\DeltaRMT}, -0.2)$) -- node[midway,below]
        {$\Delta_\text{RMT}$} ++($(\DeltaRMT, 0)$);
        \draw[thick,dashed] plot (\x, {\ymax*exp(-0.5*(\x-\meanbeta)^2)});
        \draw[|-|] ($(\meanbeta, \ymax/1.5)+(-0.8,0)$) -- node[midway,above] {$\Delta_\text{spec}$}
            ++(1.6, 0);
    \end{tikzpicture}
    \caption{\label{fig:interval}
    Sketch of energy intervals, transitions, and density of states that contribute to the conductance from $S$ to $\bar{S}$.}
\end{figure}

\paragraph{Case \texorpdfstring{$E_r \in [E_\mL+\Delta_{\rm RMT}, E_\mR-\Delta_{\rm RMT}]$}{}.}
The conductance is lower-bounded with Assumptions~\ref{as:dos_ratio} and \ref{as:dos_gibbs} as
\begin{align}
\label{eq:Phi_increase_energy}
    \Phi(S) 
    &\geq 
    2Q(S, \bar S) 
    \geq 
    2(\pi_{S_r} P_{r\to r+1} + \pi_{S_\ell} P_{\ell\to \ell - 1}) 
    \geq 
    2\pi_{S_r} P_{r\to r+1}
    \nonumber\\
    &= 
    \frac{2}{r} \sum_{E_r\in S_r, E_{r+1}\in S_{r+1}} \pi_{E_r}
    \eta_{r,r+1}^2
    \frac{|f_{r,r+1}|^2}{D_{r, r+1}}
    \nonumber\\
    &= 
    \frac{2}{r}
    \int_{E_{r-1}}^{E_r} \diff E\,D_\beta(E) 
    \int_{E_r}^{E_{r+1}} \diff E'\, D(E') \eta_{E'-E}^2 \frac{|f(E - E')|^2}{D\left( \frac{E+E'}{2} \right)}
    \nonumber\\
    &\hspace{-.7em} \overset{\text{Ass.~\ref{as:dos_ratio}}}{\geq}
    \frac{2}{r R_D} \int_{E_{r-1}}^{E_r} \diff E\, D_\beta(E) \int_{E_r}^{E_{r+1}} \diff E'\, \eta_{E'-E}^2 |f(E - E')|^2
    \nonumber\\
    &\hspace{-1.2em} \overset{\text{Ass.~\ref{as:dos_gibbs}\ref{as:dos_gibbs_bulk}}}{\geq}
    \frac{2D_\beta^{\rm min}}{r R_D} 
    \int_{E_{r-1}}^{E_r} \diff E \int_{E_r}^{E_{r+1}} \diff E'\, \eta_{E'-E}^2 |f(E - E')|^2
    \nonumber\\
    &\hspace{-1.2em} \overset{\nu=E'-E}{=}
    \frac{2 D_\beta^{\rm min}}{r R_D} \int_{E_{r-1}}^{E_r} \diff E\, \int_{E_r-E}^{E_{r+1}-E} \diff \nu\, \eta_{\nu}^2 |f_{\nu}|^2
    \nonumber\\
    &\hspace{-1.2em} \overset{\mu=E_r-E}{=}
    \frac{2 D_\beta^{\rm min}}{r R_D} \int_{0}^{\Delta_\text{RMT}} \diff \mu \int_{\mu}^{\Delta_\text{RMT}+\mu} \diff\nu\, \eta_{\nu}^2 |f_{\nu}|^2
    \nonumber\\
    &\ge
    \frac{2 D_\beta^{\rm min}}{r R_D}
    \int_{0}^{\Delta_\text{RMT}} \diff \mu \int_0^{\Delta_\text{RMT}} \diff \nu\, \eta_{\nu}^2 |f_{\nu}|^2
    \nonumber\\
    &\ge  
    \frac{2 D_\beta^{\rm min}\e^{-\beta\Delta_\mathrm{RMT}}\Delta_\mathrm{RMT}}{r R_D}
    \int_0^{\Delta_\text{RMT}} \diff \nu\, \eta_{-\nu}^2 |f_{\nu}|^2
    \nonumber\\
    &\hspace{-.7em} \overset{\eqref{eq:spectral_gap_of_average_L_rshift_upper}}{=}
    \Omega\left( 
    \frac{\e^{-\beta\Delta_\mathrm{RMT}}\Delta_\mathrm{RMT}}{\Delta_\mathrm{spec}}
    \right).
\end{align}
The derivations of such bounds only account for energy-increasing transitions, as seen in the first line of the derivations.

\paragraph{Case \texorpdfstring{$E_\ell \in [E_\mL+\Delta_{\rm RMT}, E_\mR-\Delta_{\rm RMT}]$}{}.}

Following similar steps, the conductance is lower-bounded as
\begin{align}
\label{eq:Phi_decrease_energy}
    \Phi(S)
    &\ge
    2\pi_{S_\ell}P_{\ell\to \ell-1}
    =
    \frac{2}{r}
    \sum_{E_\ell\in S_\ell, E_{\ell-1}\in S_{\ell-1}}
    \pi_{E_\ell}
    \eta_{\ell-1,\ell}^2
    \frac{|f_{\ell,\ell-1}|^2}{D\big( \frac{E_\ell+E_{\ell-1}}{2} \big)}
    \ge
    \Omega \left( 
    \frac{\Delta_\text{RMT}}{\Delta_\text{spec}}
    \right).
\end{align}
Note that the energy-increasing conductance bound \eqref{eq:Phi_increase_energy} is smaller than the energy-decreasing one \eqref{eq:Phi_decrease_energy}.

\paragraph{Case \texorpdfstring{$E_r \in (-\infty, E_\mL]$}{}.}

Under Assumption~\ref{as:dos_gibbs}, $Q(S, \bar{S})$ is lower-bounded as,
\begin{align}
    Q(S, \bar{S})
    &\ge
    \pi_{S_r}P_{r\to r+1}
    \ge
    \frac{1}{r}\min_{E_r\in S_r}\Delta_\mathrm{RMT}D_\beta(E_r)
    \e^{-\beta\Delta_\mathrm{RMT}}
    \int_0^{\Delta_\mathrm{RMT}}\diff\nu \eta_{-\nu}^2|f(\nu)|^2.
\end{align}
and
\begin{align}
    \pi_S 
    \le 
    \int_{-\infty}^{E_r} \diff E' D_\beta(E') 
    =
    {\cal O} \big(\Delta_{\rm spec}D_\beta(E_r)\big).
\end{align}
Thus, the conductance is bounded as
\begin{align}
    \Phi(S) 
    =
    \frac{Q(S, \bar{S})}{\pi_S}
    \ge
    \Omega\left(\frac{\e^{-\beta\Delta_\mathrm{RMT}}\Delta_\mathrm{RMT}}{\Delta_\mathrm{spec}}
    \right).
\end{align}

\paragraph{Case \texorpdfstring{$E_\ell \in [E_\mR, \infty )$}{}.}
Following similar steps as for the previous case leads to the same lower bound~\eqref{eq:Phi_decrease_energy}.

\vspace{1em}

To deal with unions of contiguous energy sets, let us first consider the case with two such sets, $S_1$ and $S_2$, that are not contiguous to each other (and thus do not overlap). We wish to relate the conductance $\Phi(S_{12} \coloneq S_1 \cup S_2)$ to the individual conductances $\Phi(S_1)$ and $\Phi(S_2)$. To do so, we note that given Assumption~\ref{as:f_ETH}\ref{as:f_ETH_a}, transitions $P_{i \to j}$ between any $i \in S_1$ to any $j \in S_2$ are suppressed such that $Q(S_{12}, \bar{S}_{12}) = Q(S_{1}, \bar{S}_{1}) + Q(S_{2}, \bar{S}_{2})$. Hence
\begin{equation}
    \Phi(S_{12}) 
    = 
    \frac{Q(S_{1}, \bar{S}_{1}) + Q(S_{2}, \bar{S}_{2})}{\pi_{S_1} + \pi_{S_2}}
    \geq 
    \min \set{\Phi(S_{1}), \Phi(S_{2})}.
\end{equation}
This can extended to unions of more than two non-contiguous sets, and shows that taking such a union cannot decrease the conductance.
Overall, accounting for all the cases studied, we get that the conductance is lower-bounded as
\begin{align}
    \Phi(S)
    \ge
    \Omega\left(\frac{\e^{-\beta\Delta_\mathrm{RMT}}\Delta_\mathrm{RMT}}{\Delta_\mathrm{spec}}
    \right),
\end{align}
which thus gives the minimum conductance, i.e., 
\begin{align}
    \phi
    \ge
    \Omega\left(\frac{\e^{-\beta\Delta_\mathrm{RMT}}\Delta_\mathrm{RMT}}{\Delta_\mathrm{spec}}
    \right).
\end{align}
Therefore, using Cheeger's inequality~\eqref{eq:Cheeger} together with the relationship between the classical Markov chain and the diagonal part of the averaged Lindbladian~\eqref{eq:spectral_gap_of_average_L_rshift}, we obtain the lower-bound for the spectral gap
\begin{align}
\label{eq:bound_gap}
    \Delta^{\rm diag}
    \ge
    \frac{r\phi^2}{2}
    =
    \Omega\left(\frac{\e^{-2\beta\Delta_\mathrm{RMT}}\Delta_\mathrm{RMT}^2\Gamma}{\Delta_\mathrm{spec}^2}
    \right).
\end{align}
Using $\Delta_\mathrm{RMT}=\Theta(1/\beta)$ (Assumption~\ref{as:f_ETH}\ref{as:f_ETH_b})
and the lower-bound on the overlap integrals~\eqref{eq:app:overlap_lower_bound},
we get
\begin{align}
\label{eq:bound_gap_simplified_diag}
    \Delta^{\rm diag}
    =
    \Omega\left( \frac{1}{n \beta^2} \right) ,
\end{align}
Here, we further used a standard scaling of the spectral width of a local Hamiltonian, $\Delta_\mathrm{spec} =\calO(\sqrt{n})$~\cite{Dymarsky2019,Banuls2020,Lu2021}.
Finally, comparing the scaling of Eq.~\eqref{eq:bound_gap_simplified_diag} to one of the eigenvalues of the off-diagonal part of the averaged Lindbladian~\eqref{eq:bound_gap_simplified_offdiag}, we conclude that $\Delta^{\rm diag}$ provides the lower bound of the spectral gap of the average Lindbladian,
\begin{align}
\label{eq:bound_gap_simplified}
    \Delta_{\bE_R \calL}
    =
    \Omega\left(\frac{\e^{-2\beta\Delta_\mathrm{RMT}}\Delta_\mathrm{RMT}^2\Gamma}{\Delta_\mathrm{spec}^2}
    \right)
    \overset{\text{Ass.}\ref{as:f_ETH}\ref{as:f_ETH_b}}{=}
    \Omega\left( \frac{1}{n \beta^2} \right).
\end{align}
Recall that this spectral gap bound is derived under the assumption of ETH~\eqref{eq:ETH_app} and Assumptions~\ref{as:f_ETH}, \ref{as:dos_ratio}, and \ref{as:dos_gibbs}.

\subsection{Concentration: upper bound of \texorpdfstring{$\Vert \calL - \bE_R\calL\Vert_{\sigma_\beta^{-1},2\to2}$}{}}
\label{app:ssec:concentration_ETH_average}

We derive an upper bound on the operator distance between the Lindbladian $\calL$ and its ETH average $\bE_R\calL$.
A convenient distance measure for our purposes is the weighted induced quantum channel norm $\| \cdot \|_{\sigma_\beta^{-1},2\to2}$~\cite{Temme2010}, which is defined by
\begin{align}
\label{eq:def_weighted_channel_norm}
    \|{\cal C}\|_{\sigma_\beta^{-1},2\to2}
    \coloneq
    \max_{X}\frac{\|{\cal C}[X]\|_{\sigma_\beta^{-1},2}}{\|X\|_{\sigma_\beta^{-1},2}},
    \quad \text{with} \quad 
    \| X \|_{\sigma_\beta^{-1},2} \coloneq \sqrt{\expval{X, X}_{\sigma_\beta^{-1}}} \overset{\eqref{eq:DB_inner_product_general}}{=} \sqrt{\Tr[ X^\dagger \sigma_\beta^{-1/2} X \sigma_\beta^{-1/2} ]},
\end{align}
for a quantum channel ${\cal C}$ and the maximization is performed over Hermitian operators $X$.
Later, this norm will allow us to bound the trace distance between the steady states $\SteadyState$  and $\sigma_\beta$.

Note that the ETH average does not affect the system Hamiltonian, such that we can simplify
\begin{align}
\label{eq:app:L_dist_to_D_dist}
    \|{\calL} - \bE_R{\calL}\|_{\sigma_\beta^{-1},2\to2} = \|{\cal D} - \bE_R{\cal D}\|_{\sigma_\beta^{-1},2\to2} .
\end{align}
where $\calD$ is the dissipative part of $\calL$, which form we recall here:
\begin{align}
\label{eq:lindblad_dissipative_app}
    \calD[\rho]
    =
    \frac{1}{|{\bm A}|}\sum_a
    \Big(L^a\rho L^{a\dag} - \frac{1}{2}\{L^{a\dag} L^a, \rho\}\Big)
    =
    \frac{1}{|{\bm A}|}\sum_a\sum_{\nu_1,\nu_2}\eta_{\nu_1}\eta_{\nu_2}
    \Big(
        A^a_{\nu_1}\rho A^{a\dag}_{\nu_2} - \frac{1}{2}\{A^{a\dag}_{\nu_1} A^a_{\nu_2}, \rho\}
    \Big).
\end{align}
Defining ${\cal K}[\cdot]\coloneq\sigma_\beta^{-1/4}{\cal D}[\sigma_\beta^{1/4}\cdot \sigma_\beta^{1/4}]\sigma_\beta^{-1/4}$, and $K$ its vectorization, we can verify that\footnote{
Let $\delta{\cal D}\coloneq {\cal D} - \bE_R{\cal D}$ and $\delta{\cal K}\coloneq {\cal K} - \bE_R{\cal K}$. The first equality in~\eqref{eq:relate_norms} follows from
\begin{align}
    \|\delta {\cal K}\|_{2\to 2}
    =
    \max_{X} \frac{\|\delta{\cal K}[X]\|_2}{\|X\|_2}
    =
    \max_X \frac{\|\sigma_\beta^{-1/4}\delta{\cal D}[\sigma_\beta^{1/4}X\sigma_\beta^{1/4}]\sigma_\beta^{-1/4}\|_2}{\|X\|_2}
    =
    \max_{\tilde{X}} \frac{\|\sigma_\beta^{-1/4}\delta{\cal D}[\tilde{X}]\sigma_\beta^{-1/4}\|_2}{\|\sigma_\beta^{-1/4}\tilde{X}\sigma_\beta^{-1/4}\|_2}
    =
    \max_{\tilde{X}} \frac{\|\delta{\cal D}[\tilde{X}]\|_{\sigma_\beta^{-1},2}}{\|\tilde{X}\|_{\sigma_\beta^{-1},2}}
\end{align}
with $\tilde{X}\coloneq\sigma_\beta^{1/4}X\sigma_\beta^{1/4}$.
In the present work, the maximization of $X$ ($\tilde{X}$) to define the induced norm is restricted to Hermitian operators.
}
\begin{align}
\label{eq:relate_norms}
    \| \mathcal{D} - \bE_R\mathcal{D}\|_{\sigma_\beta^{-1},2\to2}
    =
    \|\mathcal{K} - \bE_R\mathcal{K}\|_{2\to2}
    =
    \|K - \bE_R K\|_{\infty}.
\end{align}
We note that, for any $\sigma_\beta$-DB dissipative channel ${\cal D}$, the channel ${\cal K}$ is self-adjoint with respect to the Hilbert-Schmidt inner product for all bounded operators.
Furthermore $K$, the vectorization of ${\cal K}$, is given by\footnote{
    The vectorization of ${\cal D}$ is
    \begin{align}
    \label{eq:app:vectorization_D}
        D
        =
        \frac{1}{|{\bm A}|}
        \sum_a\sum_{\mu, \nu}\eta_\mu\eta_\nu\Big(
            (A^a_\nu)^* \otimes A^a_\mu
            -\frac{1}{2}I\otimes (A^{a}_\mu)^\dag A^a_\nu
            -\frac{1}{2}(A^a_\nu)^\intercal (A^a_\mu)^* \otimes I
        \Big).
    \end{align}
    Then, the vectorization of ${\cal K}$ is $K=((\sigma_\beta^{-1/4})^\intercal\otimes\sigma_\beta^{-1/4})D((\sigma_\beta^{1/4})^\intercal\otimes\sigma_\beta^{1/4})$.
}
\begin{align}
\label{eq:app:vectorization_K}
    K
    =
    \frac{1}{|{\bm A}|}
    \sum_a\sum_{\mu, \nu}\eta_\mu\eta_\nu\Big(
        \e^{\frac{\beta}{4}(\mu+\nu)} (A^a_\nu)^* \otimes A^a_\mu
        -\frac{\e^{-\frac{\beta}{4}(\mu-\nu)} }{2}I\otimes (A^a_\mu)^\dag A^a_\nu
        -\frac{\e^{\frac{\beta}{4}(\mu-\nu)} }{2}(A^a_\nu)^\intercal (A^a_\mu)^* \otimes I
    \Big),
\end{align}
where $A^\intercal$ is the transpose of $A$.

Our strategy to bound Eq.~\eqref{eq:relate_norms} is to partition the energy spectrum into a set of non-overlapping energy intervals $\{S_i\}_i$ (Fig.~\ref{fig:interval}) using the projectors
\begin{align}
\label{eq:energy_projector}
    \Pi_{S_i} \coloneq \sum_{E_i\in[i\Delta_{\rm RMT},(i+1)\Delta_{\rm RMT}]}\ket{E_i}\bra{E_i}.
\end{align}
We start with the terms involving $A^*\otimes A$. For $\Delta K_{A\otimes A}\coloneq\sum_a\sum_{\mu,\nu}\eta_\mu\eta_\nu\e^{\frac{\beta}{4}(\mu+\nu)} ((A^a_\nu)^*\otimes A^a_\mu - \bE_R[(A^a_\nu)^*\otimes A^a_\mu])$ and an even integer $p$,
\begin{align}
    &\left(\bE_R\big\|(\Pi_{S_k}\otimes\Pi_{S_l})\Delta K_{A\otimes A}(\Pi_{S_i}\otimes\Pi_{S_j})\big\|_\infty\right)^p
    \nonumber\\
    &\le
    \bE_R\big\|(\Pi_{S_k}\otimes\Pi_{S_l})\Delta K_{A\otimes A}(\Pi_{S_i}\otimes\Pi_{S_j})\big\|_\infty^p
    \le
    \bE_R\big\|(\Pi_{S_k}\otimes\Pi_{S_l})\Delta K_{A\otimes A}(\Pi_{\bar{E}_1}\otimes\Pi_{S_j})\big\|_p^p
    \nonumber\\
    &\overset{\eqref{eq:sum_gaussian}}{\le}
    \bE_R \left\|(\Pi_{S_k}\otimes\Pi_{S_l})
        \frac{1}{|{\bm A}|}\sum_a\sum_{\mu,\nu}\e^{\frac{\beta}{4}(\mu+\nu)} \eta_\mu\eta_\nu
        ((A^a_\nu)^*\otimes A^a_\mu - (A'^a_\nu)^*\otimes A'^a_\mu)
    (\Pi_{S_i}\otimes\Pi_{S_j})\right\|_p^p
    \nonumber\\
    &\overset{\eqref{eq:B_gaussian}}{=}
    \bE_R \left\|(\Pi_{S_k}\otimes\Pi_{S_l})
        \frac{2}{|{\bm A}|}\sum_a\sum_{\mu,\nu}\e^{\frac{\beta}{4}(\mu+\nu)} \eta_\mu\eta_\nu
        (B^a_\nu)^*\otimes B'^a_\mu
    (\Pi_{S_i}\otimes\Pi_{S_j})\right\|_p^p
    \nonumber\\
    &\overset{\eqref{eq:G_variance}}{\le}
    \bE_R \left\|
        \frac{2}{|{\bm A}|}\sum_a G^a_{S_k,S_i}\otimes G'^a_{S_l,S_j}
    \right\|_p^p
    \nonumber\\
    &\overset{\eqref{eq:decouple_gaussian}}{\le}
    \left(\frac{2}{\sqrt{|{\bm A}|}}\right)^p
    \bE_R\|G_{S_k,S_i}\|_p^p
    \,
    \bE_R\|G_{S_l,S_j}\|_p^p
    \nonumber\\
    &\overset{\eqref{eq:bound_p_norm}}{\le}
    \left(\frac{2}{\sqrt{|{\bm A}|}}\right)^p
    \left(
        \max\{\Tr[\Pi_{S_k}],\Tr[\Pi_{S_i}]\}
        \bE_R[|G_{S_k,S_i}|^2]
        \max\{\Tr[\Pi_{S_l}],\Tr[\Pi_{S_j}]\}
        \bE_R[|G_{S_l,S_j}|^2]
    \right)^{p/2}.
\end{align}
In the third line, we used that, for the zero-mean random matrix $X\coloneq (A^{a}_\nu)^*\otimes A^a_\mu - \bE[(A^{a}_\nu)^*\otimes A^a_\mu]$
\begin{align}
\label{eq:sum_gaussian}
    \bE_R\|X\|_p^p \le \bE_R\|X-X'\|_p^p,
\end{align}
where $X$ and $X'$ are i.i.d. random matrices (all the diagonal entries are zero).
In the fourth line, we wrote $A=(B+B')/\sqrt{2}$ and $A'=(B-B')/\sqrt{2}$ with zero-mean random matrices $B$ and $B'$. Then, it follows
\begin{align}
\label{eq:B_gaussian}
    A^*\otimes A - A'^*\otimes A'
    = 
    B^*\otimes B' + B'^*\otimes B.
\end{align}
In the fifth line, we defined the zero-mean random matrices $G_{\bar{F},\bar{E}}$ whose entries all have the same variance,
\begin{align}
\label{eq:G_variance}
    \bE_R[|(G_{S',S}^a)_{ij}|^2]
    =
    \max_{F\in S', E\in S} (\e^{\frac{\beta}{4}(F-E)}\eta_{F-E})^2 \bE_R[|(B^a_{F-E})_{ij}|^2]
    =
    \max_{F\in S', E\in S} (\e^{\frac{\beta}{4}(F-E)}\eta_{F-E})^2  \frac{|f(F-E)|^2}{D\big(\frac{F+E}{2}\big)}.
\end{align}
In the sixth line, the Gaussian matrices are decoupled as
\begin{align}
\label{eq:decouple_gaussian}
\begin{split}
    &\bE_R \Big\|\sum_a c_a G^a\otimes G'^a\Big\|_p^p
    \\
    &=
    \sum_{a_1,\dots,a_p} c_{a_1}c_{a_2}^*\dots c_{a_{p-1}}c_{a_p}^*
    \bE_R\Tr[G^{a_1} G^{a_2\dag}\dots G^{a_{p-1}} G^{a_p\dag}]
    \bE_R\Tr[G'^{a_1} G'^{a_2\dag}\dots G'^{a_{p-1}} G'^{a_p\dag}]
    \\
    &\le
    \big(\sum_a|c_a|^2\big)^{p/2}
    \sum^{\text{pairs}}_{a_1,\dots,a_p}
    \bE_R\Tr[G^{a_1} G^{a_2\dag}\dots G^{a_{p-1}} G^{a_p\dag}]
    \bE_R\Tr[G'^{a_1} G'^{a_2\dag}\dots G'^{a_{p-1}} G'^{a_p\dag}]
    \\
    &\le
    \big(\sum_a|c_a|^2\big)^{p/2}
    \sum^{\text{pairs}}_{a_1,\dots,a_p}
    \bE_R\Tr[G^{a_1} G^{a_2\dag}\dots G^{a_{p-1}} G^{a_p\dag}]
    \sum^{\text{pairs}}_{b_1,\dots,b_p}
    \bE_R\Tr[G'^{b_1} G'^{b_2\dag}\dots G'^{b_{p-1}} G'^{b_p\dag}]
    \\
    &\le
    \big(\sum_a|c_a|^2\big)^{p/2}
    \bE_R\|G\|_p^p\, \bE_R\|G'\|_p^p.
\end{split}
\end{align}
The quantity $\bE_R\Tr[G'^{b_1}\dots G'^{b_p\dag}]$ is non-vanishing upon averaging over independent random matrices $R$ or equivalently $G^a$ only when all the indices make pairs. Accordingly, the sum $\sum^\text{pairs}_{a_1,\dots,a_p}$ only runs over such indices.
The last inequality holds because $\bE_R\Tr[G^{a_1} G^{a_2\dag}\dots G^{a_{p-1}} G^{a_p\dag}]\ge 0$ holds when the indices are contracted.
The seventh line follows from Fact~D.1 in \cite{Chen2021ETH}. For a $d_1\times d_2$ matrix~$G$ of the maximum variance $\Var_G$,
\begin{align}
    \bE_R \|G\|_p^p \le \min\{d_1,d_2\}\Var_G^{p/2}\big(\big(c_1\sqrt{\max\{d_1,d_2\}}\big)^p + (c_2\sqrt{p})^p\big),
\end{align}
for constants $c_1, c_2$.
Setting $p$ such that $\min\{d_1,d_2\}^{1/p}\le 1/c_1$, we have
\begin{align}
\label{eq:bound_p_norm}
    (\bE_R \|G\|_p^p)^{1/p} 
    \le 
    \sqrt{\Var_G\cdot\max\{d_1,d_2\}}.
\end{align}

Therefore, we can upper-bound the fluctuation in each eigensector,
\begin{align}
\label{eq:bound_local_AxA}
    &\bE_R \big\|(\Pi_{S_k}\otimes\Pi_{S_l})\Delta K_{A\otimes A}(\Pi_{S_i}\otimes\Pi_{S_j})\big\|_\infty
    \nonumber\\
    &\le
    \frac{\Delta_{\rm RMT} R_D}{\sqrt{|{\bm A}|}}
    \sqrt{
        \max_{E\in S_k\cup S_i} D(E)
        \max_{F\in S_l\cup S_j} D(F)
    }
    \nonumber\\
    &\quad\times
    \max_{E_k\in S_k, E_i\in S_i}
    \frac{\e^{\frac{\beta}{4}(E_k-E_i)} \eta_{E_k-E_i}|f(E_k-E_i)|}{\sqrt{D\big(\frac{E_k+E_i}{2}\big)}}
    \max_{E_l\in S_l, E_j\in S_j}
    \frac{\e^{\frac{\beta}{4}(E_l-E_j)} \eta_{E_l-E_j} |f(E_l-E_j)|}{\sqrt{D\big(\frac{E_l+E_j}{2}\big)}}
    \nonumber\\
    &\leq
    \frac{\Delta_{\rm RMT} R_D^2}{\sqrt{|{\bm A}|}}
    \e^{\frac{\beta}{4}(\hat E_k-\hat E_i)} \eta_{\hat E_k-\hat E_i} |f(\hat E_k-\hat E_i)|
    \cdot
    \e^{\frac{\beta}{4}(\hat E_l-\hat E_j)} \eta_{\hat E_l-\hat E_j} |f(\hat E_l-\hat E_j)|,
\end{align}
where $\hat E_k, \hat E_i$ is the pair of energies solving the first maximization of the third line, and $\hat E_l, \hat E_j$ solves the second maximization of the third line.
Recalling Assumption~\ref{as:dos_ratio}, for the second line we used that the number of states in each energy interval can be expressed as
\begin{align}
\label{eq:dos_entropy_projectors}
    \Tr[\Pi_{S_i}] \le D(E_i)R_D\Delta_{\rm RMT}, \quad\text{for all } E_i\in S_i.
\end{align}
The last line in Eq.~\eqref{eq:bound_local_AxA} holds only if $\hat E_i\in S_i$ and $\hat E_k\in S_k$ are close, which is true because $f(\hat E_k-\hat E_i)$ vanishes for $|\hat E_k - \hat E_i|> \Delta_\mathrm{RMT}$ (Assumption~\ref{as:f_ETH}\ref{as:f_ETH_a}). Same for $\hat E_j$ and $\hat E_l$.

Now, we add all the energy eigensectors.
\begin{align}
    &\bE_R\|\Delta K_{A\otimes A}\|_\infty
    \nonumber\\
    &=
    \bE_R\left\|\sum_{i,j,k,l}(\Pi_{S_k}\otimes\Pi_{S_l})\Delta K_{A\otimes A}(\Pi_{S_i}\otimes\Pi_{S_j})\right\|_\infty
    \nonumber\\
    &\le
    \sum_{m,n}\bE_R\left\|\sum_{i,j}
    (\Pi_{S_{i+m}}\otimes\Pi_{S_{j+n}})\Delta K_{A\otimes A}(\Pi_{S_i}\otimes\Pi_{S_j})\right\|_\infty
    \nonumber\\
    &\le
    \sum_{m,n}\max_{i,j}\bE_R\left\|
    (\Pi_{S_{i+m}}\otimes\Pi_{S_{j+n}})\Delta K_{A\otimes A}(\Pi_{S_i}\otimes\Pi_{S_j})\right\|_\infty
    \nonumber\\
    &\hspace{-.6em}\overset{\eqref{eq:bound_local_AxA}}{\le}
    \frac{\Delta_{\rm RMT} R_D^2}{\sqrt{|{\bm A}|}} \Big(
        \sum_{m}\max_{i}
        \e^{\frac{\beta}{4}(\hat E_{i+m}-\hat E_i)} \eta_{\hat E_{i+m}-\hat E_{i}} |f(\hat E_{i+m}-\hat E_i)|
    \Big)^2
    \nonumber\\
    &\le
    \frac{\Delta_{\rm RMT} R_D^2}{\sqrt{|{\bm A}|}} \Big(
        \sum_{m}
        \e^{\frac{\beta}{4}(\hat E_{\hat\imath+m}-\hat E_{\hat\imath})} \eta_{\hat E_{\hat\imath+m}-\hat E_{\hat\imath}} |f(\hat E_{\hat\imath+m}-\hat E_{\hat\imath})|
    \Big)^2
    \label{eq:Delta_K_AA_i_maximization}
    \\
    &\le
    \frac{\Delta_{\rm RMT} R_D^2}{\sqrt{|{\bm A}|}}
    \Big(\sum_{m}
    \e^{\beta\hat{\nu}_m/4}\eta_{\hat{\nu}_m}|f(\hat{\nu}_m)|\Big)^2
    \label{eq:Delta_K_AA_nu_maximization}
    \\
    &\hspace{-.5em}\overset{\text{Ass.}\ref{as:dos_ratio}}=
    \calO\left(
        \frac{1}{\Delta_{\rm RMT}\sqrt{|{\bm A}|}}
        \left(\int_{-\infty}^{\infty}\diff\nu\, \e^{\beta\nu/4}\eta_{\nu}|f(\nu)|\right)^2
    \right)
    \qquad
    \text{where we used }\Delta_{\rm RMT}\sum_{m}\approx \int\diff\nu
    \nonumber\\
    &\hspace{-.6em}\overset{\eqref{eq:CSineq_integral}}{=}
    \calO\left(
        \frac{1}{\Delta_{\rm RMT}\sqrt{|{\bm A}|}}
    \right).
    \label{eq:global_AxA}
\end{align}
In the inequality \eqref{eq:Delta_K_AA_i_maximization}, we defined $\hat\imath(m)\coloneq\underset{i}{\rm argmax}\,\e^{\frac{\beta}{4}(\hat E_{i+m}-\hat E_i)} \eta_{\hat E_{i+m}-\hat E_{i}} |f(\hat E_{i+m}-\hat E_i)|$, and suppressed its $m$-dependence for the notational simplicity.
In the inequality~\eqref{eq:Delta_K_AA_nu_maximization}, $\hat\nu_m$ is given by $\nu_m\in\bigcup_i\{E_{i+m}-E_i| E_i\in S_i, E_{i+m}\in S_{i+m}\}$ that maximizes $\e^{\beta \nu_m/4}\eta_{\nu_m}|f(\nu_m)|$.
The last equality follows from the Cauchy-Schwarz inequality,
\begin{align}
\label{eq:CSineq_integral}
    \left(
    \int_{-\infty}^{\infty}\diff\nu\,\e^{\beta\nu/4}\eta_{\nu}|f(\nu)|\right)^2
    \le
    \underbrace{\int_{-\infty}^{\infty}\diff\nu\, \eta_\nu^2}_{=1}
    \int_{-\infty}^{\infty}\diff\nu'\, \e^{\beta\nu'/2}|f(\nu')|^2
    \overset{\eqref{eq:f_scale_derivation}}{=} 
    \calO(1).
\end{align}

Next, we bound the norm of $\Delta K_{I\otimes AA}\coloneq\sum_a\sum_{\mu,\nu}\eta_\mu\eta_\nu I\otimes\big((A^a_\mu)^\dag A^a_\nu - \bE_R[(A^a_\mu)^\dag A^a_\nu]\big)$.
\begin{align}
    &\left(\bE_R\big\|(\Pi_{S_k}\otimes\Pi_{S_l})\Delta K_{I\otimes AA}(\Pi_{S_i}\otimes\Pi_{S_j})\big\|_\infty\right)^p
    \nonumber\\
    &\le
    \bE_R\big\|(\Pi_{S_k}\otimes\Pi_{S_l})\Delta K_{I\otimes AA}(\Pi_{S_i}\otimes\Pi_{S_j})\big\|_p^p
    \nonumber\\
    &\le
    \bE_R \left\|\Pi_{S_l}
        \frac{1}{|{\bm A}|}\sum_a\sum_{\mu,\nu}\e^{-\frac{\beta}{4}(\mu-\nu)} \eta_\mu\eta_\nu
        \big((A^a_\mu)^\dag A^a_\nu - (A'^a_\mu)^\dag A'^{a}_\nu\big)
    \Pi_{S_j}\right\|_p^p
    \nonumber\\
    &=
    \bE_R \left\|
        \frac{2}{|{\bm A}|}\sum_a
            \sum_{\mu,\nu}\e^{-\frac{\beta}{4}(\mu-\nu)} \eta_\mu\eta_\nu 
            \Pi_{S_l} (B^a_\mu)^\dag B'^a_\nu\Pi_{S_j}
    \right\|_p^p
    \nonumber\\
    &\le
    \bE_R \left\|
        \frac{2}{|{\bm A}|}\sum_a\sum_{h}
        G_{S_l,S_h} G_{S_h,S_j}
    \right\|_p^p
    \nonumber\\
    &\overset{\eqref{eq:decouple_gaussian_IxAA}}{\le}
    \left(\frac{2}{\sqrt{|{\bm A}|}}\right)^{p} 
    \bE_R\left\|\sum_{h}G_{S_l,S_h}G_{S_h,S_j}\right\|_p^p
    \nonumber\\
    &\overset{\eqref{eq:p_to_infty_norm}}{\le}
    \left(\frac{2}{\sqrt{|{\bm A}|}}\right)^{p}
    \Big(
        \sum_{h}
        \max\{\Tr[\Pi_{S_l}],\Tr[\Pi_{S_h}]\}
        \bE_R[|G_{S_l,S_h}|^2]
        \max\{\Tr[\Pi_{S_h}],\Tr[\Pi_{S_j}]\}
        \bE_R[|G_{S_h,S_j}|^2]
    \Big)^{p/2}.
\end{align}
In the second last inequality, the Gaussian matrices are decoupled as
\begin{align}
\label{eq:decouple_gaussian_IxAA}
\begin{split}
    &\bE_R \Big\|\sum_a c_a G^a G'^a\Big\|_p^p
    \\
    &=
    \sum_{a_1,\dots,a_p} c_{a_1}c_{a_2}^*\dots c_{a_{p-1}}c_{a_p}^*
    \bE_R\Tr[G^{a_1}G'^{a_1} G^{a_2\dag}G'^{a_2\dag}\dots G^{a_{p-1}}G'^{a_{p-1}} G^{a_p\dag}G'^{a_p\dag}]
    \\
    &\le
    \big(\sum_a|c_a|^2\big)^{p/2}
    \sum^{\text{pairs}}_{a_1,\dots,a_p}
    \bE_R\Tr[G^{a_1}G'^{a_1} G^{a_2\dag}G'^{a_2\dag}\dots G^{a_{p-1}}G'^{a_{p-1}} G^{a_p\dag}G'^{a_p\dag}]
    \\
    &=
    \big(\sum_a|c_a|^2\big)^{p/2}
    \big(\bE_R\|G^1 G'^1\|_p^p\big)^2.
\end{split}
\end{align}
The sum $\sum^\text{pairs}_{a_1,\dots,a_p}$ only runs over the indices that consist only of pairs.
The last inequality holds because $\bE_R\Tr[G^{a_1} G^{a_2\dag}\dots G^{a_{p-1}} G^{a_p\dag}]\ge 0$ holds when all the indices are paired.
The last inequality follows from,
\begin{align}
\label{eq:p_to_infty_norm}
    \bE_R\big\|G G'\|_p^p 
    \le 
    d_2\bE_R\big\|G\|_\infty^p \cdot \bE_R\big\|G'\|_\infty^p
    \overset{\eqref{eq:bound_p_norm}}{\le} 
    \sqrt{\Var_G\Var_{G'}\cdot\max\{d_1,d_2\}\max\{d_2,d_3\}},
\end{align}
for a $d_1\times d_2$ matrix $G$ and a $d_2\times d_3$ matrix $G'$. The even integer $p$ is chosen so that $d_2^{1/p}\le1$.

Therefore, we can upper-bound the fluctuation in each eigensector,
\begin{align}
\label{eq:bound_local_IxAA}
    &\bE_R\big\|(\Pi_{\bar{F}_1}\otimes\Pi_{S_l})\Delta K_{I\otimes AA}(\Pi_{\bar{E}_1}\otimes\Pi_{S_j})\big\|_\infty
    \nonumber\\
    &\le
    \frac{\Delta_{\rm RMT} R_D}{\sqrt{|{\bm A}|}}
    \sum_{S_h}
    \sqrt{
        \max_{E\in S_l\cup S_h}D(E)
        \max_{F\in S_h\cup S_j}D(F)
    }
    \nonumber\\
    &\times
    \max_{E_l\in S_l, E_h\in S_h}
    \frac{\e^{\frac{\beta}{4}(E_l-E_h)}\eta_{E_l-E_h}\big|f\big(\frac{E_l+E_h}{2}, E_l-E_h\big)\big|}{\sqrt{D\big(\frac{E_l+E_h}{2}\big)}}
    \max_{E_h\in S_h, E_j\in S_j}
    \frac{\e^{\frac{\beta}{4}(E_h-E_j)}\eta_{E_h-E_j}\big|f\big(\frac{E_h+E_j}{2}, E_h-E_j\big)\big|}{\sqrt{D\big(\frac{E_h+E_j}{2}\big)}}
    \nonumber\\
    &=
    \frac{\Delta_{\rm RMT}R_D^2}{\sqrt{|{\bm A}|}}\sum_{h}
    \e^{\frac{\beta}{4}(\hat{E}_l-\hat{E}_h)}\eta_{\hat{E_l}-\hat{E}_h}|f(\hat{E}_l-\hat{E}_h)|
    \ 
    \e^{\frac{\beta}{4}(\hat{E}_h-\hat{E}_j)}\eta_{\hat{E}_h-\hat{E}_j}|f(\hat{E}_h-\hat{E}_j)|,
\end{align}
where $\hat E_h, \hat E_j, \hat E_l$ are the energies solving the maximization of the third line.

Summing the contributions from all the energy eigenspaces, we find
\begin{align}
    &\bE_R\|\Delta K_{I\otimes AA}\|_\infty
    \nonumber\\
    &=
    \bE_R\left\|\sum_{j,l}(I\otimes\Pi_{S_l})\Delta K_{I\otimes AA}(I\otimes\Pi_{S_j})\right\|_\infty
    \nonumber\\
    &\le
    \sum_{m}\bE_R\left\|\sum_{j}(I\otimes\Pi_{S_{j+m}})\Delta K_{I\otimes AA}(I\otimes\Pi_{S_j})\right\|_\infty
    \nonumber\\
    &\le
    \sum_{m}\max_{j}\bE_R\left\|
    (I\otimes\Pi_{S_{j+m}})\Delta K_{I\otimes AA}(I\otimes\Pi_{S_j})\right\|_\infty
    \nonumber\\
    &\hspace{-.6em}\overset{\eqref{eq:bound_local_IxAA}}{\le}
    \frac{\Delta_{\rm RMT} R_D^2}{\sqrt{|{\bm A}|}}
    \sum_{m, h}
    \e^{\frac{\beta}{4}(\hat E_{\hat\jmath+m}-\hat E_h)}\eta_{\hat E_{\hat\jmath+m}-\hat E_h}|f(\hat E_{\hat\jmath+m}-\hat E_h)|
    \e^{\frac{\beta}{4}(\hat E_h-\hat E_{\hat\jmath})}\eta_{\hat E_h-\hat E_{\hat\jmath}}|f(\hat E_h-\hat E_{\hat\jmath})|
    \label{eq:Delta_K_IxAA_j_maximization}
    \\
    &=    
    \frac{\Delta_{\rm RMT} R_D^2}{\sqrt{|{\bm A}|}}
    \sum_{m, h}
    \e^{\frac{\beta}{4}(\hat E_{\hat\jmath+h+m}-\hat E_{\hat\jmath+h})}\eta_{\hat E_{\hat\jmath+h+m}-\hat E_{\hat\jmath+h}}|f(\hat E_{\hat\jmath+h+m}-\hat E_{\hat\jmath+h})|
    \e^{\frac{\beta}{4}(\hat E_{\hat\jmath+h}-\hat E_{\hat\jmath})}\eta_{\hat E_{\hat\jmath+h}-\hat E_{\hat\jmath}}|f(\hat E_{\hat\jmath+h}-\hat E_{\hat\jmath})|
    \nonumber\\
    &\le
    \frac{\Delta_{\rm RMT} R_D^2}{\sqrt{|{\bm A}|}}
    \sum_{h}
    \e^{\beta \hat\nu_m/4}
    \eta_{\hat\nu_m}|f(\hat\nu_m)|
    \sum_{m}
    \e^{\beta \hat\nu_h/4}
    \eta_{\hat\nu_h}|f(\hat\nu_h)|
    \label{eq:Delta_K_IxAA_nu_maximization}
    \\
    &\hspace{-.5em}\overset{\text{Ass.}\ref{as:dos_ratio}}=
    \calO\left(
    \frac{1}{\Delta_{\rm RMT} \sqrt{|{\bm A}|}}
    \left(\int_{-\infty}^{\infty}\diff\nu\,
    \e^{\beta\nu/4} \eta_{\nu}|f(\nu)|\right)^2
    \right)
    \qquad
    \text{where we used }\Delta_{\rm RMT}\sum_{h}\approx \int\diff\nu
    \nonumber\\
    &\hspace{-.6em}\overset{\eqref{eq:CSineq_integral}}{=}
    \calO\left(\frac{1}{\Delta_{\rm RMT} \sqrt{|{\bm A}|}}
    \right).
\end{align}
In the inequality \eqref{eq:Delta_K_IxAA_j_maximization}, $\hat\jmath(m)$ solves the maximization, $\hat\jmath(m)\coloneq\underset{j}{\rm argmax}\,\sum_{h}\e^{\frac{\beta}{4}(\hat E_{j+m}-\hat E_h)}\eta_{\hat E_{j+m}-\hat E_h}|f(\hat E_{j+m}-\hat E_h)|\e^{\frac{\beta}{4}(\hat E_h-\hat E_j)}\eta_{\hat E_h-\hat E_j}|f(\hat E_h-\hat E_j)|$, and its $m$-dependence is suppressed for the notational simplicity.
In the inequality~\eqref{eq:Delta_K_IxAA_nu_maximization}, $\hat\nu_x$ for an integer $x$ is given by $\nu_x\in\bigcup_i\{E_{i+x}-E_i| E_i\in S_i, E_{i+x}\in S_{i+x}\}$ that maximizes $\e^{\beta \nu_x/4}\eta_{\nu_x}|f(\nu_x)|$.

Similarly, one can obtain the bound
\begin{align}
    \bE_R\|\Delta K_{AA\otimes I}\|_\infty
    &\le
    \calO\left(\frac{1}{\Delta_{\rm RMT}\sqrt{|{\bm A}|}}\right).
\end{align}

Overall, we have
\begin{align}
\begin{split}
    \bE_R\|{\cal D} - \bE_R{\cal D}\|_{\sigma_\beta^{-1},2\to2}
    &=
    \bE_R\|\Delta K_{A\otimes A}
    + \Delta K_{I\otimes AA}
    + \Delta K_{AA\otimes I} \|_{\infty}
    \le
    \calO\left(\frac{1}{\Delta_{\rm RMT}\sqrt{|{\bm A}|}}\right)
\end{split}
\end{align}
Furthermore, a careful concentration analysis in~\cite{Chen2021ETH} together with Eq.~\eqref{eq:app:L_dist_to_D_dist} shows that
\begin{align}
\label{eq:L_fluctuation}
\begin{split}
    \|{\calL} - \bE_R{\calL}\|_{\sigma_\beta^{-1},2\to2}
    =
    \|{\calD} - \bE_R{\calD}\|_{\sigma_\beta^{-1},2\to2}
    \le
    \calO\left(\frac{1}{\Delta_{\rm RMT}\sqrt{|{\bm A}|}}
    \right)
\end{split}
\end{align}
holds with the probability exponentially close to 1 with respect to $n$.
Using $\Delta_{\rm RMT} = \Theta(1/\beta)$, we can simplify Eq.~\eqref{eq:L_fluctuation} to 
\begin{equation}
\label{eq:L_fluctuation_simplified}
\| \calL - \bE_R \calL\|_{\sigma_\beta^{-1},2\to2}
=
\calO \left(
\frac{\beta}{\sqrt{|\bm{A}|}}
\right).
\end{equation}

\subsection{Bounding the convergence to the Gibbs state}
\label{app:ssec:spectral_gap_and_mixing_time}

In the previous sections~\ref{app:ssec:gap_EL} and \ref{app:ssec:concentration_ETH_average} we showed that the gap of the average Lindbladian $\bE_R\calL$ is lower-bounded by a polynomial in $\beta$ and the system size $n$, whereas the channel distance between the actual Lindbladian $\calL$ and its average $\bE_R\calL$ is upper-bounded by the inverse square root of the number of jump operators $|\bm{A}|$. In this section we will use both results to polynomially bound the spectral gap of the true Lindbladian $\calL$ and to show a high convergence accuracy of the Lindblad evolution, in the sense that its exact steady state $\SteadyState$ is close to the target Gibbs state $\sigma_\beta$, which is the steady state of $\bE_R\calL$.

Let us start with the spectral gap.
From the gap $\Delta_{\bE_R\calL}$ of the ETH-averaged Lindbladian~\eqref{eq:bound_gap_simplified} and the channel distance~\eqref{eq:L_fluctuation_simplified}, we obtain a lower bound on the gap $\Delta_{\cal L}$ of the Lindbladian ${\cal L}=\bE_R{\cal L} + ({\cal L}-\bE_R{\cal L})$ via
\begin{align}
\label{eq:app:L_gap_lower_bound_deriv}
    \Delta_\calL
    =
    \Omega\left(\frac{\e^{-2\beta\Delta_\mathrm{RMT}}\Delta_\mathrm{RMT}^2\Gamma}{\Delta_\mathrm{spec}^2}
    \right)
    -
    \calO\left(\frac{1}{\Delta_{\rm RMT} \sqrt{|{\bm A}|}}
    \right).
\end{align}
Hence, there exists a number of jump operators,
\begin{align}
\label{eq:num_jump}
    |{\bm A}|
    =
    \Omega\left(\frac{\e^{4\beta\Delta_\mathrm{RMT}}\Delta_\mathrm{spec}^4}{\Delta_\mathrm{RMT}^6\Gamma^2}
    \right),
\end{align}
for which the spectral gap~\eqref{eq:app:L_gap_lower_bound_deriv} is lower bounded by
\begin{align}
\label{eq:L_gap}
    \Delta_\calL
    =
    \Omega\left(\frac{\e^{-2\beta\Delta_\mathrm{RMT}}\Delta_\mathrm{RMT}^2\Gamma}{\Delta_\mathrm{spec}^2}
    \right),
\end{align}
with high probability.
Finally, using Eq.~\eqref{eq:mixing_time_and_spectral_gap} and $\log(\|\sigma_\beta^{-1}\|_{\infty}) \approx \beta\|H\|_\infty$, the lower bound of the spectral gap is turned into an upper bound of the mixing time,
\begin{align}
\label{eq:app:mixing_time_bound}
    t_\mathrm{mix}
    \le
    \calO\left(\Big(\beta\|H\|_\infty + \log\frac{1}{\epsilon}\Big) \frac{\e^{2\beta\Delta_\mathrm{RMT}}\Delta_\mathrm{spec}^2}{\Delta_\mathrm{RMT}^2\Gamma}
    \right).
\end{align}
With Assumption~\ref{as:f_ETH}\ref{as:f_ETH_b}, $\Delta_{\rm spec}\propto1/\sqrt{n}$, spectral gap~\eqref{eq:L_gap}, the sufficient number of jump operators~\eqref{eq:num_jump}, the mixing time~\eqref{eq:app:mixing_time_bound}, and the bounds on the overlap integral $\Gamma$~\eqref{eq:app:overlap_lower_bound} and \eqref{eq:app:overlap_upper_bound}, we find
\begin{align}
\label{eq:app:mixing_time_bound_final}
    \Delta_\calL
    =
    \Omega\left( \frac{1}{n \beta^2} \right),
    \qquad
    |{\bm A}| = \Omega(n^2 \beta^6),
    \qquad
    t_\mathrm{mix}
    = \calO \left(
    n\beta^2\big(\beta\|H\|_\infty + \log(1/\epsilon)\big)
    \right).
\end{align}
As a reminder, we used the assumption of ETH~\eqref{eq:ETH_app} and Assumptions~\ref{as:f_ETH}, \ref{as:dos_ratio}, and \ref{as:dos_gibbs} to derive the spectral gap.
This proves Eq.~\eqref{eq:thm:mixing_time} of Thm.~\ref{thm:convergence}.

Let us now come to the distance between Gibbs state and steady state of $\calL$.
The inequality~\eqref{eq:L_fluctuation} allows us to bound the $\chi^2$ divergence, $\|\sigma_\beta - \SteadyState\|_{\sigma_\beta^{-1},2}$, between Gibbs state and the steady state $\SteadyState$ of $\calL$, which in turn, bounds the trace distance, $\|\sigma_\beta - \SteadyState\|_1\le \|\sigma_\beta - \SteadyState\|_{\sigma_\beta^{-1},2}$~\cite{Temme2010}.
Let us simplify notation, $\calL' = \bE_R\calL$, and let $\sigma$ and $\sigma'$ be states such that $\calL[\sigma]=\calL'[\sigma']=0$.
Then
\begin{align}
\label{eq:bound_chi_square}
    \|\sigma - \sigma'\|_{\sigma^{-1},2}
    =
    \|\e^{t\calL}[\sigma] - \e^{t\calL'}[\sigma']\|_{\sigma^{-1},2}
    \le
    \|\e^{t\calL}[\sigma - \sigma']\|_{\sigma^{-1},2}
    +
    \|(\e^{t\calL} - \e^{t\calL'})[\sigma']\|_{\sigma^{-1},2}.
\end{align}
To bound the second term, we note
\begin{align}
\begin{split}
    \|\e^{t\calL} - \e^{t\calL'}\|_{\sigma^{-1},2\to2}
    &\le
    t\int_{0}^{1}\diff s \|\e^{st\calL}(\calL - \calL')\e^{(1-s)t\calL'}\|_{\sigma^{-1},2\to2}
    \\
    &\le
    t\int_{0}^{1}\diff s \|\e^{st\calL}\|_{\sigma^{-1},2\to2}\cdot \|\calL - \calL'\|_{\sigma^{-1},2\to2} \cdot \|\e^{(1-s)t\calL}\|_{\sigma^{-1},2\to2}
    \\
    &\le
    t\|\calL - \calL'\|_{\sigma^{-1},2\to2},
\end{split}
\end{align}
where in the last inequality we used
\begin{align}
    \|\e^{x\calL}[\rho]\|_{\sigma^{-1},2}
    =
    \|\e^{x\calL}[\rho - \sigma]\|_{\sigma^{-1},2} + 1
    \le
    \|\rho - \sigma\|_{\sigma^{-1},2} + 1
    =
    \|\rho\|_{\sigma^{-1},2}
\end{align}
for $x\ge0$, leading to
\begin{align}
    \|\e^{x\calL}[\rho]\|_{\sigma^{-1},2\to2} 
    =
    \max_\rho \frac{\|\e^{x\calL}[\rho]\|_{\sigma^{-1},2}}{\|\rho\|_{\sigma^{-1},2}}
    \le 1.
\end{align}
The same inequality holds for $\calL'$.
Thus, the second term in Eq.~\eqref{eq:bound_chi_square} is bounded by
\begin{align}
    \|(\e^{t\calL} - \e^{t\calL'})[\sigma']\|_{\sigma^{-1},2}
    \le
    t\|\calL - \calL'\|_{\sigma^{-1},2\to 2} \cdot
    \|\sigma'\|_{\sigma^{-1},2}.
\end{align}
Note that $\|\sigma'\|_{\sigma^{-1},2}$ can be evaluated as follows:
\begin{align}
    \|\sigma'\|_{\sigma^{-1},2} 
    \le
    \|\sigma\|_{\sigma^{-1},2} + \|\sigma' - \sigma\|_{\sigma^{-1},2}
    =
    1 + \|\sigma'-\sigma\|_{\sigma^{-1},2}.
\end{align}
To bound the first term in \eqref{eq:bound_chi_square}, we set $t=\Delta_\calL^{-1}\log(\|\sigma^{-1}\|_{\infty}/\epsilon)$ for the spectral gap $\Delta_\calL$ of $\calL$~\eqref{eq:L_gap} so that the following inequality holds,
\begin{align}
    \|\e^{t\calL}[\sigma'-\sigma]\|_{\sigma^{-1},2}
    \le
    \e^{-t\Delta_\calL}\|\sigma'-\sigma\|_{\sigma^{-1},2}
    \le
    \epsilon
\end{align}
Putting these together, we have the bound on~\eqref{eq:bound_chi_square},
\begin{align}
\begin{split}
    \|\sigma - \sigma'\|_{\sigma^{-1},2}
    &\le
    \epsilon
    +
    t\|\calL - \calL'\|_{\sigma^{-1},2\to 2}
    (1 + \|\sigma-\sigma'\|_{\sigma^{-1},2})
    \\
    &\le
    \frac{\epsilon
    +
    \Delta_\calL^{-1}\log(\|\sigma^{-1}\|_{\infty}/\epsilon)\|\calL - \calL'\|_{\sigma^{-1},2\to 2}}
    {1 - \Delta_\calL^{-1}\log(\|\sigma^{-1}\|_{\infty}/\epsilon)\|\calL - \calL'\|_{\sigma^{-1},2\to 2}},
\end{split}
\end{align}
which also upper bounds the trace distance due to $\|\sigma - \sigma'\|_{1}\le\|\sigma - \sigma'\|_{\sigma^{-1},2}$.

Setting $\sigma'=\SteadyState$, $\sigma=\sigma_\beta$, and using Eqs.~\eqref{eq:L_fluctuation_simplified} and~\eqref{eq:app:mixing_time_bound_final},
we can bound
\begin{align}
\label{eq:app:trace_distance_deriv}
    \Delta_\calL^{-1}\log(\|\sigma^{-1}\|_{\infty}/\epsilon)\|\calL - \calL'\|_{\sigma^{-1},2\to 2}
    =
    \calO \left( \frac{n\beta^3}{\sqrt{|\bm{A}|}} \left( 
    \beta \|H\|_\infty + \log(1/\epsilon)
    \right) \right),
\end{align}
where we used $\log(\|\sigma_\beta^{-1}\|_{\infty}) \approx \beta\|H\|_\infty$.
With this, we obtain the following upper bound on the trace distance,
\begin{align}
\begin{split}
\label{eq:app:trace_distance_simplified}
    \|\SteadyState - \sigma_\beta\|_{1}
    &\le
    \|\SteadyState - \sigma_\beta\|_{\sigma_\beta^{-1},2}
    \overset{\mathrm{w.h.p}}\le
    \epsilon
    +
    \big(\Delta_\calL^{-1}\log(\|\sigma_\beta^{-1}\|_{\infty}/\epsilon)\|\calL - \calL'\|_{\sigma^{-1},2\to 2}
    \\
    &=
    \epsilon + \calO \left(\frac{n\beta^3}{\sqrt{|\bm{A}|}} \left( 
    \beta \|H\|_\infty + \log(1/\epsilon)
    \right) \right)
    =
    2\epsilon.
\end{split}
\end{align}
This proves Eq.~\eqref{eq:thm:convergence_accuracy} in Thm.~\ref{thm:convergence}.
The second inequality holds with a high probability.
In the last equality, we set the number of jump operators as
\begin{align}
\label{eq:num_jump_2}
    |{\bm A}|
    = 
    \Omega\left(\frac{n^2\beta^8\|H\|_\infty^2}{\epsilon}\right),
\end{align}
so that $\rho_\infty$ is $2\epsilon$-close to $\sigma_\beta$ in trace distance. Choosing the number of jump operators as Eq.~\eqref{eq:num_jump_2} automatically guarantees the second equation in Eq.~\eqref{eq:app:mixing_time_bound_final}.

\section{Details on numerical experiments and further results}
\label{app:numerics}

We discuss further numerical results for the mixing time and convergence accuracy of the Lindblad dynamics~\eqref{eq:lindblad_eq}. We expand the analysis of Sec.~\ref{sec:numerics} by investigating the impact of the locality $k$ of the jump operators, and of the initial state in App.~\ref{app:numerics_mixing_time}.
Furthermore, we extend our analysis to more parameter points, exploring further regular limits and the edge of chaos of the model~\eqref{eq:ising_model}.
We start this section by giving a more detailed overview of fractal dimensions and eigenstate delocalization, from which we identify the various dynamical regimes of our mixed-field Ising model~\eqref{eq:ising_model}.
In App.~\ref{app:numerics_details}, we give further details on our Lindbladian simulation scheme.
Throughout this section we fix $\beta=(2J)^{-1}$, as in the main text.

\subsection{Eigenstate delocalization}
\label{app:numerics_spectral_statistics}

\begin{figure*}
    \centering
    \includegraphics[width=\textwidth]{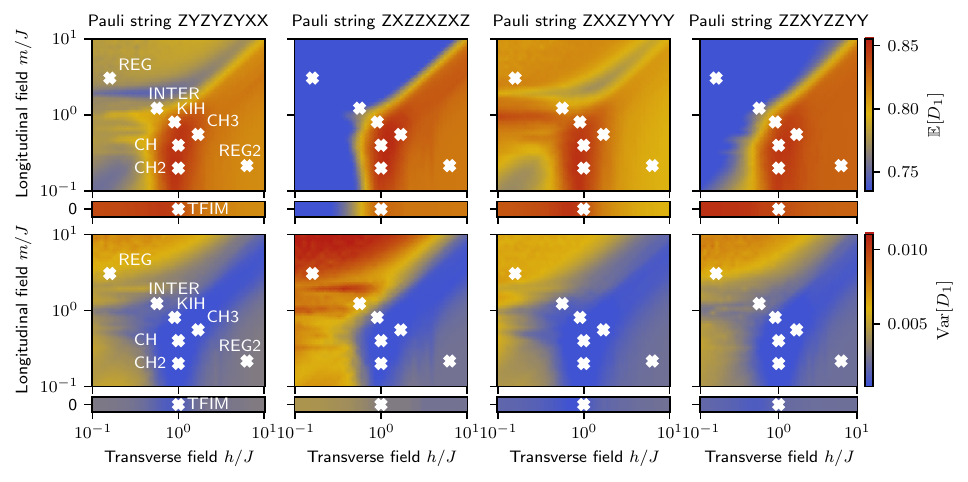}
    \caption{
    Fractal dimension analysis of the mixed-field Ising model~\eqref{eq:ising_model} for $n=8$ qubits. The upper and lower rows show, respectively, mean $\mathbb{E}[D_1]$ and variance $\text{Var}[D_1]$ of the first fractal dimension, Eq.~\eqref{eq:fractal_dimensions}, evaluated over the inner $80\%$ of the spectrum, as a function of transverse and longitudinal field strengths.
    The four columns correspond to four different Pauli strings (given in each column's title) which define the basis in which the fractal dimension is evaluated.
    Large mean values $\mathbb{E}[D_1]$
    in conjunction with small $\text{Var}[D_1]$ are a signature of quantum chaos.
    We identify a clear region of uniform eigenstate delocalization below the center ($m = h = 1$) of the phase diagrams, which is present for all considered Pauli bases.
    Within the phase diagram, we identify eight parameter choices (white crosses), with broadly varying eigenstate characteristics, of interest, which we will consider in more depth in the following convergence analysis. They are summarized in Tab.~\ref{tab:hamiltonian_parameters}.
    } \label{fig:app_spectral_statistics_n_qubits8}
\end{figure*}

Eigenstate thermalization is expected to accurately describe the behavior of quantum chaotic systems~\cite{Srednicki1999}.
Typical signatures of quantum chaos comprise, for example, spectral statistics~\cite{oganesyan_localization_2007,atas_distribution_2013} or the delocalization behavior of energy eigenstates~\cite{kolovsky_quantum_2004,atas_multifractality_2012,beugeling_global_2015,atas_quantum_2017,beugeling_statistical_2018,pausch_chaos_2021,brunner_many-body_2023}.
We identify sets of parameters $(h,m)$ of the mixed-field Ising model~\eqref{eq:ising_model} with distinct quantum chaotic properties by following the approach advocated in \cite{pausch_chaos_2021,brunner_many-body_2023}, based on delocalization of energy eigenstates.
This approach directly assesses the crucial properties underlying the ETH, while simultaneously being independent of a specific choice of observables, i.e. of the Lindbladian jump operators in our case.

\subsubsection{Eigenstate delocalization and fractal dimension}
Quantum chaos is signaled by a strong and uniform delocalization of a substantial part of the energy eigenstates in almost any basis of the Hilbert space (excluding pathological cases such as the eigenbasis itself) \cite{kolovsky_quantum_2004,atas_multifractality_2012,beugeling_global_2015,atas_quantum_2017,beugeling_statistical_2018,pausch_chaos_2021,brunner_many-body_2023}.
A way to quantify eigenstate delocalization in a given Hilbert space basis $\ket{\zeta}$ is via so-called fractal dimensions.
For an eigenstate $\ket{E_i} = \sum_{\zeta = 1}^{\dim\calH} \psi_\zeta^i \ket{\zeta}$ of~\eqref{eq:ising_model} the finite-size generalized fractal dimensions are defined by
\begin{equation}
\label{eq:fractal_dimensions}
    D_q^{(i)} = \frac{1}{1-q} \log_{\dim \calH} R_q^{(i)} \,, \quad R_q^{(i)} = \sum_{\zeta=1}^{\dim \calH} |\psi^i_\zeta|^{2q}\,, \quad q \in \mathbb{R} \,,
\end{equation}
where $\dim \calH$ is the dimension of the Hilbert space.
In the following we will omit the superscript $(i)$ of $D_q$ if it is not necessary.
For $q=1$ the definition is given by taking the limit, $D_1 = \lim_{q \rightarrow 1} D_q$, which coincides with the normalized Shannon entropy of the probability distribution $\{|\psi_\zeta|^2\}_\zeta$.
In the large-system limit one can identify three different regimes: localized (the eigenstate is supported only on a vanishingly small portion of basis states, with $\lim_{\dim\calH \rightarrow \infty} D_q = 0$ for all $q \geq 1$), multifractal or extended non-ergodic ($0 \leq \lim_{\dim\calH \rightarrow \infty} D_q \leq 1$ dependent on $q$) and ergodic ($\lim_{\dim\calH \rightarrow \infty} D_q = 1$ for all $q$, which holds for the limit of the uniform distribution $|\psi^i_\zeta|^2 = 1/\dim\calH$ for all $\zeta$).

For uniform eigenstate delocalization in a given expansion basis, two statistical signatures are crucial: Firstly, a large average fractal dimension $\mathbb{E}[D_q]$, where the average $\bE$ is taken over the eigenstates $\{\ket{E_i}\}_i$, indicates that almost all eigenstates are strongly delocalized. Secondly, a small variance $\Var[D_q]$ over the spectrum signals that this delocalization happens uniformly over all eigenstates. In the following we will focus on the first fractal dimension, $q=1$, which is sufficient for our purposes.

To gain a comprehensive picture of the dynamical landscape of the mixed-field Ising model~\eqref{eq:ising_model}, we analyze $D_1$ in the full parameter space of our Hamiltonian.
To filter out potentially untypical eigenstates at the spectral edges, we focus on the bulk of the spectrum, and compute mean $\mathbb{E}[D_1]$ and variance $\text{Var}[D_1]$ over the inner $80\%$ (in terms of energy) of the eigenstates.
To consolidate our analysis in Fig.~\ref{fig:spectral_statistics_n_qubits8}, where we consider the fractal dimension in the $n$-qubit $Z$-basis, we evaluate $D_1$ in four additional random $n$-qubit Pauli bases. The results are shown in Fig.~\ref{fig:app_spectral_statistics_n_qubits8} for $n=8$.
The four columns correspond to the four chosen Pauli bases, given in the column titles.
For each column, we show the mean value $\mathbb{E}[D_1]$ (top row) and variance $\text{Var}[D_1]$ (bottom row), for the Ising Hamiltonian~\eqref{eq:ising_model} in the parameter range $m \in \{0\} \cup [10^{-1}, 10^1]$, $h \in [10^{-1}, 10^1]$.
As before, we also include the value $m=0$, which corresponding to the (integrable) transverse-field Ising model.

As in Fig.~\ref{fig:spectral_statistics_n_qubits8}, the parameter domain $0.7 \lesssim h/J \lesssim 2$, $0.2 \lesssim m/J \lesssim 0.9$ shows strong uniform eigenstate delocalization for all considered Pauli bases.
In this regime we observe a consistent maximization of $\bE[D_1]$, which is accompanied by a drop of variance $\Var[D_1]$, signaling uniform properties of the eigenstates over the bulk of the spectrum.
The fact that this happens independently of the chosen basis, confirms that this region is indeed characterized by strong quantum chaotic behavior.
Consequently, we expect our model to align well with the ETH predictions in this parameter regime.
Note that the remaining parameter regions show different signatures for each chosen basis. This is expected. The fractal dimensions are, by construction, basis-dependent and we expect traces of this basis-dependence to be visible within the non-chaotic regimes of the model.
Physically, we can identify three \textit{regular limits} of the Ising model~\eqref{eq:ising_model}; $m/J \rightarrow \infty$,  $h/J \rightarrow \infty$ or $m, h\rightarrow 0$.
In these limits the Hamiltonian is effectively given by only a single term.
In general, we expect ETH to be not applicable in the regular limits.
In the limit $m/J \rightarrow 0$, the mixed-field Ising model approaches the transverse field Ising model, which can be mapped to free fermions via the Jordan-Wigner transformation. This is an integrable (i.e. non-chaotic) model---however, it can exhibit critical behavior.

\begin{table}
\centering
\begin{tabular}{c|c|ccccccc}
    Key & Hamiltonian parameters & Step size
    &$n=3$ &$n=4$ &$n=5$ &$n=6$ &$n=7$ &$n=8$ \\ \hline
    \texttt{TFIM} &
    $h/J = 1.0$, $m/J = 0.0$ &
    &0.25 &0.25 &0.125 &0.125 &0.125 &0.125 \\
    \texttt{CH2} &
    $h/J = 1.0$, $m/J = 0.2$ &
    &0.25 &0.25 &0.125 &0.125 &0.125 &0.125 \\
    \texttt{CH} &
    $h/J = 1.0$, $m/J = 0.4$ &
    &0.25 &0.25 &0.125 &0.125 &0.125 &0.125 \\
    \texttt{KIH} &
    $h/J = 0.9045$, $m/J = 0.8090$ &
    $J\Delta t = $
    &0.25 &0.125 &0.125 &0.125 &0.125 &0.0625 \\
    \texttt{REG} &
    $h/J = 0.1585$, $m/J = 3.062$ &
    &0.125 &0.0625 &0.0625 &0.0625 &0.0625 &0.03125 \\
    \texttt{INTER} &
    $h/J = 0.5623$, $m/J = 1.230$ &
    &0.25 &0.125 &0.125 &0.125 &0.0625 &0.0625 \\
    \texttt{CH3} &
    $h/J = 1.698$, $m/J = 0.5551$ &
    &0.125 &0.125 &0.125 &0.0625 &0.0625 &0.0625 \\
    \texttt{REG2} &
    $h/J = 6.310$, $m/J = 0.2158$ &
    &0.0625 &0.03125 &0.03125 &0.03125 &0.03125 &0.015625 \\
\end{tabular}
\caption{
Summary of the Hamiltonian parameter configurations considered in our numerical experiments.
For each configuration, we give the label, the corresponding coordinates $h/J$ and $m/J$, and the time step sizes for the considered values of $n$.
The step sizes are determined by our numerical scheme as described in App.~\ref{app:numerics_details}.}
\label{tab:hamiltonian_parameters}
\end{table}

\subsubsection{Hamiltonian parameter points for the numerical investigation}
To investigate the dynamical regimes of the mixed-field Ising model, we consider several parameter configurations, covering different dynamical regimes of the Hamiltonian~\eqref{eq:ising_model}, for our simulations.
These configurations are shown in Fig.~\ref{fig:app_spectral_statistics_n_qubits8} as white crosses with corresponding labels.
We consider points along an anti-diagonal through the center of the phase diagram ($m = h = 1$) (\texttt{REG}, \texttt{INTER}, \texttt{CH3}, \texttt{REG2}).
The points \texttt{REG} and \texttt{REG2} are exemplary for the two regular model limits, where either the longitudinal or the transverse field dominates the dynamics, cf. Eq.~\eqref{eq:ising_model}.
The point \texttt{CH3} is located within the chaotic domain of the model.
Additionally, we investigate the point \texttt{KIH}, considered in
\cite{kim_ballistic_2013} as a ``robustly non-integrable'' parameter point based on the distribution of level spacing ratios.
From the standpoint of eigenstate delocalization, this parameter point seems to be located close to the edge of chaos.
We further explore the transition towards the transverse-field Ising model ($m \rightarrow 0$) with the three points \texttt{TFIM}, \texttt{CH2} and \texttt{CH}, located on the vertical line $h=1$ for three different values of $m$.
All eight parameter points are summarized in Tab.~\ref{tab:hamiltonian_parameters}.
According to our fractal dimension analysis, we expect that the configurations \texttt{CH}, \texttt{CH2} and \texttt{CH3}, show the clearest signatures of quantum chaos and, therefore, that ETH is applicable at these points.

\subsubsection{Distribution of level spacing ratios}
\label{sec:app_mean_level_spacing_ratio}

\begin{figure*}
    \centering
    \includegraphics[width=0.52\textwidth]{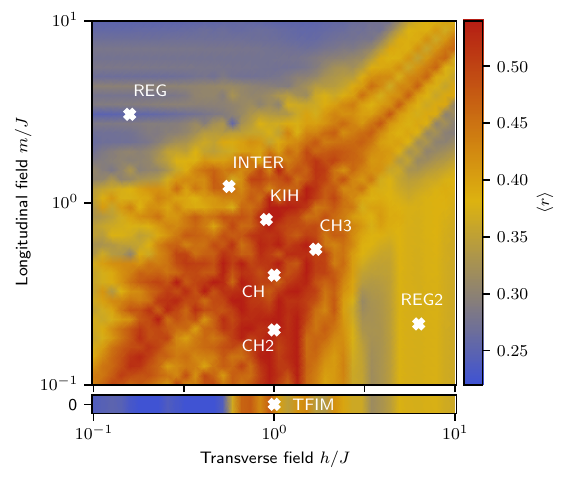}
    \caption{
    Mean level spacing ratio of the even-parity symmetry sector of the mixed-field Ising model~\eqref{eq:ising_model} for $n=8$ qubits. The level spacing ratio distriution is computed in the even parity symmetry sector of the model.
    In the quantum chaotic parameter region, its mean value $\expval{r}$ approaches the prediction from the Gaussian orthogonal ensemble (GOE), given by $\expval{r}_\mathrm{GOE} = 0.5307$ \cite{pausch_chaos_2021,atas_distribution_2013}.
    We observe that the region identified as quantum chaotic by the mean level spacing ratio is well-aligned with the regime identified from the eigenstate delocalization analysis shown in Fig.~\ref{fig:app_spectral_statistics_n_qubits8}.
    The considered parameter points used for our numerical investigation (white crosses) are shown again for comparison. They are summarized in Tab.~\ref{tab:hamiltonian_parameters}.
    } \label{fig:app_spectral_statistics_level_spacing_ratio_n_qubits8}
\end{figure*}

For comparison, we provide numerical data for another frequently employed signature of quantum chaos, the level spacing ratio distribution $r$~\cite{oganesyan_localization_2007,atas_distribution_2013}.
Let $E_i$ denote the ordered energies (in ascending order) of the Hamiltonian, in our case the mixed-field Ising model~\eqref{eq:ising_model}, and $s_i = E_{i+1} - E_i$ the $i$-th level spacing.
The level spacing ratios are defined by $r_i = \mathrm{min}(s_{i+1} / s_i, s_i / s_{i+1})$.
In quantum chaotic systems, the distribution of level spacing ratios follows the corresponding distribution of an appropriate ensemble of random matrices.
In our case of the mixed-field Ising model~\eqref{eq:ising_model}, which is time-reversal symmetric, the appropriate ensemble is the Gaussian orthogonal ensemble (GOE).
Typically, one simply considers the mean value of level spacing ratios $\expval{r}$ as quantum chaos signature, which is predicted to be $\expval{r}_\mathrm{GOE} = 0.5307$ by GOE.

Figure~\ref{fig:app_spectral_statistics_level_spacing_ratio_n_qubits8} shows $\expval{r}$ for the mixed-field Ising model in the parameter range $m \in \{0\} \cup [10^{-1}, 10^1]$, $h \in [10^{-1}, 10^1]$, as considered in Fig.~\ref{fig:app_spectral_statistics_n_qubits8}.
Note that symmetries of the model will strongly alter the distribution of energy levels, hence one first has to remove those symmetries, respectively work in the individual symmetry sectors.
Our considered mixed-field Ising model with open boundary conditions~\eqref{eq:ising_model} is symmetric with respect to reflection at the center of the chain. We remove this symmetry by projecting the Hamiltonian to the even parity symmetry sector.
The observed region of quantum chaos, $\expval{r} \approx \expval{r}_\mathrm{GOE}$, aligns well with the parameter region identified in Fig.~\ref{fig:app_spectral_statistics_n_qubits8} by the transition to uniformly delocalized energy eigenstate.

\subsection{Mixing time and convergence accuracy}
\label{app:numerics_mixing_time}

\begin{figure*}
    \centering
    \includegraphics[width=\textwidth]{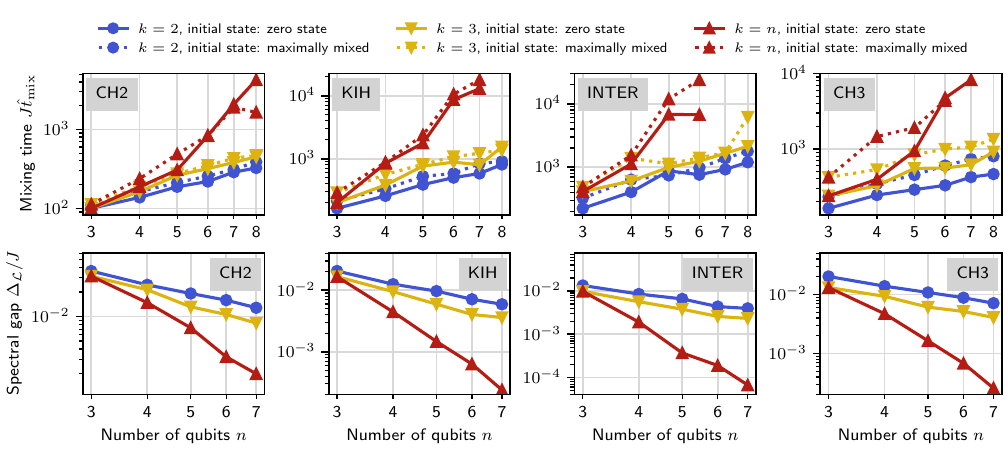}
    \caption{
    Scaling of mixing time (top row) and spectral gap (bottom row) for varying $k = 2,3,n$, indicated by color, and two different initial states, the maximally mixed state $I/2^n$ and the zero state $\ket{0}^{\otimes n}$.
    We observe a general increase of the mixing time, and a corresponding decrease of $\Delta_\calL$, with the locality $k$ of the jump operators. Especially non-local jump operators, $k=n$, show a significantly faster increase (decrease) of the mixing time (spectral gap), as compared to local operators.
    In most cases, the zero state converges faster to the steady state.
    This effect is most prominent for the chaotic Hamiltonian parameter point \texttt{CH3}, see Tab.~\ref{tab:hamiltonian_parameters}.
    }
    \label{fig:app:mixing_time_appendix_random_k_body_Pauli}
\end{figure*}

We extend our analysis of mixing time and convergence accuracy in Secs.~\ref{sec:numerics:mixing_time} and~\ref{sec:numerics:convergence_accuracy}.
Especially, we take a closer look at the influence of the locality $k$ of our jump operators and of the initial state on the Lindblad evolution
in App.~\ref{app:non-local_jump_operators}.
In App.~\ref{app:bohr_frequencies} we discuss the influence of the Bohr frequencies and the initial state's energy distribution on convergence accuracy and mixing time.
We also investigate the influence of the coherent term $-\im[H, \rho]$ in the Lindbladian on the mixing dynamics
in App.~\ref{app:coherent_term}.

\subsubsection{Non-local jump operators and initial state}
\label{app:non-local_jump_operators}

\begin{figure*}
    \centering
    \includegraphics[width=0.99\textwidth]{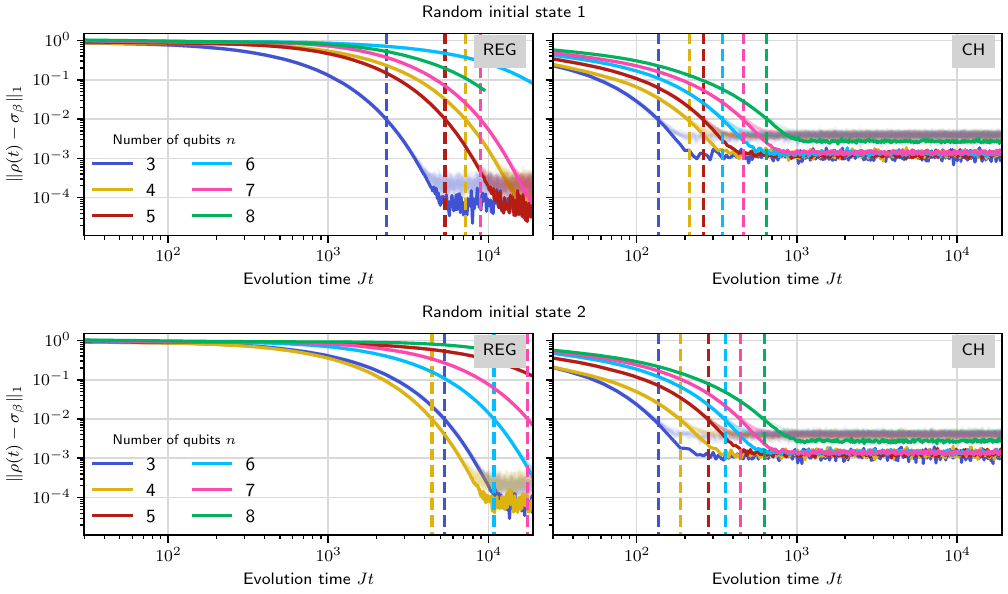}
    \caption{
    Convergence of the Lindblad dynamics for two random pure initial states (rows).
    We plot trace distance versus evolution time (in units of $1/J$) for Hamiltonian parameters \texttt{REG} (left) and \texttt{CH} (right), as in Fig.~\ref{fig:time_evolution_random_k_body_Pauli_k_order2}.
    For \texttt{CH}, the mixing time is comparable for both states (as also for the maximally mixed initial state considered in Fig.~\ref{fig:time_evolution_random_k_body_Pauli_k_order2}), and increases monotonically with system size $n$.
    For the point \texttt{REG}, the mixing time strongly fluctuates with varying $n$.
    The mechanism behind this is explained in more depth in App.~\ref{app:bohr_frequencies}.
    }
    \label{fig:app:random_initial_states}
\end{figure*}

The upper row of Fig.~\ref{fig:app:mixing_time_appendix_random_k_body_Pauli} shows the scaling of the mixing time estimate $\hat{t}_\mathrm{mix}$~\eqref{eq:mixing_time} with $n$, for two initial states and three values of $k=2,3,n$.
As Hamiltonian parameter points we consider \texttt{CH2}, \texttt{KIH}, \texttt{INTER} and \texttt{CH3}, cf. Tab.~\ref{tab:hamiltonian_parameters}, i.e. the parameter sets that are not considered in Sec.~\ref{sec:numerics}.
Note that the point \texttt{REG2}, which is dominated by the transverse field, fails to converge below convergence accuracy $\Vert \sigma_\beta - \rho(t) \Vert_1= 10^{-2}$ within the maximal time horizon considered in our simulations for all cases, see App.~\ref{app:numerics_details} for details.
Therefore, we do not have mixing time estimates for this point.
We fix $|\bm{A}| = 50$ random $k$-local Pauli jump operators, cf. Sec.~\ref{sec:numerics}, in Fig.~\ref{fig:app:mixing_time_appendix_random_k_body_Pauli}.
As an initial state, we compare the maximally mixed state $I/2^n$ (solid lines), with $I$ the $2^n\times 2^n$ identity, to the zero state $\ket{0}^{\otimes n}$ (dotted lines). The different values of $k$ are indicated by color.
Firstly, we observe that the chaotic parameter points \texttt{CH2} and \texttt{CH3} show faster convergence than the other two points, i.e. shorter mixing time $\hat{t}_\mathrm{mix}$. The point \texttt{KIH}, according to our analysis in Fig.~\ref{fig:app_spectral_statistics_n_qubits8} located at the edge of chaos, still shows a faster mixing than the non-chaotic point \texttt{INTER}.
These observations are in line with the results discussed in Fig.~\ref{fig:mixing_time_spectral_gap_random_k_body_Pauli_k_order2}.
The locality of the jump operators has a huge influence on the mixing time.
Note that the ETH ansatz~\eqref{eq:ETH} assumes local operators with $k \ll n$.
Indeed, for local jump operators, $k=2,3$, $\hat{t}_\mathrm{mix}$ scales polynomially in $n$ with approximately the same small exponent for both initial states and both values of $k$.
This confirms our observations Fig.~\ref{fig:mixing_time_spectral_gap_random_k_body_Pauli_k_order2}, showing that our analytical bounds are applicable in this case.
In turn, non-local jump operators, $k = n$, show a a much faster growth of the mixing time with $n$ (for both initial states), potentially even with a super-polynomial dependency on $n$.
To validate this observation, simulations of larger system sizes would be necessary.
Regarding the influence of the initial state, for the chaotic point \texttt{CH3}, we observe that choosing the zero state as the initial state gives consistently shorter mixing times as compared to starting from the maximally mixed state. For the remaining parameter points, this effect is less pronounced, but we still observe a weak increase in mixing time for the maximally mixed state in most cases.

In the lower row of Fig.~\ref{fig:app:mixing_time_appendix_random_k_body_Pauli} we show the spectral gap $\Delta_\calL$ as a function of $n$ for the same Hamiltonian parameter points.
For $k=2,3$ we observe, similar to Fig.~\ref{fig:mixing_time_spectral_gap_random_k_body_Pauli_k_order2}, a polynomial decrease of the gap with $n$, in line with our theoretically predicted bound~\eqref{eq:L_gap}.
For non-local jump operators, $k=n$, we observe a much faster decay of $\Delta_\calL$ with $n$.
Whether this decay is polynomially or super-polynomially, we cannot reliably assess based on the limited number of data points.
As expected, the gap mirrors the behavior of the mixing time; larger mixing times correspond to smaller gaps, and vice verse.
This supports the theoretical bound~\eqref{eq:mixing_time_and_spectral_gap}, stating that the mixing time is mainly controlled by the spectral gap of the Lindbladian generator.

Figure~\ref{fig:app:mixing_time_appendix_random_k_body_Pauli} shows that for the considered Hamiltonian parameters, the initial state does not have a strong influence on the mixing time. Moreover, mixing time increases, as expected, monotonically with system size.
Note, however, that the four parameter combinations are all sufficiently far away from the regular limits of our model, and the above described behavior can strongly change in these limits.
For example, the left panel of Fig.~\ref{fig:time_evolution_random_k_body_Pauli_k_order2} shows that for the regular point \texttt{REG}, mixing time does not monotonically increase with system size, e.g. for $n=5$ the dynamics converges faster to the steady state than for $n=3$.

To underline this effect, we compare in Fig.~\ref{fig:app:random_initial_states} the convergence dynamics for parameter points \texttt{REG}, \texttt{CH} and for two pure initial states sampled from the Haar distribution.
While for the chaotic point \texttt{CH}, the mixing time of all initial states shows roughly the same dependence on system size $n$, for \texttt{REG} the mixing time strongly fluctuates as $n$ is varied.
For example, $n=8$ for random state 1 shows comparably fast mixing, while $n=5$ for random state 2 mixes extremely slowly.
We observed similar effects for various other random initial states.
The reason behind this is a strong sensitivity of the dynamics to the initial state's energy distribution in the regular limits.
We describe this phenomenon in more detail in App.~\ref{app:bohr_frequencies} by connecting the initial states' energy distribution with the Bohr frequency spectrum of the Hamiltonian.

\subsubsection{Bohr frequencies and convergence}
\label{app:bohr_frequencies}

\begin{figure*}
    \centering
    \includegraphics[width=\textwidth]{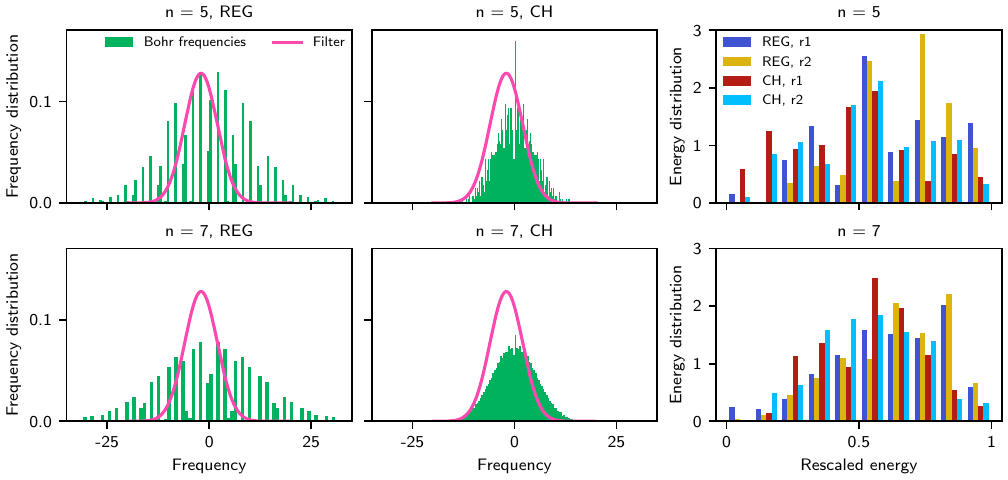}
    \caption{
    Bohr frequencies and initial state energy distribution.
    The left and middle column show the distribution of Bohr frequencies (as normalized histograms) of the mixed-field Ising model~\eqref{eq:ising_model} for parameter points \texttt{REG} and \texttt{CH} and $n=5, 7$ (rows). The pink curves schematically depict the Gaussian filter function~\eqref{eq:eta_nu} of our Lindblad operators~\eqref{eq:L_operators} in the frequency domain.
    The regular limit \texttt{REG} shows a discontinuous distribution due to many energy levels that are almost degenerate. In comparison, the Bohr frequency distribution of the chaotic point \texttt{CH} is smooth with a Gaussian-like shape.
    The right column shows the energy distribution (as normalized histograms) of two exemplary random pure initial states (state 1 and state 2 in Fig.~\ref{fig:app:random_initial_states}) for the two parameter points.
    The interplay between initial energy distribution, Bohr frequencies and filter function determines the mixing properties of the dynamics.
    In the regular point \texttt{REG} the initial state has a much stronger effect on the dynamics due to the discreteness of the Bohr frequencies, as can be seen in Fig.~\ref{fig:app:random_initial_states}.
    }
    \label{fig:app_bohr_frequencies}
\end{figure*}

\paragraph{Mixing time}
The left and middle column of Fig.~\ref{fig:app_bohr_frequencies} show the distribution of Bohr frequencies for Hamiltonian parameters \texttt{REG} and \texttt{CH} and $n=5, 7$.
A schematic representation of the filter function $\eta_\nu$, see Eqs.~\eqref{eq:g_gaussian} and~\eqref{eq:eta_nu}, in the frequency domain is added as pink curve.
For \texttt{REG} there are many energy levels, and correspondingly Bohr frequencies, that are almost degenerate, resulting in a strongly discretized and broadly spaced distribution, with large regions that are not covered by Bohr frequencies of the model.
In comparison, for the chaotic point \texttt{CH} there are no degeneracies and the Bohr frequency distribution follows an almost continuous and densely packed Gaussian-like shape.

The mixing time of the system, i.e., the ability of the system to transform an initial energy distribution to the Gibbs distribution, is controlled by those Bohr frequencies that are selected by the filter function $\eta_\nu$.
For \texttt{REG}, due to the discreteness of the distribution, there are only a few selected Bohr frequencies that participate effectively in the dynamics.
Thus, the mixing time strongly depends on how these few Bohr frequencies can transform the initial states' energy distribution to the Gibbs distribution.
On the other hand, the broadband of Bohr frequencies that are active in the chaotic point \texttt{CH} ensures that the mixing time is to relatively robust to the energy distribution of the initial state.
To exemplify this, we show in the right column of Fig.~\ref{fig:app_bohr_frequencies} the energy distribution of random initial state 1 (\texttt{r1}) and random initial state 2 (\texttt{r2}) used in Fig.~\ref{fig:app:random_initial_states} (identified in the legend), for \texttt{REG} and \text{CH}.
To compare different system sizes and parameter points, we rescale energy to the interval $[0, 1]$ and sample $10^4$ energies from the energy distribution of each initial state. The resulting distributions are shown as histograms with $10$ bins on the rescaled energy axis.
We observe that in the regular point \texttt{REG}, the state \texttt{r2} (yellow bars) shows a relatively strong concentration of the distribution for $n=5$ in the sixth and eighth bin, while the corresponding distribution for $n=7$ is broader.
This statistical feature is accompanied by a larger mixing time for $n=5$ than for $n=7$, as visible from the red and pink curves in the left column of Fig.~\ref{fig:app:random_initial_states}.
In general, the precise in interplay of energy distribution, Bohr frequencies and the filter function is complex, so concrete mixing time predictions will be difficult to make for points such as \texttt{REG}, with many almost degenerate levels.

\begin{figure*}
    \centering
    \includegraphics[width=\textwidth]{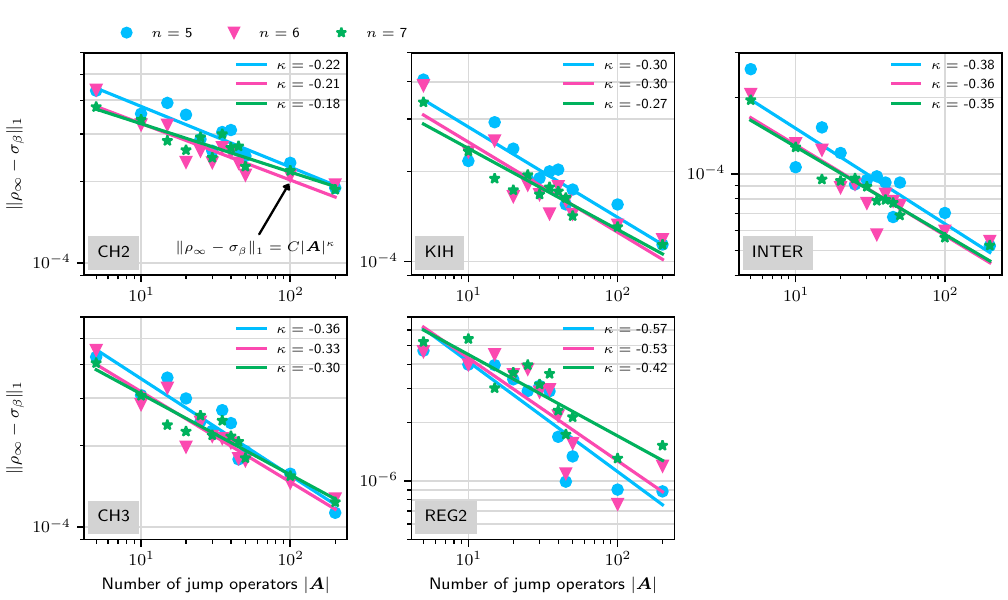}
    \caption{
    Scaling of the convergence accuracy $\Vert \rho_\infty - \sigma_\beta \Vert_1$ with the number of jump operators $|\bm{A}|$, for the parameter configurations not considered in Fig.~\ref{fig:distance_steady_gibbs_random_k_body_Pauli_k_order2}, and $n=5,6,7$ (denoted by color).
    We consider $(k=2)$-local jump operators.
    As in Fig.~\ref{fig:distance_steady_gibbs_random_k_body_Pauli_k_order2}, we observe a polynomial decay with $|\bm{A}|$.
    The leading order polynomial exponents $\kappa$ (obtained from a fit) are given in the legend.
    We observe $-0.3 < \kappa$ for the close-to-chaotic configurations \texttt{CH2} and \texttt{KIH}, and slightly larger (in modulus) $\kappa$ for the less chaotic points \texttt{INTER} and \texttt{REG2}.
    In turn, these points, generally, show lower distances $\Vert \rho_\infty - \sigma_\beta \Vert_1$, i.e. a higher accuracy.
    For \texttt{REG2}, the distance is around two orders of magnitude smaller than in the chaotic regime. The same effect was observed for the other regular limit \texttt{REG} of our model~\eqref{eq:ising_model}, cf. Fig.~\ref{fig:distance_steady_gibbs_random_k_body_Pauli_k_order2}.
    We give an explanation of this effect in App.~\ref{app:numerics_mixing_time}.
    }
    \label{fig:app_distance_steady_gibbs_random_k_body_Pauli_k_order2}
\end{figure*}

\paragraph{Convergence accuracy}
Let us turn to the convergence accuracy of the Lindblad evolution, i.e. the trace distance $\Vert \sigma_\beta - \rho_\infty \Vert_1$ between Gibbs state $\sigma_\beta$, Eq.~\eqref{eq:gibbs_state}, and the steady state $\rho_\infty$ of $\calL$, Eq.~\eqref{eq:lindbladian_general}.
In Fig.~\ref{fig:app_distance_steady_gibbs_random_k_body_Pauli_k_order2} we plot this distance as a function of $|\bm{A}|$ for varying $n = 5,6,7$, indicated by color. We show results for the five Hamiltonian parameter points not discussed in the main text in Sec.~\ref{sec:numerics}.
We set the locality to $k=2$. Other values of $k$ exhibit a similar behavior.
In all cases we observe a polynomial decrease of $\Vert \sigma_\beta - \rho_\infty \Vert_1$ with the number of jump operators $|\bm{A}|$, which is qualitatively in line with our analytical prediction in Eq.~\eqref{eq:upper_bound_trace_distance}.
The approximate leading order exponent $\kappa$ of the decay is obtained from a linear fit in the double-logarithmic scale and is given for each curve in the legend.
We generally observe, as in Fig.~\ref{fig:distance_steady_gibbs_random_k_body_Pauli_k_order2}, slightly smaller exponents (in absolute value) between $\kappa \approx -0.2$ and $\kappa \approx - 0.4$ for all parameter points except the point \texttt{REG2}.
In comparison, our upper bound~\eqref{eq:upper_bound_trace_distance} predicts an exponent of $\kappa = -1/2$.
The parameter point \texttt{REG2} shows the predicted scaling with good accuracy. However, this point resides in a regular limit of our model~\eqref{eq:ising_model}, where the Hamiltonian is mainly controlled by the transverse field only, and we cannot expect ETH to be a faithful description for the system at this point.

Similar to Fig.~\ref{fig:distance_steady_gibbs_random_k_body_Pauli_k_order2}, we observe a trace distance $\Vert \sigma_\beta - \rho_\infty \Vert_1$ of order $10^{-4}$.
Only the regular limit \texttt{REG2}, similar to the other regular limit \texttt{REG} in Fig.~\ref{fig:distance_steady_gibbs_random_k_body_Pauli_k_order2}, shows a significantly lower trace distance of order $10^{-6}$.
First, let us consider \texttt{REG} by setting $(h/J,m/J)=(0,3)$ to understand its qualitative convergence behavior. The corresponding Hamiltonian has the Bohr frequencies, $B_H=\{\pm 2\ell_J J\pm 2\ell_m m\}_{\ell_J,\ell_m\in\{0,\pm1,\pm2,\dots\}}$. The Lindblad operators are given by Eq.~\eqref{eq:Gibbs_filter_appA},
\begin{align}
    L^a 
    = 
    \sum_{\nu\in B_H}\eta_\nu A_\nu^a
    \approx
    \eta_{-2J}A_{-2J}^a.
\end{align}
We used that the jump operator with $\nu=-2J$ dominantly contributes because the others are suppressed by $\eta_\nu\propto\e^{-(\beta\nu+1)^2/8}$ [Eq.~\eqref{eq:eta_nu}] for $\beta=(2J)^{-1}$.
With this Lindblad operator, the dissipative part of the Lindbladian~\eqref{eq:lindbladian_dissipative} approximately takes the Davies form,
\begin{align}
    \calD[\rho] 
    \approx
    \sum_{a\in\bm{A}}\gamma_a 
    \eta_{-2J}^2\Big(A_{-2J}^a\rho A_{-2J}^{a\dag} 
    -\frac{1}{2}\{A_{-2J}^{a\dag}A_{-2J}^a,\rho\}\Big).
\end{align}
One can readily show that the approximated generator (right-hand side) is exactly $\sigma_\beta$-DB (i.e. $\sigma_\beta$ being the exact steady state).
Thus, the steady state of $\calL$ is comparably close to the Gibbs state.
Similarly, we can qualitatively understand the convergence at \texttt{REG2} by considering $(h/J,m/J)=(6,0)$. At the leading order in $h/J$, where the energy eigenstates are those of $h\sum_iX_i$, the Bohr frequencies are $B_H=\{\pm 2\ell_h h\}_{\ell_h\in\mathbb{Z}}$. The Lindblad operators are given by Eq.~\eqref{eq:Gibbs_filter_appA}, $L_a \approx \eta_{-2h}A_{-2h}^a$ with $\eta_{-2h}\propto\e^{-(-6\beta J+1)^2/8}$. This again explains why the convergence accuracy is high at \texttt{REG2}. Moreover, since even the most dominant transition rate is strongly suppressed by $\eta_{-2h}$, the applications of the jump process rarely happen, and hence, the mixing time is large.

\subsubsection{Influence of the coherent term}
\label{app:coherent_term}

\begin{figure*}
    \centering
    \includegraphics[width=0.99\textwidth]{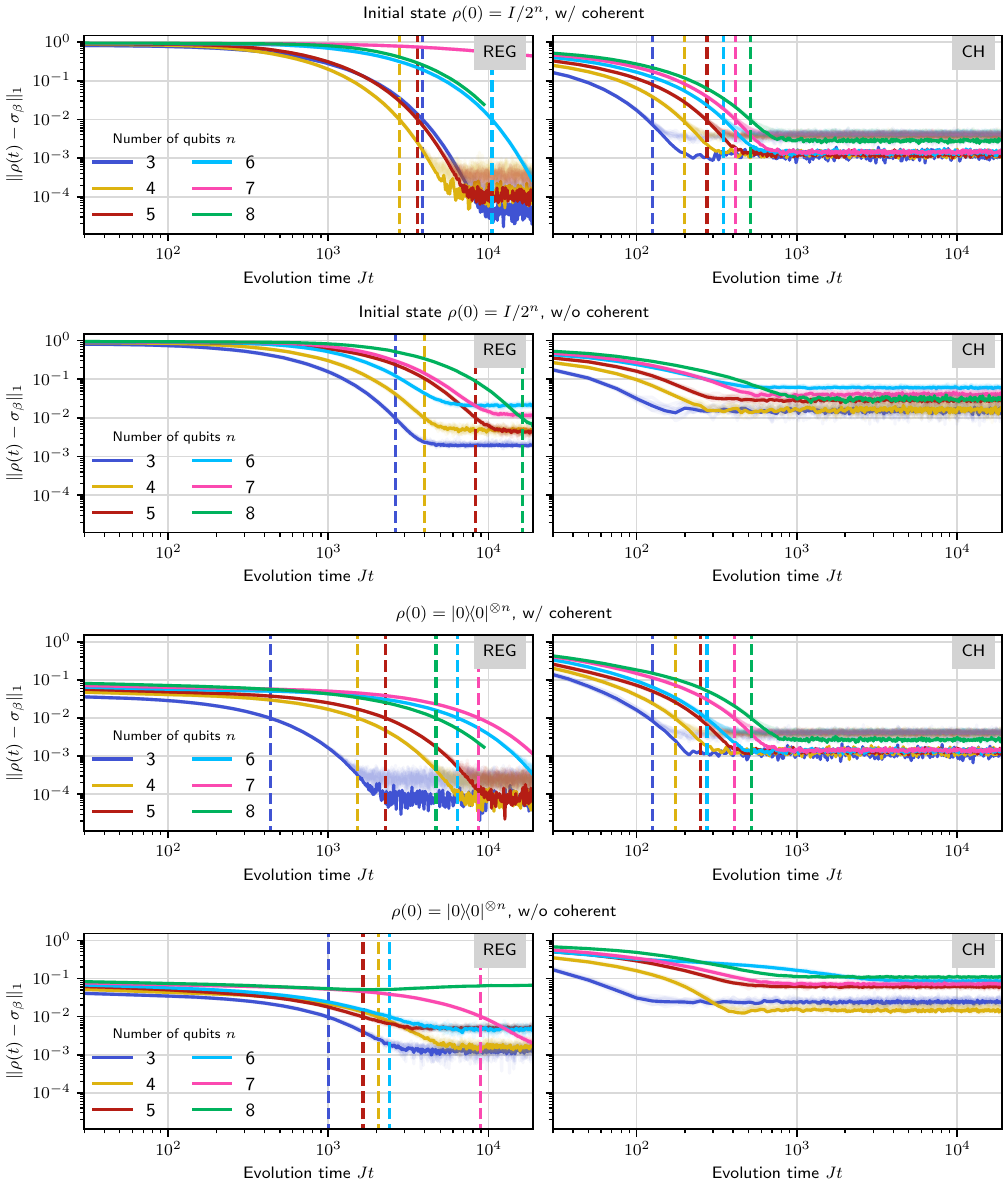}
    \caption{
    Influence of the coherent term on the convergence of the Lindblad dynamics. We plot trace distance versus evolution time (in units of $1/J$) with parameters as in Fig.~\ref{fig:time_evolution_random_k_body_Pauli_k_order2}. The four rows show results for: (i) $\rho(0) = I/2^n$, with coherent term; (ii) $\rho(0) = I/2^n$, without coherent term; (iii) $\rho(0) = \op{0}^{\otimes n}$, with coherent term; (iv) $\rho(0) = \op{0}^{\otimes n}$, without coherent term (with $I$ being the $2^n\times 2^n$ identity operator).
    For both initial states and both parameter configurations, inclusion of the coherent term significantly improves the convergence towards the target Gibbs state $\sigma_\beta$.
    The speed of convergence, as discernible from the early time behavior of the curves, appears to be less influenced by the coherent term.
    }
    \label{fig:app:time_evolution_influence_coherent_term}
\end{figure*}

Note that the coherent term $-\im[H, \rho]$ of the Lindbladian $\calL$ does not enter explicitly our derivation of the spectral gap bound and the fluctuations around the ETH average in Secs.~\ref{app:ssec:gap_EL} and~\ref{app:ssec:concentration_ETH_average}.
This is because the coherent term only contributes a complex phase in the calculations and, thus, does not influence the decay behavior of the classical Markov chain, cf. Sec.~\ref{sec:app:map_classical_mc}, capturing the action of the ETH-averaged Lindbladian $\bE_R\calL$.
Moreover, in Sec.~\ref{app:ssec:concentration_ETH_average}, the influence of the coherent term directly drops out in the difference between $\calL$ and the ETH-averaged Lindbladian
$\bE_R\calL$.
Nevertheless, the coherent term can have a significant influence on the convergence speed of the Lindblad dynamics in practice, see for example~\cite{Ding2023,fang_mixing_2024}.

In the following we numerically study the influence of the coherent term on the evolution.
Figure~\ref{fig:app:time_evolution_influence_coherent_term} compares the dynamics for two exemplary parameter configurations, \texttt{REG} and \texttt{CH}, with and without the coherent term for the Lindbladian~\eqref{eq:lindbladian_general}.
In the four rows we plot the time evolution of the trace distance $\Vert \rho(t) - \sigma_\beta \Vert_1$ for two initial states (the maximally mixed state $\rho(0) = I /2^n$ and the zero state $\ket{0}^{\otimes n}$) and, each case, with and without the coherent term in $\calL$.
In each row we show the two parameter configurations \texttt{REG} and \texttt{CH} and consider $n=3, \dots, 8$.
The first row replicates Fig.~\ref{fig:time_evolution_random_k_body_Pauli_k_order2}.

We observe that, in all cases, the coherent term leads to a more accurate convergence towards the Gibbs state $\sigma_\beta$ in the long time limit.
Without the coherent term, the simulation does not converge to accuracies below $10^{-2}$.
In terms of convergence speed, which can be estimated from the slope of the curves for early and intermediate times, the influence of the coherent term is less clear.
One can potentially discern a slightly faster convergence for \texttt{CH}, $n \geq 5$ and using $\ket{0}^{\otimes n}$ as initial state, if the coherent term is present. For the other cases, there is no significant difference in convergence speed with and without coherent term.

\subsection{Details of the numerical scheme}
\label{app:numerics_details}

In this section we describe our adaptive numerical scheme for simulating the Lindblad dynamics.
Key to the procedure is a Hermitianity test during the time evolution, which allows us to monitor the stability of the numerics and to stop the simulation in case the temporal step size is chosen too large.
As can be seen from Tab.~\ref{tab:hamiltonian_parameters}, the optimal step size for simulation can vary substantially for different Hamiltonian parameter configurations, such that this ability is crucial for our numerical studies.
In Sec.~\ref{sec:app:adaptive_rk_scheme} we introduce our adaptive scheme based on a fourth-order Runge-Kutta step.
Our method, which we denote by dmRK4 for brevity, is based on simulating the dynamics of the full density matrix, in combination with randomization over jump operators, similar to our quantum algorithm described in Sec.~\ref{sec:single_ancilla_protocol}.
In Sec.~\ref{sec:app:adaptive_rk_scheme:step_size_dependence} we provide numerical support for our proposed step size update criterion.
Note that, for large systems, the processing of the full density matrix of size $2^n \times 2^n$ becomes a limiting factor, such that so-called Monte Carlo wave function (MCWF) techniques (also called quantum jump approach)~\cite{dalibard_wave-function_1992,dum_monte_1992,molmer_monte_1993,plenio_quantum-jump_1998} need to be considered.
Therefore, in Sec.~\ref{sec:app:adaptive_rk_scheme:comparison_WVMC}, we compare the accuracy and runtime of our method and a standard MCWF solver.
Our numerical evidence shows that, for intermediate-sized systems, the proposed dmRK4 exhibits for our setting a substantial performance advantage over the MCWV method, which typically requires a much larger number of trajectories to obtain accurate results.
Independently from the dmRK4 solver, Sec.~\ref{app:numerics_details_Lindbladian_gap} gives a brief discussion on how the spectral gap of the Lindbladian $\calL$ is computed in our work.

\subsubsection{Adaptive Runge-Kutta scheme}
\label{sec:app:adaptive_rk_scheme}

Our adaptive dmRK4 scheme proceeds in two steps. Firstly, the density matrix of the next time step is computed based on a Runge Kutta 4 update step. Secondly, we test whether the predicted density matrix is still Hermitian. If this is not the case, this indicates that the simulation is unstable and that the time step size $\delta t_\mathrm{RK}$ needs to be reduced. The simulation is started again from the beginning, or from the previous time step, with a reduced step size.
In Sec.~\ref{sec:app:adaptive_rk_scheme:step_size_dependence} we show numerical evidence that in case of a too large step size, the simulation stops early and this procedure yields an efficient way to determine a suitable step size $\delta t_\mathrm{RK}$.

As jump operators, we consider a set of random $k$-local Pauli jump operators $A^a = P_1 \otimes \dots \otimes P_n$ with $P_i \in \lbrace I_{\rm 1Q}, X, Y, Z \rbrace$ and $| \lbrace i \,:\, P_i \neq I_{\rm 1Q} \rbrace | = k$, indexed by $a \in \bm{A}$.
Here, $I_{\rm 1Q}$ is the $2\times2$ identity matrix.
A random $k$-local Pauli is specified by: 1) $k$-tuple $(i_1,\dots, i_k)$ of distinct positions (in the 1D spin chain) from the set $\lbrace 1, \dots, n \rbrace$ picked uniformly at random from the set of all such $k$-tuples; 2) a $k$-tuple picked uniformly at random from the set $\lbrace I_{\rm 1Q}, X, Y, Z \rbrace^k$.
For our numerical studies, we vary the size of $\bm{A}$ and the locality $k$ of the jump operators, cf. Secs.~\ref{sec:numerics} and~\ref{app:numerics_mixing_time}.
To advance the system state for a single time step, we use a randomized approach similar to our quantum algorithm described in Sec.~\ref{sec:single_ancilla_protocol}.
Let $\rho^{(i)}(j\delta t_\mathrm{RK})$ describe the time-evolved state of a single trajectory $i = 1,\dots, N_\mathrm{tra}$ at time step~$j$.
To advance each trajectory in time, we do not apply the full Lindbladian~\eqref{eq:lindbladian_general}, but, instead, choose $N_\mathrm{tra}$ random jump operators $A^{a_i}$ with $a_i \in \bm{A}$, $i=1,\dots, N_\mathrm{tra}$ with probability $p_{a_i} = 1/|\bm{A}|$, and perform a fourth-order Runge-Kutta step with the Lindbladian generated by this single jump, $\calL^{a_i} = -\im[H, \cdot] + \calD^{a_i}$ [with $\calD^{a_i}$ as defined in Eq.~\eqref{eq:random_dissipator}] for each $i$.
This procedure approximates the dynamics generated by~\eqref{eq:lindbladian_general} for $\gamma_a = \gamma p_a$ with overall dissipation strength parameter $\gamma = 1$.
In contrast to our quantum protocol, we do not factorize the coherent and dissipative parts of the one-step time evolution as in Eq.~\eqref{eq:lindblad_random}, but instead apply the Lindbladian $\calL^{a_i}$ for each $i$ according to the scheme
\begin{align}
    \begin{split}
    \label{eq:app:runge_kutte4}
        \rho^{(i)}((j+1)\delta t_\mathrm{RK}) &=
        \rho^{(i)}(j\delta t_\mathrm{RK}) + \frac{\delta t_\mathrm{RK}}{6} \left( k^{a_i}_1 + 2k^{a_i}_2 + 2k^{a_i}_3 + k^{a_i}_4 \right) \\
        k^{a_i}_1 &= \calL^{a_i} [\rho^{(i)}(j\delta t_\mathrm{RK})] \\
        k^{a_i}_2 &= \calL^{a_i} [\rho^{(i)}(j\delta t_\mathrm{RK}) + \delta t_\mathrm{RK} k^{a_i}_1 / 2] \\
        k^{a_i}_3 &= \calL^{a_i} [\rho^{(i)}(j\delta t_\mathrm{RK}) + \delta t_\mathrm{RK} k^{a_i}_2 / 2] \\
        k^{a_i}_4 &= \calL^{a_i} [\rho^{(i)}(j\delta t_\mathrm{RK}) + \delta t_\mathrm{RK} k^{a_i}_3] .
    \end{split}
\end{align}
This advances each trajectory $i$ one step forward in time.
The average state at time point $(j+1)$ is given by $\rho((j+1)\delta t_\mathrm{RK}) = \sum_{i=1}^{N_\mathrm{tra}} \rho^{(i)}((j+1)\delta t_\mathrm{RK}) / N_\mathrm{tra}$.
To obtain an approximate solution to Eq.~\eqref{eq:lindblad_eq}, we have to average over a suitable large number of trajectories.
In practice, we observe that only a very small number of trajectories, of the order of ten, suffices for our purposes, as can be seen for example from the mixing time estimates (vertical dashed lines) in Fig.~\ref{fig:time_evolution_random_k_body_Pauli_k_order2}.
For all our simulations, we each use $N_\mathrm{tra} = 10$ trajectories for system sizes $n = 3,\dots, 7$, and $N_\mathrm{tra} = 3$ trajectories for $n=8$.

To obtain a suitable step size $\delta t_\mathrm{RK}$, we start with a fixed starting step size $\delta t_\mathrm{RK}^0$ and evolve the system as described above under Eq.~\eqref{eq:app:runge_kutte4} until the average state $\rho(j \delta t_\mathrm{RK})$ becomes non-Hermitian.
More precisely, at step $j$ we first check if $\rho(j \delta t_\mathrm{RK})$ is Hermitian (up to a given accuracy) and, if so, normalize the state, $\tr[\rho(j \delta t_\mathrm{RK})] = 1$, and project it to the Hermitian subspace via $\rho(j \delta t_\mathrm{RK}) = (\rho(j \delta t_\mathrm{RK}) + \rho(j \delta t_\mathrm{RK})^\dagger)/2$ to remove a potential small overlap with the non-Hermitian sector.
If $\rho(j \delta t_\mathrm{RK})$ is non-Hermitian, the simulation is stopped and we start again, either from the previous step or from the initial state $\rho_0$ (resetting the simulation completely) with reduced step size $\delta t_\mathrm{RK}^1 = \delta t_\mathrm{RK}^0 / 2$.
This procedure is repeated until $\rho(j \delta t_\mathrm{RK})$ stays Hermitian for all time steps $j$.
We numerically observe that, if $\delta t_\mathrm{RK}$ is too large, a small number of steps~\eqref{eq:app:runge_kutte4} evolves $\rho(j \delta t_\mathrm{RK})$ into a state $\rho( (j+1) \delta t_\mathrm{RK})$ which is far away from the Hermitian subspace, cf. Sec.~\ref{sec:app:adaptive_rk_scheme:step_size_dependence}.
On the other hand, if the step size is sufficiently small, the state stays Hermitian for all times.

Note that we can make the scheme exact by performing Eq.~\eqref{eq:app:runge_kutte4} with the full Lindbladian $\calL$, instead of the randomly selected $\calL^{a_i}$.
More precisely, for the exact scheme we apply Eq.~\eqref{eq:app:runge_kutte4} with the full $\calL = -\im[H, \cdot] + \calD = -\im[H, \cdot] + \sum_{a\in \bm{A}} \calD^{a} / |\bm{A}|$.
The full dissipator $\calD$ is given in \eqref{eq:lindbladian_dissipative} with $\gamma_a = 1 / |\bm{A}|$ (see also Eq.~\eqref{eq:dissipator_app_convergence_under_ETH} and the discussion before Eq.~\eqref{eq:lindblad_random}, with $\gamma_a = \gamma p_a$, $\gamma = 1$ and $p_a = 1/|\bm{A}|$).
In this case we do not have to average over trajectories, and $\rho((j+1)\delta t_\mathrm{RK})$ denotes the state obtained by this procedure at iteration step $j+1$.
In Secs.~\ref{sec:app:adaptive_rk_scheme:step_size_dependence},~\ref{sec:app:adaptive_rk_scheme:comparison_WVMC} we verify that this scheme is consistent, i.e. that the exact scheme converges up to machine precision to the exact steady state of $\calL$, and that the randomized scheme converges to the exact result for increasing number $N_\mathrm{tra}$ of trajectories.

\subsubsection{Step size dependence and Hermitianity check}
\label{sec:app:adaptive_rk_scheme:step_size_dependence}

\begin{figure*}
    \centering
    \includegraphics[width=\textwidth]{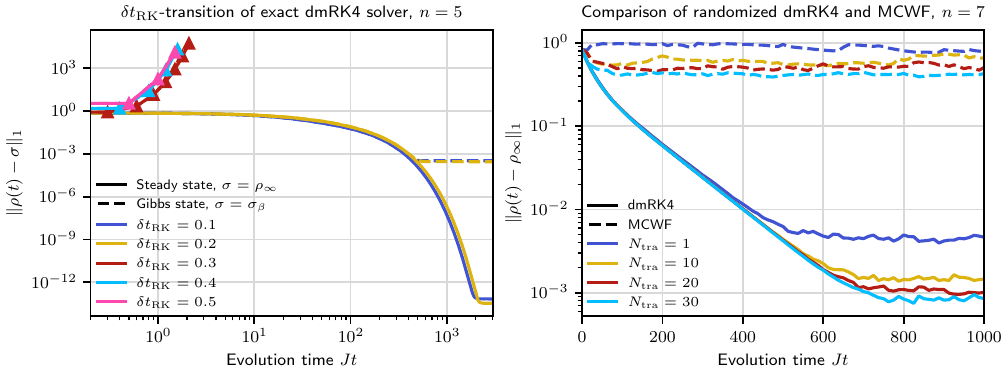}
    \caption{
    Convergence in trace distance of the dmRK4 and MCWF methods towards the exact steady state of $\calL$.
    The left plot shows the trace distance to the exact steady state (solid) and the target Gibbs state (dashed) as a function of time $Jt$ for the exact dmRK4 solver (cf. Sec.~\ref{sec:app:adaptive_rk_scheme}, for varied step size $\delta t_\mathrm{RK}$ and $n=5$. For $\delta t_\mathrm{RK} \geq 0.3$, the simulation stops after a small number of time steps, due to the violation of our proposed Hermitianity check. For $\delta t_\mathrm{RK} \leq 0.2$ the state stays Hermitian for all times and the simulation converges up to machine precision to the exact steady state of $\calL$. The distance towards the Gibbs state plateaus at a fixed value around $10^{-3}$ which is approximately the distance between exact steady state and the Gibbs state.
    The right plot compares the time dependence of the trace distance for the MCWF and our randomized dmRK4 scheme (cf. Sec.~\ref{sec:app:adaptive_rk_scheme} for $n=7$.
    We simulate the dynamics for varied number of trajectories $N_\mathrm{tra}$ (indicated by color).
    The dmRK4 scheme converges already for a single trajectory to precision $10^{-2}$. The accuracy of the MCWF method is orders of magnitude worse, even for a larger number $N_\mathrm{tra} = 30$ of trajectories. For a more complete analysis of the convergence dependence on $N_\mathrm{tra}$ see Fig.~\ref{fig:app_convergence_vs_timing}.
    }
    \label{fig:app_time_evolution_comparison_and_step_size_dependence}
\end{figure*}

To justify the proposed Hermitianity criterion and step size update, we analyze in Fig.~\ref{fig:app_time_evolution_comparison_and_step_size_dependence} (left panel) the convergence of the exact scheme for varied step sizes $\delta t_\mathrm{RK} = 0.1, 0.2, \dots, 0.5$.
We consider the parameter point \texttt{CH}.
We plot the trace distance between time-evolved state and the exact steady state (solid lines) versus the evolution time (in units of $1/J$).
We observe a clear transition in the convergence behavior for $\delta t_\mathrm{RK}$ between $0.2$ and $0.3$.
For $\delta t_\mathrm{RK} \geq 0.3$, the trace distance grows quickly for a couple of time steps after which the simulation stops because our Hermitianity criterion is violated. For $\delta t_\mathrm{RK} \leq 0.2$, the simulation continues and converges up to machine precision to the exact steady state of the model.
We also show the distance between time-evovled state and the Gibbs state (which is not identical to the exact steady state, since $\calL$ is not exacty DB) as dashed lines, which plateau at a value of order $10^{-3}$.

The optimal step size $\delta t_\mathrm{RK}$ for each simulation needs to be chosen to be below, but close, to the transition between divergence and convergence. We approximate this value by initializing the step size with a given value, and halving this value when the simulation fails and restarts, cf Sec.~\ref{sec:app:adaptive_rk_scheme}.
Table~\ref{tab:hamiltonian_parameters} shows the in this way obtained step sizes $\delta t_\mathrm{RK}$ for the considered Hamiltonian parameter configurations and system sizes.

\subsubsection{Comparison to the Monte Carlo wave function method}
\label{sec:app:adaptive_rk_scheme:comparison_WVMC}

\begin{figure*}
    \centering
    \includegraphics[width=\textwidth]{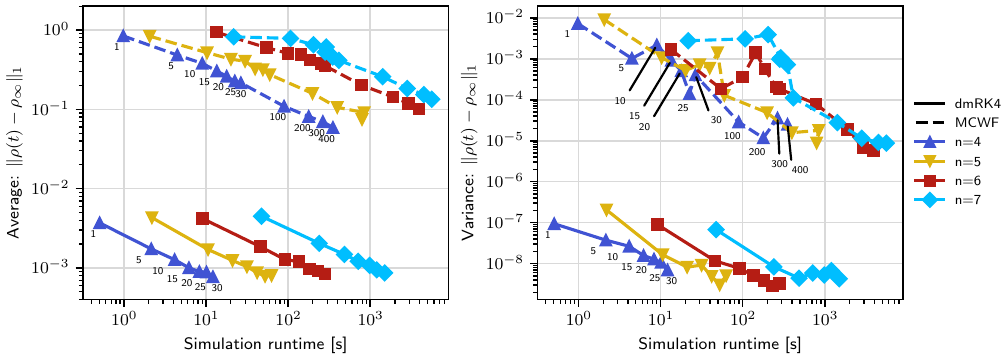}
    \caption{
    Comparison of convergence accuracy and fluctuations versus simulation runtime comparison between dmRK4 and MCWF.
    The left plot shows the convergence accuracy, i.e. the average trace distance within the plateau in the right plot of Fig.~\ref{fig:app_time_evolution_comparison_and_step_size_dependence} for $Jt \geq 700$. The left plot shows the corresponding variance around the mean value.
    Solid (dashed) lines show the results for dmRK4 (MCWF). The points along the lines correspond to different values of $N_\mathrm{tra}$, ordered from small (upper left) to large (lower right), and indicated as small numbers for the dark blue lines. Colors indicate the system size.
    For all $n$, our dmRK4 method achieves higher accuracies and smaller fluctuations than MCWF by orders of magnitude, even if much larger $N_\mathrm{tra}$ are used for MCWF.
    A straight-forward extrapolation of the dashed lines suggests that much larger simulation runtimes (around $> 10^{6}$ seconds) would be required to achieve comparable accuracies to our dmRK4 method.
    }
    \label{fig:app_convergence_vs_timing}
\end{figure*}

The MCWF method is based on simulating the time evolution of state vectors (i.e. wave functions) using an effective Hamiltonian and applying the jump operators (quantum jumps) in a similar randomized fashion as in our dmRK4 scheme.
The full (mixed-state) Lindblad time evolution is then realized by taking appropriate averages over many realizations of these pure state trajectories.
A review of the method and applications is given for example in \cite{plenio_quantum-jump_1998}.
The advantage of the method lies in the fact that only pure states are evolved and need to be stored, i.e. state vectors of size $2^n$. In comparison, our dmRK4 scheme requires storing and manipulating the full density matrix of size $2^n \times 2^n$, which becomes the computational bottleneck for large $n$.
While our ansatz has this limitation, below we discuss numerical evidence showing that for intermediate system sizes the dmRK4 scheme can be significantly more performant than MCWF.
For our simulations we use the QuTiP implementation of MCWF~\cite{johansson_qutip_2013,lambert_qutip_2024}.

The right panel of Fig.~\ref{fig:app_time_evolution_comparison_and_step_size_dependence} shows the trace distance between the steady state and the time-evolved state for the randomized dmRK4 solver (solid), cf. Sec.~\ref{sec:app:adaptive_rk_scheme}, and the MCWF simulation (dashed).
As before, we consider the parameter point \texttt{CH}.
We show exemplary results for varied numbers of trajectories $N_\mathrm{tra}$ between $1$ and $30$, and consider a system of $n=7$ qubits.
The dmRK4 method shows significantly better convergence than MCWF.
Already a single trajectory, $N_\mathrm{tra} = 1$, suffices for dmRK4 to converge below accuracy $10^{-2}$, while larger numbers of trajectories lead to accuracies below $10^{-3}$.
On the other hand, for all shown $N_\mathrm{tra}$, MCWF does not converge below $0.4$.

A more complete picture is given in Fig.~\ref{fig:app_convergence_vs_timing}, where we compare convergence and runtime of both methods for varied $n$ and $N_\mathrm{tra}$ (again for \texttt{CH}).
To quantify the convergence accuracy and fluctuations we compute the mean and variance of the time signals, shown in the right panel of Fig.~\ref{fig:app_time_evolution_comparison_and_step_size_dependence}, in the region of large evolution times $Jt > 700$ (where the signals are all converged to their respective long-time value).
We compute these quantifiers for $n=4, 5, 6, 7$ and $N_\mathrm{traj} = 1, 5, 10, 15, 20, 25, 30$ for dmRK4 and for $N_\mathrm{traj} = 1, 5, 10, 15, 20, 25, 30, 100, 200, 300, 400$ for MCWF.
For each case we also measure the runtime of the simulation (on a standard laptop).
Figure~\ref{fig:app_convergence_vs_timing} shows average (left panel) and variance (right panel) of the trace distance versus the runtime.
For each $n$ (indicated by color), the points along the lines correspond to increasing values of $N_\mathrm{tra}$, ordered from small (upper left) to large (lower right).
For the dark blue lines we indicate the $N_\mathrm{tra}$ values as small numbers next to the points. Note that we are using a larger $N_\mathrm{tra}$-range for MCWF (dashed lines).
The data shows a clear separation between dmRK4 and MCWF, both, in terms of average trace distance and variance.
Generally, the average trace distances are around two orders of magnitude smaller for the dmRK4 solver, and this even though much larger $N_\mathrm{tra}$ are used for MCWF.
A rough extrapolation of the dashed lines shows that, to obtain similar accuracies as dmRK4, the MCWF method will likely require runtimes of around $> 10^{6}$ seconds.
A similar picture is given for the variance around the average trace distance in the right plot.

\subsubsection{Comments and improved classical simulation techniques}
\label{sec:app:improved_classical_simulation_techniques}

Let us add a few comments.
Interestingly, in the left panel of Fig.~\ref{fig:app_convergence_vs_timing} we see that the accuracy of the simulation does not degrade for larger $n$ for dmRK4, while there is a weak degradation effect for MCWF.
At the same time, for dmRK4 the variance is slightly reduced for increasing $n$ (for fixed $N_\mathrm{tra}$), which hints at a concentration effect that would be interesting to explore further in the future.
In addition, note that the Hermitianity check cannot be applied to MCWF methods, which only evolve and store the state vector. Hence, the corresponding density matrix estimate, given by averaging pure states over trajectories, is Hermitian by construction.

An interesting future avenue is to explore the combination of tensor network methods with our dmRK4 scheme, similar as proposed by
\cite{sander_large-scale_2025,zhan2025rapid}.
The requirements for such a scheme are similar to \cite{sander_large-scale_2025}, namely the ability to encode Hamiltonian and Lindblad operators as low-bond-order tensor network operators, and a limited bond order increase during time evolution for the tensor network representation of $\rho(j\delta t_\mathrm{RK})$.
This poses unique challenges for tensor network methods. Specifically,
for quantum chaotic Hamiltonians we expect the Lindblad operators, given in form of operator Fourier transforms~\eqref{eq:L_operators}, to be strongly entangling. In this case, one would require large bond orders to represent the Lindblad operators as tensor network operators.
For example~\cite{zhan2025rapid} develop a tensor network protocol for simulating dissipative dynamics under the OFT Lindblad operators~\eqref{eq:L_operators}.
One step of discretized Lindblad dynamics entails a computational cost of $\calO(D_B^5)$ (cf. Sec.~V in~\cite{zhan2025rapid}), where $D_B$ is the bond order of the matrix product operator representation of the Lindblad operators.
Even if $D_B$ stays sub-exponentially during the evolution---a questionable assumption in the case of chaotic Hamiltonians as discussed above---this constitutes a large polynomial overhead cost in terms of bond order, and further improvements of such schemes would be desirable for future applications.

\subsubsection{Spectral gap of the Lindbladian}
\label{app:numerics_details_Lindbladian_gap}

To determine the spectral gap of the Lindbladian $\calL$, Eq.~\eqref{eq:lindbladian_general}, we construct the vectorization $L$ of $\calL$, cf. Eqs.~\eqref{eq:app:vectorization_D} and~\eqref{eq:app:vectorization_K}, as
\begin{align}
\label{eq:app:vectorization_L}
    L
    =
    - \im \left(I\otimes H - H^\intercal\otimes I \right)
    + 
    \sum_a\sum_{\mu, \nu}\eta_\mu\eta_\nu\Big(
        (A^a_\nu)^* \otimes A^a_\mu
        -\frac{1}{2}I\otimes (A^{a}_\mu)^\dag A^a_\nu
        -\frac{1}{2}(A^a_\nu)^\intercal (A^a_\mu)^* \otimes I
    \Big) .
\end{align}
In general, $L$ is not Hermitian and, thus, has complex eigenvalues $\ell_i \in \mathbb{C}$, $i = 1,\dots, 2^{2n}$.
If $L$ has a unique steady state (as we always assume here), there is a single eigenvalue which is zero, and all other eigenvalues have negative real parts.
This is, indeed, the case for all settings considered here.
Without loss of generality, let $\ell_1 = 0$.
Then we compute the spectral gap of $L$ as the distance between zero and the second largest real part of the spectrum,
\begin{equation}
    \Delta_\calL = \min \lbrace |\Re [\ell_i]| \,\mid \, i = 2,\dots, 2^{2n} \rbrace .
\end{equation}

\section{Noise study}
\label{app:noise_study}

In this appendix, we provide additional details and data supporting the noise study presented in  Sec.~\ref{sec:noise}. 
After recalling the setup of interest (App.~\ref{app:noise_setup}), we derive and discuss bounds for the convergence accuracy of a state prepared through Lindbladian evolution subject to global depolarization noise (App.~\ref{app:noise_conv}).
These bounds are generalized to stochastic noise (App.~\ref{app:gen_noise_comp}).
We summarize the main result here as a theorem:
\begin{theorem}\label{thm:noise_bound}
    For a Hamiltonian $H$ and inverse temperature $\beta$, let $\Gamma^M:=\e^{M\delta t\calL}$ be a Lindblad dynamics that admits the Gibbs state $\sigma_\beta$ as a steady state. Assume the state converges to the Gibbs state as $\|\Gamma^M[\rho_0]-\sigma_\beta\|_1\le B\e^{-\alpha M}$ for an arbitrary initial state $\rho_0$ and $n$-dependent variables $\alpha$ and $B$. Then the modified dynamics $(\Lambda_\lambda\circ\Gamma)^M$ subject to a stochastic noise channel $\Lambda_\lambda$~\eqref{eq:proba model}, with probability of no error incurring $\lambda$, admits the steady state $\tilde{\rho}_\infty$ that satisfies $\|\tilde{\rho}_\infty-\sigma_\beta\|_1 \le \frac{B\lambda}{1-(1-\lambda)\e^{-\alpha}}$.
\end{theorem}
We then provide details of the scaling study (App.~\ref{app:noise_scaling}).
We present bounds obtained through a more generic noise analysis and based on what would be obtained for unitary circuits  (App.~\ref{app:noise_comp} and~\ref{app:noise_comp_unit}, respectively). These are used for the Fig.~\ref{fig:noise_study} of the main text.

\subsection{Setup}
\label{app:noise_setup}

The noiseless dynamics corresponding to $M$ steps of Lindblad evolution, each of duration $\delta t$, is given by $\Gamma_M =(\e^{\delta t\lindblad})^M$, with $\lindblad$ the Lindbladian admitting the target Gibbs state $\sigma_\beta$ as a steady state. We also denote $\Gamma:=\Gamma_1 = \e^{\delta t\lindblad}$. 
The noisy dynamics is defined as $\widetilde{\Gamma}_M = (\Lambda_{\lambda}  \circ \Gamma)^M$, and accordingly $\widetilde{\Gamma} = \Lambda_{\lambda}  \circ \Gamma$, where each step of ideal evolution is followed by a noise channel $\Lambda$.
We start by considering a global depolarization channel
\begin{equation}
\label{eq:global_depo}
    \Lambda_{\lambda}[X] \coloneq (1-\lambda) X + \lambda \Tr[X]\frac{I}{2^n}, 
\end{equation}
although the following derivations apply (with minor adjustments) to generic stochastic noise, as detailed in App.~\ref{app:gen_noise_comp}.
Under this simplified noise model, the probability $\lambda\in [0,1]$ appearing in Eq.~\eqref{eq:global_depo} can be related to the probability of a random error happening per step of evolution. 
With $N_{\rm g}$ the number of noisy quantum operations (typically $2$-qubit gates) involved in the circuit simulating such evolution, and denoting as $\lambda_{\rm g}$ the probability of error per gate, the probability $\lambda$ satisfies
\begin{equation}\label{eq:app_proba_err}
    (1-\lambda) = (1-\lambda_{\rm g})^{N_{\rm g}}.
\end{equation}

Let $\rho_M\coloneq\Gamma_M[\rho(0)]$ be the state obtained under noiseless evolution starting with the maximally mixed state $\rho(0)=I/2^n$.
As $M$ increases, $\rho_M$ converges towards $\sigma_\beta$, and we assume such convergence to be captured in trace distance by:
\begin{equation}\label{eq:app_bound_td}
    \|\rho_M - \sigma_\beta \|_1 \leq B \e^{-\alpha M},\end{equation}
where both the convergence rate $\alpha$ and prefactor $B$ can depend on the system size $n$.
We highlight that the convergence rate $\alpha$ depends implicitly on the choice of evolution step $\delta t$, through the definition of $\rho_M$. Denoting this dependence explicitly, we have $\alpha (a \delta t)=a \alpha(\delta t)$ for any $a>0$. For now, we assume that $\delta t$ is fixed and drop it in our notations.
Let us denote $M^*(\varepsilon)$ the smallest number of steps such that $\|\rho_{M^*} - \sigma_\beta \|_1\leq \epsilon$. From Eq.~\eqref{eq:app_bound_td} it is obtained as:
\begin{equation}\label{eq:mixing_steps}
    M^\ast(\epsilon)=\frac{\ln(B/\epsilon)}{\alpha}.
\end{equation}

Finally, let $\tilde{\rho}_M=\widetilde{\Gamma}_M[\rho(0)]$ be the state obtained after $M$ steps of noisy evolution starting again from $\rho(0)$ and let us define $\tilde{\rho}_\infty$ to be the steady state of the noisy evolution such that $\widetilde{\Gamma}_M[\tilde{\rho}_\infty]= \tilde{\rho}_\infty$.

\subsection{Convergence of the noisy states towards the target Gibbs state}
\label{app:noise_conv}

Our aim is to understand how the noise affects the convergence of the noisy dynamics. In particular, we want to bound the distance $\norm{\tilde{\rho}_M - \sigma_\beta}_1$ between the noisy state and the target Gibbs state.
To do so, we can relate the states prepared through noisy and noiseless evolutions. For instance, after one step of noisy evolution, we get:
\begin{equation*}
    \rho_0 \xrightarrow[]{\Gamma:=\e^{\delta t\lindblad}} \rho_1 \xrightarrow[]{\Lambda_{\lambda}} \tilde{\rho}_1 
    = 
    \lambda \rho_0 + (1-\lambda) \rho_1.
\end{equation*}
Iterating such step, we get that after $M$ steps of noisy evolution, the resulting state has the form
\begin{align}
\begin{split}\label{eq:noisy_state_dynamics}
    \tilde{\rho}_M 
    & = 
    (1-\lambda)^M \rho_M + \lambda(1-\lambda)^{M-1} \rho_{M-1} +  \dots + \lambda (1-\lambda) \rho_1 + \lambda \rho_0 
    \\
    &= 
    (1-\lambda)^M \rho_M + \lambda \sum_{m=0}^{M-1} (1-\lambda)^m \rho_m.
\end{split}
\end{align}
Making use of Eq.~\eqref{eq:app_bound_td} and the triangle inequality, we obtain an upper bound for $\|\tilde{\rho}_M  - \sigma_\beta\|_1$ at arbitrary $M$:
\begin{align}
\label{eq:main_depo_distance}
\begin{split}
    \|\tilde{\rho}_M  - \sigma_\beta\|_1 \leq \widetilde{B}_M 
    &\coloneq 
    (1-\lambda)^M \norm{\rho_M - \sigma_\beta}_1 + \lambda \sum_{m=0}^{M-1} (1-\lambda)^m \norm{\rho_m - \sigma_\beta}_1  
    \\
    &= 
    B \left[(1-\lambda)^M \e^{-\alpha M} + \lambda \sum_{m=0}^{M-1} (1-\lambda)^m \e^{-\alpha m}\right]  
    \\
    &= 
    B \left[u_0^M + \lambda \frac{1-u_0^M}{1-u_0} \right]
    \quad\quad\quad\quad\quad\quad\quad\quad  \text{with   } 
    u_0 \coloneq (1-\lambda)\e^{- \alpha} \in [0,1) 
    \\ 
    &= 
    B \left[u_0^M \left(1-\frac{\lambda}{1-u_0}\right) + \frac{\lambda}{1-u_0} \right]
    \quad\quad  \text{where   } 
    \frac{\lambda}{1-u_0}\in [0,1).
\end{split}
\end{align}

Let us now comment on the convergence displayed in Eq.~\eqref{eq:main_depo_distance}. Given that $0\le u_0<1$, the bound on the distance decreases monotonically towards the value
\begin{equation}
\label{eq:asympt_distance}
    \widetilde{B}_{\infty}\coloneq B \frac{\lambda}{1-u_0},
\end{equation}
that corresponds to a bound on the distance $\norm{\tilde{\rho}_\infty  - \sigma_\beta}_1$ between noiseless and noisy steady states.
Note that this bound also holds for the case of general stochastic noise channels, as shown in App.~\ref{app:gen_noise_comp}.
Such distance depends both on the error rate $\lambda$ and the convergence rate $\alpha$. 
As would be expected the smaller the errors are, the closer the fixed state of the noisy dynamics is compared to the Gibbs state. For $\lambda=0$, they coincide. 

Additionally, we see in Eq.~\eqref{eq:main_depo_distance} that at fixed noise level $p$, the term $\widetilde{B}_{\infty}$ decreases as the decay rate $\alpha$ increases. Further inspection reveals features of the mixing time of the noisy dynamics. In particular, we see that the distance between $\tilde{\rho}_M$ and $\tilde{\rho}_\infty$ decays as $u_0^M$, that adopts a scaling $\e^{-\tilde{\alpha}}$ analogous to Eq.~\eqref{eq:app_bound_td} but now with a convergence rate
\begin{equation}
\label{eq:app_noisy_conv_rate}
    \tilde{\alpha} = \alpha + N_{\rm g} C
    \quad \text{  where  }
    C = -\ln(1-\lambda_{\rm g}) >0.
\end{equation} 
Notably, this convergence rate is always larger than the noiseless one. That is, while the noise is detrimental in that it perturbs the steady state of the dynamical evolution, reaching the noisy steady state is never slowed.

\subsection{Generalization of  the convergence analysis to probabilistic noise models}
\label{app:gen_noise_comp}

While presented for a global depolarization noise model, the convergence analysis of App.~\ref{app:noise_conv} can readily be ported to more general models.
Let us consider the setup of App.~\ref{app:noise_setup}, but now with stochastic noise of the form
\begin{equation}
\label{eq:proba model}
    \Lambda_{\lambda}[X] \coloneq (1-\lambda) X + \sum_{l} \lambda_l U_l X U_l^{\dagger}
\end{equation}
where each of the $U_i$ is unitary and with $\lambda_l \in [0,1]$ together with $\lambda = \sum_{l} \lambda_l \leq 1$. Such a model can be understood as the probabilistic application of a unitary $U_l$ (or the identity) occurring with a probability $\lambda_l$ (or $1- \lambda$). This encompasses global or local depolarization, any Pauli noise model and many more.

As before, let us denote as $\rho(0)$ the initial state, as $\rho_M = \Gamma_M[\rho(0)]$ the state obtained after $M$ steps of ideal evolution, and as $\tilde{\rho}_M = \tilde{\Gamma}_M[\rho(0)]$ the state obtained after $M$ steps of noisy evolution. We define $\chi_{m,0}\coloneq (1/\lambda)\sum_l \lambda_l U_l \Gamma[\tilde{\rho}_{m-1}] U_l^{\dagger}$, which is a valid state ensuring that one step one noisy evolution yields
\begin{equation}
\label{eq:mixture_error_state}
    (\Lambda_{\lambda}\circ \Gamma)[\tilde{\rho}_{m-1}] \coloneq (1-\lambda) \Gamma[\tilde{\rho}_{m-1}] + \lambda \chi_{m,0},
\end{equation}
and further define $\chi_{m, K} \coloneq \Gamma_K[\chi_{m,0}]$, which is also a state. 
With these notations, the state obtained after $M$ steps of noisy evolution, akin to Eq.~\eqref{eq:noisy_state_dynamics} for the global depolarization case, can now be written as
\begin{align}
\begin{split}
    \tilde{\rho}_M 
    & = 
    (1-\lambda)^M \rho_M + \lambda(1-\lambda)^{M-1} \chi_{1,M-1} +  \dots + \lambda (1-\lambda) \chi_{M-1,1} + \lambda \chi_{M,0} 
    \\
    &= 
    (1-\lambda)^M \rho_M + \lambda \sum_{K=0}^{M-1} (1-\lambda)^K \chi_{M-K, K}.
\end{split}
\end{align}
Assuming that the convergence dynamics of Eq.~\eqref{eq:app_bound_td} holds for any initial state $\chi_{m, 0}$ (rather than the single $\rho(0)=I/2^n$ that was needed in the previous case), such that
\begin{equation}\label{eq:app_bound_td_arbitary}
    \|\chi_{m, K} - \sigma_\beta \|_1 \leq B \e^{-\alpha K},
\end{equation}
we recover exactly the same results as before in terms of bounds for the trace distance between the target Gibbs state compared to the state prepared along noisy evolution, as per Eq.~\eqref{eq:main_depo_distance}, or compared to the steady state of the noisy dynamics, as per Eq.~\eqref{eq:asympt_distance}.

\subsection{Scaling study}
\label{app:noise_scaling}  

Our aim is to study the impact of the errors due to noise as the system sizes increase. For that, we wish to evaluate the bounds of Eq.~\eqref{eq:asympt_distance} for different $n$. This requires specifying values for the convergence rates $\alpha$, the prefactors $B$, and the probabilities of error $\lambda$. For the two first quantities, we will resort to data obtained through the numerical simulations for fixed initial states and small system sizes that are then \emph{extrapolated} to larger ones.

In Fig.~\ref{fig:noise_app_fit}~(left panel), we show the fits of $\alpha$ and $B$ for the numerical simulations performed for $n=3, \dots, 8$ qubits for the Hamiltonian's configuration \texttt{CH} given in Tab.~\ref{tab:hamiltonian_parameters}, for the initial state $\rho(0)=I/2^n$,
and for $\beta=(2J)^{-1}$.
We recall that this configuration corresponds to the chaotic regime. As can be seen the dynamics of Eq.~\eqref{eq:app_bound_td} closely matches the numerical data. Then in Fig.~\ref{fig:noise_app_fit}~(right panel) we report geometric fits for these two quantities as a function of $n$.
In both cases, albeit performed on a small number of data points, we see reasonably good fits, in particular for $\alpha$.
For the probability of error $\lambda$, we resort to Eq.~\eqref{eq:app_proba_err} and take a number of gates per unit of time ($\delta t=1$) scaling linearly in $n$, namely $N_{\rm g}=50 n$.
With these, and fixing the value of the probability of error $p_g$ per noisy gate, we can evaluate Eq.~\eqref{eq:asympt_distance} for arbitrary $n$.
These are the data points used for  Fig.~\ref{fig:noise_study} in the main text (solid lines).

\subsection{Comparison to generic noise bounds}
\label{app:noise_comp}

We wish to compare Eq.~\eqref{eq:main_dk} to bounds obtained for generic noise models.
In the following, we first port the bounds of Ref.~\cite{Chen2023thermal} (Lemma II.1) to the discrete settings considered here and then evaluate those for the here considered noise models.

Let us denote as $\mathcal{A}$ and as $\widetilde{\mathcal{A}}= \Lambda_{\lambda} \circ \mathcal{A}$ a step of noiseless and noisy dissipative evolution, respectively.
Accordingly, let $\rho_\infty$ and $\tilde{\rho}_\infty$ be their steady states that satisfy $\mathcal{A}[\rho_\infty]=\rho_\infty$ and $\widetilde{\mathcal{A}}[\tilde{\rho}_\infty]=\tilde{\rho}_\infty$. The mixing time of the noiseless dynamics in the discrete case is characterized by the number of steps $M_{\mathrm{mix}}(\epsilon)$ ensuring that for any state $\rho$ we have $\|\mathcal{A}^{M_\mathrm{mix}(\epsilon)}[\rho] - \rho_\infty\|_1 \leq \epsilon$. For now let us drop the dependency in $\epsilon$. To bound $\|\tilde{\rho}_\infty - \rho_{\infty}\|_1$ we proceed as follow 
\begin{align}
\label{eq:step_1_b}
\begin{split}
    \| \tilde{\rho}_\infty - \rho_{\infty} \|_1 
    &= 
    \big\| \widetilde{\mathcal{A}}^{M_{\mathrm{mix}}}[\tilde{\rho}_\infty] - \rho_{\infty} \big\|_1 
    \\
    &\leq 
    \big\| \widetilde{\mathcal{A}}^{M_{\mathrm{mix}}}[\tilde{\rho}_\infty] - \mathcal{A}^{M_{\mathrm{mix}}}[\tilde{\rho}_\infty] \big\|_1 
    + \big\| \mathcal{A}^{M_{\mathrm{mix}}}[\tilde{\rho}_\infty] - \rho_{\infty} \big\|_1  
    \\
    &\leq 
    \big\| (\widetilde{\mathcal{A}}^{M_{\mathrm{mix}}} - \mathcal{A}^{M_{\mathrm{mix}}})[\tilde{\rho}_\infty] \big\|_1 
    + \epsilon,
\end{split}
\end{align}
where we made use of the triangle inequality and of the definition of $M_{\mathrm{mix}}$.
We further bound 
\begin{align}
\label{eq:step_2_b}
\begin{split}
    \big\| \widetilde{\mathcal{A}}^{M_{\mathrm{mix}}} - \mathcal{A}^{M_{\mathrm{mix}}} \big\|_{1\rightarrow 1} 
    &\leq  
    \big\| \mathcal{A}\circ( \widetilde{\mathcal{A}}^{M_{\mathrm{mix}}-1} - \mathcal{A}^{M_{\mathrm{mix}}-1}) \big\|_{1\rightarrow 1} 
    + \big\| ( \widetilde{\mathcal{A}} - \mathcal{A} )\circ \widetilde{\mathcal{A}}^{M_{\mathrm{mix}}-1} \big\|_{1\rightarrow 1}  
    \\
    &\leq  
    \big\| \widetilde{\mathcal{A}}^{M_{\mathrm{mix}}-1} - \mathcal{A}^{M_{\mathrm{mix}}-1}\big\|_{1\rightarrow 1} 
    + \big\| \widetilde{\mathcal{A}} - \mathcal{A} \big\|_{1\rightarrow 1}
    \\
    &\leq 
    M_{\mathrm{mix}} \big\| \widetilde{\mathcal{A}} - \mathcal{A} \big\|_{1\rightarrow 1}.
\end{split}
\end{align}
To get the second line, we used the triangle inequality first, and then sub-multiplicativity of the induced norm $\norm{\cdot}_{1 \rightarrow 1}$ together with $\| \mathcal{A} \|_{1 \rightarrow 1}\leq 1$ and $\|\widetilde{\mathcal{A}}\|_{1 \rightarrow 1}\leq{1}$ as both $\mathcal{A}$ and $\widetilde{\mathcal{A}}$ are quantum channels. The third line follows through recursion. Overall we get:
\begin{align}
\label{eq:asympt_distance_generic}
\begin{split}
    \| \widetilde{\rho}_\infty - \rho_{\infty}\|_1  
    \leq 
    M_\mathrm{mix}(\epsilon) \big\|\widetilde{\mathcal{A}} - \mathcal{A} \big\|_{1\rightarrow 1} 
    + \epsilon.
\end{split}
\end{align}

Due to the repeated use of the triangle inequality in Eq.~\eqref{eq:step_2_b}, each evolution step incurs a contribution $\|\widetilde{\mathcal{A}} - \mathcal{A} \|_{1\rightarrow 1}$ to the resulting bound. This may significantly over-estimate the effect of the noise, as it corresponds to a worst-case scenario where each step incurs the maximum deviation possible and noise contributions between the steps never average out. However, it allows us to deal with arbitrary discrepancies between $\mathcal{A}$ and $\widetilde{\mathcal{A}}$.
This is in contrast to the bounds obtained in Eq.~\eqref{eq:main_depo_distance}, where part of the deviations induced by the noise, especially the ones occurring early on, are mitigated during the evolution. 

Let us assume a bounded channel distance $\|\widetilde{\mathcal{A}} - \mathcal{A} \|_{1\rightarrow 1} \leq C$. The constant $C$ has to be computed for a given noise model.
With this we obtain the following bound on the distance between noisy and noiseless state:
\begin{align}
\label{eq:step_4_b}
\begin{split}
    \norm{ \tilde{\rho}_\infty - \sigma_\beta}_1 
    &\leq 
    C M_\mathrm{mix}(\epsilon) + \epsilon.
\end{split}
\end{align}
Substituting the expression $M^\ast(\epsilon)$ for $M_\mathrm{mix}(\epsilon)$ from Eq.~\eqref{eq:mixing_steps},
and minimizing the bound over the choice of $\epsilon$ we obtain 
\begin{align}
\begin{split}
    \widetilde{B}^{'}_{\infty} 
    \coloneq  
    \min \left\{B,\, \frac{C}{\alpha} \Big(\ln\Big(\frac{B\alpha}{C}\Big) + 1\Big)\right\}
\end{split}
\end{align}
as a bound for $\|\tilde{\rho}_\infty - \sigma_\beta\|_1$, with the minimum achieved for $\varepsilon'= \min \{C / \alpha, B \}$.

Let us specialize to the stochastic noise channel~\eqref{eq:proba model}.
From the definition we see that
\begin{align}
\label{eq:dist_1_1}
\begin{split}
    \big\|(\widetilde{\mathcal{A}} - \mathcal{A})[X] \big\|_{1} 
    &= 
    \big\| \lambda Y - \sum_l \lambda_l U_l Y U_l^\dag \big\|_{1} 
    \leq 
    \lambda \norm{Y}_1 + \sum_l \lambda_l \norm{Y}_1  \leq 2 \lambda \norm{X}_1,
\end{split}
\end{align}
where we have introduced $Y=\mathcal{A}[X]$, used $\Tr[Y] \leq \norm{Y}_1 \leq \norm{X}_1$ together with the unitary invariance of the trace norm.
Overall, we get that $\|\widetilde{\mathcal{A}} - \mathcal{A} \|_{1\rightarrow 1} \leq 2\lambda$.
Thus, for this model can set $C = 2\lambda$, which yields
\begin{align}
\label{eq:app_bounds_generic}
\begin{split}
    \widetilde{B}^{'}_{\infty} 
    \coloneq  
    \min \left\{B,\, \frac{2\lambda}{\alpha} \Big(\ln\Big(\frac{B\alpha}{2\lambda}\Big) + 1\Big)\right\}
\end{split}
\end{align}
for the distance $\|\tilde{\rho}_\infty - \sigma_\beta\|_1$.

Let us compare the bound~\eqref{eq:asympt_distance}, explicitly taking into account damping of early time errors through the dissipative Lindblad channel, to Eq.~\eqref{eq:app_bounds_generic}.
In particular, the latter will always be tighter for values of $B$ large enough (at fixed $\lambda$ and $\alpha$). However for small enough $B$, where the bounds are informative, the former may be significantly tighter. This is what is seen in our numerics:
The bounds~\eqref{eq:asympt_distance} and~\eqref{eq:app_bounds_generic} (both for the case of a global depolarizing channel) are reported as solid and dotted lines, respectively, in Fig.~\ref{fig:noise_study}~(right panel) in the main text.
Values used for their evaluations were discussed in Sec.~\ref{app:noise_scaling}. As can be seen, for this setting, the bound~\eqref{eq:app_bounds_generic} significantly overestimate the impact of the noise.

\begin{figure}
\centering
\includegraphics[width=1\textwidth]{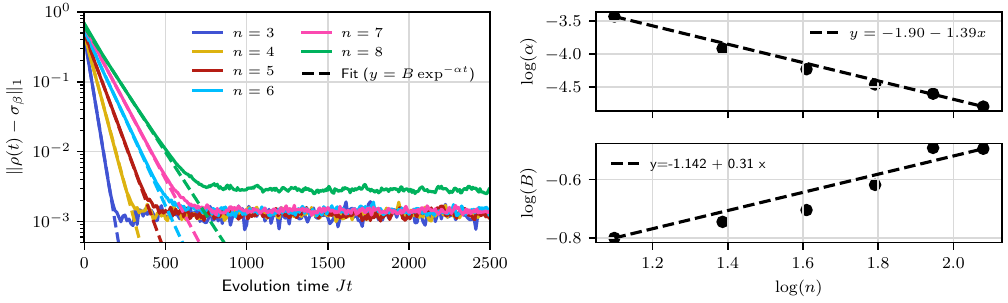}
\caption{
(Left panel) From the numerical data of App.~\eqref{app:numerics} for the setup \texttt{CH} (see details in Tab.~\ref{tab:hamiltonian_parameters}), we plot the trace distance of the prepared state along the Lindblad evolution, compared to the target Gibbs state (solid lines). For each of the system sizes probed (colors in legend), we fit the parameters $\alpha$ and $B$ from Eq.~\eqref{eq:app_bound_td} and display the resulting fit (dashed lines). 
(Right panel) Values of the parameters obtained for $\alpha$ (upper panel) and $B$ (lower panel) are plotted and fitted as a function of $n$. 
}
\label{fig:noise_app_fit}
\end{figure}

\subsection{Comparison to unitary circuits}
\label{app:noise_comp_unit}

We wish to compare deviations induced by noise for the Lindbladian protocols to deviations that would occur in a unitary circuit.
We note the limits of such a comparison, as preparation of Gibbs states would require some amount of non-unitarity in the dynamics in the first place. Nonetheless, for the sake of comparison, we will consider unitary circuits that are assumed to have the same complexity (i.e., the same number of noisy gates) and to produce the same outputs as our protocols. 
Recall that in the noiseless case, the number of steps $M_\mathrm{mix}(\epsilon)$ required to prepare a state $\sigma_{\epsilon}$ that differs by  $\epsilon$ in trace distance from $\sigma_\beta$ was provided in Eq.~\eqref{eq:mixing_steps}. Given that each evolution step requires $N_{\rm g}$ noisy gates, this incurs a total of $N_{\rm tot}=M_\mathrm{mix}(\epsilon) N_{\rm g}$ noisy gates. Upon the global depolarization model~\eqref{eq:global_depo}, and for a unitary circuit having $N_{\rm tot}$ gates, the state $\sigma_{\epsilon}$ prepared by the noiseless circuit would become $\tilde{\sigma}_{\epsilon} = \lambda_{\rm tot} I/2^n + (1-\lambda_{\rm tot})\sigma_{\epsilon}$ with
\begin{equation}
    (1- \lambda_{\rm tot}) = (1-\lambda_{\rm g})^{M_\mathrm{mix}(\epsilon)N_g}.
\end{equation}
 Compared to the Gibbs state $\sigma_\beta$, this state would deviate in trace distance by 
\begin{align}\label{eq:app_comp_noisy_unit}
\begin{split}
 \norm{ \tilde{\sigma}_{\epsilon} - \sigma_\beta}_1 &= \norm{\lambda_{\rm tot} \Big(\frac{I}{2^n} -\sigma_\beta\Big) + (1-\lambda_{\rm tot})(\sigma_{\epsilon} -\sigma_\beta)}_1
\leq \lambda_{\rm tot} B + (1-\lambda_{\rm tot}) \epsilon.
\end{split}
\end{align}
In contrast to the noisy Lindblad dynamics.~\ref{eq:main_depo_distance}, as we increase the number of steps $M_\mathrm{mix}(\epsilon)$ we see competitive effects between a decrease of the convergence accuracy $\epsilon$ and an increase of the errors due to the noise.
To be conservative in our study, we minimize the right-hand side of Eq.~\eqref{eq:app_comp_noisy_unit} over $\epsilon$, or equivalently $M_\mathrm{mix}(\epsilon)$, to make this bound as tight as possible. Results of the bounds obtained through this minimization for different system sizes are reported as dashed lines in Fig.~\ref{fig:noise_study}~(right panel) in the main text. 
When compared to Eq.~\eqref{eq:asympt_distance}, we see a noticeable increase in the distances of the prepared state towards the targeted one for the unitary case. That is, we probed the added resilience of the Lindblad dynamics compared to unitary circuits.

\section{Details on algorithmic errors in the quantum circuit simulation}
\label{app:circ_sim_alg_errors}

We analyze the circuit error dependence on the main algorithmic parameters, evolution step $\delta t$, and OFT discretization step $\Delta t$ for our randomized single-ancilla Lindblad simulation protocol discussed in the main text. As discussed in Sec.~\ref{sec:alg_errors}, we use the mixed-field Ising model~\eqref{eq:ising_model},
set $\beta=(2J)^{-1}$,
fix the integration domain of the OFT to $JT = 1.6$, and simulate the circuit evolution up to a time point $Jt = 500$.
The circuit error is quantified by the distance $\| \rho^\mathrm{circ}_\infty - \sigma_\beta \|_1$ between the target Gibbs state $\sigma_\beta$ and the steady state of the circuit $\rho^\mathrm{circ}_\infty$.
We compute $\epsilon_{ij} := \| \rho^\mathrm{circ}_\infty(\delta t_i, \Delta t_j) - \sigma_\beta \|_1$ for various evolution steps $\delta t_i$ and OFT discretization steps $\Delta t_j$ in the range $10^{-2} \leq J\delta t \leq 10^1$ and $0.06 \leq J\Delta t\leq 2$.
More precisely, we consider $J\delta t_i = 10^{\ell^{(\delta)}_i}$ and $J\Delta t_j = 10^{\ell^{(\Delta)}_j}$ with $\{\ell^{(\delta)}_i\}_{i=1,\dots 10}$ and $\{\ell^{(\Delta)}_j\}_{j = 1,\dots, 30}$ evenly discretizing the intervals $[-2, 1]$ and $[\log_{10} 0.06, \log_{10} 2]$, respectively.

\subsection{Error fit}
\label{app:circ_sim_alg_errors_fit}

To test the applicability of our theoretical bound~\eqref{eq:error_sources2}, we fit the function~\eqref{eq:algorithmic_errors_fit} to the data. As we are mainly interested in the logarithmic dependency of the error on $\delta t$ and $\Delta t$, we minimize the following loss function
\begin{equation}
\label{eq:app:circuit_error_loss_function}
    (\alpha^\mathrm{opt}_1, \alpha^\mathrm{opt}_2, \alpha^\mathrm{opt}_3, \alpha^\mathrm{opt}_4) 
    = 
    \underset{\alpha_1, \alpha_2, \alpha_3, \alpha_4}{\argmin} \sqrt{\sum_{i,j} \big| \log_{10} f_{\alpha_1, \alpha_2, \alpha_3, \alpha_4} (\delta t_i, \Delta t_j) - \log_{10} \epsilon_{ij} \big|^2} .
\end{equation}
Note that, since the logarithm is strictly monotonically increasing, closeness of $\log_{10} f_{\alpha_1, \alpha_2, \alpha_3, \alpha_4} (\delta t_i, \Delta t_j)$ and  $\log_{10} \epsilon_{ij}$ also implies that $f_{\alpha_1, \alpha_2, \alpha_3, \alpha_4} (\delta t_i, \Delta t_j)$ and $\epsilon_{ij}$ are close.
Introducing the logarithm gives a stronger relative weight to the regime where $\delta t$ and $\Delta t$ are small, in which the values of $f_{\alpha_1, \alpha_2, \alpha_3, \alpha4}(\delta t_i, \Delta t_j)$ and $\epsilon_{ij}$ are orders of magnitude smaller than for large $\delta t$ and $\Delta t$.
An alternative would be to consider a relative error loss function, which results in a very similar fit.
Note that the bound on the discretization error (the fourth term in Eq.~\eqref{eq:algorithmic_errors_fit}) is only valid for small $\Delta t$, satisfying $2\pi\frac{\beta}{\Delta t} - 2\beta \|H\| - 1 > 0$, which corresponds to $J\Delta t \lesssim 0.37$ for our parameter choice. When approaching this limit, the error bound diverges.
Hence, we limit the domain of the fit and take into account only $\delta t_i$ and $\Delta t_j$ with $J\delta t_i \leq 0.3$ and $J\Delta t_j \leq 0.37$ for this.
With these restrictions, we obtain optimal parameters $\alpha^{\rm opt}_1 = 2.6\times10^{-3}$, $\alpha^{\rm opt}_2 = 1.8\times10^{-2}$, $\alpha^{\rm opt}_3 = 4.1\times10^{-4}$, $\alpha^{\rm opt}_4 = 1.5\times10^{-4}$ from Eq.~\eqref{eq:app:circuit_error_loss_function}.

\subsection{OFT discretization error decomposition}
\label{app:circ_sim_alg_errors_decomposition}

\begin{figure*}[t]
    \centering
    \includegraphics[width=\textwidth]{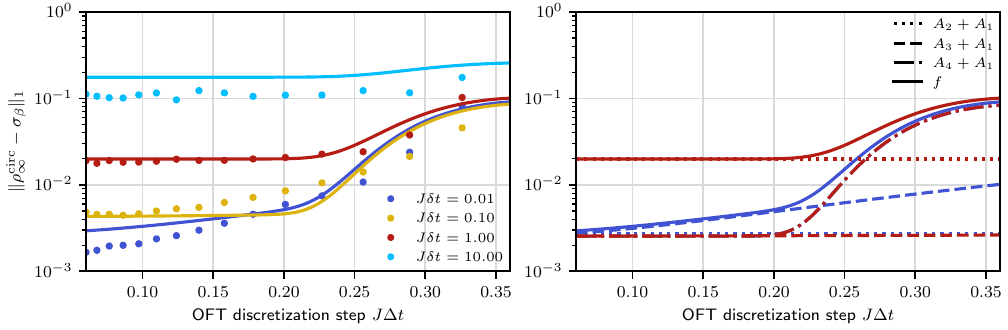}
    \caption{
    Operator Fourier transform discretization error dependence of the randomized single-ancilla protocol for $n=5$ qubits.
    The left plot shows the same data as Fig.~\ref{fig:algorithmic_error_scaling_OFT_discretization_and_trotter_step}, on a limited domain $J\Delta t \leq 0.37$.
    The right plot decomposes the total error~\eqref{eq:algorithmic_errors_fit} (solid curve), for two values $J\delta t = 0.01$ and $J\delta t = 1$, into individual terms as described below Eq.~\eqref{eq:app:algorithmic_errors_fit_decomposition} (shown as dotted, dashed and dashed-dotted curves, respectively).
    }
    \label{fig:app:algorithmic_error_scaling_OFT_discretization_and_trotter_step}
\end{figure*}

In Fig.~\ref{fig:app:algorithmic_error_scaling_OFT_discretization_and_trotter_step} we take a closer look at the individual error terms in Eq.~\eqref{eq:algorithmic_errors_fit}.
We focus on the error dependence on the OFT discretization step $\Delta t$.
The left plot shows the same data as in the right plot of Fig.~\ref{fig:algorithmic_error_scaling_OFT_discretization_and_trotter_step}, where we limit the domain to $J\Delta t \leq 0.37$.
The right plot decomposes the fit (solid lines in the left plot) into its individual terms
\begin{align}
\label{eq:app:algorithmic_errors_fit_decomposition}
    \begin{split}
        &f_{\alpha_1, \alpha_2, \alpha_3, \alpha_4} (\delta t, \Delta t) = A_1 + A_2 + A_3 + A_4, \\
        &A_1 = \alpha_1, \quad
        A_2 = \alpha_2 \delta t, \quad
        A_3 = \alpha_3 T\frac{\Delta t^2}{\delta t}, \quad
        A_4 = \alpha_4 \sqrt{\beta} |B_H|\e^{-\frac{1}{8}\big(2\pi\frac{\beta} {\Delta t}-2\beta\|H\| - 1 \big)^2} .
    \end{split}
\end{align}
We consider the cases $J\delta t = 0.01$ and $1$ (blue and red lines) as examples and plot $A_1+A_2$ (dotted), $A_1+A_3$ (dashed) and $A_1+A_4$ (dashed-dotted line) in the right panel.
For each curve we include the shift due to the constant term $A_1$ to make the individual curves comparable to the total error shown as solid curves.
For small $J\delta t = 0.01$, the error due to the Trotterization, randomization and dilation, quantified by $A_2$ (blue dotted), is negligibly small. Therefore, the shape of the blue curve is controlled by the discretization error $A_4$ (blue dotted-dashed) and the polynomial $\Delta t$ dependence of $A_3$ (blue dashed), which is the dominant contribution for $J\Delta t \lesssim 0.25$.
The term $A_3$ comes from the coherent evolution under the system Hamiltonian for time step $\Delta t$ to compute the discretized OFT~\eqref{eq:L_operator_discretized} [cf. discussion below Eq.~\eqref{eq:error_sources2}].
This weak polynomial decrease of the error for $J\Delta t \lesssim 0.25$ is clearly visible in the dark blue data points in the left plot (although with a slightly larger slope as predicted by the fit).
The other exemplary case is for larger $J\delta t = 1$, shown in red.
In this case, the Trotter error $A_2$ (red dotted) is larger than $A_3$ (red dashed), such that for $J\Delta t \lesssim 0.25$ the error becomes independent from $\Delta t$.
This behavior is clearly shown by the red data points in the left plot.
For larger $\Delta t$, the discretization error $A_4$ becomes the dominant contribution.

\bibliographystyle{quantum}
\bibliography{main}

\begin{thebibliography}{100}

\bibitem{Dalzell2023}
Alexander~M. Dalzell, Sam McArdle, Mario Berta, Przemyslaw Bienias, Chi-Fang Chen, András Gilyén, Connor~T. Hann, Michael~J. Kastoryano, Emil~T. Khabiboulline, Aleksander Kubica, Grant Salton, Samson Wang, and Fernando G. S.~L. Brandão.
\newblock ``{Quantum Algorithms: A Survey of Applications and End-to-end Complexities}''.
\newblock \href{https://dx.doi.org/10.1017/9781009639651}{Cambridge University Press}. Cambridge, UK~(2025).
\newblock  \href{http://arxiv.org/abs/2310.03011}{arXiv:2310.03011}.

\bibitem{abbasChallengesOpportunitiesQuantum2024}
Amira Abbas, Andris Ambainis, Brandon Augustino, Andreas B{\"a}rtschi, Harry Buhrman, Carleton Coffrin, Giorgio Cortiana, Vedran Dunjko, Daniel~J. Egger, Bruce~G. Elmegreen, Nicola Franco, Filippo Fratini, Bryce Fuller, Julien Gacon, Constantin Gonciulea, Sander Gribling, Swati Gupta, Stuart Hadfield, Raoul Heese, Gerhard Kircher, Thomas Kleinert, Thorsten Koch, Georgios Korpas, Steve Lenk, Jakub Marecek, Vanio Markov, Guglielmo Mazzola, Stefano Mensa, Naeimeh Mohseni, Giacomo Nannicini, Corey O'Meara, Elena~Pe{\~n}a Tapia, Sebastian Pokutta, Manuel Proissl, Patrick Rebentrost, Emre Sahin, Benjamin C.~B. Symons, Sabine Tornow, V{\'i}ctor Valls, Stefan Woerner, Mira~L. {Wolf-Bauwens}, Jon Yard, Sheir Yarkoni, Dirk Zechiel, Sergiy Zhuk, and Christa Zoufal.
\newblock ``Challenges and opportunities in quantum optimization''.
\newblock \href{https://dx.doi.org/10.1038/s42254-024-00770-9}{Nat. Rev. Phys. {\bf 6}, 718–735}~(2024).

\bibitem{Terhal2000}
Barbara~M. Terhal and David~P. DiVincenzo.
\newblock ``Problem of equilibration and the computation of correlation functions on a quantum computer''.
\newblock \href{https://dx.doi.org/10.1103/PhysRevA.61.022301}{Phys. Rev. A {\bf 61}, 022301}~(2000).

\bibitem{Poulin2009}
David Poulin and Pawel Wocjan.
\newblock ``{Sampling from the Thermal Quantum Gibbs State and Evaluating Partition Functions with a Quantum Computer}''.
\newblock \href{https://dx.doi.org/10.1103/PhysRevLett.103.220502}{Phys. Rev. Lett. {\bf 103}, 220502}~(2009).

\bibitem{Temme2009}
K.~Temme, T.~J. Osborne, K.~G. Vollbrecht, D.~Poulin, and F.~Verstraete.
\newblock ``{Quantum Metropolis Sampling}''.
\newblock \href{https://dx.doi.org/10.1038/nature09770}{Nature {\bf 471}, 87}~(2011).

\bibitem{Bilgin2010}
Ersen Bilgin and Sergio Boixo.
\newblock ``Preparing thermal states of quantum systems by dimension reduction''.
\newblock \href{https://dx.doi.org/10.1103/PhysRevLett.105.170405}{Phys. Rev. Lett. {\bf 105}, 170405}~(2010).

\bibitem{Yung2012}
Man-Hong Yung and Alán Aspuru-Guzik.
\newblock ``{A quantum–quantum Metropolis algorithm}''.
\newblock \href{https://dx.doi.org/10.1073/pnas.1111758109}{Proc. Natl. Acad. Sci. U.S.A. {\bf 109}, 754--759}~(2012).

\bibitem{Kastoryano2016}
Michael~J. Kastoryano and Fernando G. S.~L. Brand{\~a}o.
\newblock ``{Quantum Gibbs Samplers: The Commuting Case}''.
\newblock \href{https://dx.doi.org/10.1007/s00220-016-2641-8}{Commun. Math. Phys. {\bf 344}, 915--957}~(2016).

\bibitem{Brandao2016}
Fernando G. S.~L. Brand\~ao and Michael~J. Kastoryano.
\newblock ``{Finite Correlation Length Implies Efficient Preparation of Quantum Thermal States}''.
\newblock \href{https://dx.doi.org/10.1007/s00220-018-3150-8}{Commun. Math. Phys. {\bf 365}, 1--16}~(2019).

\bibitem{Chowdhury2016}
Anirban~Narayan Chowdhury and Rolando~D. Somma.
\newblock ``{Quantum algorithms for Gibbs sampling and hitting-time estimation}''.
\newblock \href{https://dx.doi.org/10.26421/QIC17.1-2-3}{Quant. Inf. Comput. {\bf 17}, 0041--0064}~(2017).

\bibitem{Motta2019}
Mario Motta, Chong Sun, Adrian Teck~Keng Tan, Matthew J.~O' Rourke, Erika Ye, Austin~J. Minnich, Fernando G. S.~L. Brand\~ao, and Garnet Kin-Lic Chan.
\newblock ``{Determining eigenstates and thermal states on a quantum computer using quantum imaginary time evolution}''.
\newblock \href{https://dx.doi.org/10.1038/s41567-019-0704-4}{Nat. Phys. {\bf 16}, 205--210}~(2019).

\bibitem{Gilyen2019}
András Gilyén, Yuan Su, Guang~Hao Low, and Nathan Wiebe.
\newblock ``Quantum singular value transformation and beyond: exponential improvements for quantum matrix arithmetics''.
\newblock In Proceedings of the 51st {Annual} {ACM} {SIGACT} {Symposium} on {Theory} of {Computing}.
\newblock \href{https://dx.doi.org/10.1145/3313276.3316366}{Pages 193--204}.
\newblock New York, NY, USA~(2019). Association for Computing Machinery.

\bibitem{lu_algorithms_2021}
Sirui Lu, Mari~Carmen Bañuls, and J.~Ignacio Cirac.
\newblock ``Algorithms for {Quantum} {Simulation} at {Finite} {Energies}''.
\newblock \href{https://dx.doi.org/10.1103/PRXQuantum.2.020321}{PRX Quantum {\bf 2}, 020321}~(2021).

\bibitem{Coopmans2022}
Luuk Coopmans, Yuta Kikuchi, and Marcello Benedetti.
\newblock ``{Predicting Gibbs-State Expectation Values with Pure Thermal Shadows}''.
\newblock \href{https://dx.doi.org/10.1103/PRXQuantum.4.010305}{PRX Quantum {\bf 4}, 010305}~(2023).

\bibitem{Holmes2022}
Zoe Holmes, Gopikrishnan Muraleedharan, Rolando~D. Somma, Yigit Subasi, and Burak \c{S}ahino\u{g}lu.
\newblock ``{Quantum algorithms from fluctuation theorems: Thermal-state preparation}''.
\newblock \href{https://dx.doi.org/10.22331/q-2022-10-06-825}{Quantum {\bf 6}, 825}~(2022).

\bibitem{Zhang2023}
Daniel Zhang, Jan~Lukas Bosse, and Toby Cubitt.
\newblock ``{Dissipative Quantum Gibbs Sampling}''~(2023).
\newblock  \href{http://arxiv.org/abs/2304.04526}{arXiv:2304.04526}.

\bibitem{schuckert_probing_2023}
Alexander Schuckert, Annabelle Bohrdt, Eleanor Crane, and Michael Knap.
\newblock ``Probing finite-temperature observables in quantum simulators of spin systems with short-time dynamics''.
\newblock \href{https://dx.doi.org/10.1103/PhysRevB.107.L140410}{Phys. Rev. B {\bf 107}, L140410}~(2023).

\bibitem{ghanem_robust_2023}
Khaldoon Ghanem, Alexander Schuckert, and Henrik Dreyer.
\newblock ``Robust {Extraction} of {Thermal} {Observables} from {State} {Sampling} and {Real}-{Time} {Dynamics} on {Quantum} {Computers}''.
\newblock \href{https://dx.doi.org/10.22331/q-2023-11-03-1163}{Quantum {\bf 7}, 1163}~(2023).

\bibitem{hemery_measuring_2024}
Kévin Hémery, Khaldoon Ghanem, Eleanor Crane, Sara~L. Campbell, Joan~M. Dreiling, Caroline Figgatt, Cameron Foltz, John~P. Gaebler, Jacob Johansen, Michael Mills, Steven~A. Moses, Juan~M. Pino, Anthony Ransford, Mary Rowe, Peter Siegfried, Russell~P. Stutz, Henrik Dreyer, Alexander Schuckert, and Ramil Nigmatullin.
\newblock ``Measuring the {Loschmidt} {Amplitude} for {Finite}-{Energy} {Properties} of the {Fermi}-{Hubbard} {Model} on an {Ion}-{Trap} {Quantum} {Computer}''.
\newblock \href{https://dx.doi.org/10.1103/PRXQuantum.5.030323}{PRX Quantum {\bf 5}, 030323}~(2024).

\bibitem{Metropolis1953}
Nicholas Metropolis, Arianna~W. Rosenbluth, Marshall~N. Rosenbluth, Augusta~H. Teller, and Edward Teller.
\newblock ``{Equation of State Calculations by Fast Computing Machines}''.
\newblock \href{https://dx.doi.org/10.1063/1.1699114}{J. Chem. Phys. {\bf 21}, 1087--1092}~(1953).

\bibitem{Hastings1970}
W.~K. Hastings.
\newblock ``{Monte Carlo Sampling Methods Using Markov Chains and Their Applications}''.
\newblock \href{https://dx.doi.org/10.2307/2334940}{Biometrika {\bf 57}, 97--109}~(1970).

\bibitem{Moussa2019}
Jonathan~E. Moussa.
\newblock ``{Low-Depth Quantum Metropolis Algorithm}''~(2019).
\newblock  \href{http://arxiv.org/abs/1903.01451}{arXiv:1903.01451}.

\bibitem{Jiang2024}
Jiaqing Jiang and Sandy Irani.
\newblock ``{Quantum Metropolis Sampling via Weak Measurement}''~(2024).
\newblock  \href{http://arxiv.org/abs/2406.16023}{arXiv:2406.16023}.

\bibitem{Lloyd1996}
Seth Lloyd.
\newblock ``Universal quantum simulators''.
\newblock \href{https://dx.doi.org/10.1126/science.273.5278.1073}{Science {\bf 273}, 1073--1078}~(1996).

\bibitem{Lindblad1975}
Goran Lindblad.
\newblock ``{On the Generators of Quantum Dynamical Semigroups}''.
\newblock \href{https://dx.doi.org/10.1007/BF01608499}{Commun. Math. Phys. {\bf 48}, 119}~(1976).

\bibitem{Kliesch2011}
M.~Kliesch, T.~Barthel, C.~Gogolin, M.~Kastoryano, and J.~Eisert.
\newblock ``Dissipative quantum church-turing theorem''.
\newblock \href{https://dx.doi.org/10.1103/PhysRevLett.107.120501}{Phys. Rev. Lett. {\bf 107}, 120501}~(2011).

\bibitem{Childs2016}
Andrew~M. Childs and Tongyang Li.
\newblock ``{Efficient simulation of sparse Markovian quantum dynamics}''.
\newblock \href{https://dx.doi.org/10.26421/QIC17.11-12-1}{Quant. Inf. Comput. {\bf 17}, 0901--0947}~(2017).

\bibitem{Cleve2019}
Richard Cleve and Chunhao Wang.
\newblock ``Efficient {{Quantum Algorithms}} for {{Simulating Lindblad Evolution}}''~(2019).
\newblock  \href{http://arxiv.org/abs/1612.09512}{arXiv:1612.09512}.

\bibitem{Wocjan2021}
Pawel Wocjan and Kristan Temme.
\newblock ``{Szegedy Walk Unitaries for Quantum Maps}''.
\newblock \href{https://dx.doi.org/10.1007/s00220-023-04797-4}{Commun. Math. Phys. {\bf 402}, 3201--3231}~(2023).

\bibitem{Shtanko2021}
Oles Shtanko and Ramis Movassagh.
\newblock ``{Preparing thermal states on noiseless and noisy programmable quantum processors}''~(2021).
\newblock  \href{http://arxiv.org/abs/2112.14688}{arXiv:2112.14688}.

\bibitem{Rall2023}
Patrick Rall, Chunhao Wang, and Pawel Wocjan.
\newblock ``Thermal {{State Preparation}} via {{Rounding Promises}}''.
\newblock \href{https://dx.doi.org/10.22331/q-2023-10-10-1132}{Quantum {\bf 7}, 1132}~(2023).

\bibitem{Chen2023thermal}
Chi-Fang Chen, Michael~J. Kastoryano, Fernando G. S.~L. Brand\~ao, and Andr\'as Gily\'en.
\newblock ``{Quantum Thermal State Preparation}''~(2023).
\newblock  \href{http://arxiv.org/abs/2303.18224}{arXiv:2303.18224}.

\bibitem{Chen2023efficient}
Chi-Fang Chen, Michael~J. Kastoryano, and Andr\'as Gily\'en.
\newblock ``{An efficient and exact noncommutative quantum Gibbs sampler}''~(2023).
\newblock  \href{http://arxiv.org/abs/2311.09207}{arXiv:2311.09207}.

\bibitem{Ding2023}
Zhiyan Ding, Chi-Fang Chen, and Lin Lin.
\newblock ``Single-ancilla ground state preparation via {{Lindbladians}}''.
\newblock \href{https://dx.doi.org/10.1103/PhysRevResearch.6.033147}{Phys. Rev. Research {\bf 6}, 033147}~(2024).

\bibitem{Ding2024}
Zhiyan Ding, Bowen Li, and Lin Lin.
\newblock ``{Efficient quantum Gibbs samplers with Kubo--Martin--Schwinger detailed balance condition}''.
\newblock \href{https://dx.doi.org/10.1007/s00220-025-05235-3}{Commun. Math. Phys. {\bf 406}, 67}~(2025).
\newblock  \href{http://arxiv.org/abs/2404.05998}{arXiv:2404.05998}.

\bibitem{Ding2024open}
Zhiyan Ding, Xiantao Li, and Lin Lin.
\newblock ``{Simulating Open Quantum Systems Using Hamiltonian Simulations}''.
\newblock \href{https://dx.doi.org/10.1103/PRXQuantum.5.020332}{PRX Quantum {\bf 5}, 020332}~(2024).

\bibitem{Gilyen2024}
Andr\'as Gily\'en, Chi-Fang Chen, Joao~F. Doriguello, and Michael~J. Kastoryano.
\newblock ``{Quantum generalizations of Glauber and Metropolis dynamics}''~(2024).
\newblock  \href{http://arxiv.org/abs/2405.20322}{arXiv:2405.20322}.

\bibitem{Chen2024random}
Hongrui Chen, Bowen Li, Jianfeng Lu, and Lexing Ying.
\newblock ``{A Randomized Method for Simulating Lindblad Equations and Thermal State Preparation}''~(2025).
\newblock  \href{http://arxiv.org/abs/2407.06594v3}{arXiv:2407.06594v3}.

\bibitem{Verstraete2009}
Frank Verstraete, Michael~M. Wolf, and J.~Ignacio~Cirac.
\newblock ``Quantum computation and quantum-state engineering driven by dissipation''.
\newblock \href{https://dx.doi.org/10.1038/nphys1342}{Nat. Phys. {\bf 5}, 633--636}~(2009).

\bibitem{Kraus2008}
B.~Kraus, H.~P. B\"uchler, S.~Diehl, A.~Kantian, A.~Micheli, and P.~Zoller.
\newblock ``Preparation of entangled states by quantum {Markov} processes''.
\newblock \href{https://dx.doi.org/10.1103/PhysRevA.78.042307}{Phys. Rev. A {\bf 78}, 042307}~(2008).

\bibitem{Harrington2022}
Patrick~M. Harrington, Erich~J. Mueller, and Kater~W. Murch.
\newblock ``Engineered dissipation for quantum information science''.
\newblock \href{https://dx.doi.org/10.1038/s42254-022-00494-8}{Nat. Rev. Phys. {\bf 4}, 660--671}~(2022).

\bibitem{Bravyi2021}
Sergey Bravyi, Anirban Chowdhury, David Gosset, and Pawel Wocjan.
\newblock ``{Quantum Hamiltonian complexity in thermal equilibrium}''.
\newblock \href{https://dx.doi.org/10.1038/s41567-022-01742-5}{Nat. Phys. {\bf 18}, 1367--1370}~(2022).

\bibitem{Rouze2024}
Cambyse Rouz\'e, Daniel Stilck~Fran\c{c}a, and \'Alvaro~M. Alhambra.
\newblock ``{Efficient thermalization and universal quantum computing with quantum Gibbs samplers}''~(2024).
\newblock  \href{http://arxiv.org/abs/2403.12691}{arXiv:2403.12691}.

\bibitem{Bergamaschi2024}
Thiago Bergamaschi, Chi-Fang Chen, and Yunchao Liu.
\newblock ``Quantum {Computational} {Advantage} with {Constant}-{Temperature} {Gibbs} {Sampling}''.
\newblock In 2024 {IEEE} 65th {Annual} {Symposium} on {Foundations} of {Computer} {Science} ({FOCS}).
\newblock \href{https://dx.doi.org/10.1109/FOCS61266.2024.00071}{Pages 1063--1085}.
\newblock ~(2024).

\bibitem{Rajakumar2024}
Joel Rajakumar and James~D. Watson.
\newblock ``{Gibbs Sampling gives Quantum Advantage at Constant Temperatures with $O(1)$-Local Hamiltonians}''~(2024).
\newblock  \href{http://arxiv.org/abs/2408.01516}{arXiv:2408.01516}.

\bibitem{Srednicki1994}
Mark Srednicki.
\newblock ``{Chaos and Quantum Thermalization}''.
\newblock \href{https://dx.doi.org/10.1103/PhysRevE.50.888}{Phys. Rev. E {\bf 50}, 888}~(1994).

\bibitem{Srednicki1995}
Mark Srednicki.
\newblock ``{Thermal fluctuations in quantized chaotic systems}''.
\newblock \href{https://dx.doi.org/10.1088/0305-4470/29/4/003}{J. Phys. A: Math. Gen. {\bf 29}, L75--L79}~(1996).

\bibitem{Srednicki1999}
Mark Srednicki.
\newblock ``The approach to thermal equilibrium in quantized chaotic systems''.
\newblock \href{https://dx.doi.org/10.1088/0305-4470/32/7/007}{J. Phys. A: Math. Gen. {\bf 32}, 1163}~(1999).

\bibitem{DAlessio2015}
Luca D'Alessio, Yariv Kafri, Anatoli Polkovnikov, and Marcos Rigol.
\newblock ``{From quantum chaos and eigenstate thermalization to statistical mechanics and thermodynamics}''.
\newblock \href{https://dx.doi.org/10.1080/00018732.2016.1198134}{Adv. Phys. {\bf 65}, 239--362}~(2016).

\bibitem{spohnIrreversibleThermodynamicsQuantum1978}
Herbert Spohn and Joel~L. Lebowitz.
\newblock ``Irreversible {{Thermodynamics}} for {{Quantum Systems Weakly Coupled}} to {{Thermal Reservoirs}}''.
\newblock In Advances in {{Chemical Physics}}.
\newblock \href{https://dx.doi.org/10.1002/9780470142578.ch2}{Volume~38, pages 109--142}.
\newblock John Wiley \& Sons, Ltd~(1978).

\bibitem{Alicki1976}
Robert Alicki.
\newblock ``On the detailed balance condition for non-hamiltonian systems''.
\newblock \href{https://dx.doi.org/https://doi.org/10.1016/0034-4877(76)90046-X}{Rept. Math. Phys. {\bf 10}, 249--258}~(1976).

\bibitem{Kossakowski1977}
Andrzej Kossakowski, Alberto Frigerio, Vittorio Gorini, and Maurizio Verri.
\newblock ``{Quantum detailed balance and KMS condition}''.
\newblock \href{https://dx.doi.org/10.1007/BF01625769}{Commun. Math. Phys. {\bf 57}, 97--110}~(1977).

\bibitem{Carlen2017}
Eric~A. Carlen and Jan Maas.
\newblock ``Gradient flow and entropy inequalities for quantum markov semigroups with detailed balance''.
\newblock \href{https://dx.doi.org/https://doi.org/10.1016/j.jfa.2017.05.003}{Journal of Functional Analysis {\bf 273}, 1810--1869}~(2017).

\bibitem{Fagnola2007}
Franco Fagnola and Veronica Umanit\`{a}.
\newblock ``{Generators of Detailed Balance Quantum Markov Semigroups}''.
\newblock \href{https://dx.doi.org/10.1142/S0219025707002762}{Infin. Dimens. Anal. Quantum Probab. Relat. Top. {\bf 10}, 335--363}~(2007).

\bibitem{Fagnola2010}
Franco Fagnola and Veronica Umanit{\`a}.
\newblock ``{Generators of KMS Symmetric Markov Semigroups on $\mathcal{B}({\rm h)}$ Symmetry and Quantum Detailed Balance}''.
\newblock \href{https://dx.doi.org/10.1007/s00220-010-1011-1}{Commun. Math. Phys. {\bf 298}, 523--547}~(2010).

\bibitem{Temme2010}
K.~Temme, M.~J. Kastoryano, M.~B. Ruskai, M.~M. Wolf, and F.~Verstraete.
\newblock ``{The $\chi^2$-divergence and Mixing times of quantum Markov processes}''.
\newblock \href{https://dx.doi.org/10.1063/1.3511335}{J. Math. Phys. {\bf 51}, 122201}~(2010).

\bibitem{wolfQuantumChannelsOperations2012}
Michael~M Wolf.
\newblock ``Quantum {{Channels}} \& {{Operations Guided Tour}}''.
\newblock Lecture notes.
\newblock Niels-Bohr Institute, Copenhagen~(2012).
\newblock  url:~\url{https://mediatum.ub.tum.de/doc/1701036/document.pdf}.

\bibitem{Pocrnic2023}
Matthew Pocrnic, Dvira Segal, and Nathan Wiebe.
\newblock ``{Quantum simulation of Lindbladian dynamics via repeated interactions}''.
\newblock \href{https://dx.doi.org/10.1088/1751-8121/adebc4}{J. Phys. A: Math. Theor. {\bf 58}, 305302}~(2025).
\newblock  \href{http://arxiv.org/abs/2312.05371}{arXiv:2312.05371}.

\bibitem{ramkumarMixingTimeQuantum2024}
Akshar Ramkumar and Mehdi Soleimanifar.
\newblock ``Mixing time of quantum {{Gibbs}} sampling for random sparse {{Hamiltonians}}''~(2024).
\newblock  \href{http://arxiv.org/abs/2411.04454}{arXiv:2411.04454}.

\bibitem{Chen2021ETH}
Chi-Fang Chen and Fernando G. S.~L. Brand\~ao.
\newblock ``{Fast Thermalization from the Eigenstate Thermalization Hypothesis}''~(2021).
\newblock  \href{http://arxiv.org/abs/2112.07646}{arXiv:2112.07646}.

\bibitem{rouzeOptimalQuantumAlgorithm2024}
Cambyse Rouz{\'e}, Daniel~Stilck Fran{\c c}a, and {\'A}lvaro~M. Alhambra.
\newblock ``Optimal quantum algorithm for {{Gibbs}} state preparation''~(2024).
\newblock  \href{http://arxiv.org/abs/2411.04885}{arXiv:2411.04885}.

\bibitem{znidaric_relaxation_2015}
Marko Znidaric.
\newblock ``Relaxation times of dissipative many-body quantum systems''.
\newblock \href{https://dx.doi.org/10.1103/PhysRevE.92.042143}{Phys. Rev. E {\bf 92}, 042143}~(2015).

\bibitem{smid_polynomial_2025}
{\v S}t{\v e}p{\'a}n {\v S}m{\'i}d, Richard Meister, Mario Berta, and Roberto Bondesan.
\newblock ``Polynomial {Time} {Quantum} {Gibbs} {Sampling} for {Fermi}-{Hubbard} {Model} at any {Temperature}''~(2025) \href{http://arxiv.org/abs/2501.01412}{arXiv:2501.01412}.

\bibitem{hagan_thermodynamic_2025}
Matthew Hagan and Nathan Wiebe.
\newblock ``The {Thermodynamic} {Cost} of {Ignorance}: {Thermal} {State} {Preparation} with {One} {Ancilla} {Qubit}''~(2025) \href{http://arxiv.org/abs/2502.03410}{arXiv:2502.03410}.

\bibitem{Li2024}
Hao-En Li, Yongtao Zhan, and Lin Lin.
\newblock ``{Dissipative ground state preparation in ab initio electronic structure theory}''~(2024).
\newblock  \href{http://arxiv.org/abs/2411.01470}{arXiv:2411.01470}.

\bibitem{plenio_quantum-jump_1998}
M.~B. Plenio and P.~L. Knight.
\newblock ``The quantum-jump approach to dissipative dynamics in quantum optics''.
\newblock \href{https://dx.doi.org/10.1103/RevModPhys.70.101}{Rev. Mod. Phys. {\bf 70}, 101--144}~(1998).

\bibitem{znidaric_dephasing-induced_2010}
Marko Znidaric.
\newblock ``Dephasing-induced diffusive transport in the anisotropic {Heisenberg} model''.
\newblock \href{https://dx.doi.org/10.1088/1367-2630/12/4/043001}{New J. Phys. {\bf 12}, 043001}~(2010).

\bibitem{finazzi_corner-space_2015}
S.~Finazzi, A.~Le~Boité, F.~Storme, A.~Baksic, and C.~Ciuti.
\newblock ``Corner-{Space} {Renormalization} {Method} for {Driven}-{Dissipative} {Two}-{Dimensional} {Correlated} {Systems}''.
\newblock \href{https://dx.doi.org/10.1103/PhysRevLett.115.080604}{Phys. Rev. Lett. {\bf 115}, 080604}~(2015).

\bibitem{gravina_adaptive_2024}
Luca Gravina and Vincenzo Savona.
\newblock ``Adaptive variational low-rank dynamics for open quantum systems''.
\newblock \href{https://dx.doi.org/10.1103/PhysRevResearch.6.023072}{Phys. Rev. Research {\bf 6}, 023072}~(2024).

\bibitem{weimer_simulation_2021}
Hendrik Weimer, Augustine Kshetrimayum, and Román Orús.
\newblock ``Simulation methods for open quantum many-body systems''.
\newblock \href{https://dx.doi.org/10.1103/RevModPhys.93.015008}{Rev. Mod. Phys. {\bf 93}, 015008}~(2021).

\bibitem{sander_large-scale_2025}
Aaron Sander, Maximilian Fröhlich, Martin Eigel, Jens Eisert, Patrick Gelß, Michael Hintermüller, Richard~M. Milbradt, Robert Wille, and Christian~B. Mendl.
\newblock ``Large-scale stochastic simulation of open quantum systems''~(2025) \href{http://arxiv.org/abs/2501.17913}{arXiv:2501.17913}.

\bibitem{zhan2025rapid}
Yongtao Zhan, Zhiyan Ding, Jakob Huhn, Johnnie Gray, John Preskill, Garnet Kin-Lic Chan, and Lin Lin.
\newblock ``Rapid quantum ground state preparation via dissipative dynamics''~(2025) \href{http://arxiv.org/abs/2503.15827}{arXiv:2503.15827}.

\bibitem{kolovsky_quantum_2004}
A.~R. Kolovsky and A.~Buchleitner.
\newblock ``Quantum chaos in the {Bose}-{Hubbard} model''.
\newblock \href{https://dx.doi.org/10.1209/epl/i2004-10265-7}{EPL {\bf 68}, 632}~(2004).

\bibitem{atas_multifractality_2012}
Y.~Y. Atas and E.~Bogomolny.
\newblock ``Multifractality of eigenfunctions in spin chains''.
\newblock \href{https://dx.doi.org/10.1103/PhysRevE.86.021104}{Phys. Rev. E {\bf 86}, 021104}~(2012).

\bibitem{pausch_chaos_2021}
Lukas Pausch, Edoardo~G. Carnio, Alberto Rodríguez, and Andreas Buchleitner.
\newblock ``Chaos and {Ergodicity} across the {Energy} {Spectrum} of {Interacting} {Bosons}''.
\newblock \href{https://dx.doi.org/10.1103/PhysRevLett.126.150601}{Phys. Rev. Lett. {\bf 126}, 150601}~(2021).

\bibitem{brunner_many-body_2023}
Eric Brunner, Lukas Pausch, Edoardo~G. Carnio, Gabriel Dufour, Alberto Rodríguez, and Andreas Buchleitner.
\newblock ``Many-{Body} {Interference} at the {Onset} of {Chaos}''.
\newblock \href{https://dx.doi.org/10.1103/PhysRevLett.130.080401}{Phys. Rev. Lett. {\bf 130}, 080401}~(2023).

\bibitem{Raghunandan2020}
Meghana Raghunandan, Fabian Wolf, Christian Ospelkaus, Piet~O. Schmidt, and Hendrik Weimer.
\newblock ``Initialization of quantum simulators by sympathetic cooling''.
\newblock \href{https://dx.doi.org/10.1126/sciadv.aaw9268}{Sci. Adv. {\bf 6}, eaaw9268}~(2020).

\bibitem{Polla2021}
Stefano Polla, Yaroslav Herasymenko, and Thomas~E. O'Brien.
\newblock ``Quantum digital cooling''.
\newblock \href{https://dx.doi.org/10.1103/PhysRevA.104.012414}{Phys. Rev. A {\bf 104}, 012414}~(2021).

\bibitem{Mi2023}
X.~Mi et~al.
\newblock ``{Stable quantum-correlated many-body states through engineered dissipation}''.
\newblock \href{https://dx.doi.org/10.1126/science.adh9932}{Science {\bf 383}, adh9932}~(2024).

\bibitem{Cubitt2023}
Toby~S. Cubitt.
\newblock ``{Dissipative ground state preparation and the Dissipative Quantum Eigensolver}''~(2023).
\newblock  \href{http://arxiv.org/abs/2303.11962}{arXiv:2303.11962}.

\bibitem{Granet2024}
Etienne Granet and Henrik Dreyer.
\newblock ``{A noise-limiting quantum algorithm using mid-circuit measurements for dynamical correlations at infinite temperature}''~(2024).
\newblock  \href{http://arxiv.org/abs/2401.02207}{arXiv:2401.02207}.

\bibitem{Kashyap2024}
Vikram Kashyap, Georgios Styliaris, Sara Mouradian, J.~Ignacio Cirac, and Rahul Trivedi.
\newblock ``Accuracy {Guarantees} and {Quantum} {Advantage} in {Analog} {Open} {Quantum} {Simulation} with and without {Noise}''.
\newblock \href{https://dx.doi.org/10.1103/PhysRevX.15.021017}{Phys. Rev. X {\bf 15}, 021017}~(2025).

\bibitem{Kastoryano2013}
Michael~J. Kastoryano and Kristan Temme.
\newblock ``{Quantum logarithmic Sobolev inequalities and rapid mixing}''.
\newblock \href{https://dx.doi.org/10.1063/1.4804995}{J. Math. Phys. {\bf 54}, 052202}~(2013).

\bibitem{Campbell2019}
Earl Campbell.
\newblock ``Random compiler for fast hamiltonian simulation''.
\newblock \href{https://dx.doi.org/10.1103/PhysRevLett.123.070503}{Phys. Rev. Lett. {\bf 123}, 070503}~(2019).

\bibitem{oganesyan_localization_2007}
Vadim Oganesyan and David~A. Huse.
\newblock ``Localization of interacting fermions at high temperature''.
\newblock \href{https://dx.doi.org/10.1103/PhysRevB.75.155111}{Phys. Rev. B {\bf 75}, 155111}~(2007).

\bibitem{atas_distribution_2013}
Y.~Y. Atas, E.~Bogomolny, O.~Giraud, and G.~Roux.
\newblock ``Distribution of the {Ratio} of {Consecutive} {Level} {Spacings} in {Random} {Matrix} {Ensembles}''.
\newblock \href{https://dx.doi.org/10.1103/PhysRevLett.110.084101}{Phys. Rev. Lett. {\bf 110}, 084101}~(2013).

\bibitem{wallman_noise_2016}
Joel~J. Wallman and Joseph Emerson.
\newblock ``Noise tailoring for scalable quantum computation via randomized compiling''.
\newblock \href{https://dx.doi.org/10.1103/PhysRevA.94.052325}{Phys. Rev. A {\bf 94}, 052325}~(2016).

\bibitem{qujax}
Samuel Duffield, Gabriel Matos, and Melf Johannsen.
\newblock ``{qujax: Simulating quantum circuits with JAX}''.
\newblock \href{https://dx.doi.org/10.21105/joss.05504}{J. Open Source Softw. {\bf 8}, 5504}~(2023).

\bibitem{jax2018github}
James Bradbury, Roy Frostig, Peter Hawkins, Matthew~James Johnson, Chris Leary, Dougal Maclaurin, George Necula, Adam Paszke, Jake Vander{P}las, Skye Wanderman-{M}ilne, and Qiao Zhang.
\newblock ``{JAX}: composable transformations of {P}ython+{N}um{P}y programs''~(2018).
\newblock Version 0.3.13. \url{http://github.com/jax-ml/jax} (accessed 2024-12-21).

\bibitem{Sivarajah_2021}
Seyon Sivarajah, Silas Dilkes, Alexander Cowtan, Will Simmons, Alec Edgington, and Ross Duncan.
\newblock ``{t$|$ket$\rangle$}: a retargetable compiler for {NISQ} devices''.
\newblock \href{https://dx.doi.org/10.1088/2058-9565/ab8e92}{Quantum Sci. Technol. {\bf 6}, 014003}~(2020).

\bibitem{Childs2001}
Andrew~M. Childs, Edward Farhi, and John Preskill.
\newblock ``Robustness of adiabatic quantum computation''.
\newblock \href{https://dx.doi.org/10.1103/PhysRevA.65.012322}{Phys. Rev. A {\bf 65}, 012322}~(2001).

\bibitem{Roland2005}
J\'er\'emie Roland and Nicolas~J. Cerf.
\newblock ``Noise resistance of adiabatic quantum computation using random matrix theory''.
\newblock \href{https://dx.doi.org/10.1103/PhysRevA.71.032330}{Phys. Rev. A {\bf 71}, 032330}~(2005).

\bibitem{Trivedi2022}
Rahul Trivedi, Adrian~Franco Rubio, and J.~Ignacio Cirac.
\newblock ``{Quantum advantage and stability to errors in analogue quantum simulators}''.
\newblock \href{https://dx.doi.org/10.1038/s41467-024-50750-x}{Nature Commun. {\bf 15}, 6507}~(2024).

\bibitem{Kechedzhi2023}
K.~Kechedzhi, S.~V. Isakov, S.~Mandr\`a, B.~Villalonga, X.~Mi, S.~Boixo, and V.~Smelyanskiy.
\newblock ``{Effective quantum volume, fidelity and computational cost of noisy quantum processing experiments}''.
\newblock \href{https://dx.doi.org/10.1016/j.future.2023.12.002}{Future Gener. Comput. Syst. {\bf 153}, 431--441}~(2024).

\bibitem{Schiffer2024}
Benjamin~F. Schiffer, Adrian~Franco Rubio, Rahul Trivedi, and J.~Ignacio Cirac.
\newblock ``{The quantum adiabatic algorithm suppresses the proliferation of errors}''~(2024).
\newblock  \href{http://arxiv.org/abs/2404.15397}{arXiv:2404.15397}.

\bibitem{Granet2024dilution}
Etienne Granet and Henrik Dreyer.
\newblock ``Dilution of {Error} in {Digital} {Hamiltonian} {Simulation}''.
\newblock \href{https://dx.doi.org/10.1103/PRXQuantum.6.010333}{PRX Quantum {\bf 6}, 010333}~(2025).

\bibitem{Chertkov2024}
Eli Chertkov, Yi-Hsiang Chen, Michael Lubasch, David Hayes, and Michael Foss-Feig.
\newblock ``{Robustness of near-thermal dynamics on digital quantum computers}''~(2024).
\newblock  \href{http://arxiv.org/abs/2410.10794}{arXiv:2410.10794}.

\bibitem{hashim_randomized_2021}
Akel Hashim.
\newblock ``Randomized {Compiling} for {Scalable} {Quantum} {Computing} on a {Noisy} {Superconducting} {Quantum} {Processor}''.
\newblock \href{https://dx.doi.org/10.1103/PhysRevX.11.041039}{Phys. Rev. X {\bf 11}, 041039}~(2021).

\bibitem{bakshi_learning_2023}
Ainesh Bakshi, Allen Liu, Ankur Moitra, and Ewin Tang.
\newblock ``Learning quantum {Hamiltonians} at any temperature in polynomial time''~(2023) \href{http://arxiv.org/abs/2310.02243}{arXiv:2310.02243}.

\bibitem{Murthy2019}
Chaitanya Murthy and Mark Srednicki.
\newblock ``Bounds on chaos from the eigenstate thermalization hypothesis''.
\newblock \href{https://dx.doi.org/10.1103/PhysRevLett.123.230606}{Phys. Rev. Lett. {\bf 123}, 230606}~(2019).

\bibitem{Davies1974}
E.~B. Davies.
\newblock ``{Markovian master equations}''.
\newblock \href{https://dx.doi.org/10.1007/BF01608389}{Commun. Math. Phys. {\bf 39}, 91--110}~(1974).

\bibitem{Davies1976}
E.~B. Davies.
\newblock ``{Markovian master equations. II}''.
\newblock \href{https://dx.doi.org/10.1007/BF01351898}{Math. Ann. {\bf 219}, 147--158}~(1976).

\bibitem{Dymarsky2019}
Anatoly Dymarsky and Hong Liu.
\newblock ``New characteristic of quantum many-body chaotic systems''.
\newblock \href{https://dx.doi.org/10.1103/PhysRevE.99.010102}{Phys. Rev. E {\bf 99}, 010102}~(2019).

\bibitem{Banuls2020}
Mari~Carmen Ba\~nuls, David~A. Huse, and J.~Ignacio Cirac.
\newblock ``Entanglement and its relation to energy variance for local one-dimensional hamiltonians''.
\newblock \href{https://dx.doi.org/10.1103/PhysRevB.101.144305}{Phys. Rev. B {\bf 101}, 144305}~(2020).

\bibitem{Lu2021}
Sirui Lu, Mari~Carmen Ba\~nuls, and J.~Ignacio Cirac.
\newblock ``Algorithms for quantum simulation at finite energies''.
\newblock \href{https://dx.doi.org/10.1103/PRXQuantum.2.020321}{PRX Quantum {\bf 2}, 020321}~(2021).

\bibitem{beugeling_global_2015}
W.~Beugeling, A.~Andreanov, and Masudul Haque.
\newblock ``Global characteristics of all eigenstates of local many-body {Hamiltonians}: participation ratio and entanglement entropy''.
\newblock \href{https://dx.doi.org/10.1088/1742-5468/2015/02/P02002}{J. Stat. Mech. {\bf 2015}, P02002}~(2015).

\bibitem{atas_quantum_2017}
Y.~Y. Atas and E.~Bogomolny.
\newblock ``Quantum {Ising} model in transverse and longitudinal fields: chaotic wave functions''.
\newblock \href{https://dx.doi.org/10.1088/1751-8121/aa81f6}{J. Phys. A: Math. Theor. {\bf 50}, 385102}~(2017).

\bibitem{beugeling_statistical_2018}
Wouter Beugeling, Arnd Bäcker, Roderich Moessner, and Masudul Haque.
\newblock ``Statistical properties of eigenstate amplitudes in complex quantum systems''.
\newblock \href{https://dx.doi.org/10.1103/PhysRevE.98.022204}{Phys. Rev. E {\bf 98}, 022204}~(2018).

\bibitem{kim_ballistic_2013}
Hyungwon Kim and David~A. Huse.
\newblock ``Ballistic {Spreading} of {Entanglement} in a {Diffusive} {Nonintegrable} {System}''.
\newblock \href{https://dx.doi.org/10.1103/PhysRevLett.111.127205}{Phys. Rev. Lett. {\bf 111}, 127205}~(2013).

\bibitem{fang_mixing_2024}
Di~Fang, Jianfeng Lu, and Yu~Tong.
\newblock ``Mixing {Time} of {Open} {Quantum} {Systems} via {Hypocoercivity}''.
\newblock \href{https://dx.doi.org/10.1103/PhysRevLett.134.140405}{Phys. Rev. Lett. {\bf 134}, 140405}~(2025).
\newblock  \href{http://arxiv.org/abs/2404.11503}{arXiv:2404.11503}.

\bibitem{dalibard_wave-function_1992}
Jean Dalibard, Yvan Castin, and Klaus Mølmer.
\newblock ``Wave-function approach to dissipative processes in quantum optics''.
\newblock \href{https://dx.doi.org/10.1103/PhysRevLett.68.580}{Phys. Rev. Lett. {\bf 68}, 580--583}~(1992).

\bibitem{dum_monte_1992}
R.~Dum, P.~Zoller, and H.~Ritsch.
\newblock ``Monte {Carlo} simulation of the atomic master equation for spontaneous emission''.
\newblock \href{https://dx.doi.org/10.1103/PhysRevA.45.4879}{Phys. Rev. A {\bf 45}, 4879--4887}~(1992).

\bibitem{molmer_monte_1993}
Klaus Mølmer, Yvan Castin, and Jean Dalibard.
\newblock ``Monte {Carlo} wave-function method in quantum optics''.
\newblock \href{https://dx.doi.org/10.1364/JOSAB.10.000524}{J. Opt. Soc. Am. B {\bf 10}, 524--538}~(1993).

\bibitem{johansson_qutip_2013}
J.~R. Johansson, P.~D. Nation, and Franco Nori.
\newblock ``{QuTiP} 2: {A} {Python} framework for the dynamics of open quantum systems''.
\newblock \href{https://dx.doi.org/10.1016/j.cpc.2012.11.019}{Comput. Phys. Commun. {\bf 184}, 1234--1240}~(2013).

\bibitem{lambert_qutip_2024}
Neill Lambert, Eric Giguère, Paul Menczel, Boxi Li, Patrick Hopf, Gerardo Suárez, Marc Gali, Jake Lishman, Rushiraj Gadhvi, Rochisha Agarwal, Asier Galicia, Nathan Shammah, Paul Nation, J.~R. Johansson, Shahnawaz Ahmed, Simon Cross, Alexander Pitchford, and Franco Nori.
\newblock ``{QuTiP} 5: {The} {Quantum} {Toolbox} in {Python}''~(2024) \href{http://arxiv.org/abs/2412.04705}{arXiv:2412.04705}.

\end{thebibliography}

\end{document}